\documentclass[%
  aps,%
  pra,%
  superscriptaddress,%
  reprint,%
  longbibliography,%
  nofootinbib,%
  notitlepage]{revtex4-2}

\usepackage[preset=reset,pkgset=extended,hyperrefdefs={defer,noemail}]{phfnote}

\usepackage[normalem]{ulem}

\usepackage[thmset=rich,smallproofs=false,proofref={marginbottom}]{phfthm}
\renewcommand\proofonname[2]{\textit{(Proof on #2.)}}

\usepackage{phflplx}
\DeclareGraphicsExtensions{.lplx,.pdf,.png,.jpeg,.jpg}

\makeatletter
\RequirePackage{dsfont}
\RequirePackage{mathrsfs}
\RequirePackage{mathtools}

\usepackage[sfdefault]{sourcesanspro}

\usepackage[reset]{geometry}
\Gm@restore@org

\DeclareMathOperator\Var{Var}
\def\leq{\leqslant}
\def\geq{\geqslant}

\DeclarePairedDelimiterXPP\opnorm[1]{}{\lVert}{\rVert}{}{{#1}}
\DeclarePairedDelimiterXPP\onenorm[1]{}{\lVert}{\rVert}{_1}{{#1}}

\DeclarePairedDelimiterXPP\neighbors[1]{}{\langle}{\rangle}{}{{#1}}

\DeclarePairedDelimiterXPP\Ftwo@nodefault[2]{F}{(}{)}{}{{%
    {#1}\mathclose{}\,{;}\;\mathopen{}{#2}%
}}
\newcommand\Ftwo@withdefaultsize[1][\big]{\Ftwo@nodefault[#1]}
\def\Ftwo{\@ifstar{\Ftwo@nodefault*}{\Ftwo@withdefaultsize}}
\robustify\Ftwo
\def\FtwoBig{\Ftwo@nodefault[\Big]}

\usepackage{ragged2e}
\DeclareCaptionJustification{myjust}{\justify}
\renewenvironment{acknowledgments}{%
  \paragraph*{Acknowledgments.}%
}{%
  \par
}

\DeclarePairedDelimiterX\groupgen[1]{\langle}{\rangle}{{#1}}

\DeclareMathOperator\wgt{wgt}

\newcommand\FIqty[3][]{F#1_{\mathrm{#2},\mkern 2mu\relax{#3}}}

\newcommand\FIqtyp[3][]{F^{\mathrm{#1}}_{{#2},\mkern 2mu\relax{#3}}}

\newcommand\FIloss[2]{\Delta F_{\mathrm{#1},\mkern 2mu\relax{#2}}}
\newcommand\FIlossp[3][]{\Delta F^{\mathrm{#1}}_{{#2},\mkern 2mu\relax{#3}}}

\def\@Alphnoi#1{\ifcase #1\or A\or B\or C\or D\or E\or F\or G\or H\or J\or K\or
  L\or M\or N\or P\or Q\or R\or S\or T\or U\or V\or W\or X\or Y\or Z\else
  \@ctrerr\fi\relax}
\def\Alphnoi#1{\expandafter\@Alphnoi\csname c@#1\endcsname}

\pretocmd\appendix{%
  \clearpage
  \newgeometry{hmargin=1.25in,vmargin=1in,marginparsep=12pt,marginparwidth=45pt}%
  \onecolumngrid
  \setstretch{1.12}%
}{}{***PATCH FAIL***}
\apptocmd\appendix{%
}{}{***PATCH FAIL***}

\usepackage{hyperref}
\usepackage[nameinlink,capitalize]{cleveref}

\renewcommand\phfthmproofref@placephantommarginpar{%
  \begingroup\def\@captype{figure}\marginpar{%
    \vspace{-\baselineskip}\rule{0pt}{0pt}}\endgroup%
}
\renewcommand\phfthmproofref@placemarginpar[2]{%
  \begingroup\def\@captype{table}\marginpar[#1]{#2}\endgroup%
}

\crefname{section}{}{}
\Crefname{section}{Section}{Sections}
\creflabelformat{section}{#2\S\,#1#3}
\crefrangelabelformat{section}{#3\S\,#1#4--#5#2#6}
\crefname{subsection}{}{}
\Crefname{subsection}{Section}{Sections}
\creflabelformat{subsection}{#2\S\,#1#3}
\crefrangelabelformat{subsection}{#3\S\,#1#4--#5#2#6}
\crefname{subsubsection}{}{}
\Crefname{subsubsection}{Section}{Sections}
\creflabelformat{subsubsection}{#2\S\,#1#3}
\crefrangelabelformat{subsubsection}{#3\S\,#1#4--#5#2#6}

\makeatother

\DeclareFontFamily{OMX}{MnSymbolE}{}
\DeclareSymbolFont{MnLargeSymbols}{OMX}{MnSymbolE}{m}{n}
\SetSymbolFont{MnLargeSymbols}{bold}{OMX}{MnSymbolE}{b}{n}
\DeclareFontShape{OMX}{MnSymbolE}{m}{n}{
    <-6>  MnSymbolE5
   <6-7>  MnSymbolE6
   <7-8>  MnSymbolE7
   <8-9>  MnSymbolE8
   <9-10> MnSymbolE9
  <10-12> MnSymbolE10
  <12->   MnSymbolE12
}{}
\DeclareFontShape{OMX}{MnSymbolE}{b}{n}{
    <-6>  MnSymbolE-Bold5
   <6-7>  MnSymbolE-Bold6
   <7-8>  MnSymbolE-Bold7
   <8-9>  MnSymbolE-Bold8
   <9-10> MnSymbolE-Bold9
  <10-12> MnSymbolE-Bold10
  <12->   MnSymbolE-Bold12
}{}
\let\llangle\relax
\let\rrangle\relax
\DeclareMathDelimiter{\llangle}{\mathopen}%
                     {MnLargeSymbols}{'164}{MnLargeSymbols}{'164}
\DeclareMathDelimiter{\rrangle}{\mathclose}%
                     {MnLargeSymbols}{'171}{MnLargeSymbols}{'171}
\begin{document}
\title{%
Time-energy uncertainty relation for noisy quantum metrology
}

\author{Philippe Faist}
\affiliation{Dahlem Center for Complex Quantum Systems, Freie Universit\"at Berlin,
  Berlin, Germany}
\affiliation{Institute for Quantum Information and Matter, Caltech, Pasadena, CA, USA}
\affiliation{Institute for Theoretical Physics, ETH Zurich, Zurich, Switzerland}

\author{Mischa P. Woods}
\affiliation{Institute for Theoretical Physics, ETH Zurich, Zurich, Switzerland}
\affiliation{University Grenoble Alpes, Inria, 38000 Grenoble, France}

\author{Victor V. Albert}
\affiliation{Joint Center for Quantum Information and Computer Science, NIST and University of Maryland, College Park, MD, USA}
\affiliation{Institute for Quantum Information and Matter, Caltech, Pasadena, CA, USA}
\affiliation{Walter Burke Institute for Theoretical Physics, Caltech, Pasadena, CA, USA}

\author{Joseph M. Renes}
\affiliation{Institute for Theoretical Physics, ETH Zurich, Zurich, Switzerland}

\author{Jens Eisert}
\affiliation{Dahlem Center for Complex Quantum Systems, Freie Universit\"at Berlin,
  Berlin, Germany}
  
\author{John Preskill}
\affiliation{Institute for Quantum Information and Matter, Caltech, Pasadena, CA, USA}
\affiliation{Walter Burke Institute for Theoretical Physics, Caltech, Pasadena, CA, USA}
\affiliation{AWS Center for Quantum Computing, Caltech, Pasadena, CA, USA}

\date{Feb~2, 2024}

\begin{abstract}
  Detection of very weak forces and precise measurement of time are two of the
  many applications of quantum metrology to science and technology. To sense an
  unknown physical parameter, one prepares an initial state of a probe system,
  allows the probe to evolve as governed by a Hamiltonian $H$ for some time $t$,
  and then measures the probe. If $H$ is known, we can estimate $t$ by this
  method; if $t$ is known, we can estimate classical parameters on which $H$
  depends. The accuracy of a quantum sensor can be limited by either intrinsic
  quantum noise or by noise arising from the interactions of the probe with its
  environment. In this work, we introduce and study a fundamental trade-off
  which relates the amount by which noise reduces the accuracy of a quantum
  clock to the amount of information about the energy of the clock that leaks to
  the environment. Specifically, we consider an idealized scenario in which a
  party Alice prepares an initial pure state of the clock, allows the clock to
  evolve for a time that is not precisely known, and then {transmits} the clock
  through a noisy channel to {a party} Bob. Meanwhile, the environment (Eve)
  receives any information about the clock that is lost during transmission.  We
  prove that Bob's loss of quantum Fisher information about the elapsed time is
  equal to Eve's gain of quantum Fisher information about a complementary energy
  parameter. We also prove a similar, but more general, trade-off that applies
  when Bob and Eve wish to estimate the values of parameters associated with two
  noncommuting observables. We derive the necessary and sufficient conditions
  for the accuracy of the clock to be unaffected by the noise, which form a
  subset of the Knill-Laflamme error-correction conditions. A state and its
    local time-evolution direction, if they satisfy these conditions, are said
    to form a metrological code. We provide a scheme to construct
    metrological codes in the stabilizer formalism.
    We show that there are
    metrological codes that cannot be written as a quantum error-correcting code
    with similar distance in which the Hamiltonian acts as a logical operator,
    potentially offering new schemes for constructing states that do not lose
    any sensitivity upon application of a noisy channel.  We discuss
  applications of the trade-off relation to sensing using a quantum many-body
  probe subject to erasure or amplitude-damping noise.
\end{abstract}
\maketitle

\section{Introduction}
\label{z:n25nnlPW3OKO}

Quantum mechanics places fundamental limits on how well we can measure a physical quantity when using a quantum system as a probe%
~\cite{R0}. \emph{Quantum metrology} is an active research area addressing how physical quantities can be estimated based on observations of a probe system~\cite{R1,R2,R3}. As methods for {accurately} controlling quantum systems steadily advance, increasingly sophisticated measurement strategies
{are becoming}
feasible~\cite{R4,R5}, leading for example to more sensitive gravitational wave detectors~\cite{R6}, 
improved frequency standards~\cite{R7},
and ultraprecise quantum clocks~\cite{R8}. These
 technological developments accentuate the need for a precise theoretical understanding of the potential of quantum metrology and of the ultimate limits on measurement accuracy.

Fundamental accuracy limits in quantum metrology can
often be phrased in terms of
uncertainty relations, 
wherein the accuracy of one physical quantity trades off against the accuracy of a complementary quantity. For example, a particle with a definite position has a highly uncertain momentum, and vice versa. Such trade-offs may be captured conveniently by entropic uncertainty relations~\cite{R9,R10}. One may envision a two-party scenario, where the entropic uncertainty relation connects the first party's ignorance about a quantity $A$ with the second party's 
lack of knowledge
about a complementary quantity $B$. Typically these quantities are values of noncommuting observables. 

In this work, we 
focus on a related but fundamentally different type of uncertainty relation.
Rather than a trade-off between the values of two observables, we consider an information-theoretic trade-off between time and energy.  Specifically, we envision preparing a probe state $\rho_{\mathrm{init}}$, which then evolves for a time $t$ as determined by some Hamiltonian $H$. By measuring the probe $\rho(t)$
at time $t$, we attempt to infer the value of
$t$~\cite{R11}. The time-energy uncertainty relation
relates the accuracy of our estimate of $ t$ to the energy fluctuations of the
probe state $\rho(t)$~\cite{R12,R13};
a state with larger energy fluctuations evolves more rapidly, allowing the
elapsed time to be estimated more precisely. Here, too, it is helpful to
envision two parties, one attempting to measure time, the other attempting to
measure energy. Indeed, such entropic time-energy uncertainty relations have
recently been
established~\cite{R14,R15}.

For our purposes, a \emph{clock} is a quantum system used to measure a time
interval.  The clock is initialized at some %
initial time and is measured at a
later %
time, with the aim of the measurement being to reveal the difference in time
between the initialization and the measurement.  
We are particularly interested in how a noise channel affects the accuracy of a clock.
For that purpose we consider the following idealized scenario, involving three
parties referred to as Alice, Bob, and Eve, which is amenable to precise
mathematical analysis (see \cref{z:FnHui0ahNLxW}). Alice prepares a
noiseless clock in the pure state vector $\lvert {\psi_{\mathrm{init}}}\rangle $, then allows
that clock to evolve until some (\emph{a priori} unknown) time $t$. Rather than
measuring the clock herself for the purpose of estimating $t$, Alice stops the
evolution of the clock and sends it to Bob through a noisy quantum channel
$\mathcal{N}_{A\to B}$. As with any noisy channel, we can represent
$\mathcal{N}_{A\to B}$ as an isometric map from Alice's system $A$ to $BE$,
where $B$ is Bob's system and $E$ is the channel's environment, after which $E$
is discarded. In our scenario, Bob receives $B$ and Eve receives $E$. We wish to
study the trade-off between what Bob can learn about the elapsed time by
measuring $B$ and what Eve can learn about the energy of the clock by measuring
$E$. Intuitively, such a trade-off is expected, because leakage to the
environment of information about the clock's energy causes the clock to dephase
in the energy-eigenstate basis, obscuring its evolution.

\begin{figure}
  \centering
  \includegraphics{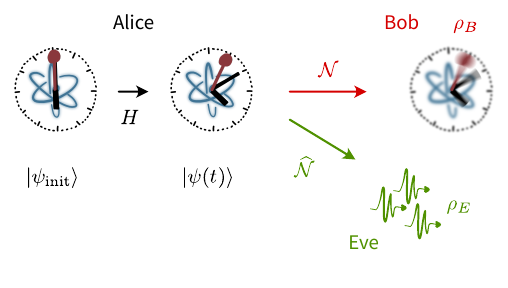}
  \caption{%
    A noiseless clock is initialized by Alice in $\lvert {\psi_{\mathrm{init}}}\rangle $ and evolves for a time $t$ under the
    Hamiltonian $H$. Then Alice sends the clock through an
    instantaneous noisy channel $\mathcal{N}_{A\to B}$ to Bob, who receives the
    state $\rho_B$, measures it, and estimates $t$.
    The complementary channel $\widehat{\mathcal{N}}_{A\to E}$ describes the quantum
    information that leaks to the environment. Eve receives the state
    $\rho_E = \widehat{\mathcal{N}}(\psi(t))$, measures, and estimates the energy parameter of $\lvert {\psi(t)}\rangle $.  Our main result describes the trade-off between Bob's ability to estimate the time and  Eve's ability to estimate the energy.%
}
  \label{z:FnHui0ahNLxW}
\end{figure}

We consider the setting of \emph{local parameter estimation}. This means that the value of a parameter is already approximately known, and we wish to determine it to greater accuracy. In this setting, the optimal estimate of the parameter is determined by the \emph{quantum Fisher information} (QFI). For example, if $\FIqty{Alice}{t}$ denotes the QFI of Alice's state with respect to the parameter $t$, then by performing the optimal measurement on her state, Alice can estimate the value of $t$ with a mean-square error of $1/\FIqty{Alice}{t}$. For the purpose of locally estimating $t=t_0+\Delta t$ to first order in $\Delta t$, it suffices to know the quantum state $\rho(t_0)$ and its first time derivative, and indeed the QFI is determined by just these quantities. 

Bob's noisy clock, degraded by transmission through the noisy channel $\mathcal{N}_{A\to B}$, has a reduced QFI compared to Alice's clock, and correspondingly Bob's optimal measurement yields a less accurate estimate of the time $t$ than Alice's. 
On the other hand, Eve receives the state of Alice's clock after transmission through the \emph{complementary} noisy channel $\widehat{\mathcal{N}}_{A\to E}$, the channel obtained if $B$ is discarded after $A$ is isometrically mapped to $BE$. We imagine that Eve wishes to learn about the \emph{energy} of Alice's clock, rather than about the elapsed time. More precisely, Eve's goal is to determine an ``energy parameter'' denoted $\eta$ and defined in Sec.~\ref{z:ip.xvGb3bzRb}, which is complementary to the time $t$. Because Eve, like Bob, receives a state of the clock degraded by noise, the QFI of her state with respect to $\eta$ is in general less than Alice's.

Our main result is an equality relating Bob's QFI about $t$ to Eve's QFI about $\eta$ given by
\begin{align}
  \frac{\FIqty{Bob}{t}}{\FIqty{Alice}{t}} +
  \frac{\FIqty{Eve}{\eta}}{\FIqty{Alice}{\eta}}
  = 1\ .
  \label{z:OUF5HnwJdMZI}
\end{align}
This time-energy uncertainty relation, derived in \cref{z:6pGtSPe4CZWB} and \cref{z:.B4MZrnrq.9S} using semidefinite programming duality, substantially differs from previous results~\cite{R11,R16,R17} %
 in that it characterizes the trade-off between Bob's and Eve's QFI, rather than the trade-off between the inherent energy variance and time uncertainty of the noiseless clock. 

\begin{figure}
  \centering
  \includegraphics{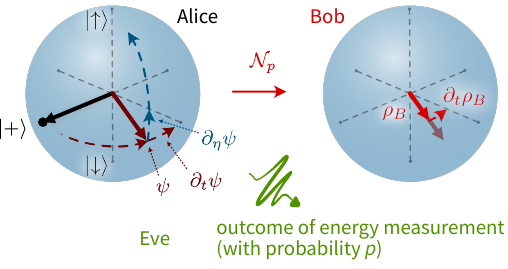}
  \caption{%
    Illustration of \cref{z:OUF5HnwJdMZI} for a single qubit subjected
    to partial dephasing. Alice's clock state is initialized as
    $\lvert {+}\rangle =\left(\lvert {\uparrow}\rangle +\lvert {\downarrow }\rangle \right)/\sqrt2$ and evolves
    according to the Hamiltonian $H=\omega {Z/2}$, 
    {where $Z$ denotes
    the qubit Pauli-$Z$
    operator.}
    At time $t$, the channel
    $\mathcal{N}_p(\cdot) = (1-p)(\cdot) + p\,\lvert {\uparrow}\rangle \mkern -1.8mu\relax \langle{\uparrow}\rvert (\cdot)\lvert {\uparrow
    }\rangle \mkern -1.8mu\relax \langle{\uparrow
    }\rvert + p\,\lvert {\downarrow}\rangle \mkern -1.8mu\relax \langle{\downarrow}\rvert (\cdot)\lvert {\downarrow}\rangle \mkern -1.8mu\relax \langle{\downarrow}\rvert $ is instantaneously applied to
    Alice's clock state. In effect, Eve measures the energy observable
    {$Z$}
    with probability $p$, and Bob receives the partially dephased
    clock.  \cref{z:OUF5HnwJdMZI} relates Bob's reduced information
    about the elapsed time to Eve's information gain about the clock's energy.
    Unitary evolution in Eve's complementary energy variable $\eta$, generated
    by an optimal local time-sensing observable, rotates the state into a
    direction that is orthogonal to the direction of the original evolution in
    time $t$ (see
    \cref{z:t.T5CjLp6B3a}.)%
  }
  \label{z:duwTyi1jsdbn}
\end{figure}

\Cref{z:duwTyi1jsdbn} illustrates the setting of
\cref{z:OUF5HnwJdMZI} in a concrete example.  Alice initializes a
single qubit in the pure state vector
$\lvert {+}\rangle =\left(\lvert {\uparrow}\rangle +\lvert {\downarrow}\rangle \right)/\sqrt2$, which evolves
under the Hamiltonian 
$H= \omega Z/2$.
Here and in the following,
    $X, Y, Z$ denote the 
    qubit Pauli-$X,Y,Z$ operators, respectively.
    The qubit basis states are denoted by
      $\lvert {\uparrow}\rangle ,\lvert {\downarrow}\rangle $ for consistency with which state is
      excited with respect to the Hamiltonian $H$, with
      $Z\lvert {\uparrow }\rangle = \lvert {\uparrow}\rangle $ and $Z\lvert {\downarrow }\rangle = -\lvert {\downarrow}\rangle $.
      Later in this work, we also use the alternative notation
      $\lvert {0}\rangle  \equiv \lvert {\uparrow}\rangle $ and $\lvert {1}\rangle  \equiv \lvert {\downarrow}\rangle $
      whenever necessary to facilitate the representation of states of multiple
      qubits using bit strings or for consistency with the literature on quantum
      error-correcting codes.
At time
$t=t_0+\Delta t$, the
partially dephasing channel
$\mathcal{N}_p = (1-p){{\mathrm{id}}} + p\mathcal{D}_Z$ is applied to Alice's
qubit, where
$\mathcal{D}_Z(\cdot) = \langle {\uparrow}\mkern 1.5mu\relax \vert \mkern 1.5mu\relax {(\cdot)}\mkern 1.5mu\relax \vert \mkern 1.5mu\relax {\uparrow}\rangle \,\lvert {\uparrow }\rangle \mkern -1.8mu\relax \langle{\uparrow }\rvert +
\langle {\downarrow}\mkern 1.5mu\relax \vert \mkern 1.5mu\relax {(\cdot)}\mkern 1.5mu\relax \vert \mkern 1.5mu\relax {\downarrow}\rangle \,\lvert {\downarrow}\rangle \mkern -1.8mu\relax \langle{\downarrow}\rvert $. We may describe this channel by
saying that the environment (Eve) measures the qubit with probability $p$ in the
energy-eigenstate basis (i.e. along the $Z$ axis of the Bloch sphere). The
partial dephasing attenuates the $t$ dependence of Bob's state $\rho_B(t)$ by
the factor $1-p$, hindering his ability to estimate the
time. \cref{z:OUF5HnwJdMZI} captures the trade-off between Bob's
information about the time (proportional to $1-p$) and Eve's information gain
about the energy (proportional to $p$).

The trade-off relation \cref{z:OUF5HnwJdMZI} 
can be a 
useful tool for deriving upper bounds on QFI. The QFI for a mixed state can
be tricky to characterize in cases where a diagonal representation
of the state is not easily obtained.
Along these lines, it is useful to note that QFI obeys a data-processing
inequality which ensures that, for any state $\rho$ and any quantum channel
$\mathcal{N}$, the QFI of $\mathcal{N}(\rho)$ is no larger than the QFI of
$\rho$~\cite{R18}. We can imagine that Eve applies a channel to her state $\rho_E$, obtaining the state $\rho^\prime_E$, which she then measures
for the purpose of estimating $\eta$. Using the data-processing inequality, we
conclude that
\begin{align}
  \frac{\FIqty{Bob}{t}}{\FIqty{Alice}{t}} +
  \frac{\FIqty[']{Eve}{\eta}}{\FIqty{Alice}{\eta}}
  \leq 1\ ,
  \label{z:zlqaTJKIS9Yc}
\end{align}
where now $\FIqty[']{Eve}{\eta}$ denotes the QFI of $\rho^\prime_E$ with
respect to $\eta$. Even if the QFI of $\rho_E$ is difficult to compute, the QFI
of $\rho^\prime_E$ may be easy to compute if the channel taking $\rho_E$ to
$\rho^\prime_E$ is artfully chosen; then \cref{z:zlqaTJKIS9Yc}
provides a computable upper bound on $\FIqty{Bob}{t}$. For example, in the case
where $\mathcal{N}_{A\to B}$ is an amplitude damping noise channel, a useful upper bound
on Bob's QFI can be derived by applying a completely dephasing channel to Eve's
state $\rho_E$. We apply this idea to an Ising spin chain in
\cref{z:Pgp-gW09n9nZ}.

One consequence of \cref{z:OUF5HnwJdMZI} is a necessary and sufficient
condition for the clock's sensitivity to be unaffected by transmission through
the noisy channel $\mathcal{N}_{A\to B}$: $\FIqty{Bob}{t}=\FIqty{Alice}{t}$ if
and only $\FIqty{Eve}{\eta} = 0$. This condition can be usefully restated in
terms of the Kraus operators $\{E_k\}$ of the channel $\mathcal{N}_{A\to
  B}$. Recall that we aim to estimate the time $t=t_0+\Delta t$ in the setting
of local parameter estimation, i.e. to linear order in $\Delta t$. Suppose that
after evolution for time $t_0$, the state of Alice's clock is $|\psi\rangle$,
and that $\lvert {\xi }\rangle = (H-\langle {H}\rangle _\psi)\lvert {\psi }\rangle = P_\psi^\perp H\lvert {\psi}\rangle $ with
$P_\psi^\perp = \mathds{1}-\lvert {\psi}\rangle \mkern -1.8mu\relax \langle{\psi}\rvert $.  Then the condition $\FIqty{Eve}{\eta} = 0$
is equivalent to
\begin{align}
  \langle {\xi }\rvert E_{k}^\dagger E_j \lvert {\psi }\rangle + \langle {\psi }\rvert E_{k}^\dagger E_j \lvert {\xi }\rangle = 0
  \qquad\text{for all \(k,j\)}\ .
  \label{z:WbYGcUl-rp7G}
\end{align}
Intuitively, \cref{z:WbYGcUl-rp7G} means that the action of the channel on the
clock cannot be confused with genuine time evolution.

\Cref{z:WbYGcUl-rp7G} may be recognized as a weakened version of the Knill-Laflamme condition for quantum error correction,
the necessary and sufficient condition for the action of a noisy channel on an encoded subspace to be reversible by a suitable recovery channel~\cite{R19}. This condition may be stated as $\Pi_L E_k^\dagger E_j \Pi_L \propto \Pi_L$ for all $k$ and $j$, where $\Pi_L$ is the projector onto the {encoded subspace}. To write \cref{z:WbYGcUl-rp7G} in a similar form, consider the two-dimensional subspace spanned by the mutually orthogonal state vectors $|\psi\rangle$ and $|\xi\rangle$; we call this two-dimensional space a ``virtual qubit.'' Using the notation $\lvert {+}\rangle _L := \lvert {\psi}\rangle $, $\lvert {-}\rangle _L := \lVert {\lvert {\xi}\rangle }\rVert ^{-1} \lvert {\xi}\rangle $, the orthogonal projector onto the virtual qubit is $\Pi_L = \lvert {+}\rangle \mkern -1.8mu\relax \langle{+}\rvert _L + \lvert {-}\rangle \mkern -1.8mu\relax \langle{-}\rvert _L$, and $Z_L=\lvert {+}\rangle \mkern -1.8mu\relax \langle{-}\rvert _L + \lvert {-}\rangle \mkern -1.8mu\relax \langle{+}\rvert _L$ is the logical $Z$ Pauli operator acting on the virtual qubit. In this language, \cref{z:WbYGcUl-rp7G} becomes 
\begin{align}
\operatorname{tr}\bigl(\Pi_L E_{k}^\dagger E_j \Pi_L \, Z_L\bigr) = 0 \qquad\text{for all \(k,j\)}\ .
\label{z:WWu9g46efaSi}
\end{align}
The condition~\cref{z:WWu9g46efaSi} is reminiscent of a recently formulated condition for quantum coding to improve how measurement sensitivity
scales with {increasing sensing
  time}~\cite{R20,R21}. In
\cref{z:JZtfKZv25too}, we explain how time-covariant quantum
error-correcting codes automatically fulfill~\cref{z:WbYGcUl-rp7G},
{providing} some simple examples.  In particular, we consider spins on a graph
with Ising or Heisenberg interactions and construct a state vector $\lvert {\psi}\rangle $
that fulfills \cref{z:WbYGcUl-rp7G},
where the noise model inflicts a single located erasure.

We {have derived} the trade-off relation \cref{z:OUF5HnwJdMZI} in a
highly idealized setting, in which noiseless evolution of Alice's clock is
followed by transmission to Bob through the noisy channel
$\mathcal{N}_{A\to B}$. For an actual clock, the noise acts continuously as the
clock evolves, rather than after the time evolution is complete. By focusing on
the idealized setting, we have been able to perform a particularly elegant
analysis of the time-energy trade-off. But in \cref{z:Xuq8JaOEWrUo}
we connect our results to the more realistic case of continuous Markovian noise
described by a master equation {in Lindblad form}, noting that the two
settings are actually equivalent, or nearly equivalent, under certain
conditions. One can decompose the Lindbladian into a Hamiltonian part and a
noise part that contains all the jump operators; if, for example, these two
parts define commuting channels, then the {Markovian} evolution for time $t$
is equivalent to Hamiltonian evolution for time $t$ followed by a noise channel
$\mathcal{N}_t$. Other cases where the Lindblad evolution is compatible with a
trade-off relation of the form \cref{z:OUF5HnwJdMZI} (at least to a
good approximation) are identified in \cref{z:Xuq8JaOEWrUo}.

\begin{figure*}
  \centering
  \includegraphics{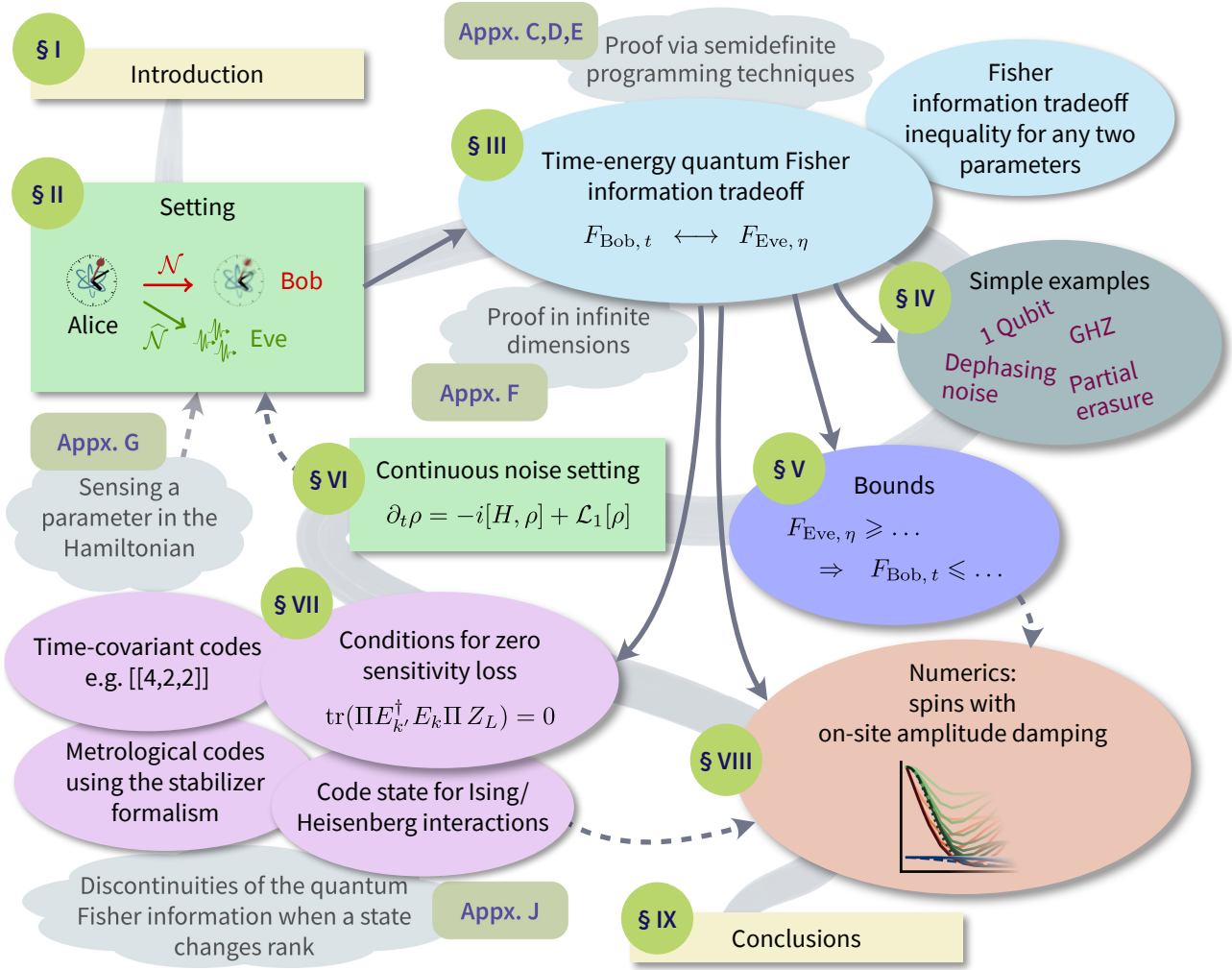}
  \caption{Overview of our main results and structure of this {work}.}
  \label{z:a4FfzgZJotxe}
\end{figure*}

Although the time-energy trade-off provided the primary motivation for this
work, we find that a trade-off relation similar to
\cref{z:OUF5HnwJdMZI} can be derived in a more general setting. Suppose
that $A$ and $B$ are Hermitian operators, and that
$\psi=|\psi\rangle\langle \psi|$ is a pure quantum state. We may consider the
``flow'' in Hilbert space generated by $A$ or by $B$. That is, we consider a
one-parameter family of pure states close to $\psi$, generated by $A$ and
parameterized by $a$, and a one-parameter family generated by $B$ and
parametrized by $b$, such that
\begin{align}
    \partial_a \psi = -i[A,\psi]\ , \quad \partial_b \psi = -i[B,\psi]\ .
\end{align}
In the setting of local parameter estimation, we suppose that Bob wishes to
estimate the parameter $a$ and Eve wants to estimate the parameter $b$, where
$a$ and $b$ are both small. Alice's QFI about $a$ is $\FIqty{Alice}{a}$, but Bob
receives the state via the noisy channel $\mathcal{N}_{A\to B}$, so his QFI
about $a$ ($\FIqty{Bob}{a}$) is in general smaller than Alice's. Alice's QFI
about $b$ is $\FIqty{Alice}{b}$, but Eve receives the state via the
complementary channel $\widehat{\mathcal{N}}_{A\to E}$, so her QFI about $b$
($\FIqty{Eve}{b}$) is in general smaller than Alice's. In
\cref{z:6pGtSPe4CZWB} we derive the trade-off relation
\begin{align}
   \frac{ \FIqty{Bob}{a} }{ \FIqty{Alice}{a} }
  + \frac{ \FIqty{Eve}{b} }{ \FIqty{Alice}{b} }
  \leq
  1 + 2 \sqrt{ 1 - \frac{ \bigl \langle {i[A,B]}\bigr \rangle _\psi^2 }{ 4\sigma_A^2\sigma_B^2 } }\ ,
  \label{z:VUDo7wgAsXzc}
\end{align}
where $\sigma_{{M}} := [\langle {{M}^2}\rangle _\psi - \langle {{M}}\rangle _\psi^2]^{1/2}$ denotes the
standard deviation of the observable ${M}$. Note that, in contrast to
\cref{z:OUF5HnwJdMZI}, this relation is an inequality rather than an
equality. It is reminiscent of the Robertson uncertainty relation, with the
commutator quantifying the incompatibility of the observables $A$ and $B$.
\Cref{z:a4FfzgZJotxe} summarizes the structure of this work and provides
an overview of our results.
In \cref{z:CKirppebju2f}, we introduce the setting of local parameter estimation,
recall some useful properties of the QFI, define the energy parameter $\eta$,
and review the concept of a complementary quantum channel.
We sketch the proof of the trade-off relation \cref{z:OUF5HnwJdMZI} and
its generalization \cref{z:VUDo7wgAsXzc} in
\cref{z:6pGtSPe4CZWB} (more details can be found in
\cref{z:.B4MZrnrq.9S}), and discuss some examples
in \cref{z:Z14T5rsnGxfi}. We use the trade-off relation to derive
upper bounds on the QFI in \cref{z:qdkRM6aFmtdz}.
In \cref{z:Xuq8JaOEWrUo} we discuss how the setting in
\cref{z:FnHui0ahNLxW} is connected with the more realistic setting of
continuous Markovian noise.  In \cref{z:JZtfKZv25too} we
derive the necessary and sufficient condition \cref{z:WbYGcUl-rp7G} for the
clock's sensitivity to be undiminished by transport through the noisy channel
$\mathcal{N}_{A\to B}$, and discuss some of the implications of this condition.
Numerical results for our upper bound on QFI in many-body systems are reported
in \cref{z:Pgp-gW09n9nZ}.
We summarize and comment on our results in \cref{z:HssWipSG8E2-}. Many further
details are {presented} 
in the appendices.

\section{Setting}
\label{z:CKirppebju2f}

We review the standard setting in quantum metrology of single-parameter
estimation.  We then introduce our noise model and the quantities that are
relevant to formulate our uncertainty relation.

\subsection{Quantum parameter estimation}
\label{z:veh1L4J.WPya}

Consider a quantum state $\rho(t)$ that depends on a single parameter $t$.
The task we study is how well the parameter $t$ can be estimated by performing
suitable measurements (\cref{z:o60mrlT1D3vR}).  In the
context of this work, the parameter $t$ is identified with physical time,
although the results hold for any general real parameter that the quantum state
might depend on.

\begin{figure}
  \centering
  \includegraphics{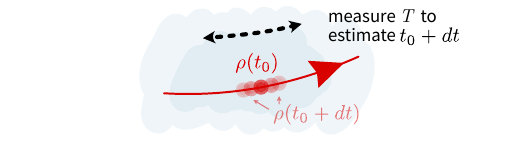}
  \caption{In the setting of quantum parameter estimation, the task is to
      infer a parameter $t$ in a one-parameter family of states $t\mapsto \rho(t)$
      through suitable measurements.  For \emph{local} parameter estimation, we
      assume the value of the parameter is already known to lie in the
      neighborhood of a given value $t_0$.  The measurement is required to refine
      the parameter estimation by optimally distinguishing $\rho(t_0)$ from
      $\rho(t_0 + dt)$ to first order in $dt$.  This setting is standard in the
      field of quantum metrology, and the optimal sensitivity is quantified by a
      quantity known as the Fisher information.}
  \label{z:o60mrlT1D3vR}
\end{figure}

We consider the setting of local sensitivity, where the goal of the quantum
measurement is to refine the precision to which we determine the parameter if
the value of the parameter is already known to be close to a given value $t_0$. 
More precisely, we seek a measurement operator $T$ with minimal variance such
that the expectation value of $T$ reveals the value of the parameter locally
around $t_0$ to first order in $dt$, i.e.,
\begin{align}
  \langle {T}\rangle _{\rho(t_0+dt)} = t_0 + dt + O(dt^2)\ .
  \label{z:c7.VBMRSmG0h}
\end{align}
Identifying the orders in $dt$ we see
that~\eqref{z:c7.VBMRSmG0h} is equivalent
to
\begin{align}
  \langle {T}\rangle _{\rho(t_0)}
  = t_0\quad\text{and}\quad
  \operatorname{tr}\bigl(T\; \partial_t\rho\, (t_0)\bigr)
  = 1\ ,
  \label{z:d4-BXp9cGlJQ}
\end{align}
using the notation
$\partial_t\rho = \frac{\partial \rho}{\partial t}$.  (We write a
partial derivative instead of a total derivative in anticipation of other
variables which will be introduced later.)
In the literature, it is common to reuse the symbol $t$ for both the parameter
on which $\rho$ depends as well as the reference value of the parameter $t_0$.
We keep the distinction for clarity. 

Here, we restricted the measurement to be projective, as described by the
Hermitian observable $T$.  A more general positive operator-valued measure (POVM)
 does not offer any more
sensitivity in sensing the
parameter~\cite{R0,R11}.

A central result in quantum metrology is the quantum Cram\'er-Rao bound, which
states that the optimal sensitivity to which one can determine the parameter $t$
locally around $t_0$ is determined by a quantity called the quantum Fisher
information~\cite{R22,R23,R0}.  The quantum Fisher
information of the state $\rho(t_0)$ with respect to a direction
$\partial_t\rho\,(t_0)$ is defined as
\begin{align}
  \Ftwo{\rho}{\partial_t\rho} = \operatorname{tr}\bigl(\rho \, R^2 \bigr)\ ,
  \label{z:.eb-akuquBNd}
\end{align}
where $R$ is any Hermitian operator that solves the equation
$\frac12\bigl\{ \rho, R \bigr\} = \frac12\bigl(\rho R + R \rho\bigr) = \partial_t\rho$, and
where the quantities $\rho$ and $\partial_t\rho$ are evaluated at
$t_0$.
The Cram\'er-Rao bound can be formulated for our purposes as follows:
For any observable $T$ that
satisfies~\eqref{z:d4-BXp9cGlJQ}, we must have
\begin{align}
  \bigl \langle {(T - t_0)^2}\bigr \rangle _{\rho(t_0)}
  \geq \frac1{\Ftwo{\rho(t_0)}{\partial_t \rho\, (t_0)}}\ ,
  \label{z:5-HFHQQdFwEu}
\end{align}
and furthermore, equality in~\eqref{z:5-HFHQQdFwEu} can always be
achieved by a suitable choice of $T$.  We refer to a choice of $T$ which is
optimal in~\eqref{z:5-HFHQQdFwEu} as an \emph{optimal local-sensing
  observable for $t$}.

The operator $R$ in~\eqref{z:.eb-akuquBNd} is called a \emph{symmetric
  logarithmic derivative}.  When $\rho$ and $\partial_t\rho$ commute, we can choose
$R = \rho^{-1} \partial_t\rho = (\partial/\partial t)\ln\rho $.
A general construction of $R$ in terms of an eigendecomposition of $\rho$ is
given as follows~\cite{R24}.  Consider an eigenbasis $\{ \lvert {k}\rangle  \}$
of $\rho$ that spans the full Hilbert space, such that
$\rho = \sum_k \lambda_k \lvert {k}\rangle \mkern -1.8mu\relax \langle{k}\rvert $ and $k=1,2,\ldots,\dim(\mathscr{H})$, then
\begin{align}
  R =
  \sum_{\substack{k,k':\\\lambda_k+\lambda_{k'}\neq 0}}
  \frac2{\lambda_k + \lambda_{k'}}\,\bigl \langle {k}\mkern 1.5mu\relax \big \vert \mkern 1.5mu\relax {\partial_t\rho}\mkern 1.5mu\relax \big \vert \mkern 1.5mu\relax {k'}\bigr \rangle \,\lvert {k}\rangle \mkern -1.8mu\relax \langle{k'}\rvert \ ,
  \label{z:d-w60nHacIQJ}
\end{align}
where the sum ranges over all pairs of indices $k,k'$ except those for which
both $\lambda_k=0$ and $\lambda_{k'}=0$.  The expression for the Fisher information becomes $\Ftwo{\rho}{\partial_t\rho}$, where 
\begin{align}
  \Ftwo{\rho}{\partial_t\rho}
  &= \sum_{\substack{k,k':\\\lambda_k+\lambda_{k'}\neq 0}}
    \frac2{\lambda_k + \lambda_{k'}} \bigl \lvert { \bigl \langle {k}\mkern 1.5mu\relax \big \vert \mkern 1.5mu\relax { \partial_t\rho }\mkern 1.5mu\relax \big \vert \mkern 1.5mu\relax {k'}\bigr \rangle  }\bigr \rvert ^2\ .
  \label{z:r7GWSaRrLfNh}
\end{align}
The solution to the anticommutator equation $\frac12\{\rho, R\} = \partial_t\rho$ is
unique up to transformations of the form
$R \mapsto R + P_\rho^\perp {M} P_\rho^\perp$ where ${M}$ is an arbitrary
Hermitian operator, where $P_\rho^\perp = \mathds{1}- P_\rho$, and where $P_\rho$
denotes the projector onto the support of $\rho$.
In the event that $P_\rho^\perp \frac{d\rho}{dt} P_\rho^\perp \neq 0$, there is no
solution for $R$.  In such a situation, the optimal estimation
variance~\eqref{z:5-HFHQQdFwEu} is zero and the Fisher information is not
defined; such cases do not arise in the setting we consider in this
work.

We review the solutions to the anticommutator equation
$\frac12\{\rho, R\} = \partial_t\rho$ in \cref{z:V.Yy8m55j7Q5}.
In \cref{z:KH4B5FtzQ1kE}, the definition and elementary properties of
the Fisher information are reviewed using simple techniques based on
semidefinite programming.  In \cref{z:9mLm-0LXJIol}, we
review a derivation of the Cram\'er-Rao bound using these methods.

Observables
$T$ that estimate the time parameter $t$ with an
accuracy that achieves the Cram\'er-Rao bound~\eqref{z:5-HFHQQdFwEu},
i.e., the optimal local-sensing observables, turn out to be the projective
measurements with outcomes associated with the eigenspaces of a symmetric
logarithmic derivative~\cite{R22,R0}.  Specifically, any optimal local-sensing
observable for $t$ is of the form
\begin{align}
  T = t_0 + \frac{1}{\Ftwo{\rho}{\frac{d\rho}{dt}}}\, R\ ,
  \label{z:fRkTtRxHhvuE}
\end{align}
where $R$ is as above any solution to the anticommutator equation
$\frac12\{\rho,R\}=d\rho/dt$ (see \cref{z:eQl8.oS.4.na}
for a review of the proof).
Due to the freedom in the choice of $R$, all optimal local-sensing observables
for $t$ differ by a term of the form $P_\rho^\perp 
{M} P_\rho^\perp$ where ${M}$ is
any Hermitian operator.

In the remaining part of this section, 
we review a few properties of the Fisher information for later use (see
\cref{z:KH4B5FtzQ1kE} for details).  First is a scaling property: If
$0<\alpha\leqslant 1$ and $\beta\in\mathbb{R}$, we have
\begin{align}
  \Ftwo{\alpha\rho}{\beta\, \partial_t\rho}
  = \frac{\beta^2}{\alpha} \Ftwo{\rho}{\partial_t\rho}\ ,
  \label{z:OBSqCtuxYn6d}
\end{align}
where the definition~\eqref{z:.eb-akuquBNd} is formally extended to
positive semidefinite operators $\rho$ that satisfy $\operatorname{tr}(\rho)\leqslant 1$.
Second, in case the state $\rho$ and derivative $\partial_t\rho$ commute, the
Fisher information takes the simple form
\begin{align}
  [\rho,\partial_t\rho] &= 0\quad\Rightarrow\quad
  \Ftwo{\rho}{\partial_t\rho} = \operatorname{tr}\bigl[\rho^{-1}\,(\partial_t\rho)^2\bigr]\ .
  \label{z:q9lQk9SttqXb}
\end{align}
Finally, for general $\rho,\partial_t\rho$, we can express the Fisher information in terms of a pair of convex optimization
problems~\cite{R25,R26,R27} as
\begin{subequations}
  \begin{align}
    \hspace*{1em}&\hspace*{-1em}
    \frac14 \Ftwo{\rho}{\partial_t\rho}
    \nonumber\\
    \begin{split}
    &=
      \max_{S=S^\dagger} \Bigl\{ 
        \operatorname{tr}\bigl[ (\partial_t\rho)\,S \bigr] - \operatorname{tr}\bigl[ \rho\,S^2 \bigr] \Bigr\}
    \end{split}
    \label{z:-JCqSeMuhqI6}
    \\
    \begin{split}
    &= \min_{L~\mathrm{arb.}} \Bigl\{
      \operatorname{tr}(L^\dagger L)\ :\ 
      \rho^{1/2}L+L^\dagger\rho^{1/2} = \partial_t\rho
    \Bigr\}\ .
    \end{split}
    \label{z:fC3N68oVdRLb}
  \end{align}
\end{subequations}
These two optimizations can be cast as semidefinite problems that are dual to
each other.  These optimizations are convenient to derive bounds on the Fisher
information, as it suffices to exhibit suitable candidates
in~\eqref{z:-JCqSeMuhqI6} or~\eqref{z:fC3N68oVdRLb}.

\subsection{Time and energy parameters of the noiseless clock}
\label{z:ip.xvGb3bzRb}
\label{z:t.T5CjLp6B3a}

Now we turn to the setup depicted in~\cref{z:FnHui0ahNLxW}, in
which Alice possesses a noiseless quantum clock which she sends to Bob through a
given noisy channel.  In this subsection, we study Alice's noiseless quantum
clock, and in the following subsection we study the effect of the noise.

\paragraph{The noiseless clock.}
Suppose that Alice prepares a quantum clock in a pure state living in a
finite-dimensional Hilbert space $\mathscr{H}_A$.  She lets it evolve according to a
Hamiltonian $H(t)$, generating a one-parameter family of state vectors $t\mapsto \lvert {\psi(t)}\rangle $.
The time evolution of
$\psi(t) = \lvert {\psi(t)}\rangle \langle {\psi(t)}\rvert $ is governed by the standard Schr\"odinger time evolution
\begin{align}
  \partial_t \psi :=
  \frac{\partial\psi}{\partial t} = -i\,[H, \psi]\ .
  \label{z:6p4-xHlpWwdA}
\end{align}
We now compute the Fisher information associated with Alice's clock locally
around a time of interest $t_0$, following the
definition~\eqref{z:.eb-akuquBNd}. %
For any $t_0$, we can choose $R = 2\,\partial_t \psi = -2i[H,\psi]$, because
$\{\partial_t \psi, \psi\} = \partial_t(\psi^2) = \partial_t \psi$.  Alice's
Fisher information $\FIqty{Alice}{t}$ for the evolution $\lvert {\psi(t)}\rangle $ at
the time of interest $t_0$ is therefore given by
\begin{align}
  \FIqty{Alice}{t} := \Ftwo[\Big]{\psi}{-i[H, \psi]}
  = 4\sigma_H^2\ ,
  \label{z:BBGAvQkdkmd8}
\end{align}
where $\psi$ and $H$ are evaluated at time $t_0$, and where again, 
we denote by
$\sigma_{{M}} = [\langle {{M}^2}\rangle _\psi - \langle {{M}}\rangle _\psi^2]^{1/2}$ the standard deviation of
an observable ${M}$.  Alternative expressions of the standard deviation are given
by
\begin{align}
  \sigma_{{M}}^2 %
  = \bigl \langle {({M} - \langle {{M}}\rangle )^2}\bigr \rangle 
  = -\bigl \langle { [{M}, \psi]^2 }\bigr \rangle \ ,
\end{align}
writing $\langle {{M}}\rangle  := \langle {{M}}\rangle _\psi$ for brevity.

Around the point $t_0$, any optimal local-sensing observable for $t$ takes the
form given by~\eqref{z:fRkTtRxHhvuE}, which we can rewrite in this
context as
\begin{align}
  T = t_0 - \frac{i [ H, \psi ]}{2\sigma_H^2}
  + P_\psi^\perp {M} P_\psi^\perp\ ,
  \label{z:w7NUzrfGUNE-}
\end{align}
where ${M}$ is any Hermitian operator.  In the case where $H$ is time independent,
then $\FIqty{Alice}{t}$ does not depend on the time of interest $t_0$,
but the optimal sensing observable $T$ depends
on $t_0$ not only directly but also indirectly through $\psi$ and $\partial_t \psi$.
In the following, we fix $t_0$ and we only consider the evolution
$\lvert {\psi(t)}\rangle $ locally at $t_0$.  Furthermore, we use the shorthand
$\lvert {\psi }\rangle := \lvert {\psi(t_0)}\rangle $.

\paragraph{The energy parameter.}
The optimal local time-sensing observable $T$
in \cref{z:w7NUzrfGUNE-}, being a Hermitian operator, can be
used to generate a different evolution in an alternative direction in the space
of quantum states.  
In our setup, we define $\eta_0 = \langle {H}\rangle _\psi$ and we consider any family of
state vectors $\eta\mapsto \lvert {\psi(\eta)}\rangle $ such that
$\lvert {\psi(\eta{=}\eta_0)}\rangle  = \lvert {\psi}\rangle  = \lvert {\psi(t{=}t_0)}\rangle $ and such that at
the point $\lvert {\psi(\eta{=}\eta_0)}\rangle $ we have
\begin{align}
  \partial_\eta \psi = i[T, \psi] \ .
  \label{z:PlQtfu9raUmQ}
\end{align}
This evolution can be interpreted as a Schr\"odinger equation with the effective
Hamiltonian $-T$.  An example of such an evolution is
\begin{align}
  \lvert {\psi(\eta)}\rangle  = {e}^{iT(\eta-\eta_0)}\,\lvert {\psi}\rangle \ .
  \label{z:DwBxsQ4hMn8m}
\end{align}
Interestingly, the evolution generated in this way locally around $\lvert {\psi}\rangle $
turns out to be complementary to time evolution in the sense that we can derive
a meaningful uncertainty relation and that the parameter $\eta$ can be
identified with the average energy of the state vector $\lvert {\psi(\eta)}\rangle $ (see
\cref{z:KI2eOYZK2.Ln}).
\begin{figure}
  \centering
  \includegraphics{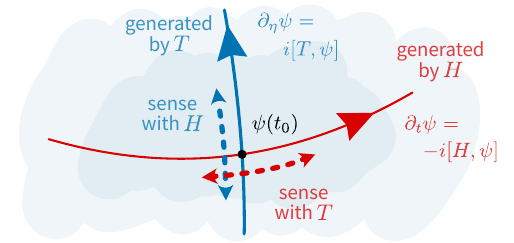}
  \caption{We define a parameter $\eta$ that is complementary to time evolution and
      that represents the energy of the state.  Consider a quantum clock modeled
      as a pure state $\psi$ evolving according to the Schr\"odinger equation
      $\partial_t \psi = -i[H, \psi]$, where $H$ is the Hamiltonian.  Locally
      around $t_0$, the observable $T$ that optimally distinguishes the
      neighboring states $\psi(t_0)$ and $\psi(t_0+dt)$ defines an \emph{optimal
        local time-sensing observable}. $T$ is the relevant measurement to carry
      out to optimally read out the information about time stored in the quantum
      clock.  We now consider locally around $\psi(t_0)$ the direction in state
      space defined by $\partial_\eta \psi = i[T, \psi]$, i.e., a
      Schr\"odinger-type evolution with $-T$ playing the role of an effective
      Hamiltonian.  It turns out that the optimal estimation procedure for the
      parameter $\eta$ is to measure $H$ itself. Therefore, the parameter $\eta$
      represents the energy of $\psi(\eta)$.  The parameters $t$ and $\eta$ are,
    therefore, complementary to each other in the sense that the generator
    associated with one parameter optimally distinguishes neighboring values of
    the other parameter and \emph{vice versa}.
    }
  \label{z:KI2eOYZK2.Ln}
\end{figure}

More formally and to clarify the dependencies of $\lvert {\psi}\rangle $ on $t$ and $\eta$,
we consider a two-parameter family of state vectors 
$(t,\eta) \mapsto \lvert {\psi(t,\eta)}\rangle $ with
$\lvert {\psi(t_0,\eta_0)}\rangle =\lvert {\psi}\rangle $ and such that at the point $(t_0,\eta_0)$ we
have
\begin{align}
  \partial_t \psi \, (t_0,\eta_0) &= -i[H, \psi]\ ,
  &
    \partial_\eta \psi \, (t_0,\eta_0) &= i[T, \psi]\ ,
  \label{z:ff7jHgn75DW4}
\end{align}
where $T$ is given by~\eqref{z:w7NUzrfGUNE-}.  For example, we
could choose
\begin{align}
  \lvert {\psi(t,\eta)}\rangle  = \exp\bigl\{-i[ (t-t_0) H - (\eta-\eta_0) T]\bigr\}\,\lvert {\psi}\rangle \ . \label{z:BocUFnzQodPW}
\end{align}

Unless indicated otherwise, the state vector $\lvert {\psi}\rangle $ and the corresponding
derivatives $\partial_t\psi, \partial_\eta\psi$ are henceforth implicitly
evaluated at $(t_0,\eta_0)$.  We use the shorthands
$\lvert {\psi(t)}\rangle  := \lvert {\psi(t,\eta_0)}\rangle $ and
$\lvert {\psi(\eta)}\rangle  := \lvert {\psi(t_0,\eta)}\rangle $ to denote the respective evolutions
according to $t$ and $\eta$ in which the other parameter is fixed to $\eta_0$ or
$t_0$, respectively; the name of the argument ($t$ or $\eta$) determines which
evolution is meant.

Let us re-express the derivative $\partial_\eta \psi$ 
of $\psi=\lvert {\psi}\rangle \langle {\psi}\rvert $
in terms of the
Hamiltonian.  Using~\eqref{z:w7NUzrfGUNE-}, we have
\begin{align}
  \partial_\eta \psi = i \bigl[T, \psi\bigr]
  = \frac{1}{2\sigma_H^2}\, \bigl[  [H,\psi] , \psi \bigr]\ .
  \label{z:Df3ZpFNLdboK}
\end{align}
A brief computation reveals that $ \bigl[[H,\psi],\psi\bigr]
= H\psi + \psi H -2\langle {H}\rangle \,\psi = \{H-\langle {H}\rangle , \psi\} $ and therefore
\begin{align}
  \partial_\eta \psi
  &= \frac1{2\sigma_H^2}\,\bigl\{H-\langle {H}\rangle , \psi\bigr\}\ .
  \label{z:SMjFM6HgZlts}
\end{align}
Alice's Fisher information with respect to the parameter $\eta$ is given by the
same expression as~\eqref{z:BBGAvQkdkmd8}, but with $H$ and $t$
replaced by $-T$ and $\eta$, to get
\begin{align}
  \FIqty{Alice}{\eta}
  &:=
    \Ftwo[\Big]{\psi}{i[T,\psi]}
    = 4\sigma_T^2
    = \frac1{\sigma_H^2}\ ,
  \label{z:lN.pbrjbNave}
\end{align}
where the last equality follows from
\begin{align}
  \sigma_T^2 = \bigl \langle {(T-t_0)^2}\bigr \rangle 
  = \frac{-\langle {[H,\psi]^2}\rangle }{4\sigma_H^4}
  = \frac1{4\sigma_H^2}\ .
\end{align}

To justify that the parameter $\eta$ in the
evolution~\eqref{z:PlQtfu9raUmQ} can be associated with the energy of
the state vector locally around $\lvert {\psi}\rangle $, we compute the optimal sensing observable
for $\eta$ and show that it is the Hamiltonian $H$ itself (up to terms lying
outside of the support of $\psi$).
The optimal local-sensing observable that distinguishes $\psi(\eta)$ from
$\psi(\eta+d\eta)$ is given by~\eqref{z:fRkTtRxHhvuE}, but with the parameter $t$ replaced by the parameter $\eta$.
Using~\eqref{z:SMjFM6HgZlts}, observe that the operator
$R' = (H-\langle {H}\rangle )/\sigma_H^2$ solves the equation
$\{\psi, R'\}/2 = \partial_\eta \psi$.  From~\eqref{z:fRkTtRxHhvuE} and
substituting $t$ by $\eta$, we see that the optimal local-sensing observable for
$\eta$ is simply $\eta_0 + H - \langle {H}\rangle  = H$.  That is, the optimal measurement
distinguishing $\lvert {\psi(\eta_0)}\rangle $ from $\lvert {\psi(\eta_0+d\eta)}\rangle $ is the
Hamiltonian $H$ itself, up to a term $P_\psi^\perp {M} P_\psi^\perp$ for any
Hermitian ${M}$.  (Alternatively, the same conclusion would have been reached had
we started from~\eqref{z:w7NUzrfGUNE-} with $t,H$ replaced by
$\eta,-T$. A more detailed computation is provided in
\cref{z:eQl8.oS.4.na}.)
Therefore, the parameter $\eta$ describes an evolution along which, locally
around $\lvert {\psi}\rangle $, we have $\eta_0 + d\eta = \langle {H}\rangle _{\psi(\eta_0+d\eta)}$.  In
this sense, $\eta$ represents the energy of the probe $\lvert {\psi(\eta)}\rangle $ locally
around $\eta_0$.

To summarize, the evolution
of $\psi(t)=\lvert {\psi(t)}\rangle \langle {\psi(t)}\rvert $
is generated by the Hamiltonian $H$;
nearby states $\psi(t_0)$ and $\psi(t_0+dt)$ are optimally distinguished by a
local time-sensing observable $T$.  The complementary evolution $\psi(\eta)$ is
one that inverts the roles of $H$ and $T$: The evolution $\psi(\eta)$ is
generated by $T$, and $H$ is the operator that optimally distinguishes
neighboring states $\psi(\eta_0)$ and $\psi(\eta_0+d\eta)$.

\paragraph{Single-qubit example.}
Consider a qubit initialized in the state  vector $\lvert {\psi_{\mathrm{init}}}\rangle  = \lvert {+}\rangle $,
where $\lvert {\pm }\rangle = [\lvert {\uparrow}\rangle \pm\lvert {\downarrow}\rangle ]/\sqrt{2}$, and let the qubit
evolve according to the Hamiltonian %
{$H=\omega Z/2$}
(i.e., Alice's system
in \cref{z:duwTyi1jsdbn}).  The time evolution of the clock is
given by $\lvert {\psi(t)}\rangle  = U_{t}\lvert {+}\rangle $, where $U_t = {e}^{-iHt}$; we see that
\begin{align}
  \lvert {\psi(t)}\rangle 
  &= \frac1{\sqrt2} \bigl[{e}^{-\frac{i\omega t}{2}}\lvert {\uparrow }\rangle +
  {e}^{\frac{i\omega t}{2}}\lvert {\downarrow}\rangle \bigr]
  \nonumber\\
  &= \cos\Bigl(\frac{\omega t}{2}\Bigr)\,\lvert {+}\rangle 
    - i\sin\Bigl(\frac{\omega t}{2}\Bigr)\,\lvert {-}\rangle \ .
    \label{z:UlWOzs-evrpA}
\end{align}
It is also convenient to note that
\begin{subequations}
  \label{z:1--J4q44Q34a}
  \begin{align}
    \psi(t)
    &= U_{t} \lvert {+}\rangle \mkern -1.8mu\relax \langle{+}\rvert  U_{t}^\dagger
      = \frac{\mathds{1}}2 + \frac12 U_{t} 
      {X}
      U_{t}^\dagger
      \label{z:P9NAwniVN.Vm}
    \\
    &= \frac{\mathds{1}}2 + \frac12\bigl[\cos(\omega t)\,
{X}
      + \sin(\omega t)\,
      {Y}
      \bigr]\ ,
      \label{z:6UjLB7XiW4LY}
  \end{align}
\end{subequations}
using the identity $\lvert {+}\rangle \mkern -1.8mu\relax \langle{+}\rvert  = [\mathds{1}+
{X}]/2$ along with
${e}^{-i a
{Z}
} 
{X}
{e}^{i a
{Z} } %
= \cos(2a)\, 
{X}
+ \sin(2a)\, 
{Y}$.
The time derivative of the state is
\begin{subequations}
  \label{z:GM3ZeRQLM9mt}
  \begin{align}
    \partial_t\psi(t)
    &= -i[H,\psi(t)]
    = -\frac{i\omega}{2}\,\bigl[
    {Z}
    \,,\, U_t \lvert {+}\rangle \mkern -1.8mu\relax \langle{+}\rvert  U_t^\dagger \bigr]
    \nonumber\\
    &= -\frac{i\omega}{2}\,U_t \Bigl[
    {Z}
    \,,\, \frac{\mathds{1}+
    {X}
    }{2} \Bigr] U_t^\dagger
      \nonumber\\
    &= \frac{\omega}{2}\,U_t 
    {Y} 
    U_t^\dagger
      \label{z:SFi9RkJisbzH}
    \\
    &= \frac\omega2 \bigl[\cos(\omega t)\,
    {Y}
    -
      \sin(\omega t)\,
      {X}\bigr]\ ,
      \label{z:.uzSQ.DqcFF7}
  \end{align}
\end{subequations}
using ${e}^{-i a
{Z} } 
{Y}
{e}^{i a
{Z} } %
= \cos(2a)\,
{Y} - \sin(2a)\,
{X} $.
The expressions~\eqref{z:P9NAwniVN.Vm}
and~\eqref{z:SFi9RkJisbzH} %
manifest the fact that
the state and the derivative evolve in time by rotation around the $Z$ axis of
the Bloch sphere, whereas we can read out from the
expressions~\eqref{z:6UjLB7XiW4LY}
and~\eqref{z:.uzSQ.DqcFF7} the information about the
time evolution of the components of the Bloch vector.
The average energy $\langle {H}\rangle _{\psi(t)}$ is
\begin{align}
  \langle {H}\rangle _{\psi(t)} = \operatorname{tr}\Bigl[U_t\lvert {+}\rangle \mkern -1.8mu\relax \langle{+}\rvert U_t^\dagger\,\frac{\omega
  {Z} }{2}\Bigr]
  = 0\ 
  \label{z:u6Va242CzB5I}
\end{align}
for all $t$, noting that $U_t$ commutes with %
${Z}$
and that
$\langle {+}\mkern 1.5mu\relax \vert \mkern 1.5mu\relax {
{Z}
}\mkern 1.5mu\relax \vert \mkern 1.5mu\relax {+}\rangle =0$.
The energy's standard deviation at time $t$ is then
\begin{align}
  \sigma_H^2
  = \bigl \langle {H^2}\bigr \rangle _{\psi(t)}
  = \frac{\omega^2}{4}\ ,
  \label{z:jQs1wEnWl.ji}
\end{align}
noting that 
{$Z^2=\mathds{1}$}.

We now compute the time sensitivity and the optimal time-sensing observable
locally around a given time $t_0$.  We write $\lvert {\psi }\rangle = \lvert {\psi(t_0)}\rangle $ for
short.  
The optimal time-sensing observable
is given by~\eqref{z:w7NUzrfGUNE-}, which we can
compute as (ignoring the degree of freedom $P_\psi^\perp {M}  P_\psi^\perp$),
\begin{align}
  T - t_0
  &= \frac1{2\sigma_H^2}\,\partial_t\psi\,(t_0)
  = \frac1{\omega} \, U_{t_0} 
  {Y}
  U_{t_0}^\dagger
  \nonumber\\
  &= \frac1\omega \bigl[ \cos(\omega t_0)\,
  {Y}
  -
     \sin(\omega t_0)\,
     {X}\bigr]\ .
  \label{z:tI5wBCGC-2n.}
\end{align}
The optimal sensing observable $T$ is therefore aligned with the direction on
the Bloch sphere that is tangent to the state's evolution.

We now determine the parameter $\eta$.  It is generated by $T$ as
per~\eqref{z:PlQtfu9raUmQ}, and we can compute the associated derivative
using~\eqref{z:SMjFM6HgZlts} as
\begin{align}
  \{ H - \langle {H}\rangle , \psi \}
  &= \frac\omega2 \Bigl\{
  {Z}
  , \frac{\mathds{1}+
  {X}
  }{2} \Bigr\}
    = \frac\omega2 
    {Z}
    \ ,
    \label{z:y11JcdWGvpsS}
  \\
  \partial_\eta \psi
  &
  = \frac{1}{2\sigma_H^2} \frac\omega2 
  {Z}
  = \frac{
  {Z}
}{\omega}\ ,
  \label{z:xVN3EzcV.KFg}
\end{align}
recalling that the Pauli matrices along different directions anticommute.  The
direction associated with the $\eta$ parameter is aligned with the $Z$ axis of
the Bloch sphere (\cref{z:duwTyi1jsdbn}) which is the
direction in which the Hamiltonian is oriented.

We can now compute the sensitivities with respect to $t$ and $\eta$
using~\eqref{z:BBGAvQkdkmd8} and~\eqref{z:lN.pbrjbNave} as
\begin{align}
  \FIqty{Alice}{t}
  &= 4\sigma_H^2 = \omega^2\ ,
  &
    \FIqty{Alice}{\eta}
  &= \frac1{\sigma_H^2}
    = \frac4{\omega^2}\ .
    \label{z:X6ybg8GSCbDf}
\end{align}
Finally, we can check that $H$ is an optimal local-sensing observable for
$\eta$.  First observe with $\eta_0 = \langle {H}\rangle _\psi = 0$ that
\begin{align}
  \langle { H }\rangle  _{\psi(t_0,\eta_0+d\eta)}
  &= d\eta \operatorname{tr}\bigl(H\, \partial_\eta \psi\bigr)
  = d\eta\ ,
\end{align}
using~\eqref{z:xVN3EzcV.KFg} along with 
{$Z^2=\mathds{1}$}.
Hence, $H$ satisfies the
condition~\eqref{z:c7.VBMRSmG0h} for the
parameter $\eta$.
The variance of this observable was computed above as
\begin{align}
  \bigl \langle {(H-\langle {H}\rangle )^2}\bigr \rangle 
  = \sigma_H^2
  = \frac{\omega^2}{4} = \frac1{\FIqty{Alice}{\eta}}\ ,
\end{align}
and therefore $H$ also saturates the Cram\'er-Rao bound. 
It is an optimal local
sensing observable.

\subsection{The noisy channel and the environment}

\paragraph{The noisy clock.}
Suppose that Alice sends the clock from its noiseless environment to a receiver
Bob through a noisy channel $\mathcal{N}_{A\to B}$
(\cref{z:FnHui0ahNLxW}).  Bob has access to the noisy clock state
\begin{align}
  \rho_B(t) = \mathcal{N}_{A\to B}( \psi(t) )\ .
\end{align}
We consider the sensitivity of Bob's clock locally around $t_0$, i.e., we ask
how well Bob can distinguish $\rho_B(t_0)$ from $\rho_B(t_0 + dt)$.  We assume
that the noisy channel $\mathcal{N}_{A\to B}$ does not depend on $t$.
This setting is nonstandard in the context of quantum metrology.  Usually, one
considers a quantum clock that is exposed to continuous noise as it evolves in
time instead of the noise being applied separately and instantaneously after the
system has evolved unitarily for a given amount of time.
This alternative setting represents the situation where Alice would like to send
a quantum reference frame to Bob over a noisy
channel~\cite{R28}.  We defer the discussion of the
connections between these two settings to \cref{z:Xuq8JaOEWrUo}.

Locally around $t_0$, Bob's optimal sensitivity is given via the Cram\'er-Rao
bound~\eqref{z:5-HFHQQdFwEu} by Bob's Fisher information with respect to
time,
\begin{align}
  \FIqty{Bob}{t} := \Ftwo{\rho_B(t_0)}{\partial_t \rho_B\,(t_0)}\ .
\end{align}
We may furthermore express $\rho_B$ and $\partial_t \rho_B$ as
\begin{align}
  \rho_B
  &= \mathcal{N}_{A\to B}( \psi )\ ,
  \nonumber\\
  \partial_t \rho_B
  &= \mathcal{N}_{A\to B}\bigl( \partial_t \psi \bigr)
    =  \mathcal{N}_{A\to B}\bigl( -i[ H, \psi ] \bigr)\ .
    \label{z:bs.I9ZrrQvK2}
\end{align}
Determining $\FIqty{Bob}{t}$ in principle requires the usage of a general
expression of the Fisher information for mixed states such
as~\eqref{z:.eb-akuquBNd} or~\eqref{z:r7GWSaRrLfNh},
which can be significantly more cumbersome to manipulate as opposed to computing
the variance of the Hamiltonian in the case of a pure-state evolution.

\paragraph{The environment.}
Any %
quantum channel $\mathcal{N}_{A\to B}$ can be expressed as a unitary evolution
over a larger system, where the environment is initialized in a pure state.
This construction is known as a Stinespring dilation.  The initial pure state of
the environment can be contracted with the global unitary to give a more concise
description of the Stinespring dilation in terms of an isometry $A\to BE$.  More
precisely, any quantum channel $\mathcal{N}_{A\to B}$ can be written as
\begin{align}
  \mathcal{N}_{A\to B}(\cdot) = \operatorname{tr}_{E}\bigl( V_{A\to BE} \, (\cdot) \, V^\dagger \bigr)\ ,
  \label{z:1i4Ojzjjs9zf}
\end{align}
where $E$ is a suitable environment system, and where $V_{A\to BE}$ is an
isometry mapping states of $A$ into $B\otimes E$.

The system $E$, which we call Eve, represents the quantum information that is
discarded by the channel $\mathcal{N}_{A\to B}$.  Instead, we can consider a quantum channel that describes what
Eve gets if Bob's system $B$ is discarded.  By tracing out $B$ instead of
$E$ in~\eqref{z:1i4Ojzjjs9zf} we obtain the \emph{complementary
  channel},
\begin{align}
   \widehat{\mathcal{N}}_{A\to E}(\cdot)
  = \operatorname{tr}_{B}\bigl( V_{A\to BE} \, (\cdot) \, V^\dagger \bigr)\ .
  \label{z:L-vRq0ojGGDS}
\end{align}
If we write the noisy channel in an operator-sum representation with Kraus
operators $\{ E_k \}$ as
\begin{align}
  \mathcal{N}(\cdot) = \sum_k E_k (\cdot) E_k^\dagger\ ,
  \label{z:jg7XnoGzLtS0}
\end{align}
we may write a corresponding complementary channel as
\begin{align}
  \widehat{\mathcal{N}}(\cdot)
  = \sum_{k,k'} \operatorname{tr}\bigl(E_{k'}^\dagger E_k (\cdot)\bigr)\, \lvert {k}\rangle \mkern -1.8mu\relax \langle{k'}\rvert _E\ ,
  \label{z:SiT0KpCikk0d}
\end{align}
for some orthonormal basis $\{ \lvert {k}\rangle  \}$ on the environment system $E$.
The complementary channel is unique up to a partial isometry on the environment
system.

Our main result involves Eve's sensitivity to the $\eta$ parameter of the state
that she obtains if Alice's quantum clock is sent to $E$ via the complementary channel.
Namely, we define
\begin{align}
  \rho_E(\eta) = \widehat{\mathcal{N}}_{A\to E}(\psi(\eta))\ .
\end{align}
Recalling~\eqref{z:SMjFM6HgZlts}, we have
\begin{align}
  \partial_\eta \rho_E
  = \frac1{2\sigma_H^2} \widehat{\mathcal{N}}_{A\to E}( \{ H - \langle {H}\rangle , \psi \} )\ .
  \label{z:aTgSsruxgou5}
\end{align}
As for $\psi$, $\partial_t \psi$, and $\partial_\eta \psi$, the states $\rho_B$,
$\rho_E$ and the derivatives $\partial_t \rho_B$, $\partial_\eta \rho_E$ are
implicitly evaluated at $(t_0,\eta_0)$ unless specified otherwise.  We also
abbreviate $\mathcal{N}_{A\to B}$ and $\widehat{\mathcal{N}}_{A\to B}$ by
$\mathcal{N}$ and $\widehat{\mathcal{N}}$ for convenience and whenever it is
unambiguous to do so.

\section{Bipartite uncertainty relation for the Fisher information}
\label{z:6pGtSPe4CZWB}

\subsection{Equality Fisher information trade-off for time and energy
  and expression for sensitivity loss}
Sending Alice's clock to Bob through the noisy channel $\mathcal{N}_{A\to B}$
reduces the clock's sensitivity to the time parameter $t$. On the other hand,
sending the clock to Eve through the complementary channel
$\widehat{\mathcal{N}}_{A\to E}$
enables Eve to gain sensitivity with respect to the energy parameter $\eta$.
Our main result
characterizes how these two effects are related:
\begin{theorem}[Bipartite time-energy uncertainty relation]
  \label{z:T-KNuYhGRM7I}
  \noproofref%
  Suppose Alice prepares a probe in a quantum state vector $\lvert {\psi}\rangle $ and consider
  the local parameters $t$, $\eta$ defining directions in state space generated
  by $H$ and $T$ and centered at $\lvert {\psi }\rangle = \lvert {\psi(t_0,\eta_0)}\rangle $ as in
  \cref{z:ff7jHgn75DW4}.
  Alice sends her probe to Bob through a channel $\mathcal{N}_{A\to B}$; let Eve
  represent the output of the corresponding complementary channel
  $\widehat{\mathcal{N}}_{A\to E}$ (see \cref{z:FnHui0ahNLxW}).  %
  Then
  \begin{align}
    \frac{\FIqty{Bob}{t}}{\FIqty{Alice}{t}} +
    \frac{\FIqty{Eve}{\eta}}{\FIqty{Alice}{\eta}}
    = 1\ ,
    \label{z:cmnZZu5eQ8hv}
  \end{align}
  provided the rank of $\mathcal{N}(\psi(t))$ does not change at $t_0$.
\end{theorem}
Recalling \cref{z:BBGAvQkdkmd8,z:lN.pbrjbNave}, our uncertainty relation is
equivalently stated as
\begin{align}
  \frac{\FIqty{Bob}{t}}{4\sigma_H^2} +
  \sigma_H^2\,\FIqty{Eve}{\eta}
  = 1\ .
  \label{z:xG6mWjjekeNN}
\end{align}
Using the Cram\'er-Rao bound~\eqref{z:5-HFHQQdFwEu} we can relate the
optimal sensing accuracies $\langle { \delta t_{\mathrm{Bob,est}}^2 }\rangle $,
$\langle { \delta \eta_{\mathrm{Eve,est}}^2 }\rangle $ associated with the parameters $t$,
$\eta$ on Bob and Eve's systems,
\begin{align}
  \frac{1}{4\sigma_H^2} \frac{1}{\langle { \delta t_{\mathrm{Bob,est}} ^2 }\rangle }
  +
  \frac{1}{4\sigma_T^2} \frac{1}{\langle { \delta \eta_{\mathrm{Eve,est}} ^2 }\rangle }
  \leq 1\ ,
\end{align}
noting that equality can be achieved with sensing strategies that saturate the
Cram\'er-Rao bound provided the rank of $\mathcal{N}(\psi(t))$ does not change
at $t=t_0$.

A proof of \cref{z:T-KNuYhGRM7I} proceeds by writing the
Fisher information on Bob's end, i.e.\@ after the application of the noise
channel, in terms of the Bures metric.  The environment Eve is introduced as the
purifying space over which the fidelity is computed via Uhlmann's theorem.  The
resulting expression is expressed as a semidefinite program as in
Refs.~\cite{R29,R30}; suitably manipulating the
corresponding dual problem yields the
relation~\eqref{z:cmnZZu5eQ8hv}.  The full proof is
provided in \cref{z:.B4MZrnrq.9S}.
We also provide an alternative proof using a semidefinite
characterization of the Fisher information.

The condition the rank of $\mathcal{N}(\psi(t))$ does not change locally at the
time $t_0$ ensures that we avoid edge cases where the correspondence between the
Fisher information and the Bures metric is
incomplete~\cite{R31,R32,R33}.  In edge cases where
this condition is violated, the uncertainty
relation~\eqref{z:cmnZZu5eQ8hv} can be shown to
hold as an inequality instead of an equality (see below and
\cref{z:.B4MZrnrq.9S}).
The no rank change condition is typically associated with situations where the
quantum Fisher information is discontinuous.  In such cases its operational
relevance can be questioned; we further discuss these points below in the
context of independent and identically distributed (i.i.d.\@) noise as well as in
\cref{z:z.6FF4fEKkkd}.

The condition that the rank of $\mathcal{N}(\psi(t))$ does not change at
$t=t_0$ is formalized by requiring that for any eigenvalue $p_k(t)$ of
$\mathcal{N}(\psi(t))$ such that $p_k(t_0)=0$ we also have
$\partial_t^2 p_k(t_0) = 0$.  This more precise formulation is the form of the
assumption that is used in the proof.  Observe that any eigenvalue $p_k(t)$ of
$\mathcal{N}(\psi(t))$ that satisfies $p_k(t_0) = 0$ necessarily also satisfies
$\partial_t p_k\,(t_0) = 0$, since the value zero is necessarily a minimum for
$p_k(t)$.

Another equivalent form of our uncertainty
relation~\eqref{z:cmnZZu5eQ8hv} is one that
quantifies directly the difference between the sensitivity of the noiseless
clock and the resulting sensitivity on Bob's end.  Let us define:
\begin{align}
  \FIloss{Bob}{t}
  = \FIqty{Alice}{t} - \FIqty{Bob}{t} = 4\sigma_H^2 - \FIqty{Bob}{t}\ .
  \label{z:QZajZctWEncO}
\end{align}
A few simple algebraic manipulations
of~\eqref{z:cmnZZu5eQ8hv} lead to
\begin{align}
  \FIqty{Alice}{t} - \FIqty{Bob}{t} =
  \frac{\FIqty{Alice}{t}}{\FIqty{Alice}{\eta}} \FIqty{Eve}{\eta}\ ,
\end{align}
which gives us an expression for $\FIloss{Bob}{t}$.  We can further spell out
this expression using~\cref{z:lN.pbrjbNave,z:BBGAvQkdkmd8} along with simple scaling
properties that follow from the definition of the Fisher information to find
\begin{align}
  \FIloss{Bob}{t}
  &= (2\sigma_H^2)^2 \FIqty{Eve}{\eta}
  = \Ftwo{ \rho_E }{ 2\sigma_H^2 \partial_\eta \rho_E }
    \nonumber\\
  &= \Ftwo{ \rho_E }{ \widehat{\mathcal{N}}(\{\bar{H},\psi\})  }\ ,
\end{align}
where we have used~\eqref{z:aTgSsruxgou5} in the last equality.

Summarizing the above argument, we obtain an alternative form of
\cref{z:T-KNuYhGRM7I} as an expression for the sensitivity loss $\FIloss{Bob}{t}$ in terms of the Fisher information that Eve
obtains with respect to a direction associated with the anticommutator of $H$
and $\psi$.

\begin{corollary}[Expression for Bob's sensitivity loss via Eve]
  \noproofref
  \label{z:dq0nj3WKijUb}
  Consider the setting of \cref{z:T-KNuYhGRM7I} and assume that
  the rank of $\mathcal{N}(\psi(t))$ does not change locally at
  $t_0$.  Then
  \begin{align}
    \FIloss{Bob}{t}
    =
    \Ftwo{ \widehat{\mathcal{N}}(\psi) }{ \widehat{\mathcal{N}}(\{ \bar{H}, \psi \}) }\ ,
    \label{z:4bic5ys.wcEe}
  \end{align}
  where $\FIloss{Bob}{t} = \FIqty{Alice}{t} - \FIqty{Bob}{t}$ and where we
  recall the shorthand $\bar{H} = H - \langle {H}\rangle $.  As a consequence,
  \begin{align}
    \FIqty{Bob}{t} =
    4\sigma_H^2
    - \Ftwo{ \widehat{\mathcal{N}}(\psi) }{ \widehat{\mathcal{N}}(\{ \bar{H}, \psi \}) }\ .
    \label{z:BtvntDGWvOI3}
  \end{align}
\end{corollary}

Two extreme cases can readily be identified. One is where there is no noise and
$\mathcal{N}={\mathrm{id}}$ is the identity process; in this case the
complementary channel is a channel that outputs a constant state regardless of
the input, $\widehat{\mathcal{N}}(\cdot) = \operatorname{tr}(\cdot)\,\tau_E$ for some state
$\tau_E$.  In this case Eve obtains no information about the probe's energy,
which can be seen in our formalism by the fact that
$\widehat{\mathcal{N}}(\{\bar{H},\psi\}) = 0$ and therefore $\FIloss{Bob}{t} = 0$.  In the
opposite extreme case, the noise destroys its input entirely and sends it to the
environment, with correspondingly $\widehat{\mathcal{N}} = {\mathrm{id}}$.  In
this case Eve has maximal sensitivity to the $\eta$ parameter,
$\FIqty{Eve}{\eta} = \FIqty{Alice}{\eta}$, and therefore
$\FIqty{Bob}{t} = 0$ and $\FIloss{Bob}{t} =4\sigma_H^2$.

\subsection{Trade-off relation in terms of a virtual qubit}
\label{z:aLDqVRGU3OD.}
In this section we simplify the setting required to produce the relation in
\cref{z:T-KNuYhGRM7I}, in an effort to identify the
fundamental concepts required for our uncertainty relation to hold.  It turns
out that \cref{z:T-KNuYhGRM7I} can be rephrased as an
uncertainty relation between Bob and Eve distinguishing states respectively
along the $Y$ and $Z$ Pauli directions of a virtual qubit space, which in the
setting of \cref{z:T-KNuYhGRM7I} is defined by the clock
state vector $\lvert {\psi}\rangle $ and its image $H\lvert {\psi}\rangle $ under application of the
Hamiltonian. %

Consider the subspace of Alice's Hilbert space spanned by the probe state
$\lvert {\psi}\rangle $ and its time derivative $\propto H \lvert {\psi}\rangle $. This subspace defines a
virtual qubit. We choose to identify the probe state with the $+1$ Pauli-$X$
eigenvector.  It turns out that our uncertainty relation admits a restatement as
a relation between the sensitivity that Bob and Eve can achieve with respect to
Pauli-$Y$ and logical Pauli-$Z$ directions of the virtual qubit.  More
precisely, we first define
\begin{align}\label{z:bbVoi4MPh36i}
  \lvert {\xi }\rangle = P_\psi^\perp H \lvert {\psi
  }\rangle = \bigl(H-\langle {H}\rangle \bigr)\,\lvert {\psi
  }\rangle = \bar{H} \lvert {\psi}\rangle \ ,
\end{align}
recalling $P_\psi^\perp = \mathds{1}- \psi$.  The norm of $\lvert {\xi}\rangle $ satisfies
\begin{align}
  \bigl \lVert {\lvert {\xi}\rangle }\bigr \rVert ^2 = \langle {\xi}\mkern 1.5mu\relax \vert\mkern 1.5mu\relax {\xi }\rangle = \sigma_H^2\ .
\end{align}
Here, we assume that $\lvert {\xi}\rangle \neq0$, otherwise the probe does not evolve in time
and all the terms in our uncertainty relation are trivially zero.  We can
write the following derivatives in terms of $\lvert {\xi}\rangle $,
\begin{subequations}
\label{z:kz9fnxT79VEp}
\begin{align}
    \partial_t\psi = -i [\bar H, \psi] &= -i\bigl(\lvert {\xi}\rangle \mkern -1.8mu\relax \langle{\psi }\rvert - \lvert {\psi}\rangle \mkern -1.8mu\relax \langle{\xi}\rvert \bigr)\ ,
  \\
  2\sigma_H^2\,\partial_\eta\psi = \{ \bar{H}, \psi \} &= \lvert {\xi}\rangle \mkern -1.8mu\relax \langle{\psi }\rvert + \lvert {\psi}\rangle \mkern -1.8mu\relax \langle{\xi}\rvert \ .
\end{align}
\end{subequations}
An orthonormal basis of the virtual qubit can be chosen as
\begin{align}
  \lvert {+}\rangle _L &= \lvert {\psi}\rangle \ ,
  &
  \lvert {-}\rangle _L &= \frac1{\sigma_H}\, \lvert {\xi }\rangle \ .
 \label{z:kuhRxuqe47Oq}
\end{align}
As the logical computational basis of the virtual qubit, we choose
\begin{align}
  \lvert {0}\rangle _L &= \frac1{\sqrt 2}\bigl[\lvert {+}\rangle _L + \lvert {-}\rangle _L\bigr]\ ,
  &
  \lvert {1}\rangle _L &= \frac1{\sqrt 2}\bigl[\lvert {+}\rangle _L - \lvert {-}\rangle _L\bigr]\ .
 \label{z:Usr3YV.YLN3v}
\end{align}
This choice of basis is motivated to match the qubit operators of a single
spin-$1/2$ particle prepared in the $+X$ eigenstate and evolving according to a
magnetic field pointing along the $Z$ axis.

Consider the logical Pauli-$X$, $Y$ and $Z$ operators defined as usual with
respect to the basis~\eqref{z:Usr3YV.YLN3v}.  They are
expressed in the $\lvert {\pm}\rangle _L$ basis as
\begin{subequations}
  \label{z:puhbpJK6ngeA}
  \begin{align}
    X_L &= \lvert {+}\rangle \mkern -1.8mu\relax \langle{+}\rvert _L - \lvert {-}\rangle \mkern -1.8mu\relax \langle{-}\rvert _L\ ,
    \\
    Y_L &= -i\lvert {-}\rangle \mkern -1.8mu\relax \langle{+}\rvert _L + i \lvert {+}\rangle \mkern -1.8mu\relax \langle{-}\rvert _L\ ,
    \\
    Z_L &= \lvert {-}\rangle \mkern -1.8mu\relax \langle{+}\rvert _L + \lvert {+}\rangle \mkern -1.8mu\relax \langle{-}\rvert _L\ ,
  \end{align}
\end{subequations}
with furthermore
\begin{align}
  Y_L &= \frac{ -i }{\sigma_H} [ \bar{H}, \psi]\ ,
  &
    Z_L &= \frac1{\sigma_H} \{\bar{H}, \psi\}\ .
          \label{z:5ChzchOJ.MLJ}
\end{align}
We see that the logical Pauli-$Y$ and Pauli-$Z$ operators are parallel to the
evolution respectively along $t$ and along $\eta$ locally at
$\lvert {\psi }\rangle = \lvert {\psi(t_0,\eta_0)}\rangle $, as we
recall~\eqref{z:kz9fnxT79VEp}.
Our uncertainty relation can be stated in terms of a metrological logical qubit
as follows.
\begin{theorem}[Uncertainty relation for the metrological logical qubit]
  \label{z:Ircbg89zc.Wy}%
  \noproofref Let $A$, $B$ and $E$ be finite-dimensional quantum systems.  Let
  $\mathcal{N}_{A\to B}$ be a completely positive, trace nonincreasing map.  Let
  $V_{A\to BE}$ be such that
  $\mathcal{N}_{A\to B}(\cdot) = \operatorname{tr}_E\bigl( V(\cdot)V^\dagger\bigr)$ and
  $V^\dagger V \leq \mathds{1}$, i.e., $V$ is a Stinespring dilation of
  $\mathcal{N}$.  Consider the complementary channel
  $\widehat{\mathcal{N}}_{A\to E}(\cdot) = \operatorname{tr}_B\bigl( V(\cdot)V^\dagger\bigr)$.
  Let $\lvert {\pm}\rangle _L$ be any two orthogonal and normalized state vectors on system $A$, and let
  $X_L, Y_L, Z_L$ be defined via~\eqref{z:puhbpJK6ngeA}.
  If $(P_{\rho_B}^\perp \otimes P_{\rho_E}^\perp) V\lvert {-}\rangle _L = 0$,
  then
  \begin{multline}
    \Ftwo{\mathcal{N}(\psi)}{\mathcal{N}(Y_L)}
    + \Ftwo{\widehat{\mathcal{N}}(\psi)}{\widehat{\mathcal{N}}(Z_L)}
    \\
    = 4\langle {-}\mkern 1.5mu\relax \vert \mkern 1.5mu\relax {\mathcal{N}^\dagger(\mathds{1})}\mkern 1.5mu\relax \vert \mkern 1.5mu\relax {-}\rangle _L\ .
    \label{z:RqzWcDReqJqy}
  \end{multline}
  If $(P_{\rho_B}^\perp \otimes P_{\rho_E}^\perp) V\lvert {-}\rangle _L \neq 0$,
  then we have the inequality
  \begin{multline}
    \Ftwo{\mathcal{N}(\psi)}{\mathcal{N}(Y_L)}
    + \Ftwo{\widehat{\mathcal{N}}(\psi)}{\widehat{\mathcal{N}}(Z_L)}
    \\
    \leq 4\langle {-}\mkern 1.5mu\relax \vert \mkern 1.5mu\relax {\mathcal{N}^\dagger(\mathds{1})}\mkern 1.5mu\relax \vert \mkern 1.5mu\relax {-}\rangle _L\ .
    \label{z:rcObY1pi1VeY}
  \end{multline}
\end{theorem}

The above theorem provides a more formal statement that generalizes the earlier
statement \cref{z:T-KNuYhGRM7I} to trace-nonincreasing
maps and to subnormalized states.  The metrological qubit construction also
provides a clearer mathematical picture of the symmetric role of Bob and Eve
in our uncertainty relation: Bob and Eve can be interchanged (i.e.,
$\mathcal{N} \leftrightarrow \widehat{\mathcal{N}}$) provided we correspondingly
interchange $\lvert {\xi }\rangle \leftrightarrow i\lvert {\xi}\rangle $.  For a state vector $\lvert {\psi}\rangle $
evolving with respect to a Hamiltonian $\bar{H}$, the state
$\lvert {\xi}\rangle = \bar {H} \lvert {\psi}\rangle $ is the derivative of $\lvert {\psi}\rangle $ with respect to
time, and $i \lvert {\xi}\rangle $ can be thought of the derivative of $\lvert {\psi}\rangle $ with
respect to \emph{imaginary time}. The symmetry in
\cref{z:Ircbg89zc.Wy} between Bob and Eve, which
involves the interchange $\lvert {\xi }\rangle \leftrightarrow i \lvert {\xi}\rangle $, is reproduced at
the level of the parameters $t$ and $\eta$ by choosing $\eta$ to parametrize the
one-family parameter of state vectors $\lvert {\psi(\eta)}\rangle $ in
\eqref{z:PlQtfu9raUmQ} governed by the imaginary-time evolution
\eqref{z:BocUFnzQodPW}.
The full proof of \cref{z:Ircbg89zc.Wy} is provided in
\cref{z:VM5-jRd.0oMe}.

In \cref{z:Ircbg89zc.Wy} a different condition is
stated for equality as in \cref{z:T-KNuYhGRM7I}, where we
require the rank of $\mathcal{N}(\psi(t))$ not to change.  These conditions
turn out to be equivalent, as shown in the following proposition.  We defer
the proof to \cref{z:VM5-jRd.0oMe}.

\begin{proposition}[Conditions for equality in the uncertainty relation]
  \noproofref
  \label{z:3JJW1LK1MXi9}
  Let $\{ E_{k} \}$ be a set of Kraus operators for $\mathcal{N}_{A\to B}$ and
  $V_{A\to BE}$ be a Stinespring dilation of $\mathcal{N}$.  The following
  conditions are equivalent:
  \begin{itemize}
  \item $(P_{\rho_B}^{\perp}\otimes P_{\rho_E}^\perp)\, V\lvert {\xi }\rangle = 0$;
  \item $\rho_B(t)$ does not change rank as a function of $t$ locally at the
    point $t_0$;
  \item For any linear combination $E = \sum c_k E_k$ (with $c_k\in\mathbb{C}$)
    such that $E\lvert {\psi }\rangle = 0$, then $P_{\rho_B}^\perp E \lvert {\xi }\rangle = 0$.
  \end{itemize}
\end{proposition}

In particular, it suffices that either $\rho_B = \mathcal{N}(\psi)$ or
$\rho_E = \widehat{\mathcal{N}}(\psi)$ has full rank to ensure that these
conditions are satisfied, and thereby that our uncertainty relation holds with
equality [\cref{z:RqzWcDReqJqy}].

As a consequence, the situations for which the
conditions~\eqref{z:3JJW1LK1MXi9} do not hold,
and correspondingly for which our uncertainty relation does not necessarily hold
with equality, are edge cases that can be infinitesimally perturbed into
situations where the corresponding conditions hold.  Indeed, one can mix
$\mathcal{N}$ with an infinitesimal amount of depolarizing noise to ensure that
Bob's state is full rank, and therefore to ensure that equality holds in our
uncertainty relation.

\subsection{General uncertainty relation for any two parameters}
\label{z:HCzmF57VVd0m}

The uncertainty relation between position and momentum can be generalized to any
arbitrary pair of observables.  The Robertson uncertainty
relation states that for any two observables $A,B$, we have
\begin{align}
  \sigma_A \, \sigma_B \geq \frac12 \bigl \lvert { \bigl \langle { i[A,B] }\bigr \rangle  }\bigr \rvert \ .
  \label{z:0pRecQmVyaWF}
\end{align}
In the same spirit, we derive a generalization
of~\eqref{z:cmnZZu5eQ8hv} that is valid for any two
observables. %
Suppose Alice prepares a pure state $\psi$ that can evolve along two possible
directions $\partial_a \psi$ and $\partial_b\psi$, and sends the state through
the noisy channel $\mathcal{N}$ to Bob as in \cref{z:FnHui0ahNLxW}.
We assume that the directions along $a,b$ are generated by two Hermitian
operators $A,B$  acting on $\psi = \lvert {\psi}\rangle \mkern -1.8mu\relax \langle{\psi}\rvert $ as
\begin{align}
  \partial_a\psi &= -i[A, \psi]\ ,
  &
    \partial_b\psi &= -i[B, \psi] \ .
  \label{z:uIClTv5cLFGk}
\end{align}
Bob is tasked with estimating a deviation locally to first order around
$\mathcal{N}(\psi)$ in the $a$ direction, whereas Eve tries to distinguish
$\widehat{\mathcal{N}}(\psi)$ from neighboring states along the $b$ direction.
The parameters $a,b$ are analogous to the parameters $t,\eta$ considered above,
but the two directions $\partial_a \psi,\partial_b \psi$ can be arbitrary.  %
\begin{theorem}[Bipartite uncertainty relation for any two parameters]
  \label{z:incm1nGt-mIE}
  \noproofref Let $\psi$ be a 
  state vector and suppose that $A,B$ are two
  Hermitian operators that generate evolutions locally at $\psi$ in directions
  $\partial_a\psi,\partial_b\psi$
  via~\eqref{z:uIClTv5cLFGk}.  Suppose we apply a
  noisy channel as depicted in \cref{z:FnHui0ahNLxW}.  Then
  \begin{align}
    \frac{ \FIqty{Bob}{a} }{ \FIqty{Alice}{a} }
    + \frac{ \FIqty{Eve}{b} }{ \FIqty{Alice}{b} }
    \leq
    1 + 2 \sqrt{ 1 - \frac{ \bigl \langle {i[A,B]}\bigr \rangle ^2 }{ 4\sigma_A^2\sigma_B^2 } }\ ,
    \label{z:4438TdZJTrfX}
  \end{align}
  where
  \begin{align}
    \FIqty{Alice}{a} &= \Ftwo{\psi}{\partial_a \psi}\ ,
    &
    \FIqty{Bob}{a} &= \Ftwo{\mathcal{N}(\psi)}{\mathcal{N}(\partial_a \psi)}\ ,
    \nonumber\\
    \FIqty{Alice}{b} &= \Ftwo{\psi}{\partial_b \psi}\ ,
    &
    \FIqty{Eve}{b}
    &= \Ftwo{\widehat{\mathcal{N}}(\psi)}{\widehat{\mathcal{N}}(\partial_b \psi)}\ .
      \nonumber
  \end{align}

  Furthermore, assume that $\mathcal{N}[\psi(a)]$ does not change rank
  locally and that
  \begin{align}
    \widehat{\mathcal{N}}\bigl( -i \bigl[ B/\sigma_B , \psi \bigr] \bigr)
    = \pm \widehat{\mathcal{N}}\bigl( \bigl\{ (A-\langle {A}\rangle )/\sigma_A , \psi \bigr\} \bigr)\ .
    \label{z:oWCGAyhkI-Md}
  \end{align}
  Then
  \begin{align}
    \frac{\FIqty{Bob}{a}}{\FIqty{Alice}{a}} +
    \frac{\FIqty{Eve}{b}}{\FIqty{Alice}{b}}
    = 1\ .
    \label{z:qdGLyd4OQ-zP}
  \end{align}
\end{theorem}
The proof of this statement is presented in
\cref{z:UxjsHdSwwZRl}.
The argument of the square root 
in~\eqref{z:4438TdZJTrfX} 
never becomes negative, thanks to the Robertson uncertainty
relation~\eqref{z:0pRecQmVyaWF} for $A$ and $B$.
The proof we present in
\cref{z:UxjsHdSwwZRl} considers in
fact a more general statement in which the two sides
of~\eqref{z:oWCGAyhkI-Md}
are proportional to one another rather than differing only by a sign.

We can identify two extreme cases of interest to gain some intuition for the
relation~\eqref{z:4438TdZJTrfX}.  First
consider $A,B$ to be two complementary observables in the sense that they
saturate the Robertson inequality~\eqref{z:0pRecQmVyaWF}.
Consider for instance the Pauli-$Y$ and the Pauli-$Z$ operators on a qubit.
In this case the
right-hand side of the
inequality~\eqref{z:4438TdZJTrfX} equals
one.  There is a trade-off between the sensitivity losses associated with Bob
sensing along the $\mathcal{N}(\partial_a\psi)$ direction and Eve sensing along
the $\widehat{\mathcal{N}}(\partial_b\psi)$ direction, as both terms on the left-hand
 side of~\eqref{z:4438TdZJTrfX}
cannot simultaneously be equal to one.  On the other hand, we can consider two
Hermitian generators $A,B$ that commute.  (Perhaps $A,B$ act on different
subsystems of Alice's noiseless clock.)  In this case, the right-hand side
of~\eqref{z:4438TdZJTrfX} evaluates to
the constant $3$.  Our uncertainty relation no longer presents any obstruction
to both Bob and Eve sensing along the respective directions $a,b$ as well as
Alice could, as there is room for both terms on the left-hand side
of~\eqref{z:4438TdZJTrfX} to be equal to
one.  This is the case, for instance, if $A,B$ act on different subsystems of
Alice's clock, and the respective subsystems are sent to Bob and Eve via the
noisy channel and its complementary channel.

We can %
recover our \cref{z:T-KNuYhGRM7I} if we
consider the two generators $A=H$ and $B=-T$, with $H,T$ defined in
\cref{z:ip.xvGb3bzRb}, leading to
$\partial_a\psi = \partial_t\psi$ and $\partial_b\psi = \partial_\eta\psi$.  To
see this, we first compute
\begin{align}
  \bigl \langle { i[H,T] }\bigr \rangle 
  &= \frac1{2\sigma_H^2}\, \bigl \langle { i[H, -i[H, \psi] ]}\bigr \rangle 
    \nonumber\\
  &= \frac1{2\sigma_H^2}\,\bigl(2\langle { H^2 }\rangle  - 2\langle { H }\rangle ^2 \bigr) = 1\ .
\end{align}
Using~\eqref{z:lN.pbrjbNave} we further see that $4\sigma_H^2\sigma_T^2 = 1$.
Therefore, the square root on the right-hand side
of~\eqref{z:4438TdZJTrfX} vanishes and
the entire right-hand side of the inequality evaluates to the constant $1$.
With the identifications $\FIqty{Alice}{t} = \FIqty{Alice}{a}$,
$\FIqty{Bob}{t} = \FIqty{Bob}{a}$, $\FIqty{Alice}{\eta} = \FIqty{Alice}{b}$,
$\FIqty{Eve}{\eta} = \FIqty{Eve}{b}$, we recover the
expression~\eqref{z:cmnZZu5eQ8hv} with an
inequality instead of an equality.
In this case, the additional
condition~\eqref{z:oWCGAyhkI-Md}
is in fact also satisfied, since $i[T, \psi] \propto \{ H - \langle {H}\rangle , \psi \}$
[cf.\@
\cref{z:Df3ZpFNLdboK,z:SMjFM6HgZlts}].
We thus fully recover the equality statement of
\cref{z:T-KNuYhGRM7I} subject to our additional condition
on the absence of a rank change of the noisy state.

The strategy of the proof of
\cref{z:incm1nGt-mIE}
(\cref{z:UxjsHdSwwZRl}) is to first
apply our main uncertainty relation
(\cref{z:Ircbg89zc.Wy}) between Bob's sensitivity to
the parameter $a$ and Eve's sensitivity to a parameter $c$ that is complementary
to $a$ using the construction in
\cref{z:ip.xvGb3bzRb,z:KI2eOYZK2.Ln}
identifying $t\to a, \eta \to c$.  We then apply a general bound relating the
quantum Fisher information with respect to two arbitrary evolution directions
(\cref{z:wXcXJs-gCbik} in \cref{z:KH4B5FtzQ1kE})
to bound the difference between Eve's sensitivity to the parameters $b$ and $c$.

One might have assumed that the
equality~\eqref{z:qdGLyd4OQ-zP}
can only be achieved if the parameters $a,b$ are complementary in the sense of
\cref{z:KI2eOYZK2.Ln}.  Yet it suffices for this property to
hold on the support of the complementary channel, as seen in the
condition~\eqref{z:oWCGAyhkI-Md}.
As a simple extreme example, consider $\mathcal{N}={\mathrm{id}}$ and
$\widehat{\mathcal{N}}(\cdot) = \operatorname{tr}(\cdot)\,\tau$ is a constant channel
preparing some fixed quantum state $\tau$.  Then our uncertainty relation
equality~\eqref{z:qdGLyd4OQ-zP}
necessarily holds for any parameters $a,b$, since Eve's sensitivity to any
parameter $b$ is zero and Bob's sensitivity to any parameter $a$ is equal to
Alice's.  This example also illustrates how the right-hand side
of~\eqref{z:4438TdZJTrfX} should
necessarily be improved to depend on the channel $\mathcal{N}$ if we wanted the
inequality to be tight for a fixed $\mathcal{N}$.  Such an improvement can be
obtained from our proof in
\cref{z:UxjsHdSwwZRl}.

We furthermore provide a proof that the general uncertainty
relation~\eqref{z:4438TdZJTrfX} also
holds in infinite-dimensional Hilbert spaces, and even for unbounded operators.
The details of this proof are given in
\cref{z:3vQhWOZ7civK}.  The proof proceeds by considering
a limiting case of the finite-dimensional setting for larger and larger system
sizes, with additional care given to the definition of the Fisher information in
the infinite-dimensional case and to the fact that the considered operators are
not necessarily bounded.%

It is expected that the
bound~\eqref{z:4438TdZJTrfX} can be
further tightened for observables that do not saturate the Robertson bound.  For
instance, consider two independent systems in a pure tensor product state, with
one system evolving with a parameter $t$ and the other with $z$: 
If we hand the
first system to Bob and the second to Eve, then there is no sensitivity loss for
either parties and the sum of the Fisher information ratios should be $2$.  But
the right-hand side of our bound is $3$.

\section{A selection of examples}%
\label{z:Z14T5rsnGxfi}

We now explore some examples illustrating the application of our main results.

\subsection{Single qubit subject to partial dephasing}

Consider the setup in \cref{z:duwTyi1jsdbn} and described in
\cref{z:t.T5CjLp6B3a}, in
which Alice prepares a pure qubit in the $\lvert {+}\rangle $ state vector and lets it evolve
according to the Hamiltonian 
{$H=\omega {Z}/2$}.
At time $t$, the qubit is
in the state $\psi(t)$ given in~\eqref{z:1--J4q44Q34a} and its
derivative $\partial_t\psi$ is given by~\eqref{z:GM3ZeRQLM9mt}.

Suppose that at time $t_0$ we apply the partially dephasing noisy channel
\begin{align}
  \mathcal{N}_p
  &= (1-p) \, {{\mathrm{id}}} +  p\, \mathcal{D}_Z\ ,
    \label{z:yIa2I8L9iHaJ}
\end{align}
where
\begin{align}
  \mathcal{D}_Z(\cdot)
  &= \langle {\uparrow}\mkern 1.5mu\relax \vert \mkern 1.5mu\relax {\cdot}\mkern 1.5mu\relax \vert \mkern 1.5mu\relax {\uparrow}\rangle  \, \lvert {\uparrow }\rangle \mkern -1.8mu\relax \langle{\uparrow }\rvert +
    \langle {\downarrow}\mkern 1.5mu\relax \vert \mkern 1.5mu\relax {\cdot}\mkern 1.5mu\relax \vert \mkern 1.5mu\relax {\downarrow}\rangle  \, \lvert {\downarrow}\rangle \mkern -1.8mu\relax \langle{\downarrow}\rvert \ .
    \label{z:kwbdc9nLdjvO}
\end{align}
In the following, we will verify that our uncertainty relation holds in this setting, by
first computing directly Bob's Fisher information with respect to $t$, and then
computing Eve's Fisher information with respect to $\eta$.

\paragraph{Direct computation of $\FIqty{Bob}{t}$.}
Using $\mathcal{D}_Z(
{X}
) = 0 = \mathcal{D}_Z(
{Y}
)$ we find
from~\eqref{z:6UjLB7XiW4LY} that Bob receives the state
\begin{align}
  \rho_B(t_0)
  &= \frac12 \begin{bmatrix}
    1 &     (1-p)\,{e}^{-i\omega t_0} \\
    (1-p)\,{e}^{i\omega t_0} & 1
  \end{bmatrix}\ .
\end{align}
Using~\eqref{z:P9NAwniVN.Vm} along with the fact that the
superoperator action of $U_t={e}^{-iHt}$ and $\mathcal{N}_p$ commute and that
$\mathcal{N}_p(
{X}) = (1-p)\,
{X}
$, we can alternatively write Bob's
state as
\begin{align}
  \rho_B(t_0)
  &= U_{t_0} \, \frac{\mathds{1}+ (1-p)\,
  {X}
  }{2}\, U_{t_0}^\dagger
  \nonumber\\
  &= \Bigl(1-\frac{p}{2}\Bigr)\bigl \lvert {+_{t_0}}\big \rangle \mkern -1.8mu\relax \big \langle{+_{t_0}}\bigr \rvert 
  + \frac{p}{2} \bigl \lvert {-_{t_0}}\big \rangle \mkern -1.8mu\relax \big \langle{-_{t_0}}\bigr \rvert \ ,
    \label{z:fF18RXy8Sd9v}
\end{align}
defining the rotated basis state vectors $\lvert {\pm_t}\rangle  := U_t \lvert {\pm}\rangle $.
For the time derivative, using~\eqref{z:SFi9RkJisbzH}
along with $\mathcal{N}_p(
{Y}
) = (1-p)\,
{Y}
$ we find
\begin{align}
  \partial_t \rho_B\,(t_0)
  = \mathcal{N}_p\bigl( \partial_t \psi\,(t_0) \bigr)
  = \frac\omega2 (1-p)\, U_{t_0} \, 
  {Y}
  \, U_{t_0}^\dagger\ .
\end{align}
We may compute the Fisher information with the
formula~\eqref{z:r7GWSaRrLfNh}, using the eigendecomposition
of $\rho_B(t_0)$ given by~\eqref{z:fF18RXy8Sd9v}
\begin{align}
  \FIqty{Bob}{t}
  &=   \begin{alignedat}[t]{1}
    \frac{\omega^2}{4}(1-p)^2\,\Bigl[&
    \frac1{1-p/2}\,\lvert {\langle {+}\mkern 1.5mu\relax \vert \mkern 1.5mu\relax {{Y}}\mkern 1.5mu\relax \vert \mkern 1.5mu\relax {+}\rangle }\rvert ^2 + {}
    \\&
    2\,\lvert {\langle {+}\mkern 1.5mu\relax \vert \mkern 1.5mu\relax {{Y}}\mkern 1.5mu\relax \vert \mkern 1.5mu\relax {-}\rangle }\rvert ^2 + {}
    2\,\lvert {\langle {-}\mkern 1.5mu\relax \vert \mkern 1.5mu\relax {{Y}}\mkern 1.5mu\relax \vert \mkern 1.5mu\relax {+}\rangle }\rvert ^2 + {}
    \\&
    \frac1{p/2}\,\lvert {\langle {-}\mkern 1.5mu\relax \vert \mkern 1.5mu\relax {{Y}}\mkern 1.5mu\relax \vert \mkern 1.5mu\relax {-}\rangle }\rvert ^2
    \Bigr]
  \end{alignedat}
  \nonumber\\
  &= \omega^2 (1-p)^2\ ,
\end{align}
using
$\langle {+}\mkern 1.5mu\relax \vert \mkern 1.5mu\relax {{Y}}\mkern 1.5mu\relax \vert \mkern 1.5mu\relax {-}\rangle  = \langle {+}\mkern 1.5mu\relax \vert \mkern 1.5mu\relax {{Y}
{Z}}\mkern 1.5mu\relax \vert \mkern 1.5mu\relax {+}\rangle  =
i\langle {+}\mkern 1.5mu\relax \vert \mkern 1.5mu\relax {
{X}}\mkern 1.5mu\relax \vert \mkern 1.5mu\relax {+}\rangle =i$ and
$\langle {+}\mkern 1.5mu\relax \vert \mkern 1.5mu\relax {
{Y} }\mkern 1.5mu\relax \vert \mkern 1.5mu\relax {+}\rangle  = 0 = \langle {-}\mkern 1.5mu\relax \vert \mkern 1.5mu\relax {
{Y}}\mkern 1.5mu\relax \vert \mkern 1.5mu\relax {-}\rangle $.

Recalling~\eqref{z:X6ybg8GSCbDf}, the ratio of the
Fisher information of the noisy versus the noiseless clock is
\begin{align}
  \frac{\FIqty{Bob}{t}}{\FIqty{Alice}{t}}
  &= (1-p)^2\ .
    \label{z:AZLTk4fVy2xa}
\end{align}

\paragraph{Computation of $\FIqty{Eve}{\eta}$.}
Now we turn to Eve's picture.  We start with computing a complementary channel
to $\mathcal{N}_p$.  We can use~\eqref{z:SiT0KpCikk0d} for
this effect from any Kraus representation of $\mathcal{N}_p$.  It is useful to
choose a representation with the fewest possible Kraus operators to simplify our
computation of $\FIqty{Eve}{\eta}$.
From~\eqref{z:yIa2I8L9iHaJ}, and using
$\mathcal{D}_Z(\cdot) = \mathds{1}(\cdot)\mathds{1}/2 + {Z}(\cdot){Z}/2$ we
can read off a representation of $\mathcal{N}_p$ with the two Kraus operators
\begin{align}
  E_0^{(p)} &= \sqrt{1-\frac{p}{2}}\,\mathds{1}\ ,
  &
  E_1^{(p)} &= \sqrt{\frac{p}{2}}\,
  {Z}\ .
\end{align}
The complementary channel constructed
via~\eqref{z:SiT0KpCikk0d} takes the form
\begin{align}
  \widehat{\mathcal{N}}_p(\cdot)
  &= \begin{bmatrix}
    \bigl(1-\frac{p}{2}\bigr) \operatorname{tr}(\cdot)
    & \sqrt{\frac{p}{2}\bigl(1-\frac{p}{2}\bigr)} \operatorname{tr}[{Z}(\cdot)] \\
    \sqrt{\frac{p}{2}\bigl(1-\frac{p}{2}\bigr)} \operatorname{tr}[{Z}(\cdot)]
    & \frac{p}{2} \operatorname{tr}(\cdot)
  \end{bmatrix}\ .
\end{align}
Hence Eve's state is
\begin{align}
  \rho_E(t_0)
  = \widehat{\mathcal{N}}_p\bigl(\psi(t_0)\bigr)
  = \begin{bmatrix}
    1 - \frac{p}{2} & 0 \\ 0 & \frac{p}{2}
  \end{bmatrix}\ .
\end{align}
The derivative in the $\eta$ direction is given by the image
of~\eqref{z:xVN3EzcV.KFg} under $\widehat{\mathcal{N}}_p$,
namely
\begin{align}
  \partial_\eta \rho_E\,(t_0)
  = \widehat{\mathcal{N}}_p\bigl(\partial_\eta\psi\,(t_0)\bigr)
  = \frac2\omega \sqrt{\frac{p}{2}\Bigl(1-\frac{p}{2}\Bigr)}\,
  {X}\ .
\end{align}
We may now directly compute $\FIqty{Eve}{\eta}$
using~\eqref{z:r7GWSaRrLfNh},
\begin{align}
  \FIqty{Eve}{\eta}
  &= \frac4{\omega^2}\,\frac{p}{2}\,\Bigl(1-\frac{p}{2}\Bigr)
    \Bigl[
    0 + 2 + 2 + 0
    \Bigr]
    \nonumber\\
  &= \frac4{\omega^2}\,\bigl( 2p - p^2 \bigr)\ .
\end{align}
Using~\eqref{z:X6ybg8GSCbDf} we find that the ratio of Eve's Fisher information to Alice's Fisher information with respect to $\eta$
is
\begin{align}
  \frac{\FIqty{Eve}{\eta}}{\FIqty{Alice}{\eta}}
  &= 2p - p^2
    = 1 - (1-p)^2\ .
    \label{z:x-jhUzhYqP5-}
\end{align}
The fact that~\labelcref{z:AZLTk4fVy2xa,z:x-jhUzhYqP5-} sum to unity
is a manifestation  of \cref{z:T-KNuYhGRM7I} in the
present setting.

Consider now %
our Fisher information loss
formula~\eqref{z:4bic5ys.wcEe}.
Using~\eqref{z:y11JcdWGvpsS} and $\bar{H} = H - \langle {H}\rangle  = H$
we have
\begin{align}
  \widehat{\mathcal{N}}_p\bigl(\{\bar{H}, \psi\}\bigr)
  &= \frac\omega2\, \widehat{\mathcal{N}}_p\bigl({Z}\bigr)
    = \omega \sqrt{\frac{p}{2}\Bigl(1-\frac{p}{2}\Bigr)}\,
    {X}\ .
\end{align}
Then we can compute
$\Ftwo{\rho_E}{ \widehat{\mathcal{N}}_p\bigl(\{\bar{H}, \psi\}\bigr) }$
using~\eqref{z:r7GWSaRrLfNh} as
\begin{align}
  \Ftwo{\rho_E}{ \widehat{\mathcal{N}}_p\bigl(\{\bar{H}, \psi\}\bigr) }
  &= \omega^2 \, \bigl(2p - p^2\bigr)\ .
    \label{z:PGpgkwghh5cq}
\end{align}
We can then verify that the difference in Fisher information between the
noiseless clock and the noisy clock is indeed
\begin{align}
  \FIloss{Bob}{t} &= \FIqty{Alice}{t} - \FIqty{Bob}{t}
  = \omega^2\bigl(1 - (1-p)^2\bigr)
         \nonumber\\
  &= \Ftwo{\rho_E}{ \widehat{\mathcal{N}}_p\bigl(\{\bar{H}, \psi\}\bigr) }\ .
\end{align}

\subsection{Single qubit subject to complete dephasing along a transversal axis}
\label{z:TaWot2w4x4dl}

Now we consider a variant of the above single-qubit example: We replace the
noisy channel by a complete dephasing along the %
 $X$ axis
(\cref{z:kCqOQRv2d.3A}).
\begin{figure}
  \centering
  \includegraphics{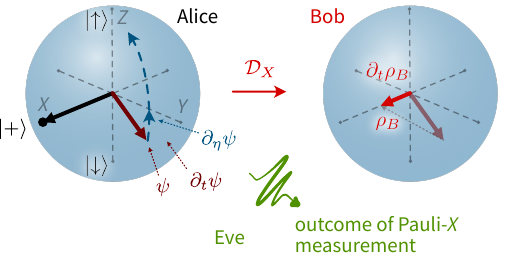}
  \caption{Single-qubit probe evolving according to the Hamiltonian
    $H=\omega{Z}/2$ and subject to complete dephasing along the $X$
    direction at time close to $t_0$.  For almost every $t_0$, the noisy probe
    remains maximally sensitive to time to first order around $t_0$.  This
    property might sound surprising, because Bob's state can be very mixed.  In
    the purified picture, Eve is given the outcome of a measurement of Alice's
    state along the $X$ axis.  Observe that in contrast to the setting in
    \cref{z:duwTyi1jsdbn}, this information does not reveal
    any information about the energy of Alice's state.}
  \label{z:kCqOQRv2d.3A}
\end{figure}
The qubit is initialized in the state vector $\lvert {\psi_{\mathrm{init}}}\rangle  = \lvert {+}\rangle $,
where $\lvert {\pm }\rangle = [\lvert {\uparrow }\rangle \pm \lvert {\downarrow}\rangle ]/\sqrt{2}$, with a Hamiltonian
$H = \frac{\omega}{2}{Z}$.  After a time $t$, the state is given
by~\eqref{z:UlWOzs-evrpA} and at all times we have
$\langle {H}\rangle _{\psi(t)} = 0$ and $\sigma_H^2 = \omega^2/4$.
At time $t\approx t_0$ the clock is completely dephased in the $X$ basis, as described by
the noisy channel
\begin{align}
  \mathcal{D}_X(\cdot)
  &= \langle {+}\mkern 1.5mu\relax \vert \mkern 1.5mu\relax {\cdot}\mkern 1.5mu\relax \vert \mkern 1.5mu\relax {+}\rangle \,\lvert {+}\rangle \mkern -1.8mu\relax \langle{+}\rvert  + \langle {-}\mkern 1.5mu\relax \vert \mkern 1.5mu\relax {\cdot}\mkern 1.5mu\relax \vert \mkern 1.5mu\relax {-}\rangle \,\lvert {-}\rangle \mkern -1.8mu\relax \langle{-}\rvert \ .
    \label{z:03pqdrL1Z7eT}
\end{align}
This completely dephasing map acts on the Pauli operator basis as
$\mathcal{D}_X(\mathds{1})=\mathds{1}$, $\mathcal{D}_X(
{X}) = 
{X}$, and
$\mathcal{D}_X(
{Y}) = 0 = \mathcal{D}_X(
{Z})$.  Bob receives the
density matrix
\begin{align}
  \rho_B = 
    \cos^2\Bigl(\frac{\omega t}{2}\Bigr) \lvert {+}\rangle \mkern -1.8mu\relax \langle{+}\rvert   +
    \sin^2\Bigl(\frac{\omega t}{2}\Bigr) \lvert {-}\rangle \mkern -1.8mu\relax \langle{-}\rvert  \ .
  \label{z:qdCKlozHjX8s}
\end{align}
Now the complementary channel of $\mathcal{D}_X$ is again
$\widehat{\mathcal{D}}_X=\mathcal{D}_X$, and so Eve gets the same density matrix
as Bob.

\paragraph{Computation of $\FIqty{Eve}{\eta}$.}
Recalling~\eqref{z:y11JcdWGvpsS}, we find
\begin{align}
  \widehat{\mathcal{D}}_X( \{ H-\langle {H}\rangle , \psi \} ) =
  \frac\omega2\,\widehat{\mathcal{D}}_X({Z}) = 0\ ,
  \label{z:vZfDiFsuC0Ah}
\end{align}
because $\mathcal{D}_X$ maps the Pauli-$Y$ and Pauli-$Z$ operators to zero.
Therefore Eve obtains zero information about $\eta$, i.e.,
$\FIqty{Eve}{\eta} = 0$.  Therefore, there is no sensitivity loss for Bob
regardless of the time $t\approx t_0$ at which the noisy channel is applied, as long as
the rank of $\rho_B(t_0)$ does not change locally at $t_0$.
The state $\rho_B$ changes rank whenever either term
of~\eqref{z:qdCKlozHjX8s} vanishes, i.e., when $t_0$ is a multiple of
$\pi/\omega$.  At those discrete points, we hit the edge cases where our main
uncertainty relation does not hold with equality and we cannot deduce that Bob
has maximal sensitivity at those points.  However, at all other points $t_0$ the
clock does not lose any sensitivity when sent to Bob.

The same conclusions apply for any noisy channel that is a complete dephasing
operation along an axis that lies in the equatorial plane, by rotational
symmetry of the problem around the $Z$ axis.  (Any axis in the equatorial plane
can be described as a rotation of the $X$ axis that is equivalent to a time
evolution of the system for some given time $t_*$.  Because the Fisher
information is invariant under unitary transformations, the calculation of Bob's
Fisher information of this qubit after complete dephasing along that given axis
at time $t_0$ is equivalent to calculating the Fisher information after a
complete dephasing along the $X$ axis at the time $t_0\to t_0 - t_*$.)

\paragraph{Check by direct computation of $\FIqty{Bob}{t}$.}
We now compute $\FIqty{Bob}{t} = \Ftwo{\rho_B}{\mathcal{D}_Z(\partial_t\psi)}$
directly, by using the definition of the Fisher information.  From~\eqref{z:.uzSQ.DqcFF7} we find
\begin{align}
  \mathcal{D}_X(\partial_t\psi)
  &= -\frac\omega2\,\sin(\omega t_0)\,
  {X}\ .
    \label{z:5xfiUugKQWOt}
\end{align}
If $\sin(\omega t_0) = 0$, which happens when $t_0$ is a multiple of
$\pi/\omega$, we find that Bob's state is locally stationary and Bob has no
sensitivity to first order in $t$.  (For this discrete set of points one could
argue that the Fisher information no longer represents the relevant sensitivity
for Bob, since the evolution should be considered to its leading order---here
the second order---and no longer only to first order.)

We now compute $\FIqty{Bob}{t}$ for all times $t_0$ where
$\sin(\omega t_0) \neq 0$. Observe that $\rho_B$ and
$\mathcal{D}_X(\partial_t\psi)$ commute.
Using~\eqref{z:q9lQk9SttqXb} and $X^2=\mathds{1}$, we find
\begin{align}
  \hspace*{1em}&\hspace*{-1em}
  \FIqty{Bob}{t}
= \frac{\omega^2}{4}\,\sin^2(\omega t_0) \, \operatorname{tr}\bigl(\rho_B^{-1}\bigr)
    \nonumber\\
  &= \omega^2\,\Bigl[\sin\Bigl(\frac{\omega t_0}{2}\Bigr)\cos\Bigl(\frac{\omega t_0}{2}\Bigr)\Bigr]^2 \,
    \biggl[
    \frac1{ \cos^2\bigl(\frac{\omega t_0}{2}\bigr) } +
    \frac1{ \sin^2\bigl(\frac{\omega t_0}{2}\bigr) }
    \biggr]
    \nonumber\\
  &= \omega^2\ ,
\end{align}
using $\sin(\omega t_0) = 2\sin(\omega t_0/2)\cos(\omega t_0/2)$ in the
second equality.

Overall, we see that Bob still has maximal sensitivity even after application of
the completely dephasing channel along the transversal $X$ axis, for all times
except for the discrete set of times $t_0$ where the rank of $\rho_B$ changes.
This conclusion matches our earlier conclusions obtained via considerations from
Eve's perspective (except for a discrete set of times $t_0$).

It might appear counterintuitive that Bob's state still has as high a
sensitivity as Alice's noiseless state for almost all $t_0$, especially as Bob's
state can get arbitrarily mixed.  Indeed, $\rho_B$ coincides with the maximally
mixed state for times $t_0$ that are midpoints between the multiples of
$\pi/\omega$.  However, we see that $\rho_B(t)$ still varies with $t$
sufficiently to enable optimal discrimination of nearby states to first order
around $t_0$.

\subsection{Probe in a GHZ state with one partial erasure}
Consider 
as initial state
an $n$-party GHZ state vector,
\begin{align}
  \lvert {\mathrm{GHZ}}\rangle  &=
  \frac1{\sqrt 2} \bigl[ \lvert {\uparrow \cdots \uparrow}\rangle 
  + \lvert {\downarrow \cdots \downarrow}\rangle  \bigr]\ ,
\end{align}
and let the system evolve according to the local Hamiltonian
$H = \sum_{i}\, (\omega/2) {Z}^{(i)}$ where ${Z}^{(i)}$ denotes the Pauli ${Z}$
operator acting on the $i$-th site.  Suppose that the first qubit is lost with
probability $p$.  This is represented by the noisy channel
\begin{align}
  \mathcal{N}(\cdot) = p\, \lvert {\phi_\perp}\rangle \mkern -1.8mu\relax \langle{\phi_\perp}\rvert \otimes\operatorname{tr}_{1}(\cdot) + (1-p)\,(\cdot)\ ,
\end{align}
where $\operatorname{tr}_1$ traces out the first qubit and where $\lvert {\phi_\perp}\rangle $ is a state vector
in a new, orthogonal dimension that has no overlap with the input state.
A Stinespring dilation of the first term in $\mathcal{N}$ is described as giving
the first qubit of Alice's system to Eve, and the remaining qubits to Bob; any
missing qubits on either Bob or Eve's side is replaced by $\lvert {\phi_\perp}\rangle $.
The complementary channel can thus be computed as
\begin{align}
  \widehat{\mathcal{N}}(\cdot) = p\,\operatorname{tr}_{2\ldots n}(\cdot) + (1-p) \operatorname{tr}(\cdot)\, \lvert {\phi_\perp}\rangle \mkern -1.8mu\relax \langle{\phi_\perp}\rvert \ .
\end{align}
We compute the sensitivity loss associated with the noise according
to~\eqref{z:4bic5ys.wcEe}.  We have
\begin{align}
  H\lvert {\psi}\rangle  = \frac{n\omega}{2\sqrt 2} \bigl[ \lvert {\uparrow \cdots \uparrow}\rangle 
  - \lvert {\downarrow \cdots \downarrow}\rangle  \bigr]
  = P_\psi^\perp H \lvert {\psi}\rangle \ ,
\end{align}
noting that $H\lvert {\psi}\rangle $ is already orthogonal to $\lvert {\psi}\rangle $ since
$\langle {H}\rangle _\psi = 0$.
The optimal noiseless sensitivity is
\begin{align}
  \FIqty{Alice}{t}
  = 4\sigma_H^2
  = 4\langle {\psi}\mkern 1.5mu\relax \vert \mkern 1.5mu\relax {H^2}\mkern 1.5mu\relax \vert \mkern 1.5mu\relax {\psi}\rangle  =  n^2\omega^2\ ,
\end{align}
exhibiting the expected Heisenberg scaling for optimally entangled probe states.
We write
$\{\psi,\bar{H}\} = P_\psi^\perp H \psi + \mathrm{h.c.} = \sigma_H
Z_L$, with $Z_L$ defined in~\eqref{z:5ChzchOJ.MLJ}.  The local
reduced operator of $\{\psi,\bar{H}\}$ on a single site is
\begin{align}
  \operatorname{tr}_{\backslash\,i}( \{\psi,\bar{H}\} )
  = \frac{n\omega}{4}\lvert {\uparrow }\rangle \mkern -1.8mu\relax \langle{\uparrow }\rvert - \frac{n\omega}{4}\lvert {\downarrow }\rangle \mkern -1.8mu\relax \langle{\downarrow }\rvert + \mathrm{h.c.}
  = \frac{n\omega}{2}\,{Z}^{(i)}\ ,
\end{align}
where $\operatorname{tr}_{\backslash\,i}$ denotes the partial trace over all subsystems
except the $i$-th subsystem.  Noting that $\operatorname{tr}(\{\psi,\bar{H}\})=0$, we obtain
\begin{align}
  \widehat{\mathcal{N}}( \{\psi,\bar{H}\} ) = p\,\frac{n\omega}{2} \,{Z}\ .
\end{align}
On the other hand, the reduced state of $\psi$ on a single site is simply the
maximally mixed state $\mathds{1}_2/2$ and thus
\begin{align}
  \rho_E = \widehat{\mathcal{N}}(\psi)
  = p\,\frac{\mathds{1}_2}{2} + (1-p)\,\lvert {\phi_\perp}\rangle \mkern -1.8mu\relax \langle{\phi_\perp}\rvert \ .
\end{align}
As $\rho_E$ and $\widehat{\mathcal{N}}( \{\psi,\bar{H}\} )$ commute, we can
use~\eqref{z:q9lQk9SttqXb} to see that
\begin{align}
\FIloss{Bob}{t}
  = \operatorname{tr}\mathopen{}\left[ \frac{2}{p} \, \Bigl(p\frac{n\omega}{2}{Z}\Bigr)^2\right]\mathclose{} =  p n^2 \omega^2\ .
\end{align}
If $p=1$, Eve is maximally disturbing and completely blocks Bob's ability to
measure time, if $p=0$ there is no sensitivity loss.  Any value in between
interpolates between these two cases.

Note that while it might appear here that Heisenberg scaling
($\FIqty{Bob}{t}\propto n^2$) is achieved for $p>0$, this is an artifact of
the lack of scaling in $n$ of our choice of noisy channel and does not
contradict the findings of, e.g.,
 Refs.~\cite{R30,R21}.

\subsection{Estimating a signal Hamiltonian term}

In this subsection, we briefly comment on the case where the parameter to
estimate is not time $t$ itself, but a parameter $f$ in the Hamiltonian
that influences time evolution.  In other words, we now account for possible
other terms in the Hamiltonian that contribute to time evolution but that reveal
nothing about the parameter of interest.  We assume that the noiseless probe
evolves according to a Hamiltonian
\begin{align}
  H_f = H_0 + f G\ ,
\end{align}
where $H_0$ does not depend on $f$, and where $H_0$ and $G$ are
time independent.
References~\cite{R34,R35} have
determined that the
Fisher information with respect to $f$ that one achieves by initializing the
system in some initial state vector
$\lvert {\psi_0}\rangle $ and letting the system evolve
according to $H_f$ for some fixed time $T$.  Let $U_f(T) = {e}^{-iH_f T}$ be the
time-evolution operator, and define $\lvert {\psi_f}\rangle  = U_f(T)\lvert {\psi_0}\rangle $.  The
question is, how much sensitivity does the family of state vectors $f\mapsto \lvert {\psi_f}\rangle $ offer
with respect to $f$?  The derivative relevant for the Fisher information is
given by~\cite{R35}
\begin{align}
  \partial_f \psi_f = -i [K_f, \psi_f]\ ,
  \label{z:lV7Y877pc1B1}
\end{align}
where
\begin{align}
  K_f = -i U_f^{-1} \frac{dU_f}{df}
  = T \sum_{k=0}^\infty \frac{(-iT)^{k}}{(k+1)!}\,\operatorname{ad}_{H_f}^k\bigl(G\bigr)\ ,
\end{align}
where $\operatorname{ad}_{{M}}(G) := [{M}, G]$ and
\begin{align}
\operatorname{ad}_{{M}}^{k}(G) := [{M}, [{M}, \ldots, [{M}, G]] 
\end{align}
is the $k$-th commutator
of ${M}$ with $G$.
The operator $K_f$ can be thought of as an effective ``Hamiltonian'' for the
parameter $f$, driving an ``evolution'' in $\lvert {\psi_f}\rangle $ with respect to $f$
according to~\eqref{z:lV7Y877pc1B1}.

If we send this probe state through a noisy channel following the setting in
\cref{z:FnHui0ahNLxW}, then our uncertainty relation can be applied,
where the complementary parameter evolution is generated by the operator
$L = -i[K_f, \psi_f]/(2\sigma_{K_f}^2)$.  That is, Bob's sensitivity to $f$ trades
off with Eve's sensitivity to the parameter generated by $L$.

\subsection{Symmetric codes against erasures via superpositions of Dicke states}
\label{z:lZkawwgrYk1u}

Based on the relevance of Dicke states for
metrology~\cite{R36,R37,R38,R39,R40} and for quantum error
correction~\cite{R41,R42,R43}, we can
ask whether our uncertainty relation can guide a search for good clock states.
To ensure good sensitivity even in the noiseless setting,
we seek probe states with a large spread over energy eigenstates.  So we
consider a general superposition of Dicke states corresponding to different
numbers of excitations.
We note an important class of permutation-invariant codes are those developed in
Refs.~\cite{R40,R41}.

Consider the $n$-spin noninteracting Hamiltonian
$H = \sum_{i=1}^n (\omega/2)\, {Z}^{(i)}$.
A Dicke state is an eigenstate of $H$ that is symmetric under permutations of
the sites.  Consider the Dicke state  %
\begin{align}
  \lvert {h^n_q}\rangle  := 
  \binom{n}{q}^{-1/2}
  \sum_{\substack{s_i = \pm 1\\\sum s_i = n-2q}} \lvert {s_1\ldots s_n}\rangle \ ,
\end{align}
where $s_i=\pm1$ represents the eigenstates of ${Z}$ and where
$q=0,\ldots,n$.
We construct our probe states as a superposition of Dicke states of different
values of $q$.
In general, such a state vector can be written as
\begin{align}
  \lvert {\psi }\rangle = \sum_{q=0}^{n} \psi_q \lvert {h^n_q}\rangle \ ,
\end{align}
for some arbitrary complex amplitudes $\{\psi_q\}$ that satisfy
$\sum_{{q}} \lvert {\psi_q}\rvert ^2 = 1$.

As a noise model, we assume that $k$ systems chosen at random are entirely
erased.  Because the probe state is completely symmetric, it does not matter
which subsystems are erased; we may assume that the first $k$ sites are erased.
The complementary channel to the erasure of $k$ subsystems is a channel that
provides those lost subsystems to Eve,
\begin{align}
  \widehat{\mathcal{N}}(\cdot)
  &= \operatorname{tr}_{k+1\ldots n}(\cdot)\ ,
\end{align}
where $\operatorname{tr}_{k+1\ldots n}$ denotes the partial trace over sites $k+1$ to $n$.

We compute numerical values for the quantities $\FIloss{Bob}{t}$ and
$4\sigma_H^2$, enabling us to infer $\FIqty{Bob}{t}$.  Consider the probe state vector
consisting of an even superposition of two Dicke states with associated
parameters $q_1,q_2$ 
\begin{align}
  \lvert {\psi }\rangle = [\lvert {h^n_{q_1}}\rangle  + \lvert {h^n_{q_2}}\rangle ]/\sqrt{2}\ .
\end{align}
The sensitivity of this probe state for $n=100$ and subject to $k=9$ erasures is
plotted as a function of $q_1,q_2$ in \cref{z:DvEqFQ7o-EPn} (with
$\omega/2=1$).  The sensitivity $\FIqty{Bob}{t}$ is obtained by computing $\FIloss{Bob}{t}$
and $\sigma_H^2$ via~\eqref{z:QZajZctWEncO}.
\begin{figure}
  \centering
  \includegraphics[width=\columnwidth]{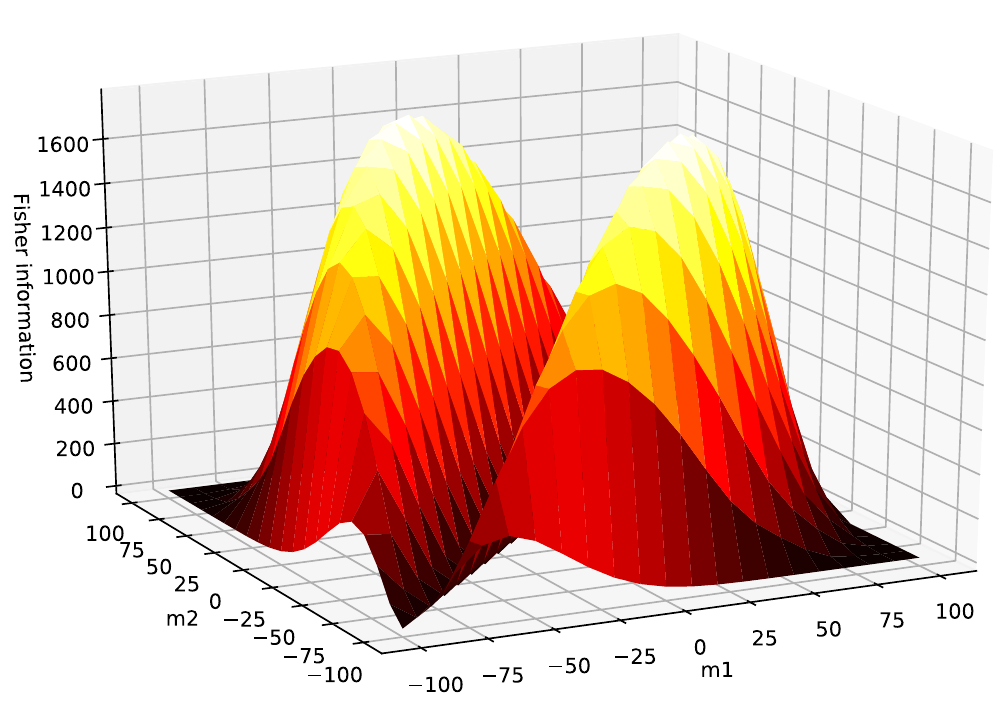}%
  \caption{Fisher information of an even superposition of two Dicke states of
    magnetizations $m_1=n-2q_1$ and $m_2=n-2q_2$ on a $n$-site noninteracting
    spin chain with local terms $H_i = (\omega/2){Z}$.  A good probe state
    has $m_1,m_2$ far from one another (for a large energy spread), but also far
    from the edges $-n$ and $n$ (to avoid decoherence caused by the erasures).
    Here we set $n=100$ total spins, $\omega/2=1$, and $k=9$ spins are lost to
    the environment.  Our trade-off relation facilitates the calculation of the
    Fisher information plotted above.  It also gives an interpretation of the
    loss in sensitivity with respect to the noiseless case (where the GHZ state
    $m_1=-m_2 = \pm n$ would be optimal; leftmost and rightmost edges of the
    plot) as the sensitivity that Eve gains with respect to the energy of the
    state.  }
  \label{z:DvEqFQ7o-EPn}
\end{figure}
On the one hand, our trade-off relation facilitates the calculation of the
remaining Fisher information after the erasures.  On the other hand, the
trade-off relation explains that the high sensitivity loss experienced for
states with a broad spread in energy ($q_1\to 0$ and $q_n\to n$) is directly
related to the fact that the environment can well infer the energy of the
state from few-site reduced states.

Because the noise is local, numerical computations only have to take place on a smaller system representing the local degrees of freedom.  Because of permutation %
symmetry globally and also locally (the reduced state also lives in the local
symmetric subspace), our computations run on $k+1$ dimensions and not on the
full $(n+1)$-dimensional symmetric subspace.
We will return to the example of permutation-invariant states on $n$ spins
in \cref{z:Pgp-gW09n9nZ}, where we consider an i.i.d.\@
amplitude-damping noise model instead of erasures.

\section{Bounds on the Fisher information}
\label{z:qdkRM6aFmtdz}

Because it might not always be simple to compute the Fisher information trade-off
quantity $\FIloss{Bob}{t}$ in~\eqref{z:4bic5ys.wcEe}, we provide
a few bounds that might be applicable to different settings, and that avoid the
calculation of the symmetric logarithmic derivative on Eve's system.

\subsection{Upper bound on Bob's sensitivity by postprocessing
  Eve's system}
\label{z:308j2vzIlz3f}

A useful bound for the Fisher information is the data-processing
inequality~\cite{R18}.  The inequality states that for any
$\rho(t)$, and for any $t$-independent completely positive,
trace-nonincreasing map $\mathcal{E}$, the sensitivity after application of the
channel can only decrease:
\begin{align}
  F(\rho(t)) \geq F(\mathcal{E}(\rho(t)))\ .
  \label{z:.Hy7vObZJO-E}
\end{align}
A trace-nonincreasing map can be used to describe only a subspace of interest of
a larger Hilbert space while accounting for leakage outside of that subspace.

Consider our setup with Alice, Bob and Eve as in
\cref{z:FnHui0ahNLxW}.  Suppose now that Eve sends her state to
another agent, Eve${}^\prime$, through a trace-nonincreasing, completely
positive map $\mathcal{N}'$ as depicted in
\cref{z:kmpiohHL8OiG}\textbf{a}.
\begin{figure}
  \centering
  \includegraphics{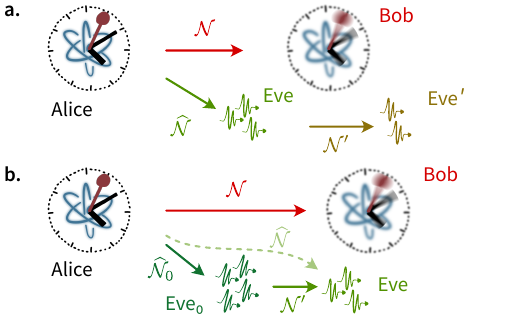}
  \caption{%
    Combining our uncertainty relation with the data processing
    inequality for the Fisher information yields new bounds for the Fisher
    information. \textbf{a.}~Suppose Eve applies a suitably chosen map
    $\mathcal{N}'$ to her system, resulting in a system we denote by Eve${}'$,
    on which the sensitivity to energy might be significantly easier to compute.
    Eve${}'$ can only have a worse sensitivity to energy than Eve, so our
    uncertainty relation gives an upper bound to Bob's sensitivity to time.
    \textbf{b.}~Suppose that Eve's output can be written as a composition of two
    maps $\widehat{\mathcal{N}}_0$ and $\mathcal{N}'$ via an intermediate system
    Eve${}_0$.  Then we obtain a lower bound on Bob's sensitivity to time by
    computing Eve${}_0$'s sensitivity to energy.}
  \label{z:kmpiohHL8OiG}
\end{figure}
The data processing inequality ensures that
$\FIqty{Eve}{\eta} \geq \FIqty{Eve'}{\eta}$.  Combining this with our
uncertainty relation~\eqref{z:cmnZZu5eQ8hv} yields
\begin{align}
  \frac{ \FIqty{Bob}{t} }{ \FIqty{Alice}{t} }
  + \frac{ \FIqty{Eve'}{\eta} }{ \FIqty{Alice}{\eta} } \leq 1\ .
\end{align}
We can also obtain this inequality by starting from the quantum Fisher
information loss on Bob's end~\eqref{z:4bic5ys.wcEe},
\begin{align}
  \FIloss{Bob}{t}
  &= \Ftwo{\rho_E}{ \widehat{\mathcal{N}}(\{\bar{H},\psi\}) }
    \nonumber\\
  &\geq \Ftwo{\mathcal{N}'(\rho_E)}{
    \mathcal{N}'\bigl(\widehat{\mathcal{N}}(\{\bar{H},\psi\})\bigr) }\ ,
  \label{z:OKz7mFdA9JgR}
\end{align}
which in turn provides an upper bound on Bob's Fisher information
via~\eqref{z:QZajZctWEncO} as
\begin{align}
  \FIqty{Bob}{t}
  \leq 4\sigma_H^2 - \Ftwo{\mathcal{N}'(\rho_E)}{
    \mathcal{N}'\bigl(\widehat{\mathcal{N}}(\{\bar{H},\psi\})\bigr) }\ .
  \label{z:t59w2cccLdg7}
\end{align}
By choosing the map $\mathcal{N}'$ suitably, one can potentially significantly
simplify the computation of the Fisher information. For instance, $\mathcal{N}'$
can be a dephasing map that ensures that $\mathcal{N}'(\rho_E)$ and
$\mathcal{N}'(\widehat{\mathcal{N}}\{\psi,H\})$ commute, therefore
enabling the use of \cref{z:q9lQk9SttqXb} and removing the
necessity of computing the symmetric logarithmic derivative.  Alternatively
$\mathcal{N}'$ can be chosen to enforce some symmetry that might be convenient
for the computation of the Fisher information.

The bound~\eqref{z:t59w2cccLdg7} can be spelled
out in the case of independent and identically distributed (i.i.d.\@) noise on a
many-body probe state.  Consider a single-site noisy channel $\mathcal{N}_1$
with Kraus operators $\{ E_x \}$ for $x=0,\ldots,m-1$.  The full noisy channel is
$\mathcal{N} = \mathcal{N}_1^{\otimes n}$.
Its Kraus operators are
$E_{\boldsymbol x}$, where $\boldsymbol x = (x_1,\ldots, x_n)$ is a collection
of indices $x_i=0,\ldots,m-1$ indicating which Kraus operator is applied on the
$i$-th site
\begin{align}
  E_{\boldsymbol x} = \bigotimes_{i=1}^n E_{x_i}\ .
\end{align}
The complementary channel $\widehat{\mathcal{N}}$ can then be written in terms
of the Kraus operators of $\mathcal{N}$ as
\begin{align}
  \widehat{\mathcal{N}}(\cdot) =
  \sum_{\boldsymbol x,\boldsymbol x'}
  \operatorname{tr}\bigl(E_{\boldsymbol x'}^\dagger E_{\boldsymbol x}\,(\cdot)\bigr)
  \lvert {\boldsymbol x}\rangle \mkern -1.8mu\relax \langle{\boldsymbol x'}\rvert \ ,
\end{align}
where $\{ \lvert {\boldsymbol x}\rangle  \}$ is a basis of the Hilbert space of $E$.

Computing the Fisher information analytically on the output of either
$\mathcal{N}$ or $\widehat{\mathcal{N}}$ might not be straightforward if the
state and its derivative are mapped to operators whose eigenbases are not
aligned in any obvious way, which would complicate the calculation of the
symmetric logarithmic derivative when computing the
expression~\eqref{z:4bic5ys.wcEe}.  Here, we see that by
completely dephasing the output of $\widehat{\mathcal{N}}$ in the computational
basis, and projecting onto the subspace of the environment associated with
low-weight Kraus operators of $\mathcal{N}$, we obtain a lower bound on
$\FIloss{Bob}{t}$ which translates into a upper bound on $\FIqty{Bob}{t}$ that
is easy to compute.  Here, we assume that the first Kraus operator $E_0$ is
close to the identity and that the other Kraus operators represent ``jump
terms.''  We mean by ``weight'' the number of Kraus operators that are jump
terms.

We now choose a suitable completely positive, trace-nonincreasing map
$\mathcal{N}'$ in order to
use~\eqref{z:t59w2cccLdg7} to obtain an upper
bound on the Fisher information at Bob's end.
In the following, we assume that $m=2$, but the argument generalizes
straightforwardly to noisy channels that have more Kraus operators.
We design the map such that it (i)~completely dephases the environment system in
the computational basis, and (ii)~projects its input onto the subspace
associated with basis vectors $\lvert {\boldsymbol x}\rangle $ with small Hamming weight
$\lvert {\boldsymbol x}\rvert $.  Fix $k>0$ and let
\begin{align}
  \mathcal{N}'(\cdot) =
  \sum_{\boldsymbol x:\ \lvert {\boldsymbol x}\rvert  \leq k}
  \lvert {\boldsymbol x}\rangle \mkern -1.8mu\relax \langle{\boldsymbol x}\rvert \,(\cdot)\,\lvert {\boldsymbol x}\rangle \mkern -1.8mu\relax \langle{\boldsymbol x}\rvert \ .
\end{align}
Then we can see that
\begin{align}
  \bigl(\mathcal{N}'\circ\widehat{\mathcal{N}}\bigr)(\cdot) =
  \sum_{\boldsymbol x:\ \lvert {\boldsymbol x}\rvert  \leq k}
  \operatorname{tr}\bigl(E_{\boldsymbol x}^\dagger E_{\boldsymbol x}\,(\cdot)\bigr)
  \lvert {\boldsymbol x}\rangle \mkern -1.8mu\relax \langle{\boldsymbol x}\rvert \ .
\end{align}

The upper bound on Bob's Fisher information with respect to time comes
from~\eqref{z:t59w2cccLdg7}.  Starting
from~\eqref{z:OKz7mFdA9JgR} and since
the two arguments of the Fisher information commute, we can
use~\eqref{z:q9lQk9SttqXb} to find
\begin{align}
  \FIloss{Bob}{t}
  &\geq \operatorname{tr}\Bigl\{
  \bigl[\bigl( \mathcal{N}'\circ\widehat{\mathcal{N}}\bigr)(\psi)\bigr]^{-1}
  \bigl[\bigl( \mathcal{N}'\circ\widehat{\mathcal{N}}\bigr)(\{\psi,\bar{H}\})\bigr]^{2}
  \Bigr\}
  \nonumber\\
  &= \sum_{\boldsymbol x:\ \lvert {\boldsymbol x}\rvert  \leq k}
    \frac{
      \bigl[ 2\operatorname{Re}\langle {\psi }\rvert \bar{H}
        E_{\boldsymbol x}^\dagger E_{\boldsymbol x}\lvert {\psi }\rangle \bigr]^2
    }{
      \operatorname{tr}(E_{\boldsymbol x}^\dagger E_{\boldsymbol x}\,\psi )
    }\ ,
    \label{z:3YnRZBPm.g17}
\end{align}
where we have used the fact that the output of
$\mathcal{N}'\circ\widehat{\mathcal{N}}$ is diagonal in the computational basis.
The completely dephasing channel ensures that the
expression~\eqref{z:3YnRZBPm.g17} is a classical Fisher information,
which is easier to compute than the quantum Fisher information in which the state
and the derivative do not commute.

The number of terms in the above sum, which corresponds to the dimension of the
subspace associated with basis vectors $\lvert {\boldsymbol x}\rangle $ satisfying
$\lvert {\boldsymbol x}\rvert \leq k$, is given by
$\binom{n}{k}+\binom{n}{k-1}+\cdots+\binom{n}{0} = O(n^k)$.  For fixed $k$, this
number scales polynomially in $n$.  The complexity of computing the numerator
and denominator in~\eqref{z:3YnRZBPm.g17} also scales only
polynomially in $n$ as long as $\lvert {\psi}\rangle $ and $\bar{H}\lvert {\psi}\rangle $ can be expressed
using a representation that enables efficient computation of local expectation
values, such as a superposition of a constant number of computational basis
vectors, or alternatively as matrix-product
states~\cite{R44}.
We discuss below the case of i.i.d.\ amplitude damping noise, where numerical
evidence indicates that for small values of $p$ (say $p\lesssim0.1$), even for
$n=50$ it can suffice to set $k=4$ to obtain meaningful bounds (see
\cref{z:Pgp-gW09n9nZ}).

\subsection{Lower bound on Bob's Fisher information by preprocessing  Eve's system}
\label{z:gJXrB9endOCQ}

Let us return to the original setting with Alice, Bob, and Eve as in
\cref{z:FnHui0ahNLxW}.  Suppose now that we can find a completely
positive, trace-preserving map $\widehat{\mathcal{N}}_0$ and a completely
positive, trace-nonincreasing $\mathcal{N}'$ such that
$\widehat{\mathcal{N}} = \mathcal{N}' \circ \widehat{\mathcal{N}}_0$. That is,
we suppose that Eve gets her state through an intermediary, which we call
Eve${}_0$ as shown in \cref{z:kmpiohHL8OiG}\textbf{b}.  The data-processing
 inequality now tells us that
$\FIqty{Eve}{\eta} \leq \FIqty{Eve_0}{\eta}$.  Combining this with our
uncertainty relation gives us
\begin{align}
  \frac{ \FIqty{Bob}{t} }{ \FIqty{Alice}{t} }
  + \frac{ \FIqty{Eve_0}{\eta} }{ \FIqty{Alice}{\eta} }
  \geq 1\ .
  \label{z:mv9tU7kxaHmU}
\end{align}
A more explicit bound on $\FIqty{Bob}{t}$ can be obtained starting
from~\eqref{z:4bic5ys.wcEe} and writing
\begin{align}
  \FIloss{Bob}{t}
  &= \Ftwo{\rho_E}{ \widehat{\mathcal{N}}(\{\bar{H},\psi\}) }
    \nonumber\\
  &= \Ftwo{\mathcal{N}'(\widehat{\mathcal{N}}_0(\psi))}{
    \widehat{\mathcal{N}}'(\widehat{\mathcal{N}}_0(\{\bar{H},\psi\})) }
    \nonumber\\
  &\leq \Ftwo{\widehat{\mathcal{N}}_0(\psi)}{
    \widehat{\mathcal{N}}_0(\{\bar{H},\psi\}) }\ .
    \label{z:U3JchlrlLghQ}
\end{align}
We present two simple example uses of this type of bound.  The first example applies to permutation-invariant systems. The second example applies to the
setting where Eve's state is reasonably close to being diagonal.

\paragraph{Permutation-invariant system.}

Consider a permutation-invariant clock state $\psi$ and Hamiltonian $H$.  If the
noise $\mathcal{N}$ acts only locally on at most $k$ known sites (or $k/2$
unknown sites), then $\widehat{\mathcal{N}}$ can be written as the composition
of a channel that traces out all but $k$ sites, and another channel that
completes the implementation of $\widehat{\mathcal{N}}$.  To see this, observe
that we can write
$\widehat{\mathcal{N}}(\cdot) = \sum_{j',j} \operatorname{tr}\bigl(E_{j'}^\dagger
E_j\,(\cdot)\bigr)\,\lvert {j'}\rangle \mkern -1.8mu\relax \langle{j}\rvert $, where $\{E_j\}$ are the Kraus operators of
$\mathcal{N}$. By assumption, $E_{j'}^\dagger E_j$ acts nontrivially on at most
$k$ sites.  Therefore, the expression $\operatorname{tr}\bigl(E_{j'}^\dagger E_j\,(\cdot)\bigr)$
depends only on the $k$-site reduced state of the 
input.  The full complementary
channel can be written as the composition of a channel that traces all but $k$
sites, and the channel
$\rho_k \mapsto \sum_{j',j} \operatorname{tr}(E_{j'}^\dagger E_j\,\rho_k)\,\lvert {j'}\rangle \mkern -1.8mu\relax \langle{j}\rvert $
(where here we reuse the notation $E_{j'}^\dagger E_{j}$ to denote the action of
those operators on only the $k$ sites where either operator acts nontrivially).
Therefore, the sensitivity loss $\FIloss{Bob}{t}$ can be upper bounded, for any
noisy channel consisting of Kraus operators of weight at most $k/2$, by the
sensitivity loss associated with $k$ located erasures.

\paragraph{If Eve's state is nearly diagonal.}
Computing useful expressions of the Fisher information when a diagonal
representation of the state is not known can be tricky.  The idea if $\rho_E$ is
reasonably close to being diagonal is to hope that one can essentially neglect
the off-diagonal elements of $\rho_E$ and still obtain a good approximation of
the Fisher information via the formula~\eqref{z:r7GWSaRrLfNh}.

Suppose we find an invertible matrix $A$ (with hopefully $A\approx\mathds{1}$) and a
diagonal matrix $\tau = \operatorname{diag}(\tau_0,\ldots,\tau_{d_E}) \geqslant 0$ such that
\begin{align}
  \rho_E = A\,\tau\,A^\dagger\ .
\end{align}
Such a matrix is given for instance by the LDLT or Cholesky decomposition of
$\rho_E$.
(The eigendecomposition of $\rho_E$ also gives such a matrix $A$, but if we can
compute an eigendecomposition one might as well use
\cref{z:r7GWSaRrLfNh} to compute the Fisher information
directly.)
Now we decompose $\widehat{\mathcal{N}}$ by including a scaling factor $\alpha$
as
\begin{align}
  \alpha\widehat{\mathcal{N}} &= \mathcal{N}'\circ \widehat{\mathcal{N}}_0\ ,
\end{align}
with $\alpha=\lVert {A}\rVert ^{-2}\lVert {A^{-1}}\rVert ^{-2}$ and with the two completely positive,
trace-nonincreasing maps
\begin{align}
  \widehat{\mathcal{N}}_0(\cdot)
  &= \frac1{\lVert {A^{-1}}\rVert ^2} A^{-1}\,\widehat{\mathcal{N}}(\cdot)\,A^{-\dagger}\ ,
  \\
  \mathcal{N}'(\cdot)
  &= \frac1{\lVert {A}\rVert ^2} A \,(\cdot)\, A^\dagger\ .
\end{align}
If $A$ is close to $\mathds{1}$ then
we have $\alpha\approx 1$.  Recalling the scaling
property~\eqref{z:OBSqCtuxYn6d} of the quantum Fisher information,
we find
\begin{align}
  \hspace*{2em}&\hspace*{-2em}
  \FIloss{Bob}{t}
  \nonumber\\
  &= \Ftwo{\rho_E}{\widehat{\mathcal{N}}(\{\bar{H},\psi\})}
  \nonumber\\
  &= \frac1\alpha\Ftwo{\alpha\widehat{\mathcal{N}}(\psi)%
    }{\alpha\widehat{\mathcal{N}}(\{\bar{H},\psi\})}
  \nonumber\\
  &\leq \frac1\alpha \Ftwo{ \widehat{\mathcal{N}}_0(\psi) }{%
    \widehat{\mathcal{N}}_0(\{\bar{H},\psi\})
    }
  \nonumber\\
  &= \frac1\alpha \Ftwo*{
    \frac1{\lVert {A^{-1}}\rVert ^2}\,\tau
    \,}{\,%
    \frac1{\lVert {A^{-1}}\rVert ^2} A^{-1} \widehat{\mathcal{N}}\bigl(\{\bar{H},\psi\}\bigr) (A^{-1})^\dagger
    }
  \nonumber\\
  &= \lVert {A}\rVert ^2 \, \Ftwo[\Big]{
    \tau
    }{%
    A^{-1} \widehat{\mathcal{N}}\bigl(\{\bar{H},\psi\}\bigr) (A^{-1})^\dagger
    }\ .
    \label{z:0kw2uyJEILUo}
\end{align}
In the last expression, the Fisher information is evaluated on a state that is
diagonal, so one can directly use~\eqref{z:r7GWSaRrLfNh}.
Furthermore, if $A$ is determined by a LDLT or Cholesky decomposition then it is
lower triangular and its inverse can be computed efficiently (matrix
multiplication of the inverse with another matrix can be done by forward
substitution).

\subsection{Bound in terms of Eve's access to the probe's energy}

In this section, a further bound on Bob's sensitivity to time is presented which
is given in terms of how well Eve can approximate a measurement of energy on the
noiseless clock state.  The properties that Eve can measure on the noiseless probe are
given by the adjoint of the complementary channel: Eve applying an operator $W$
on her system can equivalently be described as the operator
$\widehat{\mathcal{N}}^\dagger(W)$ being applied onto Alice's system, because
$\operatorname{tr}\bigl(\widehat{\mathcal{N}}(\psi)\,W\bigr) =
\operatorname{tr}\bigl(\psi\,\widehat{\mathcal{N}}^\dagger(W)\bigr)$.  One measure of how well Eve
can approximate a measurement of the Hamiltonian around $\lvert {\psi}\rangle $ with an
observable $S$ on her system is the minimum root-mean-squared error
$\min_{S=S^\dagger}\bigl[\bigl \langle { \bigl(\widehat{\mathcal{N}}^\dagger(S) - H
  \bigr)^2 }\bigr \rangle _\psi\bigr]^{1/2}$.
It turns out that the minimum square of this quantity is a lower bound to Bob's Fisher
information to time
\begin{align}
  \FIqty{Bob}{t} \geq
  \min_{S=S^\dagger}
  4\Bigl \langle { \bigl(\widehat{\mathcal{N}}^\dagger(S) - H \bigr)^2 }\Bigr \rangle _\psi\ .
  \label{z:9Ypzyg9WG4oD}
\end{align}
While this bound is aesthetically interesting, finding the optimal $S$ in this
expression is not significantly easier than directly solving the semidefinite
program~\eqref{z:-JCqSeMuhqI6}.
Furthermore, a candidate for $S$ in~\eqref{z:-JCqSeMuhqI6} immediately
provides an upper bound on $F_{\mathrm{Bob}}$, whereas a candidate
in~\eqref{z:9Ypzyg9WG4oD} does not provide any useful
bound on $F_{\mathrm{Bob}}$ because of the direction of the inequality.

The bound~\eqref{z:9Ypzyg9WG4oD} is proven as
follows.  Starting from~\eqref{z:BtvntDGWvOI3} and
using~\eqref{z:-JCqSeMuhqI6},
\begin{align}
  \hspace*{1em}
  &\hspace*{-1em}
  \frac14 \FIqty{Bob}{t}
    \nonumber\\
  &= \min_{S=S^\dagger} \bigl[
  \langle {\bar{H}^2}\rangle  - \operatorname{tr}\bigl(\psi\bigl\{\bar{H},\widehat{\mathcal{N}}^\dagger(S)\bigr\}\bigr)
  + \bigl \langle {\widehat{\mathcal{N}}^\dagger(S^2)}\bigr \rangle 
  \bigr]
    \nonumber\\
  &\geq \min_{S=S^\dagger} \bigl[
  \langle {\bar{H}^2}\rangle  - \bigl \langle {\bigl\{\bar{H},\widehat{\mathcal{N}}^\dagger(S)\bigr\}}\bigr \rangle 
  + \bigl \langle {[\widehat{\mathcal{N}}^\dagger(S)]^2}\bigr \rangle 
  \bigr]
    \nonumber\\
  &= \min_{S=S^\dagger}
    \Bigl \langle {(\bar{H} - \widehat{\mathcal{N}}^\dagger(S))^2}\Bigr \rangle \ ,
    \label{z:yrW79clKupEj}
\end{align}
where we have used
$\widehat{\mathcal{N}}^\dagger(S^2) \geq [\widehat{\mathcal{N}}^\dagger(S)]^2$
(see \cref{z:2McXzF-3arvz} in \cref{z:2dfjPpYV0kEA}).  Finally, we
can replace $\bar{H}$ by $H$ in~\eqref{z:yrW79clKupEj} because any
shifts of $\bar{H}$ by the identity can be canceled out by corresponding shifts of
$S$ by the identity.

\subsection{If Eve can measure the probe's energy almost perfectly}
If Eve has (approximate) access to the energy of the probe state, then this
(approximately) kills sensitivity on Bob's end.  Suppose we can find an
observable $S$ on Eve's system such that
$\lVert {\widehat{\mathcal{N}}^\dagger(S)-\bar H}\rVert \leq\lVert {\bar H}\rVert \,\delta$ and
$\lVert {\widehat{\mathcal{N}}^\dagger(S^2)-\bar H^2}\rVert \leq\lVert {\bar H}\rVert ^2\,\delta$.  Then
\begin{align}
  \FIqty{Bob}{t} \leq 12\delta\,\lVert {\bar H}\rVert ^2\ .
  \label{z:ZdlKhX0fnU2-}
\end{align}
To show this inequality, we first write
$\Delta = \widehat{\mathcal{N}}^\dagger(S) - \bar H$ and
$\Delta' = \widehat{\mathcal{N}}^\dagger(S^2) - \bar H^2$, with
$\lVert {\Delta}\rVert \leq \lVert {\bar H}\rVert \delta$ and
$\lVert {\Delta'}\rVert \leq\lVert {\bar H}\rVert ^2\delta$.  Then,
from~\eqref{z:BtvntDGWvOI3} and
using~\eqref{z:-JCqSeMuhqI6} we obtain
\begin{align}
  \hspace*{1em}
  &\hspace*{-1em}
    \frac14\FIqty{Bob}{t}
    \nonumber\\
  &= \min_{S=S^\dagger} \bigl[
  \langle {\bar{H}^2}\rangle  - \operatorname{tr}\bigl(\bigl\{\psi,\bar{H} \bigr\}\,\widehat{\mathcal{N}}^\dagger(S)\bigr)
  + \bigl \langle {\widehat{\mathcal{N}}^\dagger(S^2)}\bigr \rangle 
  \bigr]
    \nonumber\\
  &\leq \min_{S=S^\dagger} \mathopen{}\left\{
    -\operatorname{tr}\bigl( \bigl\{\psi,\bar H\bigr\}\,\Delta \bigr) + \operatorname{tr}\bigl(\psi\,\Delta'\bigr)
    \right\}\mathclose{}
    \nonumber\\
  &\leq 2\lVert {\bar{H}}\rVert \,\lVert {\Delta}\rVert  + \lVert {\Delta'}\rVert 
  \leq 3\delta\,\lVert {\bar H}\rVert ^2\ .
  \label{z:9UXdk4zK-qI9}
\end{align}

\subsection{Clock sensitivity loss for weak i.i.d.\@ noise}
\label{z:sZUQYA.IesEQ}

Here, we consider an $n$-site system subject to weak i.i.d.\@ noise, where each
site is affected by a noisy channel $\mathcal{N}_\epsilon$ such that
$\mathcal{N}_\epsilon \to {\mathrm{id}}$ if $\epsilon\to0$.  Clearly for
$\epsilon=0$ there is no sensitivity loss.  For a given clock state and
Hamiltonian, we develop a set of tools to understand and determine to which
order $m$ in $\epsilon$ the Fisher information loss is suppressed,
$\FIloss{Bob}{t} = O(\epsilon^m)$.

The question is partly motivated by a similar question in the context of quantum
error correction.  A quantum error-correcting code of distance $d$ can correct
any $(d-1)/2$ arbitrary single-site errors.  In the case of a weak i.i.d.\@
noisy channel $\mathcal{N}_\epsilon^{\otimes n}$ affecting the $n$ sites, a
weight-$[(d-1)/2]$ error happens with probability of order
$O(\epsilon^{(d-1)/2})$ if we assume that a single-site error happens with
probability $O(\epsilon)$.
This means that the chance of an uncorrectable error occurring is upper bounded
by $O(\epsilon^{(d-1)/2})$.  In this scenario, we see that the higher the
distance of the code, the better robustness is achieved against weak i.i.d.\@
noise.
In the context of quantum metrology, we ask the following analogous question:
Can we determine the robustness of the sensitivity of the clock to time when
affected by a weak i.i.d.\ noisy channel, a function of a certain feature
(analogous to the code distance) of the clock state, the Hamiltonian, and the
noisy channel?

There does not appear to be any obvious property of the setup (analogous to the
code distance) that immediately determines the order $m$ in the Fisher
information loss $\FIloss{Bob}{t} = O(\epsilon^m)$.  Instead, we explain a
general procedure for how to obtain a bound on $m$ when given a weak i.i.d.\@
noisy channel, a clock state and a Hamiltonian.

The simplest case presents itself if the complementary channel
$\widehat{\mathcal{N}}_\epsilon^{\otimes n}$ maps the clock state $\psi$ onto a
full-rank state $\rho_E = \sum p_{\boldsymbol x} \lvert {\boldsymbol x}\rangle \mkern -1.8mu\relax \langle{\boldsymbol x}\rvert _E$ that is
diagonal in the tensor product computational basis on $E$.  (This is equivalent
to all vectors $\{ E_{\boldsymbol x} \lvert {\psi }\rangle \}_{\boldsymbol x}$ being
orthogonal on Bob's system.)  In such a case we can
use~\eqref{z:r7GWSaRrLfNh} to express the Fisher information
loss as
\begin{align}
  \FIloss{Bob}{t}
  &= \sum_{\boldsymbol x, \boldsymbol x'}
  \frac2{p_{\boldsymbol x} + p_{\boldsymbol x'}} \,\bigl \lvert {
    \langle {\boldsymbol x}\mkern 1.5mu\relax \vert \mkern 1.5mu\relax {
        \widehat{\mathcal{N}}_\epsilon^{\otimes n}\big(\{\bar{H},\psi\})
    }\mkern 1.5mu\relax \vert \mkern 1.5mu\relax {\boldsymbol x'}\rangle 
  }\bigr \rvert ^2
  \nonumber\\
  &= \sum_{\boldsymbol x, \boldsymbol x'}
  \frac{O(\epsilon^{ 2q_{\boldsymbol x, \boldsymbol x'} })}{\Omega(\epsilon^{\min(r_{\boldsymbol x}, r_{\boldsymbol x'})})}
  \nonumber\\
  &= 
    O(\epsilon^{m})
    \label{z:1qgah0zhTYQY}
\end{align}
defining $r_{\boldsymbol{x}}$ and $q_{\boldsymbol x,\boldsymbol x'}$ via
$p_{\boldsymbol x} = \Omega(\epsilon^{r_{\boldsymbol x}})$ and
$\bigl \lvert {\langle {\boldsymbol x}\mkern 1.5mu\relax \vert \mkern 1.5mu\relax { \widehat{\mathcal{N}}_\epsilon^{\otimes
      n}\big(\{\bar{H},\psi\}) }\mkern 1.5mu\relax \vert \mkern 1.5mu\relax {\boldsymbol x'}\rangle }\bigr \rvert  = O(\epsilon^{q_{\boldsymbol x,
    \boldsymbol x'}})$, and with
\begin{align}
  m = \min_{\substack{\boldsymbol x, \boldsymbol x':\\
  r_{\boldsymbol x}\leq r_{\boldsymbol x'}}}
  \bigl\{ 2q_{\boldsymbol x, \boldsymbol x'} - r_{\boldsymbol x} \bigr\}
  \ .
\end{align}

As we can see above, it is not obvious which $\boldsymbol x,\boldsymbol x'$
minimizes the expression in the exponent above.  One might have expected that
events $\boldsymbol{x}$ whose probability of occurring vanish faster than other
events (large $r_{\boldsymbol{x}}$ compared to other $r_{\boldsymbol{x}'}$) are
less relevant and would not contribute significantly to the Fisher information
loss.  However, 
this is not the case; terms with high
$r_{\boldsymbol{x}},r_{\boldsymbol{x}'}$ can contribute to leading order to the
sensitivity loss if the corresponding term $q_{\boldsymbol{x},\boldsymbol{x'}}$
is sufficiently small.
If the state $\rho_E$ is not diagonal, then it is unclear whether or not one can easily
determine the order of the Fisher information loss.

\section{Clock sensitivity in the presence of continuous noise}
\label{z:Xuq8JaOEWrUo}

The setting presented in \cref{z:FnHui0ahNLxW} is nonstandard in
metrology, because in typical settings the noise and the signal both get
imprinted on the state in the same physical time-evolution process.  It is more
common to consider for instance a Lindbladian master equation that governs the
time evolution of the clock state, with terms that encode any noise processes
via jump operators.

Here we consider the situation where the noise is described by a Lindbladian
master equation.  Under suitable conditions, we can decompose the time evolution
into a pure unitary evolution followed by some effective noisy channel, and the
time dependence of the effective noisy channel can be neglected.  In this case
our \cref{z:T-KNuYhGRM7I} can be applied to compute the
sensitivity loss after some time $t_0$.

One can follow a similar procedure in the setting where the goal is to determine
an unknown parameter in the Hamiltonian when the overall evolution is governed
by a Lindbladian master equation.  The full derivation is presented in
\cref{z:nnVBeRY6bPwI}.
We can carry out a similar decomposition in the case of a clock sensing an
unknown parameter in the Hamiltonian, while subject to continuous noise
described by a Lindblad evolution.

\subsection{Decomposing a Lindbladian evolution of a clock into a
  pure unitary time evolution and an instantaneous noisy channel}
\label{z:.YRCo8Z9DEX4}

Consider a clock initialized at time $t=0$ in the state
vector
$\lvert {\psi_{\mathrm{init}}}\rangle $.  Suppose that the dynamics $\rho(t)$ of the clock
are given by the Lindblad master equation
\begin{subequations}
  \label{z:ddHKrh0Tuhx-}
  \begin{align}
    \partial_t\rho = \mathcal{L}_{\mathrm{tot}}[\rho]\ ,
    \label{z:Rp0XjUAxal5O}
  \end{align}
  where
  \label{z:BsZSNRM3uYvl}
  \begin{gather}
    \begin{aligned}
      \mathcal{L}_{\mathrm{tot}}
      &= \mathcal{L}_0 + \mathcal{L}_1\ ,
      &
      \mathcal{L}_0(\rho) &= -i[H, \rho]\ ,
    \end{aligned}
    \label{z:FMoskr44m.mb}
    \\
    \begin{aligned}
      \mathcal{L}_1(\rho)
      &= \sum_{j} \Bigl[ L_j \rho L_j^\dagger - \frac12 \bigl\{L_j^\dagger L_j, \rho \bigr\} \Bigr]\ .
    \end{aligned}
    \label{z:bQ3kzjvzh32s}
  \end{gather}
\end{subequations}
Here we assume that the operators $H$ and $L_j$ are time independent.
The evolution up to a time $t$ is given by the completely positive,
trace-preserving map
\begin{align}
  \mathcal{E}_{t} = {e}^{t(\mathcal{L}_0 + \mathcal{L}_1)}\ .
  \label{z:am8SEZRPvl1r}
\end{align}
The evolution driven by the Hamiltonian part $\mathcal{L}_0$ of the dynamics can
be written as ${e}^{t\mathcal{L}_0}(\cdot) = {e}^{-iHt}\,(\cdot)\,{e}^{iHt}$.

We would like to compute the sensitivity of the clock at a given time $t_0$,
meaning that the relevant quantity to compute is the Fisher information
\begin{align}
  \FIqty{clock}{t}(t_0)
  &= \Ftwo{ \rho(t_0) }{ \partial_t \rho(t_0) }\ .
    \label{z:RquV6F2IJPiu}
\end{align}

We can decompose the evolution $\mathcal{E}_{t}$ as first a unitary evolution
according to $H$ for a time $t$ followed by the instantaneous application of an
effective noisy channel $\mathcal{N}_{t}$.  Define
\begin{align}
  \mathcal{N}_{t}
  &= \mathcal{E}_t\,{e}^{-t\mathcal{L}_0}
    = {e}^{t(\mathcal{L}_0 + \mathcal{L}_1)}\,{e}^{-t\mathcal{L}_0}\ .
  \label{z:p1QyEW2YDGue}
\end{align}
Here, ${e}^{-t\mathcal{L}_0}$ is the inverse of the unitary evolution
${e}^{t\mathcal{L}_0}$.
By construction, if we apply $\mathcal{N}_{t}$ after applying
${e}^{t\mathcal{L}_0}$, then the overall effect is the same as letting the
system evolve for time $t$ under the full Lindbladian dynamics
$\mathcal{L}_0+\mathcal{L}_1$:
\begin{align}
  \mathcal{E}_{t} &= \mathcal{N}_{t}\,{e}^{t\mathcal{L}_0}\ .
  \label{z:XMdxyV4u6n1N}
\end{align}
An alternative expression for $\mathcal{N}_t$ is obtained
from~\eqref{z:p1QyEW2YDGue} using the
Baker-Campbell-Hausdorff formula,
\begin{align}
  \mathcal{N}_{t}
  &= {e}^{t\mathcal{L}_1 - \frac{t^2}{2}[\mathcal{L}_1,\mathcal{L}_0] + \ldots}\ .
\end{align}
Observe that if $[\mathcal{L}_1,\mathcal{L}_0]=0$, then we simply have
$\mathcal{N}_{t} = {e}^{t\mathcal{L}_1}$. 
This situation is known as \emph{phase-covariant dynamics}
(cf.\@ e.g.\@ Refs.~\cite{R45,R46}).
This is the case if $[L_j,H]=0$ for all jump operators $L_j$.
In other cases, the map can be determined from
\eqref{z:p1QyEW2YDGue} directly if the superoperator
$\mathcal{E}_t$ can be computed.

Let us introduce the family of states
$\psi(t) = {e}^{-iHt} \,\psi_{\mathrm{init}}\, {e}^{iHt}$ associated with the
(fictitious) pure unitary evolution of $\psi_{\mathrm{init}}$ if we artificially
turn off the noise terms.

The derivative of the quantum state
$\rho(t) = \mathcal{E}_t(\psi_{\mathrm{init}})$ can then be written as
\begin{align}
  \partial_t\rho(t)
  &= \partial_t\Bigl[ \mathcal{N}_t\bigl(\psi(t)\bigr) \Bigr]
    =  \mathcal{N}_{t}\bigl( \partial_t\psi\, (t) \bigr)
    + \bigl(\partial_t \mathcal{N}_{t}\bigr)\bigl(\psi(t)\bigr)\ .
    \label{z:7Uz5bUPxySIV}
\end{align}
Therefore, the derivative of the noisy state can be decomposed into a sum of two
terms, the first associated with the unitary dynamics $\psi(t)$, and the other
associated with the time dependence of the effective noisy channel
$\mathcal{N}_t$.  Plugging into~\eqref{z:RquV6F2IJPiu}, this
gives us
\begin{align}
  \FIqty{clock}{t}
  &= \Ftwo[\Big]{ \mathcal{N}( \psi ) }{%
  \mathcal{N}( \partial_t \psi )
  + \bigl(\partial_t \mathcal{N}\bigr)(\psi)
  }\ ,
\end{align}
where now $\FIqty{clock}{t}$,
$\mathcal{N}$, %
$\partial_t \mathcal{N}$, %
$\psi$ %
and $\partial_t \psi $ %
are all implicitly evaluated at $t_0$.

In the following, we consider settings where the local time dependence of the
state due to the time dependence of the effective noisy channel terms can be
neglected when computing $\FIqty{clock}{t}$.  (We will study in greater
depth below when exactly this situation arises.)
I.e., for now we assume that
\begin{align}
  \FIqty{clock}{t}
  \approx \Ftwo{ \mathcal{N}(\psi) }{
  \mathcal{N}( \partial_t\psi ) }
  =:
  \FIqty{clock,U}{t}
  \ .
  \label{z:svEq8bs.6Q93}
\end{align}
Expanding $\partial_t \psi$, we obtain
\begin{align}
  \FIqty{clock,U}{t}
  = \Ftwo{ \mathcal{N}(\psi) }{ \mathcal{N}(-i[H,\psi]) }\ .
\end{align}
This quantity is what we defined as $\FIqty{Bob}{t}$ in the context of 
our main
uncertainty relation.

The complementary channel $\widehat{\mathcal{N}}_{t_0}$ is directly determined
by the complementary channel of the overall evolution up to that time
$\widehat{\mathcal{E}}_{t_0}$, since the two channels differ only by a unitary
evolution ${e}^{-t_0\mathcal{L}_0}$ on their input:
\begin{align}
  \widehat{\mathcal{N}_t} &= \widehat{\mathcal{E}_t} \, {e}^{-t_0\mathcal{L}_0}\ .
\end{align}
This means that the Fisher information on Eve's end with respect to the
complementary direction can be expressed entirely in terms of the complementary
channel $\widehat{\mathcal{E}}_{t_0}$ to the entire evolution up to time $t_0$:
\begin{align}
  \FIloss{clock,U}{t} =
  \Ftwo{\widehat{\mathcal{E}}_{t_0}(\psi_{\mathrm{init}})}%
  {\widehat{\mathcal{E}}_{t_0}\bigl( \{ \bar{H}, \psi_{\mathrm{init}} \}\bigr)}\ ,
\end{align}
with $\bar{H} = H - \langle {H}\rangle _{\psi(t_0)}$, and
\cref{z:T-KNuYhGRM7I} states that
\begin{align}
  \FIqty{clock,U}{t}
  = 4\sigma_H^2 - 
  \FIloss{clock,U}{t}\ .
\end{align}

Now we turn to discussing when the
approximation~\eqref{z:svEq8bs.6Q93} is a
reasonable assumption, by characterizing the error induced on the Fisher
information.  First of all, the approximation is exact in the case of
phase-covariant dynamics, where $[\mathcal{L}_1, \mathcal{L}_0]=0$ (e.g.\@
 Refs.~\cite{R45,R46}).
In other settings, 
we can use a continuity bound of
the Fisher information in its second argument
(\cref{z:Y81L4FyklNd5} in
\cref{z:KH4B5FtzQ1kE}%
) to try to get a handle on the error terms involved in the
approximation~\eqref{z:svEq8bs.6Q93}.  Denote by
$\delta$ %
the error in the
approximation~\eqref{z:svEq8bs.6Q93},
\begin{align}
  \delta = 
  \FIqty{clock}{t}
  - \FIqty{clock,U}{t}
  \ ,
\end{align}
then we have
\begin{multline}
  \lvert {\delta}\rvert 
  \leq \Ftwo{\rho}{ (\partial_t\mathcal{N})(\psi) }
  \\
  + \mathopen{}\left[ \Ftwo{\rho}{ (\partial_t\mathcal{N})(\psi) } \,
            \FIqty{clock,U}{t}\right]\mathclose{}^{1/2}\ .
  \label{z:raWZngbty0t4}
\end{multline}
That is, the relative error in 
the
approximation~\eqref{z:svEq8bs.6Q93} is
demonstrably small if $\Ftwo{ \rho }{ (\partial_t\mathcal{N})(\psi) }$ is much
smaller than $\FIqty{clock,U}{t}$.  We can rewrite this term
using~\eqref{z:7Uz5bUPxySIV} as
\begin{align}
  (\partial_t \mathcal{N})(\psi)
  &= \partial_t \rho(t) - \mathcal{N}_t( -i[H, \psi(t)] )
    \nonumber\\
  &= \mathcal{L}_{\mathrm{tot}}[\rho(t)] - \mathcal{E}_t( -i[H, \psi_0] )\ .
    \label{z:oTTBiCd7U6dy}
\end{align}
The above expression is given in terms of the Lindbladian map and the overall
evolution map, and can aid in determining an analytical or numerical upper bound
to the quantity $\Ftwo{\rho}{(\partial_t\mathcal{N})(\psi)}$.
 In \cref{z:nnVBeRY6bPwI}, we study two single-qubit examples that
are subject to continuous dephasing along various axes in order to illustrate
the connections between the Lindbladian setting and the setting
in~\cref{z:FnHui0ahNLxW}.

\section{Error-correction conditions for zero sensitivity loss}
\label{z:JZtfKZv25too}

The uncertainty relation~\eqref{z:4bic5ys.wcEe} enables us to
provide a characterization of when the noise reduces a probe's sensitivity to
time.  In this section, we study the situation where the sensitivity loss
$\FIloss{Bob}{t}$ introduced in~\eqref{z:QZajZctWEncO} is equal to zero. This is a
situation where the probe is chosen cleverly enough such that the noise has no
effect on sensitivity. The main contribution of this section is a set of
necessary and sufficient conditions for $\FIloss{Bob}{t}=0$, which bear
resemblance to the Knill-Laflamme conditions for quantum error
correction~\cite{R19} and which are closely related to the
Hamiltonian-not-in-Lindblad-span condition of
Refs.~\cite{R20,R21}.

\subsection{Conditions for zero sensitivity leakage}

In the following, we suppose that our uncertainty relation holds with
equality, i.e., that the conditions given in
\cref{z:3JJW1LK1MXi9} hold.  Recall the
expression for the Fisher information loss on Bob's
end~\eqref{z:4bic5ys.wcEe}, and consider the
expression~\eqref{z:fC3N68oVdRLb} for the Fisher information.  If
$\FIloss{Bob}{t}=0$, then there exists an operator $L$ such that
$\operatorname{tr}(L^\dagger L)=0$ and
$\rho^{1/2}L+L^\dagger\rho^{1/2} = \widehat{\mathcal{N}}(\{\psi, \bar{H}\})$; the
former condition implies $L=0$ and thus the latter implies
$\widehat{\mathcal{N}}(\{\psi, \bar{H}\})=0$.  Therefore, we see that
$\FIloss{Bob}{t}=0$ if and only if
\begin{align}
  \widehat{\mathcal{N}}(\{\psi, \bar{H}\}) = 0\ ,
  \label{z:zFMncm5bkXDt}
\end{align}
i.e., $\{\psi, \bar{H}\}$ must lie in the kernel of the superoperator
$\widehat{\mathcal{N}}$.
It is instructive to rewrite this condition in terms of the ``virtual qubit''
  introduced in \cref{z:aLDqVRGU3OD.}.  With $Z_L$ defined in~\eqref{z:puhbpJK6ngeA},
then~\eqref{z:zFMncm5bkXDt} becomes
\begin{align}
  \widehat{\mathcal{N}}(Z_L) &= 0\ .
  \label{z:GOkK7nYKf.l7}
\end{align}
Alternatively, the above condition is equivalent to requiring that for all
operators $O$,
\begin{align}
  \operatorname{tr}\bigl[ \widehat{\mathcal{N}}^\dagger(O)\, Z_L \bigr] &= 0\ ,
  \label{z:s7BIFNKk1XcA}
\end{align}
meaning that error operations of the form
  $\widehat{\mathcal{N}}^\dagger(O)$ should not have any overlap with the
  ``logical'' $Z_L$ operator on the qubit subspace.

So the task of finding probe states that perfectly counter the noisy channel
$\mathcal{N}$ can be formulated as ensuring the logical $Z$ Pauli operator in
the logical qubit subspace spanned by $\lvert {+}\rangle =\lvert {\psi}\rangle $ and
$\lvert {-}\rangle \propto \lvert {\xi}\rangle =P_\psi^\perp H\lvert {\psi}\rangle $ is in the kernel of the
complementary channel to the noisy channel.

Note that simply looking for zero sensitivity loss is not sufficient to find the
best probe states; we still need to make sure that $\lvert {\psi}\rangle $ has as large
energy variance as possible to ensure good sensitivity.

An alternative representation of the zero sensitivity-loss condition can be
obtained if we consider an operator-sum representation of the noisy channel in
terms of Kraus operators $\{ E_k \}$ as in~\eqref{z:jg7XnoGzLtS0}.
The condition~\eqref{z:zFMncm5bkXDt} is then equivalent to
the condition
\begin{align}
  \langle {\psi }\rvert \, E_{k'}^\dagger E_k \, \lvert {\xi
  }\rangle + \langle {\xi }\rvert \, E_{k'}^\dagger E_k \, \lvert {\psi
  }\rangle = 0\ \quad\ \text{for all $k,k'$.}
  \label{z:G2tCYcCXVZFu}
\end{align}
These may be interpreted as Knill-Laflamme-like conditions for optimal
sensitivity.  Whereas for a traditional quantum error-correcting code, we require
any two code words $\lvert {\psi_i}\rangle ,\lvert {\psi_j}\rangle $ to satisfy
$\langle {\psi_i}\rvert  E_{k'}^\dagger E_k \lvert {\psi_j}\rangle  \propto \delta_{i,j}$, here we
require that the error operator $E_{k}^\dagger E_{k'}$ cannot map the state
$\lvert {\psi}\rangle $ onto the vector $\lvert {\xi}\rangle $, or at least not in a way that is not
suitably antisymmetric.  The weird antisymmetrization
in~\eqref{z:G2tCYcCXVZFu} can be expressed in a more
elegant form if we switch back to the picture of the logical qubit spanned by
$\lvert {\psi}\rangle $ and $\lvert {\xi}\rangle $.  Analogously
to~\eqref{z:GOkK7nYKf.l7}, we may rewrite the
condition~\eqref{z:G2tCYcCXVZFu} as
\begin{align}
  \operatorname{tr}\bigl[ Z_L\, \Pi_L\,E_{k'}E_k\,\Pi_L\bigr] = 0\ ,
  \label{z:mvEtMj.b-Bs6}
\end{align}
where $\Pi_L = \lvert {+}\rangle \mkern -1.8mu\relax \langle{+}\rvert _L + \lvert {-}\rangle \mkern -1.8mu\relax \langle{-}\rvert _L$ is the projector onto the virtual qubit
  subspace spanned by $\lvert {\psi}\rangle $ and $\lvert {\xi}\rangle $.  The full Knill-Laflamme
conditions applied to the subspace $\Pi_L$ would require
$\Pi_L E_{k'}E_k\Pi_L \propto \Pi_L$.  The
condition~\eqref{z:mvEtMj.b-Bs6} is simply a
weaker condition where only the corresponding projection onto the logical Pauli
operator $Z_L$ is considered and where the projection onto the other Pauli
operators is unconstrained.

The form~\eqref{z:mvEtMj.b-Bs6} also helps
clarify that for zero sensitivity loss, the terms
in~\eqref{z:G2tCYcCXVZFu} need not vanish individually.
Indeed, only the Hilbert-Schmidt projection of $\Pi_L\,E_{k'}E_k\,\Pi_L$ onto
$Z_L$ is required to vanish, and not in principle on $Y_L$ or $X_L$.  An example
below in \S\ref{z:kLDpse.aVOd5}, consisting of a single-qubit subject to transversal noise, will
illustrate this point.

The conditions~\eqref{z:mvEtMj.b-Bs6} are
reminiscent of quantum error correction for operator algebras, where we require
a code to preserve the outcomes of any operator in a given
algebra~\cite{R47,R48,R49}.  In fact, if the algebra
associated with any choice of optimal sensing operator of the
form~\eqref{z:w7NUzrfGUNE-} is preserved, then our
conditions~\eqref{z:mvEtMj.b-Bs6} are satisfied.
Indeed, suppose that $[T,\widehat{\mathcal{N}}^\dagger(W)] =0$ for any operator
$W$ on Eve and for a fixed choice of 
${M}$
in~\eqref{z:w7NUzrfGUNE-}, meaning that the Abelian algebra
generated by $T$ is correctable~\cite{R47,R48,R49}.  Then taking the expectation
value $\langle {\cdot}\rangle _\psi$ of this commutator we find
$0 = \bigl \langle { [T,\widehat{\mathcal{N}}^\dagger(W)] }\bigr \rangle  = \operatorname{tr}\bigl( [\psi,
T]\,\widehat{\mathcal{N}}^\dagger(W) \bigr) \propto \operatorname{tr}\bigl(
Z_L\,\widehat{\mathcal{N}}^\dagger(W) \bigr)$,
using~\eqref{z:Df3ZpFNLdboK}, which holds
for all $W$, and therefore our Knill-Laflamme-like
condition~\eqref{z:s7BIFNKk1XcA}
holds.  The converse implication is unclear, in part because the optimal sensing
operator is not unique and different choices can generate different algebras.
The conditions~\eqref{z:G2tCYcCXVZFu} are actually
tightly related to the Hamiltonian-not-in-Kraus-span condition of
Refs.~\cite{R29,R30,R50,R20,R21,R51,R52}. There, it was shown that there exists a clock state
vector $\lvert {\psi}\rangle $ that achieves Heisenberg scaling %
in the presence of noise using quantum error correction if and only if the
Hamiltonian signal term is not in the linear span of the Lindblad noise
operators.
Here we argue that the Hamiltonian-not-in-Kraus-span condition is
in fact equivalent to the existence of a state $\lvert {\psi}\rangle $ that satisfies
our zero sensitivity-loss
conditions~\eqref{z:G2tCYcCXVZFu}.
(In our setting, the clock state vector $\lvert {\psi}\rangle $ is a given fixed
state.)  As we have a discrete noisy channel, we consider the Kraus operators
$\{ E_k \}$ of the noisy channel instead of Lindblad operators.  If
$H=\sum\alpha_{k',k} E_{k'}^\dagger E_k$, and supposing the
conditions~\eqref{z:G2tCYcCXVZFu} are satisfied for some
$\lvert {\psi}\rangle $, then by taking a linear combination $\sum \alpha_{k',k}$ of the
conditions~\eqref{z:G2tCYcCXVZFu} we obtain
$0 = 2\langle {\psi }\rvert H \, P_\psi^\perp H \lvert {\psi }\rangle = 2\sigma_H^2$; therefore the
conditions~\eqref{z:G2tCYcCXVZFu} cannot be satisfied by
any $\psi$ that has nonzero energy variance.  Conversely, we know (see, e.g.,
Refs.~\cite{R20,R21,R51}) that if the Hamiltonian is not in the span of the
noisy channel's Kraus operators, then there is a code space $\Pi$, possibly
involving an ancilla system, with $\Pi E_{k'}^\dagger E_k \Pi = c_{k', k}\Pi$
such that $[\Pi,H]=0$ (i.e., $\Pi$ is spanned by a subset of energy
eigenvectors) and such that $\Pi$ contains a state vector $\lvert {\psi}\rangle $ with
nonzero energy variance; then for any $k,k'$ we have
$\langle {\psi}\mkern 1.5mu\relax \vert \mkern 1.5mu\relax {E_{k'}^\dagger E_k P_\psi^\perp H}\mkern 1.5mu\relax \vert \mkern 1.5mu\relax {\psi}\rangle  = \langle {\psi}\mkern 1.5mu\relax \vert \mkern 1.5mu\relax {\Pi
  E_{k'}^\dagger E_k P_\psi^\perp H \Pi}\mkern 1.5mu\relax \vert \mkern 1.5mu\relax {\psi}\rangle  = \langle {\psi}\mkern 1.5mu\relax \vert \mkern 1.5mu\relax {\Pi
  E_{k'}^\dagger E_k \Pi \, P_\psi^\perp H}\mkern 1.5mu\relax \vert \mkern 1.5mu\relax {\psi}\rangle  = c_{k', k} \langle {\psi}\mkern 1.5mu\relax \vert \mkern 1.5mu\relax {
  P_\psi^\perp H}\mkern 1.5mu\relax \vert \mkern 1.5mu\relax {\psi}\rangle  = 0$ using the fact that
$[P_\psi^\perp,\Pi] = [H,\Pi] = 0$, so the
conditions~\eqref{z:G2tCYcCXVZFu} are
satisfied. Therefore, if the Hamiltonian is not in the span of the Kraus
operators, then there exists a clock state vector $\lvert {\psi}\rangle $ that suffers no
sensitivity loss after being exposed to the noise locally at $t_0$.  This state
is constructed in the above mentioned references using a quantum error-correcting
code.

We can ask whether there is a relation between our conditions for no
  sensitivity loss and when the sensitivity can achieve Heisenberg scaling in
  the system size~\cite{R1}.  The Heisenberg scaling refers to situations where
$\FIqty{Bob}{t}$ scales like $n^2$, where $n$ is the number of systems that are
jointly prepared in the clock state vector $\lvert {\psi}\rangle _n$. %
(If no entanglement is present between the $n$ systems, the best scaling that
can be achieved is $\FIqty{Bob}{t} \propto n$.)  We assume that the clock state vector 
$\lvert {\psi}\rangle _n$ has a variance that scales quadratically in $n$, i.e.,
$[\sigma_H(\psi_n)]^2 \propto n^2$, as otherwise even the noiseless clock does
not achieve Heisenberg scaling.  Suppose the
conditions~\eqref{z:G2tCYcCXVZFu} are satisfied: Then
$\FIqty{Bob}{t} = 4[\sigma_H(\psi_n)]^2 \propto n^2$ as there is no sensitivity
loss, and the Fisher information displays Heisenberg scaling.
On the other hand, even if there is some loss of sensitivity due to the noise, the Heisenberg scaling might survive. Suppose, for example,
that we consider
two independent one-dimensional spin chains, each consisting of $n/2$ sites that are prepared
in a GHZ state and that evolve according to an on-site $Z$ Hamiltonian.  Both
spin chains are independent probes whose sensitivity each scales as $\sim n^2$, and
therefore the overall probe state exhibits Heisenberg scaling.  Now consider the
noisy channel that erases one of the spin chains.  Half the sensitivity is lost;
because there is sensitivity loss our Knill-Laflamme-like conditions cannot be
satisfied.  However, 
the single spin chain that is left for Bob still exhibits
Heisenberg scaling.
This shows that Heisenberg scaling is guaranteed if the environment has zero
sensitivity to energy (and the noiseless probe itself has Heisenberg scaling),
but that there are also situations where the environment induces sensitivity
loss without hindering the Heisenberg scaling of the probe.  In the language of
Refs.~\cite{R20,R21}, this corresponds to a Hamiltonian that might have
both a parallel component to the signal as well as a perpendicular component
that can be exploited to achieve Heisenberg scaling.  %
We see that zero sensitivity loss implies Heisenberg scaling for a family
of state vectors $\lvert {\psi}\rangle $ that are sufficiently entangled.  But there are
states that achieve the Heisenberg scaling even if some sensitivity is lost due
to the noise.

When the zero sensitivity-loss conditions~\eqref{z:zFMncm5bkXDt} hold, then by definition there must
exist a sensing observable for Bob to estimate the parameter $t$, whose
sensitivity matches that of Alice.  We can extract this optimal sensing
observable from our technical analysis using semidefinite programming (see
\cref{z:VM5-jRd.0oMe}).  Namely, in
\cref{z:r0U85JMqlcN8} we show that if the
zero sensitivity-loss conditions hold, then the operator
$i\rho\mathcal{N}(\lvert {\xi}\rangle \mkern -1.8mu\relax \langle{\psi}\rvert )$ is Hermitian.  Furthermore, the operator
\begin{align}
  R_B = -2i\mathcal{N}(\lvert {\xi}\rangle \mkern -1.8mu\relax \langle{\psi}\rvert )\rho^{-1}
  + 2i\rho^{-1}\mathcal{N}(\lvert {\psi}\rangle \mkern -1.8mu\relax \langle{\xi}\rvert ) P_\rho^\perp
\end{align}
is also Hermitian and satisfies $\frac12\{ R_B, \rho_B \} = \mathcal{N}(Y_L)$,
i.e., we obtain an explicit expression of the symmetric logarithmic derivative
on Bob's end.  The optimal sensing observable on Bob's system is then given
via~\eqref{z:fRkTtRxHhvuE} as $T_b = [\FIqty{Bob}{t}]^{-1} R_B + t_0$.
That is, when a clock state and associated Hamiltonian fulfill the metrological
code conditions for a given noise channel, we obtain an explicit expression for
the optimal measurement on Bob's end.

\subsection{Metrological codes and metrological distance}
\label{z:xoTJRt4h5gjZ}

We now introduce the concept of a \emph{metrological code.}  The idea is to
study the qubit space spanned by the vectors $\lvert {\psi}\rangle $ and
$\lvert {\xi }\rangle = \bar{H}\lvert {\psi }\rangle = \bigl(H - \langle {H}\rangle \bigr)\,\lvert {\psi}\rangle $.  If the state loses
no sensitivity upon the action of a noisy channel, one could expect these states
to span some kind of quantum error-correcting code space.  We can see that they
do not necessarily form a full error-correcting code as follows.  Consider the
single-qubit state $\lvert {\psi}\rangle =\lvert {+}\rangle =\bigl[\lvert {0}\rangle +\lvert {1}\rangle \bigr]/\sqrt2$ evolving under the
Hamiltonian $H=\omega\sigma_Z/2$, which we expose to an error channel whose
Kraus operators are proportional to $\mathds{1}$ and $X$.  We see that the
condition~\eqref{z:G2tCYcCXVZFu} is satisfied, given that
$\lvert {\xi }\rangle = \lvert {-}\rangle  = \bigl[\lvert {0}\rangle -\lvert {1}\rangle \bigr]/\sqrt2$ is orthogonal to $\lvert {+}\rangle $ and that
$\lvert {+}\rangle $ is an eigenstate of both $\mathds{1}$ and $X$.  Yet a quantum state stored
on this qubit would be corrupted by the noise, as the bit flips would be
uncorrectable.  We identify a concept that is weaker than a full
error-correcting code, which applies precisely to states that satisfy the
condition~\eqref{z:G2tCYcCXVZFu}.  Here, we assume that
the setting is specified as a pair of orthogonal states $\lvert {\psi}\rangle ,\lvert {\xi}\rangle $,
whereby $\lvert {\xi}\rangle $ is presumably obtained from a Hamiltonian $H$ as
$\lvert {\xi }\rangle = \bigl(H - \langle {H}\rangle \bigr)\lvert {\psi}\rangle $.  Specifying the full Hamiltonian is not
necessary as the relevant quantum Fisher information quantities can be fully
expressed only in terms of $\lvert {\psi}\rangle ,\lvert {\xi}\rangle $.

\begin{thmheading}{Metrological code}
Let $\mathscr{E}$ be any set of operators.  We say that the state vectors
$\lvert {\psi}\rangle $ and $\lvert {\xi}\rangle $ form
 a \emph{metrological code} against the errors
$\mathscr{E}$ if for all $E,E'\in\mathscr{E}$, we have
\begin{align}
  \operatorname{tr}\bigl[ E'^\dagger E \, \bigl( \lvert {\xi}\rangle \mkern -1.8mu\relax \langle{\psi }\rvert + \lvert {\psi}\rangle \mkern -1.8mu\relax \langle{\xi }\rvert \bigr) \bigr] = 0\ .
  \label{z:q5Fv8e19UZL5}
\end{align}
\end{thmheading}
As a consequence of the zero sensitivity-loss
condition~\eqref{z:G2tCYcCXVZFu}, a metrological code
prevents sensitivity loss against any noise channel whose Kraus operators are
linear combinations of elements in $\mathscr{E}$ (as long as the conditions of
\cref{z:3JJW1LK1MXi9} are satisfied).

A natural class of errors to consider is the set of all operators that act
on only a subset of $n$ components %
of a composite quantum system
$A = A_1\otimes A_2\otimes\cdots\otimes A_n$.
The \emph{weight} $\wgt(O)$ of an operator $O$ acting on the $n$ systems is
defined as the number of systems on which $O$ acts nontrivially.  Specifically,
if $O$ is expanded in the Pauli operator basis (or in any tensor basis using a
single-site operator basis that includes the identity matrix), all non-identity
elements in tensor products of basis operators that appear in the decomposition
of $O$ must be supported on a fixed set of $\wgt(O)$ sites.  Equivalently, the
expectation value of $O$ on any state can be computed exactly even after tracing
out all but a given set of $\wgt(O)$ sites.

We say that the pair of state vectors $\lvert {\psi}\rangle $ and %
$\lvert {\xi}\rangle $ form a \emph{metrological code of distance $d_m$} if it is a
metrological code against all operators of weight at most $d_m-1$;  i.e., for
all operators $O$ satisfying $\wgt(O)<d_m$, we have
\begin{align}
  \operatorname{tr}\bigl[ O\bigl(\lvert {\xi}\rangle \mkern -1.8mu\relax \langle{\psi }\rvert + \lvert {\psi}\rangle \mkern -1.8mu\relax \langle{\xi}\rvert \bigr) \bigr] = 0\ .
  \label{z:y2lXSpYS4Z-t}
\end{align}
Metrological codes of distance $d_m$ have the property that for any noise
channel $\mathcal{N}$ whose Kraus operators $\{E_k\}$ are such that
$\wgt(E_{k'}^\dagger E_k) < d_m$ for all $k',k$, the associated sensitivity loss
is zero (as long as \cref{z:3JJW1LK1MXi9} is
satisfied).

Metrological codes are, roughly speaking, in between classical and quantum
codes.  On one hand, they are not full-blown classical codes because
condition~\eqref{z:y2lXSpYS4Z-t} requires
protection against both \(X\)- and \(Z\)-type physical noise.  Because of this,
the pair \(\lvert {\psi}\rangle \propto \lvert {0}\rangle ^n + \lvert {1}\rangle ^n\) and
\(\lvert {\xi}\rangle \propto \lvert {0}\rangle ^n - \lvert {1}\rangle ^n\) of GHZ states \textit{is not} a
metrological code of nontrivial distance because single-qubit \(Z\) errors cause
a logical-\(X\) error, thereby
violating~\eqref{z:y2lXSpYS4Z-t}.  On the
other hand, metrological codes are not full-blown quantum codes because the
sensitivity conditions say nothing about other types of logical noise.  In other
words, noise can cause logical-\(Y\) and logical-\(Z\) errors for a metrological
code, but not for a bona-fide error-correcting code.

\subsection{Uncertainty relation equality and conditions
  for metrological codes}

In order to deduce from Eve's lack of sensitivity to energy that Bob
loses no sensitivity to time, it is necessary to ensure that the conditions of \cref{z:3JJW1LK1MXi9} hold.
When we presented
\cref{z:3JJW1LK1MXi9}, we already noted that
the situations where these conditions are not satisfied are edge cases that can
be perturbed away.  Here, we strengthen this statement for metrological codes:
If a metrological code for a given noise channel happens not to satisfy the
conditions of \cref{z:3JJW1LK1MXi9}, then the
noise channel can be infinitesimally perturbed to obtain a situation for which
these conditions hold, and furthermore, the zero sensitivity-loss conditions~\eqref{z:zFMncm5bkXDt} are
preserved.  

\begin{proposition}[{Perturbation bound for noise channels consistent
    with a metrological code}]\noproofref
  \label{z:F6rfqyfKNr41}
  Let $V_{A\to BE}$ be an isometry, let $\lvert {\psi}\rangle _A,\lvert {\xi}\rangle _A$ with
  $\langle {\psi}\mkern 1.5mu\relax \vert\mkern 1.5mu\relax {\xi}\rangle _A=0$ and let
  $\mathcal{N}(\cdot) = \operatorname{tr}_E\bigl(V\,(\cdot)\,V^\dagger\bigr)$,
  $\widehat{\mathcal{N}}(\cdot) = \operatorname{tr}_B\bigl(V\,(\cdot)\,V^\dagger\bigr)$.
  Suppose that 
  $\widehat{\mathcal{N}}(\lvert {\xi}\rangle \mkern -1.8mu\relax \langle{\psi}\rvert +\lvert {\psi}\rangle \mkern -1.8mu\relax \langle{\xi}\rvert ) = 0$.
  We furthermore assume that there exists a unitary operator $G_B$ acting on the
  system $B$ with the properties that
  ${0 = P_{\rho_B} G_B P_{\rho_B}} = P_{\zeta_B} G_B P_{\zeta_B} = P_{\rho_B} G_B
  P_{\zeta_B} = P_{\zeta_B} G_B P_{\rho_B}$, where
  $\zeta_B = \mathcal{N}(\lvert {\xi}\rangle \mkern -1.8mu\relax \langle{\xi}\rvert )$.
  Then, for any $\epsilon>0$, there exists an isometry $V'_{A\to BE}$ with
  $%
  \lVert {V' - V}\rVert  \leq \epsilon
  $
  such that
  \begin{subequations}
    \begin{align}
      \bigl(P_{\rho_B'}^\perp\otimes P_{\rho_E'}^\perp\bigr) V'\lvert {\xi }\rangle &= 0\ ;\quad\text{and}
      \\
      \widehat{\mathcal{N}}'\bigl(\lvert {\xi}\rangle \mkern -1.8mu\relax \langle{\psi}\rvert +\lvert {\psi}\rangle \mkern -1.8mu\relax \langle{\xi}\rvert \bigr) &= 0\ ,
    \end{align}
  \end{subequations}
  where
  $\rho_B' = \operatorname{tr}_E\bigl\{ V'\psi V'^\dagger \bigr\}$,
  $\rho_E' = \operatorname{tr}_B\bigl\{ V'\psi V'^\dagger \bigr\}$, and
  $\widehat{\mathcal{N}}'(\cdot) = \operatorname{tr}_B\bigl\{ V' \, (\cdot)\, V'^\dagger \bigr\}$.
\end{proposition}

The proof is presented as \cref{z:QmMajA-tmAuf} in \cref{z:DQLLjQ-XUmxW}.
Note that the existence of such an operator $G_B$ can always be ensured by
augmenting the $B$ system to include a qubit which $\mathcal{N}$ prepares in a
fixed pure state vector $\lvert {0}\rangle $ for all inputs. The operator $G_B$ can be chosen to
flip the qubit to $\lvert {1}\rangle $.  The additional qubit can represent an additional
``failure'' flag such as, for instance, an additional photon that is emitted at the
output of the noise process.

\subsection{Sensitivity loss of metrological codes under weak i.i.d.\@ noise}

\paragraph{Sensitivity loss under weak i.i.d.\@ noise.}
If we encode a logical quantum state using a quantum error-correcting code of a
distance $d$, and each site has a small probability $O(\epsilon)$ of incurring
an error, then we know that the errors that the code cannot correct occur with
probability at most $O(\epsilon^{d/2})$.  In turn, this implies that the
infidelity of recovery of the logical information also scales as
$O(\epsilon^{c d})$ with a constant $c$ depending on which convention for the
infidelity measure we choose.
It is then natural to conjecture that if $\lvert {\psi}\rangle $ and $\lvert {\xi}\rangle $ form a
metrological code of metrological distance $d_m$, then the loss in Fisher
information must similarly be upper bounded by $O(\epsilon^{c d_m})$, for some
universal constant $c$.

Interestingly, the order of the Fisher information loss in $\epsilon$ is not
directly related to the metrological distance of a metrological code.  In fact,
there are examples of metrological codes with large metrological distance, but
for which the Fisher information loss is always of order $\epsilon$.
This behavior appears to contradict the expectation that events of
vanishing probability should not significantly influence observable properties
of the system (such as its sensitivity to time).  An explanation 
stems from the fact that the operational interpretation of the Fisher information
via the Cram\'er-Rao bound involves an implicit averaging of the error over
infinitely many samples.  It might turn out in the present case that events with
vanishing probability can contribute nonnegligibly to the quantum Fisher
information.  To remedy this issue, it would be desirable to consider a measure
of sensitivity that accounts for finite data acquisition.  One such measure has been put forward in Ref.~\cite{R53}.
We refer to \cref{z:z.6FF4fEKkkd} for a more detailed
discussion.

\subsection{Clock states from time-covariant quantum error-correcting codes}
\label{z:AVGHDyn1nGzm}

Here, we explore a simple method to construct states that satisfy the zero
sensitivity-loss condition, using time-covariant quantum error-correcting codes. A code is said to be time-covariant code with respect to a given
  Hamiltonian $H$ if $H$ (and hence also time evolution generated by  $H$)
is a nontrivial logical operator.
In the following, Pauli operators $X,Y,Z$ carry an index indicating the qubit on which the operator acts.
This strategy is the one pursued by, e.g., Refs.~\cite{R50,R21,R51,R52}.

\paragraph{Four nearest-neighbor interacting qubits in a square pattern.}
As a warm-up example, we first consider how to leverage the $[[4,2,2]]$ code for
quantum metrology with a Hamiltonian on four qubits with $ZZ$ interactions
arranged in a square pattern.
\begin{figure}
  \centering
  \includegraphics{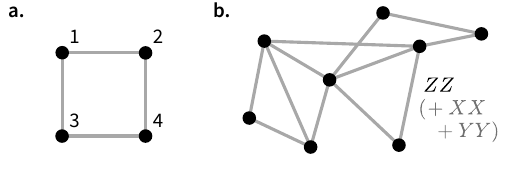}
  \caption{%
    Metrology with interacting qubits.  \textbf{a.}~Consider four qubits in a
    square with nearest-neighbor $ZZ$ Ising interactions (alternatively with
    additional $XX$ and $YY$ interactions). A clock state with maximal
    sensitivity and zero sensitivity loss under a single located erasure can be
    obtained via the time-covariant $[[4,2,2]]$ code.  \textbf{b.}~We can extend the
    construction based on the $[[4,2,2]]$ 
    code to any number of qubits interacting
    with respect to any graph of $ZZ$ interactions (alternatively with
    additional $XX$ and $YY$ interactions), while offering protection against a
    single located erasure.}
  \label{z:q6y-xgqvsyUu}
\end{figure}
Consider four qubits arranged in a square as depicted in
\cref{z:q6y-xgqvsyUu}.  The Hamiltonian is defined by placing a $ZZ$
interaction on each side of the square,
\begin{align}
  H = \omega \bigl(Z_1Z_2 + Z_1Z_3 + Z_2Z_4 + Z_3Z_4\bigr)\ .
  \label{z:8Smzekcjhom5}
\end{align}
The $[[4,2,2]]$ code~\cite{R54,R55} has stabilizers $X_1X_2X_3X_4$ and $Z_1Z_2Z_3Z_4$.
The logical operators $X_1,Z_1$ and $X_2,Z_2$ for the first and second
  logical qubits are $\overline{X}_1 = X_1X_3$, %
$\overline{X}_2 = X_1X_2$, %
$\overline{Z}_1 = Z_1 Z_2$, %
and $\overline{Z}_2 = Z_1Z_3$.

Observe that the Hamiltonian is a logical operator: The second and fourth
terms in~\eqref{z:8Smzekcjhom5} have the same action on the code
space as the first and third terms, respectively, because they differ only by
the stabilizer $Z_1Z_2Z_3Z_4$.  When acting on the code space, we have
\begin{align}
  H \, \Pi =  2\omega\bigl( \overline{Z}_1 + \overline{Z}_2 \bigr) \Pi \ .
\end{align}
Let us choose the clock state as a logical state with the largest possible
energy spread under this Hamiltonian,
\begin{align}
  \lvert {\psi }\rangle = \frac1{\sqrt 2}\,\bigl[ \lvert {\overline{00}}\rangle  + \lvert {\overline{11}}\rangle  \bigr]\ ,
  \label{z:9fL4fsRkzT18}
\end{align}
where $\lvert {\overline{00}}\rangle $ and $\lvert {\overline{11}}\rangle $ refer to logical state vectors
with the first and second logical qubits in the given logical computational
basis states.

Now we check our Knill-Laflamme-like condition.  Having distance 2, the code can
correct a single erasure at a known location.  Crucially, the operator
$\lvert {\xi }\rangle = (H-\langle {H}\rangle )\lvert {\psi }\rangle = 2\omega [\lvert {\overline{00}}\rangle  -
\lvert {\overline{11}}\rangle ]/\sqrt{2}$ is still in the code space because $H$ is a
logical operator.  Then from the Knill-Laflamme conditions we know that
$\langle {\xi}\mkern 1.5mu\relax \vert \mkern 1.5mu\relax {O_i}\mkern 1.5mu\relax \vert \mkern 1.5mu\relax {\psi}\rangle  = 0$ for any single-site operator $O_i$, because
$\lvert {\xi}\rangle $ and $\lvert {\psi}\rangle $ are orthogonal vectors in the code space, and hence our
conditions~\eqref{z:G2tCYcCXVZFu} are satisfied for
single located errors.

If we have some freedom in engineering our Hamiltonian, there are other choices
of logical operators to use in the Hamiltonian that would achieve a similar
sensitivity while also offering protection against single located erasures.  For
instance, we could ignore the second logical qubit (or treat it as a gauge
qubit) and the Hamiltonian could be chosen to act only on sites 1 and 2 as
$H = 2\omega\overline{Z}_1 = 2\omega Z_1 Z_2$.

We see that the probe state~\eqref{z:9fL4fsRkzT18} does not
lose any sensitivity to time if a system is erased at a known location.
The variance of $\lvert {\psi}\rangle $ is given by
\begin{align}
  \sigma_H^2 = \langle {\psi}\mkern 1.5mu\relax \vert \mkern 1.5mu\relax {H^2}\mkern 1.5mu\relax \vert \mkern 1.5mu\relax {\psi}\rangle  = 16\omega^2 \ .
\end{align}
Because we have not specified how this model scales with $n$, we cannot talk yet
about achieving Heisenberg scaling.

In this example, the sensitivity is in fact as good as you can get without any
noise at all, for any probe state:  The
state~\eqref{z:9fL4fsRkzT18}
is a superposition between two states
that have extremal eigenvalues with respect to $H$, which is optimal in the
absence of noise.
What is special about the state vector $\lvert {\psi}\rangle $ is that it retains its sensitivity
even after a single located error, which is not in general the case of other
probe states that would be optimal in the noiseless setting.  For instance, the
state vector
$\bigl[\lvert {0\,0\,0\,0}\rangle  + \lvert {0\,1\,1\,0}\rangle \bigr]/\sqrt2$ has the same
sensitivity as $\lvert {\psi}\rangle $ if no noise is applied, but it does not satisfy our
conditions~\eqref{z:G2tCYcCXVZFu} and so is subject to
sensitivity loss under single-site errors.

The above construction can also be applied if we include $XX$ and $YY$
interactions between the neighboring qubits on top of the existing $ZZ$
interactions (enabling us to model, e.g., Heisenberg interactions):
\begin{align}
  H = \omega \sum_{\neighbors{i,j}} \bigl[ s_x X_i X_j + s_y Y_i Y_j + Z_i Z_j \bigr]\ ,
\end{align}
with the additional coupling constants $s_x, s_y$ allowing for some anisotropy
in the interaction strengths.  In this case, the interaction terms are again all
logical operators, which can be seen from the fact that $X_1X_2X_3X_4$ and
$Y_1Y_2Y_3Y_4 = (X_1X_2X_3X_4)(Z_1Z_2Z_3Z_4)$ are stabilizers.  Our zero
sensitivity-loss conditions are therefore still satisfied.  To compute the
variance of $\lvert {\psi}\rangle $ under this new Hamiltonian, we need to determine the
action of the additional terms on $\lvert {\psi}\rangle $.  The $X$ terms give us again
$\overline{X}_1 + \overline{X}_2$ when acting on the code space following the
same argument as for the $Z$ terms.  Now $\lvert {\psi}\rangle $ is a maximally entangled
state {vector} between the two logical qubits, satisfying
$(\overline{A}_1\otimes\overline{\mathds{1}})\lvert {\psi }\rangle =
(\overline{\mathds{1}}_1\otimes\overline{A}^T_2)\lvert {\psi}\rangle $ where $(\cdot)^T$ denotes
the matrix transpose in the (logical) computational basis, and where $\overline{A}_i$
  is a logical operator acting on the $i$-th logical qubit.
For the $Y$ terms, we then find
\begin{align}
  Y_1Y_2\lvert {\psi
  }\rangle &= -\overline{X}_2\overline{Z}_1\lvert {\psi
  }\rangle = -\overline{Z}_1\overline{X}_1\lvert {\psi }\rangle = i\overline{Y}_1\lvert {\psi}\rangle \ ,
  \nonumber\\
  Y_1Y_3\lvert {\psi
  }\rangle &= -\overline{X}_1\overline{Z}_2\lvert {\psi
  }\rangle = -\overline{X}_1\overline{Z}_1\lvert {\psi }\rangle = -i\overline{Y}_1\lvert {\psi}\rangle \ ,
  \nonumber\\
  Y_2Y_4\lvert {\psi
  }\rangle &= Y_1Y_3\lvert {\psi }\rangle = -i\overline{Y}_1\lvert {\psi}\rangle \ ,
  \nonumber\\
  Y_3Y_4\lvert {\psi
  }\rangle &= Y_1Y_2\lvert {\psi }\rangle = 
  i\overline{Y}_1\lvert {\psi}\rangle \ .
\end{align}
Thus the sum of all four $Y$ interacting terms vanishes when applied onto
$\lvert {\psi}\rangle $.  The variance of $H$ is hence given by
\begin{align}
  H\lvert {\psi
  }\rangle &=
  2\omega\bigl[\,\overline{Z}_1 + \overline{Z}_2 + s_x\bigl(\overline{X}_1 + \overline{X}_2\bigr)\bigr] \lvert {\psi
    }\rangle \nonumber\\
  &= 4\omega\bigl[\,\overline{Z}_1 + s_x \overline{X}_1\bigr]\lvert {\psi}\rangle \ ,
\end{align}
using the fact that $\lvert {\psi}\rangle $ is a maximally entangled state {vector} between the two
logical qubits, and
\begin{align}
  \sigma_H^2 = \langle {\psi}\mkern 1.5mu\relax \vert \mkern 1.5mu\relax {H^2}\mkern 1.5mu\relax \vert \mkern 1.5mu\relax {\psi}\rangle  = 4\omega^2 \bigl(1 + s_x^2\bigr)\ .
\end{align}
The increase in the variance $\sigma_H^2$ when we switch on transversal
interactions can be simply associated with the increased norm of the
Hamiltonian.  Had we defined the clock
state~\eqref{z:9fL4fsRkzT18} with a $-1$ relative phase, then
the $YY$ terms would contribute instead of the $XX$ terms and we would get
$\sigma_H^2 = 4\omega^2\bigl(1+s_y^2\bigr)$.

\paragraph{Time-covariant codes lead to states with no sensitivity loss.}

The construction above based on the $[[4,2,2]]$ code exploited a key property of
that code with respect to the Hamiltonian, namely \emph{time
    covariance}~\cite{R56,R43,R57,R58,R59}.  A time-covariant code with respect to a given
  Hamiltonian $H$ is a code for which the time evolution generated by $H$
is a (nontrivial) logical operator.  If we can find a time-covariant code with respect to the system's Hamiltonian, then the clock state
 can be chosen to lie within the code space, so that
errors that affect it can be corrected, all while evolving nontrivially in time
and thus serving as a clock.

However, there are constraints on the possibility of constructing time-covariant
codes.  Consider a Hamiltonian that is a sum of terms of weight at most
  $k$, which we call a \emph{$k$-local Hamiltonian}.  Any code that can correct
up to $k$ arbitrary errors at known locations cannot be time covariant with
respect to a $k$-local Hamiltonian, because the Hamiltonian would be a sum of
correctable terms that cannot have a nontrivial action that preserves the
code space.
On the other hand, physical systems like spin chains and the anti-de Sitter/conformal field theory (AdS/CFT)
correspondence as a model for quantum gravity offer natural examples of
time-covariant codes that can approximately correct against low-weight
errors~\cite{R42}.  The above example using the
$[[4,2,2]]$ code is a concrete case of a time-covariant code with respect to a
2-local Hamiltonian and which can correct a single erasure at a known location.

We can see that whenever we can find a time-covariant code with respect to a
given Hamiltonian, then we can construct from the code a clock state with zero
sensitivity loss.  Consider a code space $\Pi$ and suppose that the Hamiltonian
$H$ is a nontrivial logical operator.  We can choose $\lvert {\psi}\rangle $ to be any
logical state vector that has nonzero variance with respect to $H$.  Let
$\lvert {\xi }\rangle = (H-\langle {H}\rangle )\lvert {\psi}\rangle $, noting that $\lvert {\xi}\rangle $ lies in the code space.
Denoting by $\{E_k\}$ the Kraus operators of $\mathcal{N}$, we see that
$\langle {\psi}\mkern 1.5mu\relax \vert \mkern 1.5mu\relax {E_{k'}^\dagger E_k}\mkern 1.5mu\relax \vert \mkern 1.5mu\relax {\xi}\rangle  \propto \langle {\psi}\mkern 1.5mu\relax \vert\mkern 1.5mu\relax {\xi}\rangle  = 0$ from
the Knill-Laflamme conditions of the code, and therefore the
conditions~\eqref{z:G2tCYcCXVZFu} are satisfied.
Therefore:
\begin{observation}[Clock state from a time-covariant code]\noproofref
  Let $\Pi$ be the projector onto a code space that 
  corrects errors of the error
  channel $\mathcal{N}$.  Assume that the code is time-covariant with respect to
  the Hamiltonian $H$.
  Then any logical state vector $\lvert {\psi}\rangle $ and associated
  $\lvert {\xi }\rangle = (H-\langle {H}\rangle )\lvert {\psi}\rangle $
  satisfy the conditions~\eqref{z:G2tCYcCXVZFu}.
  Furthermore, if $\Pi$ defines a $[[n,1,d]]$ quantum code, then $\lvert {\psi}\rangle $ and
  $P_\psi^\perp H\lvert {\psi}\rangle $ define a metrological code of metrological distance
  $d$.
\end{observation}
That is, \emph{any} logical state of the code satisfies our Knill-Laflamme-like
conditions for zero sensitivity loss.  The sensitivity is maximized by picking
the state with the largest energy variance.

If we are given an $\epsilon$-approximate quantum error-correcting code that is
time-covariant, that is, if the error-correction procedure is allowed to fail
with some probability $\epsilon>0$, then we can still use a state lying in
the code space to construct a clock state with little sensitivity loss.
Approximate quantum error-correcting codes can be characterized by the fact that
the channel that maps the code space to the environment,
$\widehat{\mathcal{N}}(\Pi(\cdot)\Pi)$, is close to a constant channel that
always outputs a fixed state~\cite{R60,R61}.
Specifically, $\widehat{\mathcal{N}}(\Pi (X) \Pi) \approx \operatorname{tr}(X)\,\tau_E$ for
all $X$, for some fixed state $\tau_E$.  If we pick a logical state vector $\lvert {\psi}\rangle $
with nonzero energy variance, then we have that
$\{\psi, \bar{H}\} = \Pi \{\psi, \bar{H}\} \Pi$ is a logical operator and
therefore
$\widehat{\mathcal{N}}\bigl(\{\psi, \bar{H}\}\bigr) = \widehat{\mathcal{N}}\bigl(\Pi
\{\psi, \bar{H}\} \Pi\bigr) \approx \operatorname{tr}\bigl(\{\psi,\bar{H}\}\bigr)\,\tau_E = 0$ since
$\langle {\bar{H}}\rangle =0$.  Therefore $\FIloss{Bob}{t}$
in~\eqref{z:4bic5ys.wcEe} satisfies $\FIloss{Bob}{t}\approx 0$,
and $\FIqty{Bob}{t} \approx \FIqty{Alice}{t} = 4\sigma_H^2$.
This choice of a clock state is hence expected to lose little sensitivity under
action of the noisy channel.
Deriving a universal quantitative bound on $\FIqty{Bob}{t}$ in this scenario in
terms of $\epsilon$ does not appear easy.  In such a scenario, a direct use of
our uncertainty relation~\eqref{z:cmnZZu5eQ8hv} [or
of a corresponding bound such as~\eqref{z:U3JchlrlLghQ}]
seems likely to be the most straightforward way to obtain useful quantitative
expressions for $\FIqty{Bob}{t}$ in the case where the clock state is prepared
using an approximate error-correcting code.

\subsection{Clock state for interacting many-body systems}%
\label{z:oz0py5ItCeCD}

\label{z:dqhjUefIA37Z}

Consider now an arbitrary interaction graph, where each vertex is associated
with a single qubit (\cref{z:q6y-xgqvsyUu}\textbf{b}) and consider the
Hamiltonian
\begin{align}
  H = \frac{J}{2}\sum_{\neighbors{i, j}} \bigl(Z_iZ_j + s_x X_iX_j + s_y Y_iY_j\bigr)\ ,
\end{align}
where the sum ranges over all graph vertices $i,j$ that are connected by an
edge, and where $s_x,s_y$ are arbitrary real coefficients.  (In fact, the
coefficients $s_x,s_y$ may also vary for each pair of sites $i,j$, though we
omit the dependence here for clarity.)  We recover the Ising model with
$s_x=s_y=0$ and the Heisenberg model with $s_x=s_y=1$.  We denote by $m$ the
number of edges in the graph, which is also the number of terms in the sum.

We define the clock state vector $\lvert {\psi}\rangle $ as follows.  Denote by $\lvert {0^n}\rangle $ and
$\lvert {1^n}\rangle $ the all-zero and the all-one state.  Choose any bit string
$\boldsymbol{x}$ and let $\lvert {\boldsymbol{x}}\rangle $ be the corresponding spin
configuration, where each bit corresponds to one of the qubit basis vectors on
the corresponding vertex. %
We assume that $\boldsymbol{x}$ violates a number $c$ out of the $m$ possible
$ZZ$-interaction terms, i.e., we denote by $c$ the number of pairs of bits in
$\boldsymbol{x}$ that differ and that are connected by an edge in the graph.
(It might not be possible to violate all the interaction terms simultaneously,
as the graph might be frustrated.)
An assumption we will need later is that the bit strings $0^n$, $1^n$, and
$\boldsymbol{x}$ all differ on at least four sites. 
Now define
\begin{align}
  \lvert {\psi }\rangle = \frac12\bigl[ \lvert {0^n}\rangle  + \lvert {1^n}\rangle  + \lvert {\boldsymbol{x}}\rangle 
  + \lvert {\widetilde{\boldsymbol{x}}}\rangle \bigr]\ ,
  \label{z:uYISiNnDt2yR}
\end{align}
where the bit string $\widetilde{\boldsymbol{x}}$ is obtained by flipping all the
bits of $\boldsymbol{x}$.
We then have
\begin{align}
  H\lvert {\psi }\rangle &= \frac{J}{4}\bigl[ m\lvert {0^n}\rangle  + m\lvert {1^n}\rangle  + (m-2c)\lvert {\boldsymbol{x}}\rangle  +
  (m-2c)\lvert {\widetilde{\boldsymbol{x}}}\rangle  \bigr]
              \nonumber\\
  &\hspace{1.5em}
    {}+ \frac{J}{2}\sum_{\neighbors{i, j}} \bigl(s_x X_iX_j + s_y Y_iY_j\bigr) \,\lvert {\psi}\rangle \ .
\end{align}
The $XX$ and $YY$ operators applied on $\lvert {\psi}\rangle $ generate terms associated with
new bit strings where, each time, two bits are flipped and a possible phase is
acquired.  These new configurations are all orthogonal to $\lvert {0^n}\rangle $,
$\lvert {1^n}\rangle $, $\lvert {\boldsymbol{x}}\rangle $, and $\lvert {\widetilde{\boldsymbol{x}}}\rangle $
thanks to our assumption that the chosen configurations differ on at least four
sites.  So we have
\begin{align}
  \langle {H}\rangle _\psi = \frac{J}{8}\bigl[2 \times m + 2 \times (m-2c)\bigr] = \frac{J}{2} (m-c)\ .
\end{align}
With $\bar{H} = H - \frac{J}{2}(m-c)\mathds{1}$ and $\lvert {\xi }\rangle = \bar{H} \lvert {\psi}\rangle $, we see that
\begin{align}
    \lvert {\xi }\rangle &= 
  \begin{alignedat}[t]{1}
    &\frac{J}{4}\bigl[ c\lvert {0^n}\rangle  + c\lvert {1^n}\rangle 
        - c\lvert {\boldsymbol{x}}\rangle  - c\lvert {\widetilde{\boldsymbol{x}}}\rangle  ]
    \\
    &+ \frac{J}{2}\sum_{\neighbors{i, j}} \bigl(s_x X_iX_j + s_y Y_iY_j\bigr) \,\lvert {\psi}\rangle \ .
  \end{alignedat}
\end{align}
To check the zero sensitivity-loss conditions~\eqref{z:zFMncm5bkXDt}, we compute the following
expression for any single-site operator $O_i$,
\begin{align}
  \langle {\psi}\mkern 1.5mu\relax \vert \mkern 1.5mu\relax {O_i}\mkern 1.5mu\relax \vert \mkern 1.5mu\relax {\xi}\rangle 
  &= \begin{alignedat}[t]{1}
    \frac{J}{8} \bigl[
    &c\,\langle {0}\mkern 1.5mu\relax \vert \mkern 1.5mu\relax {O_i}\mkern 1.5mu\relax \vert \mkern 1.5mu\relax {0}\rangle  + c\,\langle {1}\mkern 1.5mu\relax \vert \mkern 1.5mu\relax {O_i}\mkern 1.5mu\relax \vert \mkern 1.5mu\relax {1}\rangle 
    \\
    &{}-c\,\langle {x_i}\mkern 1.5mu\relax \vert \mkern 1.5mu\relax {O_i}\mkern 1.5mu\relax \vert \mkern 1.5mu\relax {x_i}\rangle  -c\,\langle {\widetilde{x}_i}\mkern 1.5mu\relax \vert \mkern 1.5mu\relax {O_i}\mkern 1.5mu\relax \vert \mkern 1.5mu\relax {\widetilde{x}_i}\rangle  \bigr]
    \\
  \end{alignedat}
  \nonumber\\
  &\hphantom{{}={}}
    {}+ \frac{J}{2}\sum_{\neighbors{i,j}} \langle {\psi}\rvert O_i\bigl(s_x X_iX_j + s_y Y_iY_j\bigr)\lvert {\psi
    }\rangle \nonumber\\
  &= \frac{J}{8}\bigl[ c\operatorname{tr}(O_i) - c\operatorname{tr}(O_i) \bigr] + 0 = 0\ ,
\end{align}
where $x_i$ (respectively, $\widetilde{x}_i$) denote the value of the $i$-th bit
in $\boldsymbol{x}$ (respectively, $\widetilde{\boldsymbol{x}}$).
The terms corresponding to $XX$ and $YY$ interactions vanish because all
configurations $0^n$, $1^n$, $\boldsymbol{x}$, and $\widetilde{\boldsymbol x}$
differ on at least four sites, and $XX$ and $YY$ flip two bits of the basis
vector state on which they are applied (with a possible phase).
Therefore the zero sensitivity-loss
conditions~\eqref{z:G2tCYcCXVZFu} are satisfied, and 
the clock state can suffer a single located erasure while retaining full
sensitivity.

The energy variance of the probe state is given by
\begin{align}
  \sigma_H^2 &= \langle {\psi}\mkern 1.5mu\relax \vert \mkern 1.5mu\relax {(H-\langle {H}\rangle )^2}\mkern 1.5mu\relax \vert \mkern 1.5mu\relax {\psi}\rangle  = \langle {\xi}\mkern 1.5mu\relax \vert\mkern 1.5mu\relax {\xi}\rangle 
  \nonumber\\
  &= \frac14 J^2 c^2 + (\text{contrib. from \(XX/YY\) terms})\ .
\end{align}
The contribution from $XX$ and $YY$ terms is zero if the configurations
$0^n,1^n,\boldsymbol{x},\widetilde{\boldsymbol x}$ all differ on at least five
sites (or in the case of Ising interactions with $s_x=s_y=0$).

The question of whether this achieves $n^2$ scaling depends on how we choose the
graph and the string $\boldsymbol{x}$ to grow with $n$.  In the case of a square
lattice with nearest-neighbor interactions, we have that the number of edges
scales like the number of vertices ($m\sim 2n$) and we can simultaneously
violate all $ZZ$ interaction terms by choosing an alternating configuration of
$0$'s and $1$'s.  In this case $\sigma_H^2 \sim J^2n^2$, achieving Heisenberg
scaling.  For other graphs, the question of whether $\sigma_H^2 \sim n^2$ is
determined by how the number of edges scales with the number of vertices in the
graph, and how many of those $ZZ$-interaction terms can be simultaneously
violated.  If there is a linear relationship between these quantities then
Heisenberg scaling is achieved, noting that only a single error at a known
location can be incurred without sensitivity loss.

\subsection{Metrological codes from stabilizer codes}

In this section, we present a general scheme to construct metrological codes
based on the stabilizer formalism~\cite{R62} and study some
simple examples.  We show that our construction is strictly more general than
constructing time-covariant error-correcting codes.  Our aim is to study and
illustrate our general construction; the Hamiltonians in our examples are not
intended as practical schemes to be engineered with near-term technology.

Consider the Pauli group $\mathcal{G}_n$ on $n$ qubits, defined as comprising
all tensor product operators on $n$ qubits of single-site Pauli operators and
the identity operator, with all possible prefactors $\pm 1$ and
$\pm i$~\cite{R62}.  Consider a subgroup
$\mathcal{S} \subset \mathcal{G}_n$ presented as
$\mathcal{S} = \groupgen{S_1, \ldots , S_\ell}$ with independent commuting
generators $S_1, \ldots , S_\ell$ such that $-\mathds{1}\not\in \mathcal{S}$.  The
\emph{normalizer} of $\mathcal{S}$ in $\mathcal{G}_n$ is
$N(\mathcal{S}) = \{ E \in \mathcal{G}_n\,:\ \forall\, g\in\mathcal{S},\
  EgE^\dagger \in \mathcal{S} \}$. %
A state is said to be \emph{stabilized} by $\mathcal{S}$ if it lies in the
simultaneous $+1$ eigenspace of all $S\in\mathcal{S}$; the elements of
$\mathcal{S}$ are called \emph{stabilizers}.  The code space associated with the
Pauli stabilizer group $\mathcal{S}$ is the subspace spanned by all states that
are stabilized by $\mathcal{S}$.
If $ \mathscr{E} \subset \mathcal{G}_n$ is a set of error operators such that
for all $E,E' \in\mathscr{E}$ either $E'^\dagger E \notin N(\mathcal{S})$ or
$E'^\dagger E$ lies in $\mathcal{S}$ up to a phase, then
a fundamental theorem of quantum error correction states that the subspace
of all common $+1$ eigenstates of the operators $\{ S_i \}$ forms a code
space that can correct any error in $\mathscr{E}$.
One defines the distance $d$ of the code as the minimal weight of an element in
$N(\mathcal{S}) \setminus \mathcal{S}$, i.e., of a nontrivial logical operation.
Then, the code can correct any $t$ errors at unknown locations as long as $2t+1 \leq d$.

As a simple example, consider the $n$-qubit GHZ state vector
$\lvert {\psi }\rangle = \bigl[ \lvert {\uparrow^n}\rangle  + \lvert {\downarrow^n}\rangle  \bigr]/\sqrt2$ and the
Hamiltonian $H = \sum_{j=1}^n Z_j$.  We have
$\lvert {\xi }\rangle \propto \bigl[ \lvert {\uparrow^n}\rangle  - \lvert {\downarrow^n}\rangle  \bigr]/\sqrt2$.
Suppose our error model consists of an arbitrary number of $X$ errors.
From~\eqref{z:G2tCYcCXVZFu}, since acting with $X$
operators on $\lvert {\psi}\rangle $ can never generate any overlap with $\lvert {\xi}\rangle $, we see
that $\lvert {\psi}\rangle ,\lvert {\xi}\rangle $ form a metrological code against any number of $X$
errors.
We now present an overview of our procedure using this example.
In our procedure, we first find a set $\{ S_i \}$ of independent commuting Pauli
operators that stabilize $\lvert {\psi}\rangle $.
We fix a set of error operators $\mathscr{E}$, which we choose in our example to
consist of all $n$-qubit Pauli operators that are a product of only $\mathds{1}$'s
and $X$'s.
Suppose that we are given an operator $H$ with the following property: For any
operators $E, E' \in\mathscr{E}$, there exists a $S\in\mathcal{S}$ such that
$\{H, S\} = 0$ and $[ E'^\dagger E, S ] = 0$.
The state vector $\lvert {\psi}\rangle $ is stabilized by the choice of commuting Pauli operators
$Z_1Z_2$, $Z_2Z_3$, \ldots, $Z_{n-1}Z_n$, $X^{\otimes n}$.  Multiplying all but
the last stabilizer by $X^{\otimes n}$, we obtain the following choice of
independent stabilizer generators
\begin{multline}
  {-Y_1Y_2 X_3 X_4 \ldots X_n} \,,\quad
  {-X_1 Y_2 Y_3 X_4 \ldots X_n} \,,\quad
  \ldots\,,\quad
  \\
  {-X_1 \ldots X_{n-2} Y_{n-1} Y_n} \,,\quad
  { X^{\otimes n} }\ .
\end{multline}
For any site $j$, the operator $Z_j$ anticommutes with all the above stabilizer
generators.  Our structural constraint turns out to apply in this case; it will
be detailed later.  Our construction then implies that the pair
$(\lvert {\psi}\rangle , H\lvert {\psi}\rangle )$ is a metrological code. Here, $\lvert {\xi }\rangle = H\lvert {\psi}\rangle $ is
in fact the state vector that is stabilized by all the operators $\{ -S_i \}$. %

\paragraph{Statement of the construction.}
Our construction is given by the following theorem.

\begin{theorem}[Metrological codes from stabilizer states]
  \label{z:Bnq1hErO9gGx}
  Let $\mathcal{S}\subset\mathcal{G}_n$ be an abelian subgroup of the Pauli
  group with $-\mathds{1}\notin\mathcal{S}$, and let $\lvert {\psi}\rangle $ be stabilized by
  $\mathcal{S}$.
  Let $H$ be any Hermitian operator such that $H\lvert {\psi}\rangle \neq0$ and let
  $\mathscr{E} \subset \mathcal{G}_n$ be any set of Pauli error operators.
  Assume that for all $E,E'\in\mathscr{E}$, there exists $S\in\mathcal{S}$ such
  that $\{H, S\} = 0$ and $[E'^\dagger E,S] = 0$.
  Then $\lvert {\psi}\rangle , H\lvert {\psi}\rangle $ form a metrological code against $\mathscr{E}$.
\end{theorem}
We recall the definition of a metrological code as satisfying the
condition~\eqref{z:q5Fv8e19UZL5}.
On the other hand, a defining property of a time-covariant code (recall definition in \cref{z:AVGHDyn1nGzm}) is that the Hamiltonian $H$ must be a logical
operator, and thus, for a stabilizer code, must commute with all the stabilisers of the code.
This is not in contradiction with \cref{z:Bnq1hErO9gGx}
since the stabiliser group $\mathcal{S}$ in the theorem is not necessarily that of the code for
which $H$ is a logical operator.  
Below, we present examples of metrological codes; some are error-correcting
time-covariant codes in disguise, yet others cannot be written as a
time-covariant error-correcting code with similar distance as the metrological
code.

 \begin{proof}[**z:Bnq1hErO9gGx]
   First, let $S_0\in\mathcal{S}$ with $\{H, S_0\} = 0$; such a stabilizer must
   exist from our assumption.
   We then have
   \begin{equation}
   \langle {\psi}\mkern 1.5mu\relax \vert \mkern 1.5mu\relax {H}\mkern 1.5mu\relax \vert \mkern 1.5mu\relax {\psi}\rangle  = \langle {\psi}\mkern 1.5mu\relax \vert \mkern 1.5mu\relax {H S_0}\mkern 1.5mu\relax \vert \mkern 1.5mu\relax {\psi}\rangle  = -\langle {\psi}\mkern 1.5mu\relax \vert \mkern 1.5mu\relax {S_0 H}\mkern 1.5mu\relax \vert \mkern 1.5mu\relax {\psi}\rangle 
   = -\langle {\psi}\mkern 1.5mu\relax \vert \mkern 1.5mu\relax {H}\mkern 1.5mu\relax \vert \mkern 1.5mu\relax {\psi}\rangle 
      \end{equation}
      and thus $\langle {H}\rangle _\psi = 0$.
   Let 
     \begin{equation}
   \lvert {\xi }\rangle = H\lvert {\psi}\rangle , 
    \end{equation}
    with $\langle {\xi}\mkern 1.5mu\relax \vert\mkern 1.5mu\relax {\psi}\rangle =0$ automatically satisfied.
   Let $E,E'\in\mathscr{E}$.  We need to show
   that~\eqref{z:q5Fv8e19UZL5} holds.  From our assumption
   there exists an $S\in\mathcal{S}$ with $\{H,S\}=0$ and $[E'^\dagger E, S] = 0$.
   We have
   \begin{align}
     \langle {\xi}\mkern 1.5mu\relax \vert \mkern 1.5mu\relax {E'^\dagger E}\mkern 1.5mu\relax \vert \mkern 1.5mu\relax {\psi
     }\rangle &= \langle {\psi}\mkern 1.5mu\relax \vert \mkern 1.5mu\relax {H E'^\dagger E S}\mkern 1.5mu\relax \vert \mkern 1.5mu\relax {\psi}\rangle 
       \nonumber\\
     &= -\langle {\psi}\mkern 1.5mu\relax \vert \mkern 1.5mu\relax {S H E'^\dagger E}\mkern 1.5mu\relax \vert \mkern 1.5mu\relax {\psi}\rangle 
       \nonumber\\
     &= -\langle {\xi}\mkern 1.5mu\relax \vert \mkern 1.5mu\relax {E'^\dagger E}\mkern 1.5mu\relax \vert \mkern 1.5mu\relax {\psi}\rangle \ ,
   \end{align}
   and thus $\langle {\xi}\mkern 1.5mu\relax \vert \mkern 1.5mu\relax {E'^\dagger E}\mkern 1.5mu\relax \vert \mkern 1.5mu\relax {\psi}\rangle =0$, confirming
   that~\eqref{z:q5Fv8e19UZL5} holds and that
   $\lvert {\psi}\rangle , \lvert {\xi}\rangle $ indeed constitute a metrological code against $\mathscr{E}$.
 \end{proof}

In the remainder of this section we review some examples of codes resulting from
the construction of \cref{z:Bnq1hErO9gGx}.
We begin by connecting our construction with error-correcting codes in which the
Hamiltonian is a nontrivial logical operator, i.e., time-covariant
error-correcting codes.  We present an example of a time-covariant code based on
the 7-qubit Steane code, and we then show that all time-covariant
error-correcting codes are special cases of
\cref{z:Bnq1hErO9gGx}.
We then show that there are metrological codes that cannot be formulated in
terms of a corresponding time-covariant error-correcting code;  i.e., there are
schemes that enable the communication of a clock state through a noisy channel
that achieve zero sensitivity loss without having to construct a full quantum
error-correcting code.

\paragraph{Example based on the $[[7,1,3]]$ Steane code.}
As an example of a time-covariant code, we consider
an example deriving from the 
Steane stabilizer code~\cite{R63}.
The latter is given by the following generators and
logical $\overline{X},\overline{Z}$ operators:
\begin{align}
  \begin{split}
   \hat{S}_1 &=  X_4 X_5 X_6 X_7,  \\
    \hat{S}_2 &=  X_2 X_3  X_6 X_7 , \\
    \hat{S}_3 &= X_1  X_3 X_5  X_7 , \\
    \hat{S}_4 &= Z_4 Z_5 Z_6 Z_7 , \\
    \hat{S}_5 &=  Z_2 Z_3 Z_6 Z_7 , \\
    \hat{S}_6 &= Z_1  Z_3  Z_5 Z_7 , \\
    \overline{X} &= X_1 X_2 X_3 X_4 X_5 X_6 X_7 , \\
    \overline{Z} &= 
    Z_1 Z_2 Z_3 Z_4 Z_5 Z_6 Z_7\ .
  \end{split}
\end{align}
Let $\lvert {\psi}\rangle = \lvert {\overline{+}}\rangle $ be the
state vector in the logical space associated with the $+1$
logical eigenspace of the $\overline{X}$ operator, and consider the Hamiltonian
\begin{align}
H = \overline{Z} \hat{S}_4 &= Z_1 Z_2 Z_3 .
\end{align}
The Hamiltonian is a logical operator, being stabilizer-equivalent to the
logical $\overline{Z}$ operator, and rotates the state vector $\lvert {\psi}\rangle $ to
$\lvert {\xi }\rangle = H\lvert {\psi }\rangle = \lvert {\overline{-}}\rangle $.  (The above choice of $H$ was
preferred to the choice $H=\overline{Z}$ because it has lower weight.)  The code
is therefore time covariant with respect to the action of $H$, and we for this reason
already know that it is a metrological code of metrological distance $3$. To
illustrate our construction, we explain how the same conclusion can be reached
by applying \cref{z:Bnq1hErO9gGx}.  We now
define a set of stabilizer generators that serve to define the state vector  $\lvert {\psi}\rangle $
of the resulting metrological code.
The new stabilizer generators $\{ S_i \}$ are obtained by multiplying each of the
$\hat{S}_i$ by $\overline{X}$, all while including $\overline{X}$ itself, as
\begin{align}
  \begin{split}
    S_1 &= \overline{X}\hat{S}_1 = X_1 X_2 X_3, \\
    S_2 &= \overline{X}\hat{S}_2 = X_1 X_4 X_5 , \\
    S_3 &= \overline{X}\hat{S}_3 =  X_2 X_4  X_6, \\
    S_4 &= \overline{X}\hat{S}_4 = X_1 X_2 X_3 Y_4 Y_5 Y_6 Y_7, \\
    S_5 &= \overline{X}\hat{S}_5 = X_1 Y_2 Y_3 X_4 X_5 Y_6 Y_7, \\
    S_6 &= \overline{X}\hat{S}_6 = Y_1 X_2 Y_3 X_4 Y_5 X_6 Y_7, \\
    S_7 &= \overline{X} = X_1 X_2 X_3 X_4 X_5 X_6 X_7 . \\
  \end{split}
\end{align}
One can verify that $H$ anticommutes with each $S_i$ listed above.
(For the application of \cref{z:Bnq1hErO9gGx},
it is convenient to use a choice of stabilizer generators that
anticommute with $H$.)
Let $\mathscr{E}$ be the set of all single-site operators.  For any
$E,E'\in\mathscr{E}$, we will show that there is a $S\in\mathcal{S}$ with
$\{H,S\}=0$ and $[E'^\dagger E,S] = 0$.  If one of the $S_i$ has support outside
of that of $E'^\dagger E$, it will do the job. Alternatively, any product of an
odd number of the $S_i$ will also do, for instance
\begin{align}
  \begin{split}
    S_1 S_2 S_3 &= X_3  
    X_5 X_6 ,\\
    S_1 S_2 S_7 &= X_1 X_6 X_7,\\
    S_2 S_3 S_7 &= X_3 X_4 
    X_7 .
  \end{split}
\end{align}
One can verify that for any two among the seven sites, at least one operator
among $S_1$, $S_2$, $S_3$, $S_1S_2S_3$, $S_1S_2S_7$, $S_2S_3S_7$ has its support
outside of those two sites.  Since these operators all anticommute with $H$, we
have that for all $E,E'\in\mathscr{E}$, there is a
$S\in\groupgen{ S_1, \ldots, S_7 }$ such that $[E'^\dagger E, S] = 0$ and
$\{ S, H \} = 0$.
From \cref{z:Bnq1hErO9gGx}, we see that
$\lvert {\psi}\rangle $ and $\lvert {\xi}\rangle =H\lvert {\psi}\rangle $ must form a metrological code against $\mathscr{E}$,
and is therefore a metrological code with metrological distance $3$.

\paragraph{Time-covariant codes.}
In this paragraph, we %
show that the assumptions of
\cref{z:Bnq1hErO9gGx} are in fact always
satisfied for time-covariant stabilizer codes like the 7-qubit Steane code example above.

Let $\hat{\mathcal{S}} = \groupgen{ \hat{S}_1, \ldots , \hat{S}_\ell }$ be a stabilizer code with a nontrivial logical operator $\overline{Z}$.  Let
$\overline{X}$ be a logical operator that anticommutes with $\overline{Z}$, and
define the stabilizer group
$\mathcal{S} = \groupgen{ \overline{X} \hat{S}_1, \overline{X} \hat{S}_2,
  \ldots, \overline{X} \hat{S}_\ell, \overline{X} }$.
Observe that
$\overline{Z}$ anticommutes with all the chosen generators for $\mathcal{S}$.
(Such an operator $\overline{X}$ must always exist, cf.\@ e.g.\@  \cite[Proposition~10.4]{R62}.)

We show the following: For any Pauli operator
$A \notin N(\hat{\mathcal{S}}) \setminus \hat{\mathcal{S}}$, there exists
$S\in\mathcal{S}$ such that $[S,A] = 0$ and $\{S, \overline{Z}\} = 0$.
This property implies that for a given set of errors $\mathscr{E}$ that are
correctable for $\hat{\mathcal{S}}$, i.e., if we have
$E'^\dagger E \notin N(\hat{\mathcal{S}}) \setminus \hat{\mathcal{S}}$ for all
$E,E'\in\mathscr{E}$, then the conditions of
\cref{z:Bnq1hErO9gGx} are satisfied, where the Hamiltonian is $H = \bar Z$.

Suppose first that $A\in \hat{\mathcal{S}} \subset \mathcal{S}$.  Then $A$
commutes with all stabilizers in $\mathcal{S}$, including $\overline{X}$ which
anticommutes with $\overline{Z}$.
Now suppose that $A\in\mathcal{G}_n$ and $A \notin N(\hat{\mathcal{S}})$, i.e.,
there is a $\hat{S}\in\hat{\mathcal{S}}$ with $\{A, \hat{S}\} = 0$.  If
$[A,\overline{X}] = 0$, then the choice $S=\overline{X} \in \mathcal{S}$
satisfies $[S,A]=0$ and $\{S,\overline{Z}\}=0$.  If, instead, we have
$\{A, \overline{X}\} = 0$, we can set $S= \overline{X} \hat{S}$ to find
$S\overline{Z} = \overline{X}\hat{S}\overline{Z} =
\overline{X}\overline{Z}\hat{S} = -\overline{Z}\overline{X}\hat{S} =
-\overline{Z} S$ and
$AS = A\overline{X}\hat{S} = -\overline{X}A\hat{S} = \overline{X}\hat{S} A =
SA$, and thus $\{S,\overline{Z}\}=0$ and $[S,A]=0$ as required.

\paragraph{Example: A $[[4,2,2]]$ code state with an auxiliary qubit.}
Whereas in the earlier 7-qubit Steane code example
the state vectors  $\lvert {\psi}\rangle $ and
$\lvert {\xi}\rangle $ both lie within a subspace of a distance $d=3$ code, the following example
illustrates a situation in which $\lvert {\psi}\rangle $ and $\lvert {\xi}\rangle $ cannot 
be contained in a code
space that can correct the same errors against which the states form a
metrological code.  I.e.,
\cref{z:Bnq1hErO9gGx} can be used to
construct metrological codes that cannot be formulated as time-covariant
quantum error-correcting codes with respect to the same errors.
Consider the 5-qubit Pauli operators
\begin{align}\label{z:hlgUBO4Uqbdv}
  \begin{split}
    S_1 &= X_1 X_2 , \\
    S_2 &= X_3 X_4 ,\\
    S_3 &= X_1  X_3 ,\\
    S_4 &=  X_5,\\
    S_5 &= Z_1 Z_2 Z_3 Z_4  .
  \end{split}
\end{align}
We can see that the stabilizer group $\mathcal{S} = \groupgen{ S_1, \ldots, S_5 }$ is generated 
by
\begin{itemize}
\item the stabilizers for the $[[4,2,2]]$ code on the first four
    qubits ($Z_1Z_2Z_3Z_4=S_5$ and $X_1 X_2 X_3 X_4 = S_1 S_2$);
\item the logical $X$ operators of the first and second
    logical qubits of that $[[4,2,2]]$  code ($X_1X_3 = S_3$ and $X_1 X_2 = S_1$); and
\item an independent stabilizer fixing the state of the 5th qubit ($X_5 = S_4$).
\end{itemize}
We choose the Hamiltonian
\begin{align}
  H &= Y_1  Z_4 Y_5 .
\end{align}
The Hamiltonian can be written as a product of three terms: A logical $Z$ operator
on both logical qubits of the $[[4,2,2]]$ code 
($Z_1 Z_4 $), a $Y$ operation on the 5th physical qubit, and a single $X_1$ on
the first physical qubit.  The Hamiltonian is not a logical operator of the
$[[4,2,2]]$ code.  Also, a suitable permutation of the qubits would make $H$
geometrically local, should this property be desired.

We can verify that $H$ anticommutes with each of the stabilizers
$S_1, \ldots , S_5$.  Furthermore, for any two sites $i,j$, one of the $S_k$ acts
as the identity on the sites $i,j$; therefore, 
for any two-site operator $A$,
there always exists a stabilizer $S$ with $[S,A] = 0$ and $\{ S, H \} = 0$.  We
can apply \cref{z:Bnq1hErO9gGx} to deduce
that $\lvert {\psi}\rangle ,\lvert {\xi}\rangle $ define a distance-$3$ metrological code.

The state vector $\lvert {\psi}\rangle $ can be expressed in terms of the logical +1 $X$ eigenvectors 
$\lvert {\overline{++}}\rangle $ of the $[[4,2,2]]$ code, and in terms of the $\lvert {\pm}\rangle _i$
physical state vectors, as
\begin{align}
  \lvert {\psi
  }\rangle &= \lvert {\overline{++}}\rangle _{1234}\otimes\lvert {+}\rangle _5
    \nonumber\\
  &= \frac1{\sqrt 2}\bigl( \lvert {{+}{+}{+}{+}{+}}\rangle  + \lvert {{-}{-}{-}{-}{+}}\rangle  \bigr)\ .
\end{align}
Recalling $Y\lvert {\pm }\rangle = \mp i\lvert {\mp}\rangle $ and $Z\lvert {\pm }\rangle = \lvert {\mp}\rangle $, we find
\begin{align}
  \lvert {\xi }\rangle = H\lvert {\psi
  }\rangle &= \frac1{\sqrt 2}\bigl( - \lvert {{-}{+}{+}{-}{-}}\rangle  + \lvert {{+}{-}{-}{+}{-}}\rangle  \bigr)
  \nonumber\\
  &= \frac1{\sqrt 2}\bigl( \lvert {{+}{-}{-}{+}}\rangle  - \lvert {{-}{+}{+}{-}}\rangle  \bigr)\otimes\lvert {-}\rangle _5\ .
\end{align}
We see that
\begin{align}
  \langle {\psi}\mkern 1.5mu\relax \vert \mkern 1.5mu\relax {X_5}\mkern 1.5mu\relax \vert \mkern 1.5mu\relax {\psi}\rangle  &= 1\ ,
  &
  \langle {\xi}\mkern 1.5mu\relax \vert \mkern 1.5mu\relax {X_5}\mkern 1.5mu\relax \vert \mkern 1.5mu\relax {\xi}\rangle  &= -1\ ,
\end{align}
so it is not possible for $\lvert {\psi}\rangle ,\lvert {\xi}\rangle $ to lie in the code space of a
distance $d=3$ quantum error-correcting code.

This example shows that metrological codes are a class of codes that is broader
than traditional error-correction codes as there are certain errors that the
former does not have to completely correct. Metrological codes might therefore
offer additional possibilities to find noise-resilient schemes for communicating
clock states across a noise channel.

\paragraph{Metrological toric code.}
A further example application of
\cref{z:Bnq1hErO9gGx}
is based on \emph{Kitaev's toric code}~\cite{R64,R65}.
\begin{figure}
  \centering
  \includegraphics{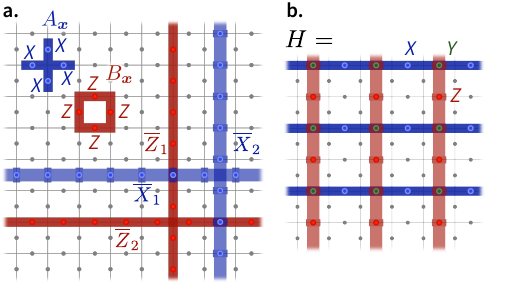}
  \caption{
    Metrological code based on the toric code. \textbf{a}.~Star
    ($A_{\boldsymbol x}$) and plaquette ($B_{\boldsymbol x}$) operators
    generate the stabilizer group of the toric code, where $\boldsymbol x$
    is a pair of integer coordinates on the two-dimensional lattice.  Two encoded logical qubits
    have logical Pauli operators $\overline{X}_1$, $\overline{Z}_1$,
    $\overline{X}_2$, $\overline{Z}_2$ corresponding to strings of physical
    Pauli $X$ or $Z$ operators that wrap around the torus.
    \textbf{b}.~In our metrological code example based on the toric code, we
    map a state from the toric code to a state of a related code which we call
    the \emph{anti-toric code}.
    The anti-toric code is the subspace stabilized by all the operators
    $\{ -A_{\boldsymbol x} \}$ and $\{ -B_{\boldsymbol x}\}$.
    Assuming the lattice side length is even, the
    depicted operator $H$ anticommutes with all star and plaquette operators,
    meaning that it maps a logical state of the toric code to a logical state of
    the anti-toric code.
    Picking $\lvert {\psi}\rangle $ in the code space of the toric code and choosing the
    depicted operator $H$ as the corresponding Hamiltonian yields an example of
    a distance-$\Omega(L^2)$ metrological code.  Interestingly, this
    metrological code cannot be phrased in terms of a time-covariant
    error-correcting code of similar distance, since $\lvert {\psi}\rangle $ and $H\lvert {\psi}\rangle $
    can be distinguished by measuring a single star or plaquette operator.  This
    example shows that there are additional possible schemes for sending a clock
    state through a noisy channel without any sensitivity loss, without
    resorting to a time-covariant
    quantum error-correcting code.} 
  \label{z:dwwtR0lWAZod}
\end{figure}
We consider a two-dimensional square lattice of dimension $L\times L$ that wraps around a torus.
We define star operators $A_{\boldsymbol x}$ and plaquette operators
$B_{\boldsymbol x}$ as depicted in \cref{z:dwwtR0lWAZod}\textbf{a},
where $\boldsymbol x$ ranges over all pairs of the lattice coordinates.

First, we can always use the toric code to form a time-covariant code, by
choosing a state vector $\lvert {\psi}\rangle $ in the code space (for instance
$\lvert {\overline{++}}\rangle $), and choosing the Hamiltonian to be a logical operator
(for instance, $\overline{Z}_1 + \overline{Z}_2$).
This code being by construction time covariant, it is necessarily a metrological code
with distance equal to the lattice side length $L$.

For the sake of the example, we construct here a metrological code from the
toric code that cannot be written as a time-covariant error-correcting code of
similar distance.
Our example is meant to
(i)~illustrate our construction as combining states that lie either in the
simultaneous $+1$ or simultaneous $-1$ eigenspaces of all the stabilizer
generators of some given stabilizer code, (ii)~furnish another example of a
metrological code that cannot be phrased in terms of a time-covariant
error-correcting code with similar distance, and (iii)~illustrate how a metrological
code can be a terrible quantum error-correcting code---any single plaquette or
star operator acts nontrivially on the subspace spanned by the state and its time
evolution state.
Our example is more of a conceptual illustration than a practical proposal, as
it requires a Hamiltonian that is highly nonlocal.

To better explain our example, we first define the \emph{anti-toric code} as the
code whose code space is stabilized by all negative star $-A_{\boldsymbol x}$
and negative plaquette operators $-B_{\boldsymbol x}$.  Being equivalent to the
standard toric code, the anti-toric code also has distance $L$ and we can see it
also has the logical operators
$\overline{X}_1, \overline{X}_2, \overline{Z}_1, \overline{Z}_2$ defined as for
the toric code.

As the state vector $\lvert {\psi}\rangle $ of our metrological code, we simply choose a logical state vector of the toric code; we can conventionally fix it to be 
$\lvert {\overline{00}}\rangle _{\mathrm{toric}}$ stabilized by
$\overline{Z}_1, \overline{Z}_2$ along with all the toric code stabilizers
$\{ A_{\boldsymbol x} \}$ and $\{ B_{\boldsymbol x} \}$.  For the Hamiltonian
we choose an operator $H$ that anticommutes with all star and all plaquette
operators. %
Such an operator is depicted in
\cref{z:dwwtR0lWAZod}\textbf{b}; we assume for convenience that $L$ is even.
The operator $H$ has the property
that it maps a code word of the toric code (i.e., a state vector $\lvert {\psi}\rangle $ satisfying
$A_{\boldsymbol x}\lvert {\psi }\rangle = \lvert {\psi }\rangle = B_{\boldsymbol x} \lvert {\psi}\rangle $) to a code
word of the \emph{anti-toric code} (we have
$A_{\boldsymbol x} H \lvert {\psi }\rangle = -H A_{\boldsymbol x} \lvert {\psi }\rangle = -H \lvert {\psi}\rangle $ and
similarly for $B_{\boldsymbol x}$).
We can verify that the assumptions of
\cref{z:Bnq1hErO9gGx} are satisfied.  The operator
$H$ anticommutes with our choice of stabilizer generators for $\lvert {\psi}\rangle $.  %
Also, for any operator $O$ of weight $< L^2/4$, there
must be a star or plaquette operator that has disjoint support with, and
therefore commutes with, $O$. (Indeed, there are $(L/2)^2$ disjoint plaquette
operators that cover all qubits; an operator that has overlapping support with
all plaquette operators must therefore have support on one qubit in each
plaquette.  The bound can presumably be improved by accounting for the star
operators as well.)  As a consequence of
\cref{z:Bnq1hErO9gGx}, the state vectors
$(\lvert {\psi}\rangle , \lvert {\xi }\rangle = H\lvert {\psi}\rangle )$ form a metrological code of distance $L^2/4$.

Is the space spanned by $(\lvert {\psi}\rangle ,\lvert {\xi}\rangle )$ secretly a code space of a
similar-distance code in which $H$ acts as a logical operator?  We can rule out
this possibility because the state vectors $\lvert {\psi}\rangle $ and $\lvert {\xi}\rangle $ can easily be
distinguished by measuring any single star or plaquette operator, 
recalling that
$A_{\boldsymbol x}\lvert {\psi }\rangle = B_{\boldsymbol x} \lvert {\psi }\rangle = \lvert {\psi}\rangle $ but that
$A_{\boldsymbol x}\lvert {\xi }\rangle = B_{\boldsymbol x} \lvert {\xi }\rangle = -\lvert {\xi}\rangle $.  The
environment only has to measure a weight-4 operator to determine whether
$\lvert {\psi}\rangle $ or $\lvert {\xi}\rangle $ was encoded.

\paragraph{Simultaneous $+1$ eigenspace and simultaneous $-1$ eigenspace of
  stabilizers.}
The intuition behind the construction in \cref{z:Bnq1hErO9gGx}
is that if we can choose
$\lvert {\psi}\rangle $ to be stabilized by $\mathcal{S} = \groupgen{S_1, \ldots, S_\ell}$,
then we might want to pick $\lvert {\xi}\rangle $ to be stabilized by the closely related
stabilizer group $\mathcal{S}' = \groupgen{ - S_1, \ldots, -S_\ell }$.
This idea was already illustrated by the example above based on the toric code,
where the Hamiltonian maps a code word of the toric code to a code word of the
anti-toric code.
We now show in general that such a construction is a special case of
\cref{z:Bnq1hErO9gGx}.

Let $\mathcal{S} = \groupgen{ S_1, \ldots, S_\ell }$
be a
subgroup of the Pauli group with $-\mathds{1}\notin\mathcal{S}$, where
$S_1, \ldots, S_\ell$ are a choice of independent commuting stabilizer
generators.  Let
$\mathcal{S}' = \groupgen{ -S_1, \ldots, -S_\ell }$. %
Let $\mathscr{E}$ denote any set of Pauli operators with the following property:
for any $E,E' \in \mathscr{E}$, there exists $S\in\mathcal{S}$ such that
$-S\in\mathcal{S}'$ and such that $[E'^\dagger E, S] = 0$.

The two stabilizer groups $\mathcal{S},\mathcal{S}'$ share many stabilizers,
including $S_1S_2$, $S_1S_3$, \ldots, $S_1S_\ell$.
We can pick a Pauli operator $H$ such that $H$ anticommutes with $S_1$ and such
that $H$ commutes with each of the operators $S_1S_2$, $S_1S_3$, \ldots,
$S_1S_\ell$ (see, e.g., Ref.~\cite[Proposition~10.4]{R62}).
Observe that for all $i=2,\ldots,\ell$, we have
\begin{equation}
 H S_i = H S_1^2 S_i = -S_1 H S_1 S_i = -S_1^2 S_i H = -S_i H
 \end{equation}
 and thus we have
that $\{H, S_i\}=0$ for all $i=1,\ldots,\ell$.
Suppose $\lvert {\psi}\rangle $ is stabilized by $\mathcal{S}$.  Then $H\lvert {\psi}\rangle $ is
stabilized by $\mathcal{S}'$, since
\begin{equation}
S_i H\lvert {\psi }\rangle = -H S_i\lvert {\psi }\rangle = - H\lvert {\psi}\rangle .
 \end{equation}
Furthermore, supposing $E,E'\in\mathscr{E}$, by assumption we have
$S\in\mathcal{S}$ such that $-S\in\mathcal{S}'$ and such that
$[E'^\dagger E, S] = 0$.  We can write $S = S_{i_1} S_{i_2} \cdots S_{i_m}$ in
terms of our choice of independent generators $S_i$ above.  Writing
$-S = (-1)^{m+1} (-S_{i_1}) (-S_{i_2}) \cdots (-S_{i_m})$, we see that $m$ must
be odd, as otherwise, we would have
\begin{equation}
-S = (-1) (-S_{i_1}) (-S_{i_2}) \cdots (-S_{i_m}) \notin \mathcal{S'}.
 \end{equation}
This observation implies that $\{H, S\} = 0$, because we can anticommute $H$
through the product of an odd number of $S_i$'s.  Therefore there exists
$S\in\mathcal{S}$ such that $[E'^\dagger E, S]=0$ and $\{H, S\}=0$.
At this point, all the assumptions of
\cref{z:Bnq1hErO9gGx} are satisfied,
implying that $\lvert {\psi}\rangle , H\lvert {\psi}\rangle $ form a metrological code that can protect
against the error set $\mathscr{E}$.

\subsection{Further examples of metrological codes}
\label{z:3LRUWXQ20v25}

We now present two additional examples of state vectors $\lvert {\psi}\rangle ,\lvert {\xi}\rangle $ that satisfy the zero
sensitivity loss
conditions~\eqref{z:mvEtMj.b-Bs6}.
These metrological codes serve to illustrate the sense in which the
conditions~\eqref{z:mvEtMj.b-Bs6} are weaker
than the conditions for quantum error correction.

\subsubsection{Single qubit subject to complete X/Y dephasing}
\label{z:kLDpse.aVOd5}

Consider the qubit example studied in
\cref{z:TaWot2w4x4dl}, where the clock
state vector $\lvert {+}\rangle $ evolves according to $H=\omega{Z}/2$ and is exposed to
complete dephasing along the $X$ axis around a given time $t_0$.
From \cref{z:vZfDiFsuC0Ah} we
immediately see that the zero sensitivity-loss
condition~\eqref{z:zFMncm5bkXDt} is satisfied for all
$t_0$.  In this setting, the clock state loses no sensitivity after complete
dephasing in the $X$ axis for any $t_0$, with the exception of possible discrete
points where the rank of $\rho_B(t)$ changes (see
\cref{z:TaWot2w4x4dl}).

Alternatively one could also check the
form~\eqref{z:G2tCYcCXVZFu} of the zero sensitivity-loss
conditions.  For any $t_0$, we have from~\eqref{z:UlWOzs-evrpA}
that
\begin{align}
  \lvert {\xi(t_0)}\rangle 
  &= H\lvert {\psi(t_0)}\rangle 
    \nonumber\\
  &= \frac\omega2\Bigl[ \cos\Bigl(\frac{\omega t_0}{2}\Bigr)\lvert {-}\rangle 
      -i \sin\Bigl(\frac{\omega t_0}{2}\Bigr)\lvert {+}\rangle  \Bigr] \,.
\end{align}
At this point, we can compute
\begin{align}
  \hspace*{1em}&\hspace*{-1em}
  \langle {\psi}\mkern 1.5mu\relax \vert \mkern 1.5mu\relax { \lvert {+}\rangle \mkern -1.8mu\relax \langle{+}\rvert  }\mkern 1.5mu\relax \vert \mkern 1.5mu\relax {\xi }\rangle + \langle {\xi}\mkern 1.5mu\relax \vert \mkern 1.5mu\relax { \lvert {+}\rangle \mkern -1.8mu\relax \langle{+}\rvert  }\mkern 1.5mu\relax \vert \mkern 1.5mu\relax {\psi
  }\rangle \nonumber\\
  &= \frac\omega2 \Bigl[
    -i\cos\Bigl(\frac{\omega t_0}{2}\Bigr)\sin\Bigl(\frac{\omega t_0}{2}\Bigr)
    +i\sin\Bigl(\frac{\omega t_0}{2}\Bigr)\cos\Bigl(\frac{\omega t_0}{2}\Bigr)
    \Bigr]
    \nonumber\\
  &=0\ ,
\end{align}
and similarly for $\langle {\psi}\mkern 1.5mu\relax \vert \mkern 1.5mu\relax { \lvert {-}\rangle \mkern -1.8mu\relax \langle{-}\rvert  }\mkern 1.5mu\relax \vert \mkern 1.5mu\relax {\xi }\rangle + \langle {\xi}\mkern 1.5mu\relax \vert \mkern 1.5mu\relax { \lvert {-}\rangle \mkern -1.8mu\relax \langle{-}\rvert  }\mkern 1.5mu\relax \vert \mkern 1.5mu\relax {\psi}\rangle =0$,
showing that~\eqref{z:G2tCYcCXVZFu} are satisfied for all
$t_0$.

An interesting aspect of this example is that there exists no recovery operation
that can restore the noiseless clock state vector $\lvert {\psi(t_0+dt)}\rangle $
accurately to first order in $dt$.  Let us consider for simplicity the point
$t_0=\pi/(2\omega)$.
Using~\eqref{z:UlWOzs-evrpA},~\eqref{z:GM3ZeRQLM9mt},
and \eqref{z:qdCKlozHjX8s}, we have at that point
\begin{align}
  \psi(t_0) &= \lvert {+i}\rangle \mkern -1.8mu\relax \langle{+i}\rvert \ ,
  &
    \rho_B(t_0) &= \frac{\mathds{1}}2\ ,
                  \nonumber\\
    \partial_t \psi\,(t_0) &= -\frac\omega2 
    {X}\ ,
  &
    \mathcal{D}_X(\partial_t \psi\,(t_0)) &= -\frac\omega2 
    {X}\ ,
\end{align}
where $\lvert {+i}\rangle  := \bigl[\lvert {\uparrow }\rangle + i\lvert {\downarrow}\rangle \bigr]/\sqrt2$.
We seek a completely positive, trace-preserving map $\mathsf{Rec}$ such that
$\mathsf{Rec}(\rho_B(t_0+dt)) = \psi(t_0+dt)+O(dt^2)$, which means that
\begin{align}
  \mathsf{Rec}\Bigl(\frac{\mathds{1}}{2}\Bigr) &= \lvert {+i}\rangle \mkern -1.8mu\relax \langle{+i}\rvert \ ,
  &
  \mathsf{Rec}(
  {X}) &= 
  {X}\ .
\end{align}
There is no completely positive map that satisfies these constraints.  If there
was such a $\mathsf{Rec}$ map, then we would have
$\mathsf{Rec}(\lvert {+}\rangle \mkern -1.8mu\relax \langle{+}\rvert ) = \mathsf{Rec}(\mathds{1}/2) + \mathsf{Rec}(
{X}/2) =
\lvert {+i}\rangle \mkern -1.8mu\relax \langle{+i}\rvert  + 
{X}/2 = \mathds{1}/2 + 
{Y} /2 + 
{X}/2$.  One can easily
check that the final expression has a negative eigenvalue, contradicting the
requirement that $\mathsf{Rec}$ be completely positive.
We conclude that in general, a metrological code does not necessarily come with
a recovery operation that enables an agent to recover the noiseless clock state,
even if the agent can sense the parameter to the same precision as before the
application of the noise.

\subsubsection{A superposition of a simple state and a generic pure state}
\label{z:VsSVmS3N9rs6}

Consider a one-dimensional chain of $n$ qubits.  Consider a generic pure state vector $\lvert {\chi}\rangle $,
chosen for instance randomly from the Haar measure on the $n$-qubit system.
For a given $d_m>0$, let us perturb the state vector $\lvert {\chi}\rangle $ to $\lvert {\tilde\chi}\rangle $
by projecting it onto the subspace of all computational basis states that do not
contain fewer than a number $d_m$ of 1's,
\begin{align}
  \lvert {\tilde\chi}\rangle  &= \tilde\Pi \lvert {\chi}\rangle \ ,
  &
    \tilde\Pi = \prod_{\lvert {\boldsymbol x}\rvert  \geq d_m} \lvert {\boldsymbol x}\rangle \mkern -1.8mu\relax \langle{\boldsymbol x}\rvert \ .
\end{align}
If $\lvert {\chi}\rangle $ is generic in some suitable sense (e.g., chosen Haar-randomly),
then $\lvert {\tilde\chi}\rangle  \approx \lvert {\chi}\rangle $ and $\lVert {\lvert {\tilde\chi}\rangle }\rVert \approx 1$.
The present example metrological code is constructed by picking
$\lvert {\psi }\rangle = \lvert {0^n}\rangle $, which is the computational basis all-zero state, and
$\lvert {\xi }\rangle = \lVert {\lvert {\tilde\chi}\rangle }\rVert ^{-1}\,\lvert {\tilde\chi}\rangle  \approx \lvert {\chi}\rangle $.

We proceed to check that the zero sensitivity-loss conditions~\eqref{z:zFMncm5bkXDt} are satisfied as
long as operators of the form $E_{k'}^\dagger E_k$ have weight at most $d_m-1$.
If $\mu\subset\{1,\ldots,n\}$ denotes a subset of at most $\lvert {\mu }\rvert = d_m-1$
systems, then the reduced operator of $\lvert {\tilde\chi}\rangle \mkern -1.8mu\relax \langle{0^n}\rvert $ on the sites
labeled by $\mu$ can be written as
\begin{align}
  \operatorname{tr}_{\setminus\mu}\bigl[\lvert {\tilde\chi}\rangle \mkern -1.8mu\relax \langle{0^n}\rvert \bigr]
  &= (\mathds{1}_{d_m}\otimes\langle {0^{n-d_m}}\rvert )\,\lvert {\tilde\chi}\rangle \langle {0^{d_m}}\rvert 
    = 0\ ,
\end{align}
because $\lvert {\tilde\chi}\rangle $ has no overlap with bit strings that have $(n-d_m)$
or more zeros.  Hence
\begin{align}
  \operatorname{tr}_{\setminus\mu}\bigl( \lvert {\psi}\rangle \mkern -1.8mu\relax \langle{\xi }\rvert + \lvert {\xi}\rangle \mkern -1.8mu\relax \langle{\psi }\rvert \bigr) = 0\ .
\end{align}
For any operator $O$ of weight $\wgt(O) < d_m$, the
condition~\eqref{z:y2lXSpYS4Z-t} is thus
satisfied and $\lvert {\psi}\rangle ,\lvert {\xi}\rangle $ form a metrological code of distance $d_m$.

An interesting observation is that this code does not form a quantum
error-correcting code in the usual sense.  The reason is that the environment,
by receiving a few sites, can tell the difference between whether the state vector 
$\lvert {\psi}\rangle $ or the state vector $\lvert {\xi}\rangle $ was prepared.  More precisely, the
environment can test whether the received qubits are all in the state vector
$\lvert {0}\rangle $.  If this is the case, it is much more likely that the original state vector
was $\lvert {\psi}\rangle $ and not $\lvert {\xi}\rangle $, as long as $\operatorname{tr}_{\setminus\nu}(\tilde\chi)$
is sufficiently distinct from $\lvert {0}\rangle \mkern -1.8mu\relax \langle{0}\rvert _\nu$ (which is the case for a Haar-random
state).

For some choices of $\lvert {\chi}\rangle $, the metrological code can be interpreted as a
quantum error-correcting code that protects only against certain types of
errors.  For instance, we can choose $\lvert {\chi }\rangle = \lvert {1}\rangle ^{\otimes n}$ in the
above and our conclusions still hold; this choice corresponds to a classical
repetition code that can correct bit flips but which is vulnerable to phase
flips.

Yet, there are choices of $\lvert {\chi}\rangle $ for which this interpretation appears more
problematic.  Consider for instance the choice $\lvert {\chi }\rangle = \lvert {+}\rangle ^{\otimes n}$.
From the above argument we have that $\lvert {\psi}\rangle =\lvert {0}\rangle ^{\otimes n}$ and
$\lvert {\xi}\rangle \approx\lvert {+}\rangle ^{\otimes n}$ form again a metrological code.  Again, the
environment can distinguish $\lvert {\psi}\rangle $ from $\lvert {\xi}\rangle $ with access only to a few
sites.  Here, the environment can use either an $X$ or a $Z$ measurement to
(imperfectly) distinguish between the two state vectors $\lvert {0}\rangle $ and $\lvert {+}\rangle $.  
It is hence not obvious how to interpret this code as an error-correcting code
that is tailored to biased noise.

While this example might illustrate the conceptual differences between quantum
error-correcting codes and metrological codes, we expect this construction of a
metrological code to be of limited practical use as it would require a
Hamiltonian that is extremely nonlocal.

\section{Many-body system subject to i.i.d.\@ amplitude damping noise}
\label{z:Pgp-gW09n9nZ}

In this section we consider a system consisting of $n$ spin-$1/2$ particles
evolving under a many-body Hamiltonian $H$ that is either noninteracting or that
has Ising interaction terms.  The system is exposed to i.i.d.\@ amplitude
damping noise.
First, we consider a noninteracting Hamiltonian with an on-site magnetic field,
and in the second part of this section we consider a Hamiltonian with Ising
interactions.

We consider %
an i.i.d.\@ amplitude damping noise model, meaning
that each site is independently exposed to the noisy channel
\begin{align}
  \begin{gathered}
    \mathcal{N}_{\mathrm{a.d.}}^{(p)}(\cdot)
    = E_0^{(p)}\, (\cdot)\, E_0^{(p)\,\dagger}
    + E_1^{(p)} \,(\cdot) \, E_1^{(p)\,\dagger}\ ,
    \quad %
    \\
    \begin{aligned}
        E_0^{(p)} &= \begin{pmatrix} \sqrt{1-p} & 0 \\ 0 & 1\end{pmatrix}\ ,
        &
        E_1^{(p)} &= \begin{pmatrix} 0 & 0 \\ \sqrt{p} & 0\end{pmatrix}\ ,
      \end{aligned}
  \end{gathered}
\end{align}
sticking to the convention that the first basis vector is $\lvert {\uparrow}\rangle $ and the
second one is $\lvert {\downarrow}\rangle $.  The amplitude-damping noise is often
also called the spontaneous emission channel.

As in \cref{z:FnHui0ahNLxW}, the system is initialized in a state
$\lvert {\psi_{\mathrm{init}}}\rangle $ and evolves according to $H$; at time $t_0$ we
apply the noisy channel $[\mathcal{N}_{\mathrm{a.d.}}^{(p)}]^{\otimes n}$ to
obtain Bob's state.  We seek to characterize the Fisher information of Bob's
state with respect to time.

In this section, we present simple numerical computations of the upper
bound~\eqref{z:3YnRZBPm.g17} for i.i.d.\@ amplitude damping noise for
different clock states.
In the first part of this section, we suppose the spins are exposed to a uniform
external magnetic field aligned along the $Z$ axis.  We present numerical
calculations of Bob's Fisher information and our lower
bound~\eqref{z:3YnRZBPm.g17} for a choice of clock states, and we
numerically optimize the initial state to achieve better output sensitivity.
In the second part of this section, we place the spins on a 1D chain with strong
Ising interactions.  We present numerical calculations of Bob's Fisher
information and our lower bound~\eqref{z:3YnRZBPm.g17} for a choice of
clock states; we numerically show that the sensitivity loss for the metrological
code state given in~\eqref{z:uYISiNnDt2yR} is
suppressed to first order in the amplitude damping parameter.

\subsection{Noninteracting Hamiltonians}

The system of $n$ spins is assumed to evolve 
under the Hamiltonian
\begin{align}
  H = \sum_{i=1}^n  \frac\omega2 %
  {Z}_i   .
  \label{z:dle6kUdtgpjO}
\end{align}
We compute Bob's Fisher information with respect to time of a selection of
states after exposure to the channel
$\mathcal{N} = [\mathcal{N}_{\mathrm{a.d.}}^{(p)}]^{\otimes n}$.  First we
consider the GHZ state, which has the optimal sensitivity if no noise is present:
\begin{align}
  \lvert {\psi_{\mathrm{GHZ}}}\rangle  
  &= \frac1{\sqrt 2} \bigl[
    \lvert {\uparrow\uparrow\cdots\uparrow}\rangle  +
    \lvert {\downarrow\downarrow\cdots\downarrow}\rangle 
    \bigr]\ .
\end{align}
The GHZ state satisfies
\begin{align}
  \FIqty{Alice}{t}[\psi_{\mathrm{GHZ}}]
  = 4 \langle { H^2 }\rangle _{\mathrm{GHZ}}
  = n^2 \omega^2\ .
\end{align}
We can also consider the product state vector of all spins pointing in the $+X$
direction,
\begin{align}
  \lvert {\psi_+}\rangle 
  = \lvert {+}\rangle ^{\otimes n} = \frac1{\sqrt{2^{n}}}
  \bigl[ \lvert {\uparrow }\rangle + \lvert {\downarrow }\rangle \bigr] ^{\otimes n}\ .
\end{align}
Then
\begin{align}
  \FIqty{Alice}{t}[\psi_{\mathrm{+}}]
  = 4 \langle { H^2 }\rangle _{\mathrm{+}}
  = n \omega^2\ .
\end{align}
Our upper bound~\eqref{z:3YnRZBPm.g17} on Bob's Fisher information for
these states is presented for $n=12$ and for $n=50$ in
\cref{z:McvK4GoEf27I}.
\begin{figure*}
  \centering
  \includegraphics{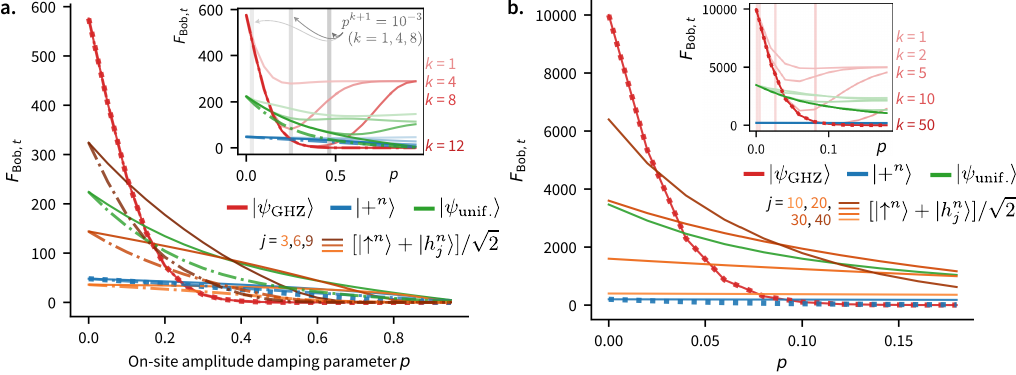}
  \caption{Quantum Fisher information of a system of $n$ spin-$1/2$
    particles with the Hamiltonian $H=\sum_i \omega%
    {Z_i}
    /2$ (with
    $\omega=2$) after the application of an amplitude damping channel of
    parameter $p$ on all sites.  \textbf{a.}~Here $n=12$.  Solid lines depict our upper
    bound~\eqref{z:3YnRZBPm.g17} with $k=n=12$ for the state vectors 
    $\lvert {\psi_{\mathrm{GHZ}}}\rangle $ (red), $\lvert {+^n}\rangle $ (blue), and
    $\lvert {\psi_{\mathrm{unif}}}\rangle $ (green), which are defined in the main text,
    as well as for the family of states corresponding to an even superposition
    of the most excited state and a symmetric state (Dicke state) $\lvert {h_j^n}\rangle $
    with a fixed number $n-j$ of excitations (shades of orange) with
    $j=3,6,9$.  Dotted lines are corresponding ad hoc lower bounds for the
    state vectors $\lvert {\psi_{\mathrm{GHZ}}}\rangle $ and $\lvert {+^n}\rangle $ (see main text).
    Dash-dotted lines are the corresponding exact values of the quantum Fisher
    information, which can still be directly computed for $n=12$.
    The curves corresponding to the superpositions of pairs of Dicke states illustrate
    situations where upper bound is not tight.
    The inset depicts our upper bound~\eqref{z:3YnRZBPm.g17} for
    different values of $k=1,4,8,12$ for the state vectors
    $\lvert {\psi_{\mathrm{GHZ}}}\rangle $, $\lvert {+^n}\rangle $ and $\lvert {\psi_{\mathrm{unif}}}\rangle $.
    The value of $k$ corresponds to a projection onto the subspace on Eve's
    system associated with errors of weight less than or equal to $k$ which is
    included in our bound~\eqref{z:3YnRZBPm.g17}.  The three vertical
    gray lines indicate values of $p$ for which $p^{k+1} = 10^{-3}$ for
    $k=1,4,8$; these lines roughly indicate the values of $p$ beyond which the
    this projection is expected to fail with a probability exceeding
    $\sim 10^{-3}$.  Indeed, for our choice of noisy channel and states, our
    bound~\eqref{z:3YnRZBPm.g17} for each $k$ is reasonably tight up
    until values of $p$ for which $p^{k+1}$ is no longer negligibly small.
    \textbf{b.}~The same computations are repeated for $n=50$ and
    $\omega=2$.  Solid lines depict our upper bound with $k=n$ for the same
    states as in (a), and dotted lines depict an ad hoc lower bound for
    $\lvert {+^n}\rangle $ and $\lvert {\psi_{\mathrm{GHZ}}}\rangle $.  The red vertical lines in
    the inset depict values for $p$ for which the total weight of the i.i.d.\@
    amplitude-damping Kraus operators $E_{\boldsymbol x}$ with
    $\lvert {\boldsymbol x}\rvert >k$ exceeds $10^{-3}$ for $k=1,2,5,10$ when applied
    onto the GHZ input state vector $\lvert {\psi_{\mathrm{GHZ}}}\rangle $.  These values of $p$
    are where our upper bound~\eqref{z:3YnRZBPm.g17} for the
    corresponding $k$ value is expected to no longer be accurate.
  }
  \label{z:McvK4GoEf27I}
\end{figure*}
In \cref{z:McvK4GoEf27I}\textbf{a} and \textbf{b}
are also depicted an ad hoc lower bound for the state vectors
$\lvert {\psi_{\mathrm{GHZ}}}\rangle $ and $\lvert {+^n}\rangle $ for the same values of $n$ and
$\omega$.  We can see that for our choice of the amplitude damping noise model
and at least for our choice of states, the upper bound on $\FIqty{Bob}{t}$
provided by~\eqref{z:3YnRZBPm.g17} is reasonably tight for $k=n$.
The inset of \cref{z:McvK4GoEf27I}\textbf{a} depicts
the same bound for different choices of the value $k=1,4,8,12$.  Recall that the
bound includes a projection onto Eve's subspace associated with error operators
of weight at most $k$.  The probability that this projection fails is the total
probability of observing an error with weight greater than $k$; this probability
is of the order of $p^{k+1}$.  In the inset of
\cref{z:McvK4GoEf27I}\textbf{a}, for $k=1,4,8$ we
display gray lines identifying the values of $p$ for which $p^{k+1} = 10^{-3}$.
Values of $p$ beyond the corresponding gray line represent situations in which
the projection is expected to fail with probability greater than the order of
$\sim 10^{-3}$.  Here, we see that our bound is indeed reasonably tight up until
the corresponding value of $p$.
In the inset of \cref{z:McvK4GoEf27I}\textbf{b}, we
determine more precisely the total weight of the events neglected by ignoring
Kraus operators of weight greater than $k$.  Namely, for $k=1,2,5,10$, we
compute the smallest value of $p$ for which
$\sum_{\lvert {\boldsymbol x}\rvert >k}\operatorname{tr}\bigl(E_{\boldsymbol x}^\dagger E_{\boldsymbol
  x}\,\lvert {\psi_{\mathrm{GHZ}}}\rangle \mkern -1.8mu\relax \langle{\psi_{\mathrm{GHZ}}}\rvert \bigr) > 10^{-3}$.  We see that these values of $p$
correspond approximately to where our upper
bound~\eqref{z:3YnRZBPm.g17} fails to accurately predict the value of
the quantum Fisher information on the state vector $\lvert {\psi_{\mathrm{GHZ}}}\rangle $.

We can ask, which state vector $\lvert {\psi}\rangle $ has the best sensitivity after application of
the noisy channel for a given value of $p$?  Here we use the understanding
brought by our main Fisher information trade-off relation.  In this case, Eve
receives any photons emitted by spontaneous emission, which tell her exactly
which sites suffered a decay.  As a consequence, if Eve observes a number $k$ of
photons, then she can safely guess that the energy of Alice's state must have
been at least the energy corresponding to $k$ excitations.  Eve has therefore
obtained information about the energy of Alice's state.  This observation
provides a simple explanation for why the GHZ state has a high Fisher
information loss even for small values of $p$: When Alice exposes a GHZ state to
the noise, then Eve can estimate the energy of Alice's GHZ state by noting
whether or not she observes a photon.  If Eve observes even a single photon,
then she can safely guess the energy associated with the all-excited state, and
if she observes no decay, she guesses the energy of the ground state.  (Her
guess is wrong with probability $(1-p)^n$, corresponding to the probability of
the all-excited state suffering no decay.)
Our trade-off relation thus tells us that we seek a state with a large energy
spread, but for which a decay would not betray the value of the total energy of
the state.  As the effect of the decay becomes more significant with increasing
$p$, some of the energy spread is sacrificed in order to make the state more
resilient to Eve's probe.

Here we consider states that are invariant under permutations of the $n$ spins,
motivated by the fact that the Hamiltonan is permutation invariant (see also~\cite{R41,R40}).
These states live in the symmetric subspace and can be written in the basis of
Dicke states of the symmetric subspace.  A Dicke state is a
permutation-invariant state with a fixed number of excitations.  More
specifically, for $q=0,\ldots, n$ we define
\begin{align}
  \lvert {h_q^n}\rangle  = \binom{n}{q} ^{-1/2}\,
    \sum_{\lvert {\boldsymbol x}\rvert  = q} \lvert {\boldsymbol x}\rangle \ ,
\end{align}
where the sum ranges over all strings $\boldsymbol x$ with
$x_i = {\uparrow},{\downarrow}$, and where $\lvert {\boldsymbol x}\rvert $ denotes the
number of sites $i$ where $x_i = {\downarrow}$.  In the standard basis
$\{\lvert {0}\rangle =\lvert {\uparrow}\rangle , \lvert {1}\rangle =\lvert {\downarrow}\rangle \}$ for spin-1/2 particles, the value
$\lvert {\boldsymbol x}\rvert $ is the Hamming weight of the corresponding computational
basis state $\boldsymbol x$.

A general pure symmetric state vector can therefore be written as
\begin{align}
  \lvert {\psi }\rangle = \sum_{q=0}^n  \psi_q \, \lvert {h_q^n}\rangle \ .
  \label{z:onOgq7042npF}
\end{align}
We consider symmetric states for convenience, although the optimal
state in such settings need not be symmetric~\cite{R66}.

We can consider an even superposition of two Dicke states (as in
\cref{z:lZkawwgrYk1u}).  Namely, for
$0\leq q_1 , q_2 \leq n$ we consider the state vector
\begin{align}
  \lvert {\widetilde\psi_{q_1;q_2}}\rangle 
  = \frac1{\sqrt2}\bigl[ \lvert {h_{q_1}^n}\rangle  + \lvert {h_{q_2}^n}\rangle  \bigr]\ .
  \label{z:zWhqAVJdt70e}
\end{align}
Our upper bound on the sensitivity of the state vector
$ \lvert {\widetilde\psi_{q_1;q_2}}\rangle $ for all $q_1,q_2$ is depicted
in~\cref{z:m9JDQIefbcfW} for $n=50$ and the
values of $p=0.01,0.05,0.1, 0.25$.
\begin{figure}
  \centering
  \includegraphics{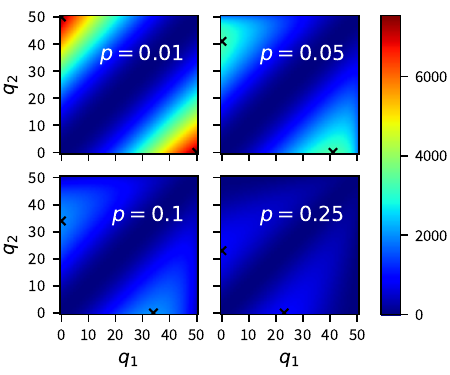}
  \caption{Upper bound~\eqref{z:3YnRZBPm.g17} on the quantum Fisher
    information of a superposition of two Dicke states on $n=50$ spin-$1/2$
    particles after being exposed to i.i.d.\@ amplitude-damping noise.  A Dicke
    state with $q$ excitations is an even superposition of all states with
    exactly $q$ excitations.  For each $q_1,q_2$, the considered state is an
    even superposition of the two Dicke states with a number $q_1$ and $q_2$ of
    downwards-pointing spins respectively.  Each plot corresponds to a different
    value of the single-site amplitude-damping noise parameter $p$.  Our bound
    attains its maximum on this family of states (black crosses) for states that are
    an even superposition
    of the highest excited state and of a
    weakly excited Dicke state that is separated from the ground state.  This
    separation hinders Eve from accurately guessing the energy of Alice's state,
    which via our trade-off relation improves the state's sensitivity after
    application of the noisy channel.  The bound is computed with $k=n=50$.}
  \label{z:m9JDQIefbcfW}
\end{figure}
In contrast to the case of erasures
(\cref{z:lZkawwgrYk1u}), the states among this family
where our bound is large have one of the terms being close to the maximally
excited state ($q_1=0$ or $q_2=0$).  (The bound is not necessarily expected
to be tight, in light of the gap that is apparent for $n=12$ in
\cref{z:McvK4GoEf27I}\textbf{a} between our bound
and the exact value of the quantum Fisher information.  The discussion that
follows aims to identify states that can potentially have high sensitivity,
while ruling out states that are certain to have low sensitivity.)  This
property can again be understood from our trade-off relation.  In the case of
erasures, Eve receives the entire reduced state of the systems that {have been}
lost.  If $q_1=0$ or $q_2=0$, then the reduced state on each subsystem is the
pure state vector $\lvert {\uparrow}\rangle $; since it is a pure state, it is easier for Eve
to distinguish it from the reduced state of the other Dicke state vector
$\lvert {h_{q_2}^n}\rangle $.  In the case of amplitude damping, Eve knows only whether a
decay happened or not on each site and she cannot access the full reduced state.
An alternative phrasing of this argument is to express erasures as a random
operation
$X,Y,Z$
applied onto each site;
equivalently, a random operation from the set $\{\sigma_+, \sigma_-, 
{Z} \}$
is applied on each site, where $\sigma_\pm = [%
{X} \pm i
{Y} ]/\sqrt2$
are the creation and annihilation operators of the qubit excitation on a
specific site.  In the case of amplitude damping, the Kraus operators have no
overlap with $\sigma_+$, meaning that physically, there is no event in which
excitations are created in the system.  Such events, however, happen in the case
of erasures.  If Eve receives the information that $k$ such events have
occurred, she can safely assert that the energy of Alice's state could not have
exceeded the energy of the state that can still accommodate $k$ further
excitations.  Thus, the state $q_1=0$ can easily be ruled out by Eve in the case
of erasures if she receives a report of even a single $\sigma_+$ event.

Another interesting choice of state is the uniform superposition of all Dicke
states, giving rise to
\begin{align}
  \lvert {\psi_{\mathrm{unif}}}\rangle 
  = \frac1{\sqrt{n+1}} \sum_{q=0}^n \lvert {h_q^n}\rangle \ .
\end{align}
The intuitive reason we expect this state to achieve a good sensitivity after
the noise is that if Eve observes emitted photons, she gains comparatively
little information about the energy of Alice's state as opposed to if the state
is a superposition of few spaced-out Dicke states.  The Fisher information of
this state after the application of the amplitude damping noisy channel is
depicted in \cref{z:McvK4GoEf27I}.

A more systematic, numerical optimization of $\FIqty{Bob}{t}$ by varying Alice's
state using different Ans\"atze for the coefficients $\{\psi_q\}$ indicate that
the Fisher information obtained by states of the
form~\eqref{z:zWhqAVJdt70e} can be marginally exceeded
for specific values of $p$ by states for which the amplitudes $\{ \psi_q \}$ are
concentrated around two values $q_1=0$ and some value $q_2$, but with some
broadening to include some weight on neighboring Dicke states to $q_2$.
Interestingly, there appears to be many states with very different profiles of
$\{\psi_q\}$ that achieve a very similar sensitivity after the application of the
noisy channel.

One particular such state is the state vector $\lvert {\psi_{\textrm{half-Gauss}}}\rangle $ of the
form~\eqref{z:onOgq7042npF} where the $\psi_q$ coefficients
are a half-Gaussian centered on the all-excited state as
\begin{align}
  \lvert {\psi_{\textrm{half-Gauss}}}\rangle 
  &= \sum_q \psi_q\,\lvert {h_q^n}\rangle \ ,
  &
  \psi_q &= \frac1c {e}^{- \frac{(q/n)^2}{2w^2} }\ ,
\end{align}
where $c$ is determined from the normalization condition.  Empirically, we find
that this state with a value of $w=0.4$ yields a sensitivity after application
of the noisy channel that is competitive with respect to the other studied
states.
The half-Gaussian spreads with over the entire Dicke basis. The
amplitude of the ground state is
$\psi_n = {e}^{-1/(2w^2)}\psi_0 \approx 0.04\psi_0$.
Here again, the state vector $\lvert {\psi_{\textrm{half-Gauss}}}\rangle $ balances a broad spread in
energy values while still preventing Eve from easily finding out the energy of Alice's
state.

For our numerical calculations, we employed the standard Python \emph{NumPy} and
\emph{SciPy} toolboxes along with \emph{QuTip}~\cite{R67,R68}. %
The permutation invariance of our setting greatly simplifies the calculation of
terms of the form $\operatorname{tr}\bigl(E_{\boldsymbol x'}^\dagger E_{\boldsymbol x} \psi\bigr)$
and $\operatorname{tr}\bigl(E_{\boldsymbol x'}^\dagger E_{\boldsymbol x} \{\bar{H},\psi\}\bigr)$
because $E_{\boldsymbol x}$ is a tensor product of single-site operators.
Similarly, the reduced operator on a given number $k$ of sites of any operator
acting on the symmetric subspace can be computed easily by combinatorial
considerations in a basis of the symmetric
subspace~\cite{R42,R41,R40}.
Even for $n=50$, our pinched bound is easy to compute even for $k\sim n$ in the
permutation-invariant setting: To determine the diagonal matrix elements
associated with $\widehat{\mathcal{N}}(\psi)$ and
$\widehat{\mathcal{N}}(\{\bar{H},\psi\})$, it suffices to compute terms of the
form $\operatorname{tr}\bigl(E_{\boldsymbol x}^\dagger E_{\boldsymbol x} \psi\bigr)$ and
$\operatorname{tr}\bigl(E_{\boldsymbol x}^\dagger E_{\boldsymbol x} \{\bar{H},\psi\}\bigr)$ for
operators $E_{\boldsymbol x}$ of the form
$E_1^{\otimes w}\otimes E_0^{\otimes (n-w)}$ (the other terms are determined by
symmetry).

\subsection{Strongly interacting Ising Hamiltonian with a noisy channel}

Consider a one-dimensional spin chain with nearest-neighbor $ZZ$ couplings, with
the Hamiltonian
\begin{align}
  H = \sum_{j=1}^{n-1} \frac J2\,%
  {Z_j}
{Z_{j+1}}.
  \label{z:6xJnS9ZmA02E}
\end{align}
Our upper bound on Bob's sensitivity to time, computed using the
expression~\eqref{z:3YnRZBPm.g17} for various states, is plotted in
\cref{z:eNrtlC61cj3z} for $n=12$ and $n=50$, with $J=2$ in
both plots.
\begin{figure*}
  \centering
  \includegraphics{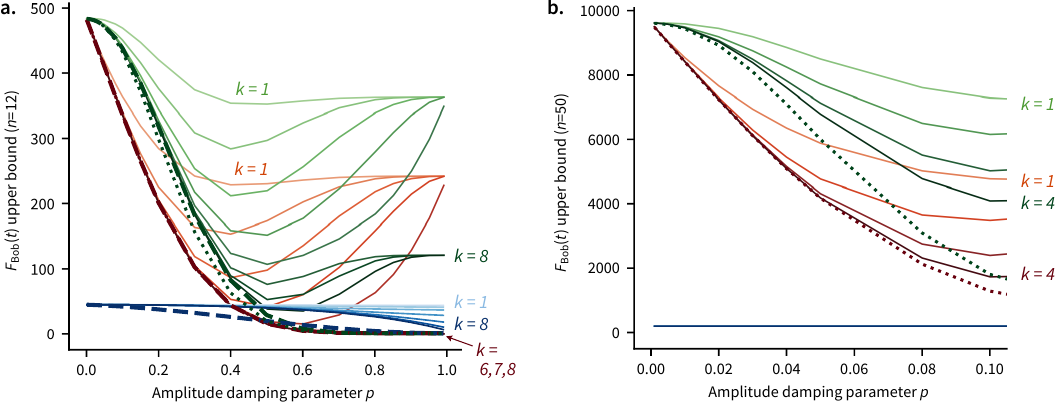}%
  \caption{Numerical calculation of the upper bound on the sensitivity of a
    many-qubit state after i.i.d.\@ amplitude-damping noise, for different 
    states evolving according to the one-dimensional Ising Hamiltonian
    $H=\sum_{j} (J/2)\, 
    {Z_j} {Z_{j+1}}$.
    \textbf{a.}~%
    Upper bound on Bob's Fisher information $F_{\textrm{Bob}}(t)$ for $n=12$ qubits as a function of
    amplitude-damping parameter $p$, for the  state
    $\lvert {\psi_{\textrm{f-af}}}\rangle $, which is an even superposition of a
    ferromagnet state and an antiferromagnet state (in red), for
    $\lvert {\psi_{\textrm{code-f-af}}}\rangle $ which satisfies our Knill-Laflamme-like
    conditions for a single located error (in green), and for the spin-coherent
    state vector $\lvert {+}\rangle ^{\otimes n}$ where $\lvert {+}\rangle $ is the $+1$ eigenvector of
{$X$}
    (in blue).  Spin-coherent states are typically used in
    non-quantum-enhanced metrology.  %
    For each  state, bounds are shown for various values of $k$, where $k$ is a parameter in the additional noisy channel acting on Eve's system which is used to derive the upper bound;
    bounds with larger values of $k$ are harder to compute
    but are tighter.  
    The upper bounds can increase again for $p\gtrsim 0.4$
    because the bound only accounts for low-weight Kraus operators, and
    higher-weight errors cannot be ignored in this regime.  Dashed curves
    indicate the true Fisher information values, determined by direct
    computation, and the dotted lines are lower bounds associated with
    $\lvert {\psi_{\textrm{f-af}}}\rangle $ and $\lvert {\psi_{\textrm{code-f-af}}}\rangle $.  The red
    curves for $k=6,7,8$, the associated true value, and the lower bound
    appear superimposed.
    \textbf{b.}~Upper bounds for the  states
      $\lvert {\psi_{\textrm{f-af}}}\rangle $ and $\lvert {\psi_{\textrm{code-f-af}}}\rangle $ for
      $n=50$ qubits, zoomed in on low values of $p$.  Dotted lines are lower
      bounds on the Fisher information for these 
      states. %
    Computing the true values of the Fisher information in this regime would
    require more advanced methods, such as tensor
    networks~\protect\cite{R27}.%
  }
  \label{z:eNrtlC61cj3z}
\end{figure*}
We first consider the  state vector corresponding to an even superposition of a
ferromagnetic all-zero state and an antiferromagnetic state vector
\begin{align}
  \lvert {\psi_{\textrm{f-af}}}\rangle  = \frac1{\sqrt2}\bigl[
  \lvert {\downarrow\downarrow\downarrow\downarrow\ldots}\rangle  +
  \lvert {\downarrow\uparrow\downarrow\uparrow\ldots}\rangle 
  \bigr]\ .
\end{align}
Since $\lvert {\downarrow\downarrow\downarrow\downarrow\ldots}\rangle $ and
$\lvert {\downarrow\uparrow\downarrow\uparrow\ldots}\rangle $ are energy eigenvectors of
respective energies $\omega(n-1)$ and $-\omega(n-1)$, we see that the state vector 
$\lvert {\psi_{\mathrm{f-af}}}\rangle $ has energy variance
$\sigma_H^2(\psi_{\textrm{f-af}}) = \omega^2(n-1)^2$.  For comparison, we
compute the true values of the Fisher information (for $n=12$), plotted as
dashed lines in \cref{z:eNrtlC61cj3z}\textbf{a}, as well as
an ad hoc lower bound, plotted as dotted lines.
As can be seen in \cref{z:eNrtlC61cj3z}, our upper bound
yields tight bounds on the time sensitivity of the many-body interacting probe,
as witnessed by its proximity to the true value and to the ad hoc lower bound,
provided $p$ is not too large and $k$ can be taken to be large enough.

The next probe state vector we consider is
\begin{align}
  \begin{alignedat}{1}
    \lvert {\psi_{\textrm{code-f-af}}}\rangle 
    = %
    \frac{1}{2} \bigl[&
    \lvert {\downarrow\downarrow\downarrow\downarrow\ldots}\rangle 
    + \lvert {\uparrow\uparrow\uparrow\uparrow\ldots}\rangle 
    \\
    &+ \lvert {\downarrow\uparrow\downarrow\uparrow\ldots}\rangle 
    + \lvert {\uparrow\downarrow\uparrow\downarrow\ldots}\rangle 
    ]\ .
  \end{alignedat}
\end{align}
The state vector $\lvert {\psi_{\textrm{code-f-af}}}\rangle $ is the
state~\eqref{z:uYISiNnDt2yR} using the
antiferromagnetic configuration as the bit string $\boldsymbol{x}$.  Recall that
this state satisfies our Knill-Laflamme-like conditions for a single located
error.  Here we study how this state's sensitivity is affected when exposed to
i.i.d.\@ amplitude-damping noise.
Our upper bound on Bob's Fisher
information via~\eqref{z:3YnRZBPm.g17} is plotted in
\cref{z:eNrtlC61cj3z} for $n=12$ and $n=50$, alongside that
of $\lvert {\psi_{\textrm{f-af}}}\rangle $.  The probe state remains almost maximally
sensitive when $p$ is small, in contrast to the probe
$\lvert {\psi_{\textrm{f-af}}}\rangle $ which immediately loses sensitivity at what appears
to be a linear rate with $p$.  This is a manifestation of the fact that the
sensitivity of the probe state is unaffected by a single error, and only in the
event that two simultaneous errors occur does the sensitivity decrease.

Finally, we consider for comparison the natural probe state given by an ensemble
of independent spins, each pointing in the $X$ direction  %
\begin{align}
  \lvert {+^n}\rangle  %
  = \lvert {+}\rangle \otimes\lvert {+}\rangle \otimes\cdots\otimes\lvert {+}\rangle \ ,
\end{align}
where $\lvert {+}\rangle  = [\lvert {\uparrow}\rangle +\lvert {\downarrow}\rangle ]/\sqrt2$ is the $+1$ eigenvector of
{$X$}.
Our upper bound computed for the spin-coherent state vector $\lvert {+^n}\rangle $ %
is plotted in blue in \cref{z:eNrtlC61cj3z}.  We can see
that this probe state performs significantly worse than the entangled probe
states for $p\lesssim 0.4$.  This is expected, since such a probe's noiseless
sensitivity scales only linearly in $n$, as opposed to the quadratic scaling of
the sensitivity of the $\lvert {\psi_{\textrm{f-af}}}\rangle $ and
$\lvert {\psi_{\textrm{code-f-af}}}\rangle $ probe states.  However, the robustness of the
spin-coherent state to the noise is significant.  At 
larger values of the
amplitude damping parameter ($p\sim 0.5$ for $n=12$), the other probe states
have all but lost their advantage in sensitivity.

For our numerical calculations, we employed the standard Python \emph{NumPy} and
\emph{SciPy} toolboxes along with \emph{QuTip}~\cite{R67,R68}. 
Our source code is published on \emph{Github}~\cite{R69}.
To compute the trace terms in~\eqref{z:3YnRZBPm.g17} we
express $\lvert {\psi}\rangle $ and $\bar{H}\lvert {\psi}\rangle $ as superpositions of a small number of
computational basis vectors over the $n$ sites.  The traces then factorize into
tensor factors enabling their efficient computation.  For the spin-coherent
state we work with the local $X$ basis instead of the $Z$ basis, such that the
spin-coherent state becomes a basis state in this picture.  The direct
computation of the Fisher information is %
performed via an eigenvalue
decomposition of the full $n$-body noisy probe state $\rho_B$ to transform the
anticommutator equation $\frac12\bigl\{\rho_B, R\bigr\} = \mathcal{N}(-i[H,\psi])$ in a
basis where $\rho_B$ is diagonal, and then solving elementwise to determine
$R$.
The ad hoc lower bound is computed by numerically solving the symmetric
logarithmic derivative in a restricted subspace consisting of the computational
basis vectors that appear in the decomposition of the probe state and those
bitstrings that are close by in Hamming distance.  The resulting value is
guaranteed to be a lower bound, because the map that projects the state down to
any subspace of the state space is a trace-nonincreasing, completely positive
map for which one can apply the data-processing inequality satisfied by the Fisher
information~\cite{R18} (see \cref{z:x2KvR6Lo6ilw}
in \cref{z:KH4B5FtzQ1kE} for details).
It is likely the quantum Fisher information in this setting can also be computed
based on existing techniques, such as those introduced in
Refs.~\cite{R20,R30,R52,R70}.

\section{Conclusions and outlook}
\label{z:HssWipSG8E2-}

Our results present a new paradigm for characterizing the sensitivity of a quantum
clock or sensor when exposed to noise, by establishing a quantitative trade-off
between the quantum Fisher information of the noisy system with respect to the
parameter of interest and the quantum Fisher information that the environment
acquires with respect to a complementary parameter.
Information trade-offs are interesting because they reveal properties of the
mathematical structure of quantum theory, which in turn determine what tasks can
be accomplished within the laws of quantum mechanics.  
Here, our results provide a guiding principle for finding noise-resilient clock
states: In order to avoid sensitivity loss due to the application of a noise
channel, clock states should hide their energy from the environment.

Energy-time uncertainty relations have historically been harder to %
formulate than position-momentum-type uncertainty principles, because there is
no global time observable in quantum mechanics in the same sense as there is a
position observable. Our work contributes an additional type of time-energy
uncertainty relation, complementing existing uncertainty relations such as
Mandelstamm-Tamm-type uncertainty
relations~\cite{R12}, Fisher-based uncertainty
relations with a single system~\cite{R11}, and
entropic uncertainty relations~\cite{R14}. 
Our relation exploits a type of
complementarity between the local optimal sensing operator for time and the
Hamiltonian (\cref{z:KI2eOYZK2.Ln}),
{in the same spirit as uncertainty relations}
derived in Refs.~\cite{R0,R11}. 
Our relation furthermore connects the estimation capabilities of two distinct
parties (Bob and Eve); in this sense our results can be seen as a Fisher
information counterpart of the entropic uncertainty relations for time and
energy~\cite{R14}.

\subsection{Summary and discussion}

An overview of the results presented in this work can be found in
\cref{z:a4FfzgZJotxe}.

\paragraph{Time-energy sensitivity trade-off.}
Our main result is a quantitative time-energy sensitivity trade-off relation in
the setting of \cref{z:FnHui0ahNLxW}.  If a quantum system is
subjected to an instantaneous noisy channel, then the loss in sensitivity to
time trades off exactly with the environment's ability to sense the energy of
the system as laid out in \cref{z:OUF5HnwJdMZI}.

The setting of our uncertainty relation (\cref{z:FnHui0ahNLxW}) is
unconventional for quantum metrology: A quantum clock usually accumulates noise
continuously as time evolves, much like a quantum probe usually accumulates
noise continuously while sensing an unknown parameter.  Our setting is instead
the communication scenario studied in Ref.~\cite{R28}: Alice
possesses a noiseless quantum clock that already encodes some time value, and
she sends it to Bob over a noisy communication channel.
In this alternative setting one can analyze the quantum information that leaks
to the environment, which is more challenging to do if we consider continuous
noise.

An appealing feature of our trade-off relation is that Bob's time sensitivity and
Eve's sensitivity to energy are related by an equality.  Concretely, this
feature means that not only does a gain in energy sensitivity imply a time
sensitivity loss by Bob, but also a loss in time sensitivity for Bob
automatically implies a gain in energy sensitivity by Eve.  In contrast,
uncertainty relations in quantum mechanics often relate two observable
uncertainties or two entropic quantities via an inequality.  For instance, a
Schr\"odinger particle in one dimension that has a large variance in the
momentum observable need not have a narrow variance in the position observable.

Our results furthermore hold for an arbitrary pure probe state vector $\lvert {\psi}\rangle $ and
Hamiltonian $H$.  We evade the question of formally optimizing over the probe
state vector $\lvert {\psi}\rangle $ itself---a central question in quantum metrology that many
contributions on using quantum error correction for metrology
address~\cite{R20,R21,R51,R71,R72}---by identifying
instead what features a probe state vector $\lvert {\psi}\rangle $ must exhibit to avoid being
affected by the noise.  %
The alternative expression of the Fisher information obtained by our
trade-off relation can potentially facilitate the computation of the Fisher
information when optimizing the clock or probe state, potentially improving
state optimization schemes such as those in
Refs.~\cite{R66,R73} in the
presence of noise.
Our trade-off relation also offers a guideline to seek good clock states,
especially in settings where it might not be possible to reliably prepare the
probe state that has the absolute best sensitivity: Noise-resilient clock states
need to hide their energy from the environment.

Importantly, it is not necessarily the probe state which is least affected by
the noise that is the most sensitive.  Another state might exhibit a better
sensitivity after the noisy channel, even if its sensitivity loss is greater, by
ensuring that it is initially sufficiently more sensitive.  As an extreme case,
this point is illustrated by the ground state of a qubit affected by
amplitude-damping noise; the ground state trivially remains unaffected by the
noise but has no sensitivity, whereas the $+X$ eigenstate has a better
sensitivity, even if it is affected by the noise.

A key technique in our approach is the formulation of the quantum Fisher
information as a semidefinite
program~\cite{R25,R27}
(see \cref{z:KH4B5FtzQ1kE}).
Semidefinite programming offers a versatile toolbox
in which an alternate expression for an optimization (known as \emph{dual
problem}) can be derived and bounds on such optimizations can be proven
more easily~\cite{R74,R75}.
The technical proof of our main trade-off result
(\cref{z:VM5-jRd.0oMe}) offers additional insight
into the meaning of the dual problem associated with the semidefinite
programming formulation of the quantum Fisher information.

\paragraph{The setting of continuous noise.}

In certain specific settings, the setup in \cref{z:FnHui0ahNLxW}
remains a good approximation of a quantum clock exposed to continuous noise
described by a Lindbladian master equation (see
\cref{z:Xuq8JaOEWrUo}).  A sufficient condition that guarantees the
accuracy of this approximation is to ensure on one hand that the Hamiltonian
part $\mathcal{L}_0$ of the evolution commutes (as a superoperator) with the
noise part $\mathcal{L}_1$ of the Lindbladian that contains all the noise
operators, and on the other hand that the time derivative of the state is
primarily driven by the Hamiltonian and not by the noise.  More precisely, in
the notation of \cref{z:Xuq8JaOEWrUo}, the sufficient condition
consists in checking that $[\mathcal{L}_0, \mathcal{L}_1]=0$ as well as ensuring
that $\partial_t\mathcal{N}$ contributes only a negligible part of the Fisher
information $\Ftwo{\rho}{\partial_t \rho}$ [for which a rigorous bound can for
instance be computed in \cref{z:raWZngbty0t4}].  The second
condition is rarely expected to be violated, as sensors are typically designed
to have their signal imprinted on their state through their Hamiltonian
evolution; noise is usually a degrading process and is typically not the
mechanism by which the signal is acquired.
If the Hamiltonian of a many-body system consists only of single-site $Z$ terms,
then both i.i.d.\@ dephasing noise and i.i.d.\@ amplitude-damping noise commute
(as a superoperator) with the Hamiltonian part of the Lindbladian.  Furthermore
if the Hamiltonian $H$ commutes with the individual Lindblad jump operators, then
the corresponding evolutions also commute as superoperators; this is the case
for instance if $H$ consists of arbitrary-weight terms containing only $Z$
operators and in the presence of i.i.d.\@ dephasing noise.
In the case where $[\mathcal{L}_0,\mathcal{L}_1]\neq 0$ the setting can still
formally be mapped onto the setting of \cref{z:FnHui0ahNLxW}, by
defining the effective noise as the full evolution map with a unitary applied on
the input, as long as the time dependence of the effective noisy channel can be
neglected.  In this case, determining the effective noisy channel in general
might be difficult.

\paragraph{Trade-off with generalized parameters.}
The trade-off relation for time and energy can be extended to other parameter
evolutions.  First of all, there is a choice in how $\psi$ evolves along the $t$
and $\eta$ parameters: Any choice of $\lvert {\psi(t,\eta)}\rangle $ such that
\cref{z:ff7jHgn75DW4} is satisfied at $(t_0,\eta_0)$ (but not
necessarily at other even neighboring points) leads to the same Fisher
information quantities $\FIqty{Alice}{t}$, $\FIqty{Alice}{\eta}$,
$\FIqty{Bob}{t}$, and $\FIqty{Eve}{\eta}$, so our trade-off relation directly
applies.
An alternative choice for the $\eta$ parameter is an evolution generated by the
Lindbladian master equation $\partial_\eta \psi = \mathcal{L}[\psi]$ with
$\mathcal{L}[\rho] = \sum_k \bigl[ L_k\rho L_k^\dagger - \{L_k^\dagger L_k,
  \rho\}/2 \bigr]$, where $L_k = \sigma_H^{-1}\sqrt{(e_k + c)} \lvert {\psi}\rangle \mkern -1.8mu\relax \langle{e_k}\rvert $,
where $\{\lvert {e_k}\rangle \}$ are eigenvectors of the Hamiltonian, where
$H = \sum e_k \lvert {e_k}\rangle \mkern -1.8mu\relax \langle{e_k}\rvert $, and where $c\geq 0$ is chosen large enough such that
$e_k+c \geq 0$ for all $k$.  (We have the opposite sign for
$\partial_\eta \psi$, but this can be corrected by redefining
$\eta\mapsto-\eta$, and this does not impact the Fisher information.)  We can
check that this choice of $\partial_\eta \psi$ satisfies
\cref{z:SMjFM6HgZlts} at $(t_0,\eta_0)$, and therefore
also \cref{z:ff7jHgn75DW4}.
Another interesting
choice for $\psi(t,\eta)$ is to set
$\lvert {\psi(t_0,\eta)}\rangle  = {e}^{(\eta-\eta_0) \bar{H} /(2\sigma_H^2)}
\lvert {\psi(t_0,\eta_0)}\rangle $ for $\eta$ in a neighborhood of $\eta_0$, recalling
$\bar{H} = H - \langle {H}\rangle _\psi$.  Again, we see that
\cref{z:SMjFM6HgZlts} is satisfied.  This evolution is
nonunitary, but one can check that it does preserve the trace of $\psi$ locally
to first order at $\eta_0$: We have
$\partial_\eta \operatorname{tr}(\psi) \bigr|_{\eta_0} = \operatorname{tr}\bigl( \{ \bar{H}, \psi \}\bigr) = 0$.
Either of these choices of evolution might be relevant depending on the specific
application, though we expect the primary application of our trade-off relation
is to help characterize Bob's Fisher information to time, in which case the
specific choice of how the clock state is stated to evolve along $\eta$ might
not be important.

The trade-off relation can further be extended to an inequality that is valid for
any two arbitrary parameters
(\cref{z:HCzmF57VVd0m}).  The trade-off between
the Fisher information that Alice and Bob respectively have with respect to
either parameter is then quantified by a value that depends on the commutator of
the generators of the two parameters
(\cref{z:incm1nGt-mIE}).  The
appearance of the commutator in this expression reinforces its central role in
quantifying the incompatibility of physical observable quantities.  
Our main time-energy trade-off relation can be %
recovered from the more
general \cref{z:incm1nGt-mIE} by
plugging in the local generators for time and energy.
While \cref{z:incm1nGt-mIE} appears
to be tight whenever the Robertson-Weyl uncertainty
relation~\eqref{z:0pRecQmVyaWF} is saturated for the two
generators, the bound can likely be improved when considering two generators
that have a small commutator.
We also present a sufficient condition under which a trade-off relation for any
two parameters can be obtained in the form of an equality, mirroring the
equality statement in our main trade-off relation for the time and energy
parameters.  One might have thought that equality in our general uncertainty
relation would happen only if the parameters are complementary in the sense of
\cref{z:ip.xvGb3bzRb,z:KI2eOYZK2.Ln}; in fact, it suffices that the parameters obey
some suitable complementarity relation on the support of the complementary
channel.  Therefore, equality in our general uncertainty relation does not
simply depend on the structure of the parameters $a,b$, but also on the noisy
channel $\mathcal{N}$.

\paragraph{Bounds on the quantum Fisher information.}
Computing the quantum Fisher information for general states involves the
calculation of the symmetric logarithmic derivative
in~\eqref{z:.eb-akuquBNd}.  This object is straightforward to determine
for pure states, it is simple when represented in the diagonal basis of the
state, and it can be computed using numerical methods such as the
Bartels-Stewart algorithm~\cite{R76}.
However, in the absence of a simple diagonal
representation of the state, it is in general difficult to characterize
analytically the Fisher information or to derive useful bounds on the Fisher
information that apply in general settings of mixed states, especially if the
state is rank deficient or close to the boundary of state space.  Our results
provide an alternative expression for the Fisher information in the scenario of
\cref{z:FnHui0ahNLxW}.  Combined with the powerful semidefinite
methods for the Fisher information reviewed in \cref{z:KH4B5FtzQ1kE},
we provide a general toolbox to characterize the Fisher information for mixed
states in a variety of situations.
For instance, for an interacting many-body system subject to noise that acts
locally, the noise process might be well approximated by an environment
that is small relative to the full many-body system.  In this case, the computation of
the Fisher information on Eve's end happens on a smaller-dimensional system.
This observation is for instance a main component of our
bound~\eqref{z:3YnRZBPm.g17}.

By applying known Fisher information bounds on Eve's system, our trade-off
relation enables us to straightforwardly obtain an opposite bound for the Fisher
information of Bob's noisy state (see \cref{z:qdkRM6aFmtdz}).  (Upper
bounds on the Fisher information can be difficult to obtain; see for instance
Refs.~\cite{R77,R30}.) An example of such a bound to apply
is the data-processing inequality for the Fisher
information~\cite{R18}: Further processing of a state that has
been exposed to the unknown parameter can only decrease the sensitivity with
respect to that parameter.  This procedure is useful when Eve obtains a state
that is not diagonal in the computational basis, making the Fisher information
harder to compute.  In such cases, we can dephase Eve's state to set all the
off-diagonal matrix elements to zero.  The resulting Fisher information for Eve
can only decrease; by our trade-off relation this immediately yields an upper
bound on Bob's Fisher information.  This bound for instance facilitates the
computation of the sensitivity loss of a state exposed to weak amplitude-damping
noise, as discussed in \cref{z:Pgp-gW09n9nZ}.

\paragraph{Metrological codes.}
Our main uncertainty relation leads to necessary and sufficient conditions for
when a clock state loses zero sensitivity when a given noisy channel is applied
(\cref{z:JZtfKZv25too}).  These conditions are a weaker
version of the Knill-Laflamme conditions for quantum error correction. Given a
clock state vector $\lvert {\psi}\rangle $ and a Hamiltonian $H$, we can consider the virtual qubit
$L$ spanned by the vectors $\lvert {\psi}\rangle $ and $H\lvert {\psi}\rangle $.  The clock state vector
$\lvert {\psi}\rangle $ loses no sensitivity under the application of a noisy channel with
Kraus operators $\{E_k\}$ if and only if all operators of the form
$E_{k'}^\dagger E_k$, when projected onto the virtual qubit, do not have any
overlap with the Pauli-$Z$ operator on the virtual qubit.
It would be in principle possible to prove these zero sensitivity-loss
conditions directly on Alice's and Bob's systems, without invoking our trade-off
relation; however, characterizing when Eve's Fisher information is zero provides
an immediate proof whose simplicity we have not been able to match with 
alternative techniques. %

The zero sensitivity-loss conditions~\eqref{z:zFMncm5bkXDt} bear similarities with classical codes,
where there is only a commutative algebra of observables that one wishes to
reproduce~\cite{R49,R47}.  Intuitively,
the conditions simply ensure that there is a measurement on Bob's system that
will reveal the time parameter as well as the local time-sensing observable on
Alice's system.  In contrast to fully quantum error correction, however, there
is in general no recovery operation that will restore the pure clock state
accurately to first order in the parameter (see
\cref{z:kLDpse.aVOd5} for a simple
counterexample).
An intriguing aspect of the zero sensitivity-loss conditions are that they do
not appear to be formally equivalent to quantum error correction with respect to
specific set of noise operators. (The results of~\cite{R21,R51} appear to indicate
that it might be possible to implement certain metrological codes as an error-correcting code
involving ancillary systems.)  In some cases, such as
the qubit example of
\cref{z:kLDpse.aVOd5}, the clock
state can be thought of as an error-correcting code that corrects only a
certain type of error ($X$ or $Y$ Pauli errors).  But this is not generally the
case---there are examples of a clock state and a (highly nonlocal) Hamiltonian
that fulfill the metrological code condition for low-weight errors, but that are
not quantum error-correcting codes with respect to neither low-weight $X$ errors
nor low-weight $Z$ errors
(\cref{z:VsSVmS3N9rs6}).

The conditions for zero sensitivity loss are closely related to the recent
series of works detailing how to use quantum error correction for metrology in
the presence of noise~\cite{R20,R21,R51,R71,R72,R78,R79}.
The main difference with our results is the setting that is being considered.
We ask which initial clock states one can prepare on the clock system such that
no sensitivity is lost when a noisy channel is applied (and what the associated
optimal sensing measurement after the application of the noisy channel is),
whereas the mentioned references consider the setting where, during the time a
probe system is exposed to the signal and the continuous noise, one can control
the probe~\cite{R80} to repeatedly apply the recovery
procedure associated with the quantum error-correcting code.

Metrological codes might be useful for ancillary measurements of error syndromes~\cite{R81}.
Consider an ancillary qudit (of dimension greater than two) which extracts a bit-valued error syndrome via an entangling gate, correlating a pair of states \(\lvert {\psi}\rangle ,\lvert {\xi}\rangle \) with respective binary syndrome values \(0,1\). If the ancillary subspace participating in syndrome extraction satisfies the zero sensitivity-loss conditions~\eqref{z:zFMncm5bkXDt} against physical noise, then, by definition, a measurement in the \(\{\lvert {\psi}\rangle ,\lvert {\xi}\rangle \}\) basis will not be affected by such noise. In other words, the loss conditions ensure protection against \(X\)-type logical noise, yielding more robust syndrome extraction using a \(Z\)-type measurement. However, such conditions do not preclude any \(Y\)-type logical noise. They also do not guarantee fault tolerance, which would require that ancilla errors not spread to any logical encoding via backaction.

\paragraph{Numerics for many-body systems.}
Characterizing the quantum Fisher information of a state exposed to a noise
channel using the bound presented in
\cref{z:308j2vzIlz3f} is convenient in
the setting of a many-body system subject to noise that acts locally.  In the
case of $n$ qubits prepared in a permutation-invariant state and exposed to an
i.i.d.\@ amplitude-damping noise channel, we empirically find that the
bound~\eqref{z:3YnRZBPm.g17} with $k=n$ appears reasonably tight for
the states that we investigated and for small values of the noise parameter $p$;
furthermore, for a selection of states including the GHZ state, the bound
appears to remain tight even in the regime of high values of $p$.
Our bound can be computed for systems of size $n\gtrsim 50$ on a standard
desktop computer.

If instead of on-site terms we consider only Ising-type nearest-neighbor
interactions, we can study the robustness of the example `metrological code'
introduced in \cref{z:dqhjUefIA37Z} to an i.i.d.\@
amplitude-damping channel with local noise parameter $p$.  This state %
retains its sensitivity after a single located error.  Our numerics show that in the
presence of i.i.d.\@ noise, the decrease in the quantum Fisher information
scales only as $p^2$, and not linearly in $p$ as for the other studied states
with similar sensitivity.
We observe that if we expose the interacting system to continuous
amplitude-damping noise, then the noise part and the unitary part of the
Lindblad evolution do not commute as superoperators  (i.e., the setting
is not that of phase-covariant noise); it is then possible that the
advantages of the metrological code state might not persist in the setting of
continuous noise.

\subsection{Outlook}

Our trade-off relation is perhaps most relevant in an intermediate regime where
the clock is exposed to a signal without the possibility for intermittent
quantum control.  In such cases, the noise is expected to spoil any Heisenberg
scaling that could be achieved using quantum error-correcting schemes due to the
lack of recovery operations during the evolution (see,
e.g., Refs.~\cite{R30,R21}).
Provided the setting can be modeled with a single noisy channel, our results
present an alternative expression for the sensitivity of the noisy probe in this
regime where the sensitivity is not yet dominated by the asymptotic scaling.
Our results might therefore help identify which states present sufficient
robustness to the noise to present an advantage in sensitivity with respect to
commonly used states (such as a GHZ state or a spin-coherent state).

In the situation where the clock evolves according to a Lindbladian master
  equation, and the Hamiltonian and noise parts of the Lindbladian fail to
  commute as superoperators, one might expect in certain cases to still
be able to consider time-dependent noise
using the following trick.  
Let us identify the system $A$ as a full copy of the bipartite system
$B\otimes E$, and let the unitary evolution of $\lvert {\psi(t)}\rangle $ cover both
systems.  A time-dependent channel $\mathcal{E}_t(\cdot)$ can be written as
$\mathcal{E}_t(\cdot) = \operatorname{tr}_E\bigl[ U(t) (\cdot) U^\dagger(t) \bigr]$ where all the
time dependence is encoded in the unitary $U(t)$ and where the environment
system $E$ is chosen suitably.  We then select the noisy channel
$\mathcal{N}_{A\to B} = \operatorname{tr}_E$ that simply performs the partial trace over $E$;
the complementary channel is correspondingly
$\widehat{\mathcal{N}}_{A\to E} = \operatorname{tr}_B$.  While writing a Markovian master
equation in this form might require a huge environment system $E$ with rapidly
mixing internal dynamics, we expect that our formalism can still account for
simple time dependence in the noisy channel in this way.  Note also that the
unitary $U(t)$ only has to approximate on $B$ the noisy channel $\mathcal{E}_t$
locally to first order around a fixed value of the parameter (e.g., $t=0$) in
order to determine the Fisher information.

A potential domain of application of our main uncertainty relation is for
\emph{quantum thermometry}~\cite{R82}, where the goal
is to estimate the temperature of a quantum system.  In the simple setting of
quantum thermometry where the temperature of the system is known to some
approximation, and a measurement is performed in order to refine that knowledge,
the optimal measurement to carry out is an energy
measurement~\cite{R82}.  Since our main
result~\eqref{z:OUF5HnwJdMZI} involves the sensitivity of a party with
respect to a parameter representing the energy, which is optimally measured
using the Hamiltonian of the noiseless system, we expect that one can leverage
our main results to yield new sensitivity bounds for quantum thermometry.

Our results are also likely to be relevant in situations where only a restricted
set of operators can be measured on a system.  Such a restriction could be
imposed by limitations in control for a given experimental platform.  Suppose we
prepare a clock state vector $\lvert {\psi}\rangle $ evolving noiselessly according to a Hamiltonian
$H$.  We would like to measure the clock at time $t\approx t_0$, but we are only
permitted to use a measurement from a given set of measurements.  What is the
optimal local sensitivity that we can achieve?  Should the set of allowed
measurement operators form an algebra, then the problem is equivalent to sending
the clock through a channel that represents the projection onto that algebra.
Our results then imply that the resulting sensitivity trades off exactly with
the sensitivity that one can achieve with the set of measurements in the
commutant of that algebra, with respect to the complementary parameter $\eta$.

It might be possible to extend our results to the multiparameter metrology
regime where more than one parameter is estimated by Bob.  There are known
uncertainty relations that determine trade-offs between the precision to which
individual parameters can be simultaneously estimated by a single
party~\cite{R83,R84,R71,R85,R86,R17}.
In fact, the $t$ and $\eta$ parameters form a so-called \emph{D-invariant
  model}~\cite{R87,R88,R85}, the latter referring to a multiparameter
quantum statistical model in which the tangent space is invariant under taking
symmetric logarithmic derivatives of the possible generated state-evolution
directions.  D-invariant models are interesting in multiparameter quantum
metrology, because different sensitivity bounds, which in general are difficult
to relate, can be shown to coincide~\cite{R85}.
It seems plausible that known multiparameter uncertainty relations can be
extended to the present bipartite setting, either where all parameters are
simultaneously estimated by Bob while Eve simultaneously estimates a set of
complementary parameters, or where a number of parties estimate each individual
parameter, where each party might be part of the output or the environment.

Our main uncertainty relation might offer a connection between the setting of
quantum metrology hindered by a noisy quantum channel and the setting of
multiparameter, noiseless quantum metrology.
It appears that a key ingredient for our main uncertainty relation is that the
$t$ and $\eta$ parameters form a D-invariant model, in the sense of the
preceding paragraph.  By construction, our uncertainty relation applies to any
pure state D-invariant model consisting of two complementary generators related
by \cref{z:w7NUzrfGUNE-}, given that we made no specific
assumptions about the parameter $t$ or its local Hermitian generator $H$.

However, it remains unclear whether our results extend to general multiparameter
D-invariant models, as our proof seems to utilize the fact that the space is
spanned by only two complementary generators $T$ and $H$.
D-invariant models have a rich geometric
structure~\cite{R87} which might prove an essential
conceptual component of our results; such connections nevertheless remain to be
better understood.
A further connection to D-invariant models appears in
\cref{z:oWCGAyhkI-Md},
which appears to be a D-invariance condition restricted onto the support
of the complementary channel $\widehat{\mathcal{N}}$.
In fact, one can view Bob's and Eve's measurements $T_B$ and $E$ as measurement
operators $\mathcal{N}^\dagger(T_B)$ and $\widehat{\mathcal{N}}^\dagger(E)$ on
Alice's system through the action of the adjoint channels $\mathcal{N}^\dagger$
and $\widehat{\mathcal{N}}^\dagger$.  We could ask whether our uncertainty
relation translates into a trade-off in how the two parameters $t$ and $\eta$
can be estimated by Alice, if the estimation of $t$ (respectively $\eta$) is
required to employ an observable in the set of operators that is specified as
the image of $\mathcal{N}^\dagger$ (respectively of
$\widehat{\mathcal{N}}^\dagger$).  It is not clear if this is the case, as the
quantum Fisher information attained by the observable $\mathcal{N}^\dagger(T_B)$
(respectively $\widehat{\mathcal{N}}^\dagger(E)$) on the state $\psi$ is not
necessarily expected to match the corresponding value of the quantum Fisher
information of $T_B$ on $\mathcal{N}[\psi]$ (respectively of $E$ on
$\widehat{\mathcal{N}}[\psi]$).
It is thus unclear if or how our main uncertainty relation is connected with
general bounds that hold in the multiparameter regime, such as multiparameter
versions of the quantum Cram\'er-Rao
bound~\cite{R85} or the Gill-Massar
inequality~\cite{R89,R90}.
We might expect that deeper connections can be developed between the setting of
parameter estimation after the application of a noisy channel and noiseless
multiparameter estimation.

Also, our results apply locally to first order around a given fixed value of the
unknown parameter; whether similar results can be derived in the global
parameter estimation regime~\cite{R91,R88,R92,R93,R94}
is unknown.  Global parameter estimation might be more relevant for applications
to atomic quantum clocks~\cite{R95,R96}.
We also anticipate extensions of our results to the finite-sample regime where
the quantum Fisher information might no longer accurately quantify the
sensitivity of a quantum state to an unknown
parameter~\cite{R97}.

Along a similar vein, there are settings where one seeks to compute
different variants of the quantum Fisher information.  For instance, the
so-called \emph{right-logarithmic derivative} (see,
e.g., Ref.~\cite{R86}) is often used to bound the standard quantum
Fisher information.  Alternative sensitivity measures include the truncated
Fisher information~\cite{R53}, which not only give useful
bounds on the standard quantum Fisher information but can be more relevant in
the regime of limited measurement data.  An interesting question would be to
study whether our results extend to such generalized sensitivity
measures.

Entropic uncertainty relations play a central role in quantum
cryptography~\cite{R10,R14,R15,R98}, and cryptographic schemes have been studied for
quantum metrology~\cite{R99}.  It is possible that
our parameter-estimation trade-off can similarly form the basis of cryptographic
schemes in which a parameter encoded in a quantum state is to be shielded from a
malevolent eavesdropper.
Furthermore, the Fisher information is closely related to relative entropy
measures~\cite{R100,R86,R101}; we might expect our trade-off
relation to translate into a statement about R\'enyi relative entropies.

The development of quantum atomic clocks as ultraprecise time
references~\cite{R102} makes it all the more important to
achieve a thorough understanding of how noise can be prevented from spoiling
sensitivity.
We also anticipate that our results will be relevant for recently developed
atomic clocks built with a lattice of interacting
atoms~\cite{R8} and correlated many-body sensing
probes~\cite{R103,R5}, as these platforms
will offer new possibilities for metrology by exploiting the strong
interactions between the particles.

\begin{acknowledgments}
The authors are grateful to
Fernando Brand\~ao,
Jonathan Conrad,
Rafa\l{} Demkowicz-Dobrza\'nski,
Richard K\"ung,
Johannes Meyer,
Yingkai Ouyang,
Renato Renner,
Ralph Silva,
Ryan Sweke,
and
Nathan Walk
for discussions.
We warmly thank Gian Michele Graf for his invaluable input for our proofs in infinite-dimensional spaces.
M.~W.\@ acknowledges support from the Swiss 
National Science Foundation (SNSF) via
an Ambizione Fellowship (PZ00P2\_179914).
M.~W.~and J.~M.~R.\@ acknowledge the National Centre of
Competence in Research QSIT. 
Ph.~F.\@ and J.~E.\@ acknowledge support from the DFG
(FOR 2724, CRC 183, EI 519/21-1), the FQXi, 
the QuantERA (HQCC),
the BMBF (RealistiQ, Hybrid, MuniQC-Atoms) 
and the Einstein Research Unit on quantum devices.  
This research is also part of the Munich Quantum Valley (K8), which is supported by the Bavarian state government with funds from the Hightech Agenda Bayern Plus. 
V.~V.~A.\@ acknowledges funding from NSF QLCI award No.\@ OMA-2120757. Contributions to this work by NIST, an agency  of  the  US  government, are  not  subject to US copyright. Any mention of commercial products  does  not  indicate  endorsement by NIST. 
J.~P.\@ acknowledges funding from the U.S.\@ Department of Energy Office of Science  (DE-NA0003525, DE-SC0020290, DE-ACO2-07CH11359, DE-SC0018407), the Simons Foundation It from Qubit Collaboration, the Air Force Office of Scientific Research (FA9550-19-1-0360), and the National Science Foundation (PHY-1733907). The Institute for Quantum Information and Matter is an NSF Physics Frontiers Center.
\end{acknowledgments}

\appendix

\section*{APPENDIX}

We first introduce some preliminaries and notation that will be used throughout
the Appendices.  All Hilbert spaces are finite dimensional unless otherwise
indicated, and all projectors are Hermitian.  A pure quantum state is a vector
$\lvert {\psi}\rangle $ in the Hilbert space that is normalized to unit norm, and the
terminology pure quantum state is also used by extension for the associated
density operator $\lvert {\psi}\rangle \mkern -1.8mu\relax \langle{\psi}\rvert $.  Quantum states are positive semidefinite
operators $\rho$ with unit trace, $\operatorname{tr}(\rho)=1$.  A subnormalized quantum state
is a positive semidefinite operator $\rho$ satisfying $\operatorname{tr}(\rho) \leq 1$.
States are normalized to unit trace unless explicitly specified as being
subnormalized.

For any Hermitian operator ${O} $, we denote by $P_{{O} }$ the projector
onto the support of ${O} $, and by $P_{{O} }^\perp = \mathds{1}- P_{{O}}$
its complement.  For any positive semidefinite operator $A$, we denote by
$A^{-1}$ its Moore-Penrose pseudoinverse, i.e., the operator obtained by taking
the inverse on the support of $A$.

We denote by $\opnorm{A}$ the maximal singular value of an operator $A$.  We
also define the Schatten one-norm as $\onenorm{A} = \operatorname{tr}\sqrt{ A^\dagger A}$.

It will prove convenient to ``vectorize'' operators by viewing them as vectors
in Hilbert-Schmidt space using the following representation.  The vector space
of operators acting on a Hilbert space $\mathscr{H}$ is isomorphic to $\mathscr{H}\otimes\mathscr{H}$.
Let $\lvert {1}\rrangle $ denote the element
\begin{align}
    \sum_{i=1}^d \lvert {i}\rangle \otimes\lvert {i}\rangle  \in \mathscr{H}\otimes\mathscr{H}, 
\end{align}
where
$\{\lvert {i}\rangle \}_{i=1}^d$ is a fixed basis of $\mathscr{H}$.  We define the ``vectorized''
representation of any operator $A$ acting on $\mathscr{H}$ as
$\lvert {A}\rrangle  = (A\otimes\mathds{1})\lvert {1}\rrangle $.  Similarly, we define
$\llangle {1}\rvert  = \sum_{i=1}^d \langle {i}\rvert \otimes\langle {i}\rvert $ and
$\llangle {A}\rvert  = \llangle {1}\rvert (A^\dagger\otimes\mathds{1})$.  We recall the useful identity
\begin{align}
(X\otimes \mathds{1})\lvert {1}\rrangle  = (\mathds{1}\otimes X^T)\lvert {1}\rrangle ,
\end{align}
and note that
$\lvert {1}\rrangle =\lvert {\mathds{1}}\rrangle $ is the vectorized operator representation of the
identity matrix $\mathds{1}$.  We denote a rank-one operator $\lvert {\phi}\rangle \mkern -1.8mu\relax \langle{\psi}\rvert $
in this representation as $\lvert {\phi,\psi}\rrangle $, with
$\lvert {\phi,\psi}\rrangle  = \lvert {\phi}\rangle \otimes(\lvert {\psi}\rangle )^*$ and
$\llangle {\phi,\psi}\rvert  = \langle {\phi}\rvert \otimes(\langle {\psi}\rvert )^*$.  The Hilbert-Schmidt inner
product in this notation is simply $\operatorname{tr}(A^\dagger B) = \llangle {A}\mkern 1.5mu\relax \vert \mkern 1.5mu\relax {B}\rrangle $.  The
matrix elements of $A$ in any basis $\{\lvert {\ell}\rangle \}$ are also simply given by
$\langle {\ell}\mkern 1.5mu\relax \vert \mkern 1.5mu\relax {A}\mkern 1.5mu\relax \vert \mkern 1.5mu\relax {\ell'}\rangle  = \llangle {\ell,\ell'}\mkern 1.5mu\relax \vert \mkern 1.5mu\relax {A}\rrangle $.  A superoperator
$\mathcal{E}$ acting on an operator 
${M}$ is denoted by $\mathcal{E} \lvert {{M}}\rrangle $.
The superoperator consisting of a left multiplication by $A$ and a
right multiplication by $B$, 
i.e., ${M}
\mapsto A{M}B$, is represented by
$\lvert {{M}}\rrangle  \mapsto (A\otimes B^T)\lvert {{M}}\rrangle $.  The identity superoperator
${\mathrm{id}}$ is represented by $\mathds{1}\otimes\mathds{1}$.  Also,
$\llangle {A}\mkern 1.5mu\relax \vert \mkern 1.5mu\relax {\mathcal{E}}\mkern 1.5mu\relax \vert \mkern 1.5mu\relax {B}\rrangle  = \llangle {B}\mkern 1.5mu\relax \vert \mkern 1.5mu\relax {\mathcal{E}^\dagger}\mkern 1.5mu\relax \vert \mkern 1.5mu\relax {A}\rrangle ^*$, where
$\mathcal{E}^\dagger$ is the usual superoperator adjoint defined by
$\operatorname{tr}({M}\mathcal{E}^\dagger({N})) = \operatorname{tr}(\mathcal{E}({M})\,{N})$.
In the following and unless otherwise stated, superoperators are expressed in
this representation, unless they are explicitly applied onto an operator with
the notation $\mathcal{E}(\cdot)$.

We now compute a few quantities that often recur throughout these appendices.
Let $\lvert {\psi}\rangle $ be a state vector and consider the evolution
$\partial_t\psi = -i[H,\psi]$ where $H$ is any Hermitian operator.  Let ${M}$ be any
Hermitian operator.  We have
\begin{align}
  \Bigl \langle { \Bigl(\frac{d\psi}{dt}\Bigr)^2 }\Bigr \rangle  =
  \bigl \langle { (-i[H,\psi])^2 }\bigr \rangle 
  &= -\operatorname{tr}\bigl[\psi (H\psi H\psi - H\psi H - \psi H^2 \psi + \psi
    H \psi H)\bigr]
    \nonumber\\
  &= \operatorname{tr}(\psi H^2) - [\operatorname{tr}(\psi H)]^2\ ;
    \label{z:XrleY6eW0oFd}
    \\[1ex]
  -i[ i[{M},\psi], \psi]
  &= ({M}\psi-\psi {M})\psi - \psi({M}\psi - \psi {M}) = \{{M}, \psi\} - 2\langle {{M}}\rangle \,\psi
    \nonumber\\
  &= \{{M}-\langle {{M}}\rangle , \psi\}\ .
    \label{z:v.5MaZHkqw71}
\end{align}

\section{Auxiliary lemmas}
\label{z:2dfjPpYV0kEA}

The notion of Schur complement will serve multiple times in these appendices, so
we state it here.

\begin{theorem}[Positive semidefiniteness via Schur complement]
  \noproofref
  \label{z:pZMoi26flsRb}
  Let $A\in \mathbb{C}^{n\times n}$, $B\in\mathbb{C}^{m\times m}$ be positive
  semidefinite matrices.  Let $W\in\mathbb{C}^{n\times m}$ be an arbitrary
  complex matrix.  The following statements are equivalent:
  \begin{enumerate}[label=(\roman*)]
  \item \label{z:diSpiJAhHEfT}
$\displaystyle \begin{bmatrix} A & W \\ W^\dagger & B \end{bmatrix} \geq 0$ ,
  \item \label{z:qRGCRW7.q0UZ}
    $W P_B^\perp = 0$\quad and\quad $A\geq WB^{-1}W^\dagger$ ,
  \item \label{z:Jz89LjC6I9Xk}
    $P_A^\perp W = 0$\quad and\quad $B\geq W^\dagger A^{-1}W$ .
  \end{enumerate}
  Moreover, \ref{z:qRGCRW7.q0UZ} implies $P_A^\perp W = 0$ and
  \ref{z:Jz89LjC6I9Xk} implies $WP_B^\perp = 0$.
\end{theorem}
For a proof, see, e.g.,\@ Ref.~\cite{R104}.  With respect to the
proof of similar statements in standard textbooks, we can see that
\ref{z:qRGCRW7.q0UZ} implies $P_A^\perp W = 0$ as follows:
Hitting the inequality with $P_A^\perp(\cdot)P_A^\perp$ and noting that
$W B^{-1} W^\dagger\geq 0$, we see that
$P_A^\perp W B^{-1} W^\dagger P_A^\perp = 0$, which implies
$P_A^\perp W B^{-1/2} =0$.  Therefore, $P_A^\perp W = 0$ using the fact that
$W P_B^\perp = 0$.  Similarly, \ref{z:Jz89LjC6I9Xk} implies
$W P_B^\perp=0$.

Now we present a simple method to relate operator inequalities before and after
the application of a completely positive map.

\begin{lemma}[Positive semidefiniteness of block matrices under completely positive maps]
  \label{z:TV-YjOuDas4P}
  Let $A,B,W\in\mathbb{C}^{n\times n}$ be complex matrices and assume that
  \begin{align}
    \begin{bmatrix}
      A & W \\
      W^\dagger & B
    \end{bmatrix}
    \geq 0\ .
   \label{z:2QBAslrmHUQN}
  \end{align}
  Let $\Phi$ be any completely positive map that maps operators on
  $\mathbb{C}^{n}$ to operators on $\mathbb{C}^{m}$.  Then
  \begin{align}
    \begin{bmatrix}
      \Phi(A) & \Phi(W) \\
      \Phi(W^\dagger) & \Phi(B)
    \end{bmatrix}
    \geq 0\ .
   \label{z:QLgSmQAh5Nwn}
  \end{align}
\end{lemma}
\begin{proof}[**z:TV-YjOuDas4P]
  The matrix~\eqref{z:QLgSmQAh5Nwn} is obtained by applying the
  completely positive map ${\mathrm{id}}_{2}\otimes\Phi$
  onto~\eqref{z:2QBAslrmHUQN}.
\end{proof}
While the above lemma is fairly trivial, 
paired with \cref{z:pZMoi26flsRb} it
enables us to show less obvious inequalities such as the following.
\begin{corollary}[Image of matrix squared under a subunital, completely positive map]
  \label{z:2McXzF-3arvz}
  Let ${M}$ be a Hermitian operator and let $\Phi$ be any completely positive map
  that satisfies $\Phi(\mathds{1})\leq\mathds{1}$.  Then $ \Phi({M}^2) \geq [\Phi({M})]^2 $.
\end{corollary}
\begin{proof}[**z:2McXzF-3arvz]
  Observe first that $\begin{bsmallmatrix} {M}^2 & {M} \\ {M} & \mathds{1}\end{bsmallmatrix} =
  \begin{bsmallmatrix}{M} & \mathds{1}\end{bsmallmatrix}^\dagger \begin{bsmallmatrix}{M}
    & \mathds{1}\end{bsmallmatrix}\geq 0$.  With the subunitality condition on
  $\Phi$ and \cref{z:TV-YjOuDas4P}, we have
  \begin{align}
    \begin{bsmallmatrix} \Phi({M}^2) & \Phi({M}) \\ \Phi({M}) & \mathds{1}\end{bsmallmatrix}
    \geq
    \begin{bsmallmatrix} \Phi({M}^2) & \Phi({M}) \\ \Phi({M}) & \Phi(\mathds{1}) \end{bsmallmatrix}
    \geq 0\ .
  \end{align}
  Thanks to \cref{z:pZMoi26flsRb}, this implies
  $\Phi({M}^2) \geq \Phi({M})\,\mathds{1}^{-1}\,\Phi({M}) = [\Phi({M})]^2$.
\end{proof}

\section{Solutions of the anticommutator equation}
\label{z:V.Yy8m55j7Q5}

In this Appendix, 
we briefly review the solutions of the anticommutator equation
\begin{align}
  \frac12\{\rho, {M}\} = {N}\ ,
\end{align}
where ${M}$ is the unknown operator, ${N}$ is a fixed operator, $\rho$ is a
subnormalized quantum state, and $\{A,B\} := AB+BA$ denotes the anticommutator.

For any subnormalized state $\rho$, it is convenient to define the
Hermiticity-preserving super-operator $\mathcal{R}_\rho$ as
\begin{align}
  \label{z:zmzb11eAMq9d}
 \mathcal{R}_\rho(\cdot) = \frac12 \bigl\{ \rho, (\cdot) \bigr\}\ .
\end{align}
Note that $\mathcal{R}_\rho$ is neither completely positive nor
trace preserving.  The operator $\mathcal{R}_\rho$ is self-adjoint, since
$\operatorname{tr}({N} \mathcal{R}_\rho({M})) = \frac12 \operatorname{tr}({N} \{\rho, {M}\}) = \frac12\operatorname{tr}(\{\rho,{N}\}
{M}) = \operatorname{tr}(\mathcal{R}_\rho({N})\,{M})$.
It is interesting to study the superoperator $\mathcal{R}_\rho$ as a linear
operator in Hilbert-Schmidt space.  In vectorized operator space, it is
represented as
\begin{align}
  \label{z:I9svpWrHYjZf}
  \mathcal{R}_\rho = \frac12\bigl( \rho\otimes\mathds{1}+ \mathds{1}\otimes\rho^T \bigr)\ .
\end{align}
This matrix is Hermitian and positive, and it is positive definite if and only
if $\rho$ has full rank.  The fact that the vectorized matrix representing
$\mathcal{R}_\rho$ is positive is not to be confused with the usual notion of a
superoperator being positive, which means preserving the positivity of its
argument.  Here, $\mathcal{R}_\rho$ has a positive semidefinite vectorized
representation, which means that $\llangle {{M}}\mkern 1.5mu\relax \vert \mkern 1.5mu\relax { \mathcal{R}_\rho }\mkern 1.5mu\relax \vert \mkern 1.5mu\relax {{M}}\rrangle \geq 0$ for
all operators $\lvert {{M}}\rrangle $.

Suppose for a moment that $\rho$ has full rank.  Then the superoperator
$\mathcal{R}_\rho$ can be inverted, because its vectorized operator matrix
representation has full rank, and we denote the inverse by
$\mathcal{R}_\rho^{-1}$.  The operator ${M}=\mathcal{R}_\rho^{-1}({N})$ is then the
unique solution to the anticommutator equation $\frac12\{\rho, {M}\} = {N}$.  If
$\{\lvert {k}\rangle \}$ is a basis of the Hilbert space that diagonalizes $\rho$ as
$\rho = \sum_k p_k\lvert {k}\rangle \mkern -1.8mu\relax \langle{k}\rvert $, then~\eqref{z:I9svpWrHYjZf} provides a
diagonal representation of $\mathcal{R}_\rho$, and we obtain the familiar
expression of $\mathcal{R}_\rho^{-1}$ as
\begin{align}
  \label{z:5DFjqResTjZ.}
  \mathcal{R}_\rho^{-1} \lvert {{N}}\rrangle 
  &= \sum_{k,k'} \frac2{p_k+p_{k'}} \lvert {k,k'}\rrangle \llangle {k,k'}\mkern 1.5mu\relax \vert \mkern 1.5mu\relax {{N}}\rrangle \ ,
    \text{\ i.e.,}
  &
  \mathcal{R}_\rho^{-1}({N})
  &= \sum_{k,k'} \frac2{p_k+p_{k'}} \langle {k}\mkern 1.5mu\relax \vert \mkern 1.5mu\relax {{N}}\mkern 1.5mu\relax \vert \mkern 1.5mu\relax {k'}\rangle \, \lvert {k}\rangle \mkern -1.8mu\relax \langle{k'}\rvert \ .
\end{align}
If $\rho$ is not full rank, then we define $\mathcal{R}_\rho^{-1}$ as the
Moore-Penrose inverse of the superoperator $\mathcal{R}_\rho$, i.e., we take the
inverse on its support.  From~\eqref{z:I9svpWrHYjZf} we can identify
the kernel $\ker\mathcal{R}_\rho$ of the superoperator $\mathcal{R}_\rho$ as the
space spanned by operators of the form $\lvert {\phi,\psi}\rrangle $ where
$P_\rho\lvert {\phi }\rangle = P_\rho\lvert {\psi}\rangle =0$, where $P_\rho$ is the projector onto the
support of $\rho$.  If $\{\lvert {k}\rangle \}$ is a basis of the Hilbert space that
diagonalizes $\rho$ as $\rho = \sum_k p_k\lvert {k}\rangle \mkern -1.8mu\relax \langle{k}\rvert $,
then~\eqref{z:I9svpWrHYjZf} is diagonal in the basis $\{\lvert {k,k'}\rrangle \}$
and we see that the expression~\eqref{z:5DFjqResTjZ.} remains the
correct expression for $\mathcal{R}_\rho^{-1}$, provided we only keep those
terms in the sum for which $p_k+p_{k'}\neq 0$.

We may now state the following useful proposition that characterizes the full
solution set of the anticommutator equation $\frac12\{\rho, {M}\} = {N}$ for ${M}$.
\begin{proposition}[Solutions to the anticommutator equation]
  \label{z:uWDssQGxkLXS}
  Let $\rho$ be any subnormalized quantum state, let ${N}$ be any operator, and
  let $\mathcal{R}_\rho$ be given by~\eqref{z:zmzb11eAMq9d}.  Let $P_\rho$ denote
  the projector onto the support of $\rho$ and let $P_\rho^\perp=\mathds{1}-P_\rho$.
  Then the set $\mathcal{S}$ of solutions of the equation
  $Y = 
  \mathcal{R}_\rho({M})$ for the operator ${M}$ is
  \begin{align}
    \mathcal{S} = \begin{cases}
      \emptyset & \textup{if \(P_\rho^\perp {N} P_\rho^\perp \neq 0\)}\ ; \\
      \bigl\{ \mathcal{R}_\rho^{-1}({N}) + P_\rho^\perp {M'} P_\rho^\perp\ :\ {M'}\ \textup{any operator}\bigr\}
      & \textup{if \(P_\rho^\perp {N} P_\rho^\perp = 0\)}\ ,
    \end{cases}
  \end{align}
  where $\mathcal{R}_\rho^{-1}$ denotes as above the Moore-Penrose pseudoinverse
  of the superoperator $\mathcal{R}_\rho$.
  Furthermore, if $Y$ is Hermitian, then the set $\mathcal{S}_H$ of Hermitian
  solutions of the equation ${N} = \mathcal{R}_\rho({M})$ for the operator ${M}$ is
  \begin{align}
    \mathcal{S}_H = \begin{cases}
      \emptyset & \textup{if \(P_\rho^\perp {N} P_\rho^\perp \neq 0\)}\ ; \\
      \bigl\{ \mathcal{R}_\rho^{-1}({N}) + P_\rho^\perp {M'}' P_\rho^\perp\ :\ {M'}'\ \textup{any Hermitian operator}\bigr\}
      & \textup{if \(P_\rho^\perp {N} P_\rho^\perp = 0\)}\ ,
    \end{cases}
  \end{align}
  where $\mathcal{R}_\rho^{-1}(Y)$ is always a Hermitian operator.
\end{proposition}

This proposition is essentially obvious if we think of superoperators as linear
operators in Hilbert-Schmidt space.  Indeed, it is well known that the general
solution to a system of equations given in matrix form can be expressed by the
matrix pseudoinverse, plus anything that is in the matrix kernel.

\begin{proof}[**z:uWDssQGxkLXS]
  Let $\mathcal{P}^\perp$ be the superoperator projector onto the kernel of
  $\mathcal{R}_\rho$.  In the vectorized-operator representation, we have
  $\mathcal{P}^\perp := P_\rho^\perp\otimes (P_{\rho}^\perp)^T$ as can be seen
  from~\eqref{z:I9svpWrHYjZf}.  Let
  $\mathcal{P} = {{\mathrm{id}}} - \mathcal{P}^\perp$ be the superoperator
  projector onto the complementary operator subspace, which is the support of
  $\mathcal{R}_\rho$.  Observe that
  $\mathcal{R}_\rho \mathcal{R}_\rho^{-1} =
  \mathcal{R}_\rho^{-1}\mathcal{R}_\rho = \mathcal{P}$ and that
  $\mathcal{R}_\rho \mathcal{P}^{\perp} = 0$.

  The claim we want to show is that if $\mathcal{P}^\perp\lvert {{N}}\rrangle  \neq 0$, then
  there is no solution to the equation $\lvert {{N}}\rrangle  = \mathcal{R}_\rho \lvert {{M}}\rrangle $;
  otherwise, then the equation is satisfied if and only if
  \begin{align}
  \lvert {{M}}\rrangle  = \mathcal{R}_\rho^{-1}\lvert {{N}}\rrangle  + \mathcal{P}^\perp \lvert {{M'}}\rrangle 
  \end{align}
  for
  some operator $\lvert {{M'}}\rrangle $.
  The condition $\mathcal{P}^\perp\lvert {{N}}\rrangle  = 0$ is necessary for any solution to
  the equation $\lvert {{N}}\rrangle  = \mathcal{R}_\rho \lvert {{M}}\rrangle $ to exist, as otherwise
  $\lvert {{N}}\rrangle $ would not be in the range of $\mathcal{R}_\rho$.  We can therefore
  assume for the rest of this proof that $\mathcal{P}^\perp\lvert {{N}}\rrangle  = 0$.

  Suppose ${M}$ solves $\mathcal{R}_\rho \lvert {{M}}\rrangle  = \lvert {{N}}\rrangle $.  Applying
  $\mathcal{R}_\rho^{-1}$ on both sides, we have
  $\mathcal{P} \lvert {{M}}\rrangle  = \mathcal{R}_\rho^{-1} \lvert {{N}}\rrangle $, which determines
  $\lvert {{M}}\rrangle $ on the operator space projected onto by $\mathcal{P}$.  On the
  complementary space (associated with $\mathcal{P}^\perp$), the operator
  $\lvert {{M}}\rrangle $ can be arbitrary because this subspace is the kernel of
  $\mathcal{R}_\rho$.  A general operator in this subspace can be written as
  $\mathcal{P}^\perp \lvert {{M'}}\rrangle $ for some operator ${M'}$.  This proves that the
  solution ${M}$ must have the form given in the claim.
  Conversely, if
  \begin{align}
  \lvert {{M} }\rrangle = \mathcal{R}_\rho^{-1}\lvert {{N}}\rrangle  + \mathcal{P}^{\perp}
  \lvert {{M'}}\rrangle 
  \end{align} for
  some operator $\lvert {{M'}}\rrangle $, then we see that
  $\mathcal{R}_\rho
  \lvert {{M}}\rrangle  = \mathcal{R}_\rho \,
  \bigl(\mathcal{R}_\rho^{-1}\lvert {{N}}\rrangle  + \mathcal{P}^{\perp}\lvert {X'}\rrangle \bigr) =
  \mathcal{P}\lvert {{N}}\rrangle  = \lvert {{N}}\rrangle $, thus proving the claim.

  If ${N}$ is Hermitian, then $\mathcal{R}_\rho^{-1}({N})$ is Hermitian because
  $\mathcal{R}_\rho$, and hence $\mathcal{R}_\rho^{-1}$, is
  Hermiticity preserving.  Any two Hermitian solutions ${M}_0,{M}_1$, as seen above,
  must differ by a term $P_\rho^\perp {M'} P_\rho^\perp$ for some arbitrary ${M'}$;
  because the difference ${M}_0-{M}_1$ is Hermitian, ${M'}$ can be chosen to be
  Hermitian as well (specifically, one can set ${M'}'=({M'}+{M'}^\dagger)/2$).
\end{proof}

We now compute the map $\mathcal{R}_\psi^{-1}(\cdot)$ in the case the reference
state is a pure (normalized) state vector $\lvert {\psi}\rangle $.

\begin{proposition}[Computing $\mathcal{R}_\rho^{-1}$ when $\rho$ is a pure state]
  \label{z:qtGZ5Y7B.DE1}
  Let $\lvert {\psi}\rangle $ be a (normalized) state vector and let
  $P_\psi^\perp = \mathds{1}- \lvert {\psi}\rangle \mkern -1.8mu\relax \langle{\psi}\rvert $.  Then for any Hermitian ${O}$ we have
  \begin{align}
    \mathcal{R}_\psi^{-1}({O}) = 2({O} - P_\psi^\perp {O} P_\psi^\perp) - \langle {{O}}\rangle _\psi \psi\ .
  \end{align}
\end{proposition}
\begin{proof}[**z:qtGZ5Y7B.DE1]
  Define $\bar{{O}} := {O} - P_\psi^\perp 
  {O} P_\psi^\perp - \langle {{O}}\rangle _\psi \psi$.  By
  linearity, we have
  \begin{align}
    \mathcal{R}_\psi^{-1}(\bar {{O}})
    = \mathcal{R}_\psi^{-1}({O}) - \langle {{O}}\rangle _\psi\,
    \psi\ ,
    \label{z:P2iJM3yraC4r}
  \end{align}
  noting that $\frac12\{\psi,\psi\} = \psi$ and therefore
  $\mathcal{R}_\psi^{-1}(\psi) = \psi$.  
  The operator 
  $\bar{{O}}$ satisfies
  $P_\psi^\perp 
  \bar {{O}} P_\psi^\perp = 0$ and $\langle {\bar{{O}}}\rangle _\psi = 0$, the latter
  implying that 
  $\bar{{O}}\psi = P_\psi^\perp 
  \bar{{O}} \psi$.  Then
  \begin{align*}
    \bar{{O}} 
    &= (\psi + P_\psi^\perp) 
    \bar{{O}} (\psi + P_\psi^\perp)
    = P_\psi^\perp 
    \bar{{O}} \psi + \psi \bar{{O}} P_\psi^\perp
    = \bigl\{\psi,
    \bar{{O}}
    \bigr\}\ ,
  \end{align*}
  and therefore $\mathcal{R}_\psi^{-1}(\bar{{O}}) = 2
  \bar{{O}}$.
  From~\eqref{z:P2iJM3yraC4r} we then find
  \begin{align}
    \mathcal{R}_\psi^{-1}({O})
    = 2\bar{{O}} + \langle {{O}}\rangle \psi
    = 2(Z - P_\psi^\perp {O} P_\psi^\perp) - \langle {{O}}\rangle  \psi\ .
    \tag*\qedhere
  \end{align}
\end{proof}

\section{Semidefinite programming methods for the Fisher information}
\label{z:KH4B5FtzQ1kE}

In this Appendix, 
we review some methods based on semidefinite
programming~\cite{R74,R75} for computing
the Fisher information, and review some elementary properties of the Fisher
information.
Let $\rho$ be any subnormalized quantum state, and let $D$ be any Hermitian
operator that satisfies $P_\rho^\perp D P_\rho^\perp=0$ (recall $P_\rho^\perp$
is the projector onto the kernel of $\rho$).  Define the quantity
\begin{align}
  \label{z:.gRi0KAw40V9}
  \Ftwo{\rho}{D}  := \operatorname{tr}\bigl(\rho R^2 \bigr)\ ,
\end{align}
where $R$ is any solution to $\frac12\bigl\{\rho,R\bigr\} = D$.  For a normalized state
$\rho$ and for traceless $D$, the quantity $\Ftwo{\rho}{D}$ corresponds to the
Fisher information associated with a one-parameter family of states
$\lambda\mapsto \rho_\lambda$ taken at a value of $\lambda$ where $\rho_\lambda=\rho$ and
$d\rho_\lambda/d\lambda = D$.  We allow subnormalized states $\rho$ and
operators $D$ with nonzero trace in the
definition~\eqref{z:.gRi0KAw40V9} for later technical convenience.  We
require that $P_\rho^\perp D P_\rho^\perp=0$ as otherwise the anticommutator
equation $\frac12\bigl\{\rho,R\bigr\} = D$ has no solution for $R$.

The definition of $\Ftwo{\rho}{D}$ does not depend on the choice of $R$ that
solves $\frac12\bigl\{\rho,R\bigr\} = D$.  Indeed,
\cref{z:uWDssQGxkLXS} guarantees that any two solutions
differ only by a term $P_\rho^\perp {M'} P_\rho^\perp$; such a term does not
contribute to the trace in~\eqref{z:.gRi0KAw40V9}.  We may therefore
write, using the notation of \cref{z:V.Yy8m55j7Q5},
\begin{align}
  \Ftwo{\rho}{D} = \operatorname{tr}\bigl( \rho \, \bigl[ \mathcal{R}_\rho^{-1}(D) \bigr]^2 \bigr)\ .
  \label{z:HWLNvE.Hsjmk}
\end{align}

We now write this expression as a pair of convex optimizations.  These
expressions have been derived in
Refs.~\cite{R25,R27}; we provide a
proof using our notation for self-consistency.

\begin{proposition}[Fisher information in terms of convex optimization problems]
  \label{z:xifd0Y80UQtS}
  Let $\rho$ be a subnormalized quantum state and $D$ be a Hermitian operator
  that satisfies $P_\rho^\perp D P_\rho^\perp = 0$.  The quantity
  $\Ftwo{\rho}{D}$ defined in~\eqref{z:.gRi0KAw40V9} is equivalently
  expressed as the following optimizations:
  \begin{subequations}
    \begin{align}
      \Ftwo{\rho}{D}
      &=
        \max_{S=S^\dagger} 4 \bigl[ \operatorname{tr}(D S) - \operatorname{tr}(\rho S^2) \bigr]
        \label{z:YdgtbMB30Vi6}
      \\
      &=
        \min\; \bigl\{ 4\operatorname{tr}(L^\dagger L) \ :\ \rho^{1/2} L + L^\dagger \rho^{1/2} = D \bigr\}\ ,
        \label{z:V4lWxgFVfrPu}
    \end{align}
    where the first optimization ranges over all Hermitian operators $S$ and where
    in the second optimization $L$ is an arbitrary complex matrix.
    Optimal choices for the variables are
    $S = (1/2) \mathcal{R}_\rho^{-1}\bigl(D\bigr)$ and $L = \rho^{1/2} S$, noting
    that $\{\rho,S\} = D$.
    Furthermore, alternative forms for the minimization are
    \begin{align}
      \Ftwo{\rho}{D}
      &=
        4 \min\, \Bigl\{ \operatorname{tr}({N})\ :\ 
        \begin{bmatrix} \rho & {O} \\ {O}^\dagger & {N}\end{bmatrix} \geq 0
     \quad\text{with}\quad
      {O}+{O}^\dagger = D\,,\ {N} \geq 0
      \Bigr\}
      \label{z:HD7JuvyAuF39}
      \\
      &=
        \min\, \Bigl\{ \operatorname{tr}(J)\ :\ 
        \begin{bmatrix} \rho & D + iK \\ D - iK & J\end{bmatrix} \geq 0
      \quad\text{with}\quad
      K=K^\dagger\,,\ J\geq 0
      \Bigr\}\ ,
      \label{z:CJMl45Ex5Bss}
    \end{align}
  \end{subequations}
  in which optimal choices are $O = \rho S$, $N=S\, \rho\, S$,
  $K=-i[\rho, S]$, and
  $J=\mathcal{R}_\rho^{-1}(D)\,\rho\,\mathcal{R}_\rho^{-1}(D)$.
\end{proposition}
Note that the condition in the
optimization~\eqref{z:CJMl45Ex5Bss}
implicitly enforces the fact that $P_\rho^{\perp}(D+iK) = 0$ (see
\cref{z:pZMoi26flsRb}); this can make it more complicated to guess a candidate
for $K$ in~\eqref{z:CJMl45Ex5Bss} if $\rho$
does not have full rank, especially if $P_\rho^\perp D \neq 0$.
Also, note that if there is any feasible choice of candidates
in~\eqref{z:HD7JuvyAuF39}, then
automatically $P^\perp_\rho {O} = 0$ and ${O}^\dagger P^\perp_\rho = 0$, such that
$P^\perp_\rho D P^\perp_\rho = P^\perp_\rho ({O}+{O}^\dagger) P^\perp_\rho = 0$.
Therefore, finding feasible candidates automatically enforces the condition in
the definition~\eqref{z:.gRi0KAw40V9}.  A similar argument holds if
feasible candidates are found
in~\eqref{z:V4lWxgFVfrPu}
or~\eqref{z:CJMl45Ex5Bss}.

\begin{proof}[**z:xifd0Y80UQtS]
  The maximization~\eqref{z:YdgtbMB30Vi6} is a
  quadratic optimization can be cast into a semidefinite program using Schur
  complements (\cref{z:pZMoi26flsRb}).  We stick closely to the formalism of
  Watrous~\cite{R75,R105}.  We introduce a
  variable $Q\geq 0$ with the constraint $Q\geq S^2$ expressed as a Schur
  complement condition:
\begin{align}
  \frac14\bigl\{
  \text{maximization in~\eqref{z:YdgtbMB30Vi6}}
  \bigr\}
  &=
  \begin{aligned}[t]
    \mathrm{maximize}:\quad
    &
      \bigl[ \operatorname{tr}( D S ) - \operatorname{tr}(\rho Q) \bigr]
    \\
    \mathrm{over~variables}:\quad
    & S=S^\dagger, Q\geq 0
    \\
    \mathrm{subject~to}:\quad
    & \begin{bmatrix} Q & -S \\ -S & \mathds{1}\end{bmatrix} \geq 0\ .
  \end{aligned}
  \label{z:YR2AWiJacgL7}
\end{align}
(The sign of $-S$ in the last constraint is for later convenience.)  We now
determine the corresponding dual problem.  Let
$\left[\begin{smallmatrix} {M} & {O}\\{O}^\dagger & {N}\end{smallmatrix}\right] \geq 0$
be the Lagrange dual variable corresponding to the primal constraint, with ${M},{N}\geq 0$
and ${O}$ arbitrary.  The primal
constraint can be written as
\begin{align}
  \begin{bmatrix} -1 & 0 \\ 0 & 0\end{bmatrix}\otimes Q
  + \begin{bmatrix} 0 & 1 \\ 1 & 0\end{bmatrix}\otimes S
  \leq \begin{bmatrix} 0 & 0 \\ 0 & 1\end{bmatrix}\otimes \mathds{1}\ .
  \label{z:7elklA8CRj8Z}
\end{align}
The dual objective is obtained by collecting the constant terms of the
constraints and taking the inner product with the corresponding dual variable.
Here we only have the right-hand side of~\eqref{z:7elklA8CRj8Z} and we
obtain the objective that is simply to minimize $\operatorname{tr}({N})$.  There are two dual
constraints, one for each primal variable $S$ and $Q$; to a Hermitian variable
corresponds an equality constraint and to a positive semidefinite variable
corresponds a positive semidefinite constraint.  The primal objective gives the
constant terms for each constraint, which are $(\ldots) = D$ and
$(\ldots) \geq -\rho$.  For the left-hand side of the first constraint we obtain
the term
$\operatorname{tr}_1\mathopen{}\left(\left[\begin{smallmatrix}{M}&{O}\\ {O}^\dagger&{N}\end{smallmatrix}\right]
\left[\begin{smallmatrix}0& 1\\ 1 & 0\end{smallmatrix}\right]\right)\mathclose{} = {O} + {O}^\dagger$.
For the left-hand side of the second constraint, we find
$-\operatorname{tr}_1\mathopen{}\left(\left[\begin{smallmatrix}{M}&{O}\\ {O}^\dagger&{N}\end{smallmatrix}\right]
\left[\begin{smallmatrix}\vphantom{{O}}1& 0\\ \vphantom{{O}^\dagger}0 &
    0\end{smallmatrix}\right]\right)\mathclose{} = -{M}$.  We thus obtain the following dual
problem: %
\begin{align}
  \text{\eqref{z:YR2AWiJacgL7}}
  &=
  \begin{aligned}[t]
    \mathrm{minimize}\quad
    & \operatorname{tr}({N})
    \\
    \mathrm{over~variables}\quad
    &
    {M}\geq 0,\ {N}\geq 0,\ {O}~{,}
    \\
    \mathrm{subject~to}\quad
    &
      \begin{gathered}[t]
        {O} + {O}^\dagger = D\ {,} \\
        {M} \leq \rho\ {,}\\
      \begin{bmatrix} {M} & {O} \\ {O}^\dagger & 
      {N}\end{bmatrix} \geq 0\ .
      \end{gathered}
  \end{aligned}
  \label{z:QBNy4lzhmGKl}
\end{align}
Equality with the primal optimization problem holds thanks to strong duality,
which is ensured by the Slater
conditions~\cite{R75,R105}.  We can further
simplify the dual problem.  First, the choice ${M}=\rho$ is optimal: Indeed, for
any optimal choices of variables with ${M}\leq\rho$, we can replace ${M}$ by $\rho$
while still achieving the same value.  Therefore{,}
\begin{align}
  \text{\eqref{z:QBNy4lzhmGKl}}
  &=
  \begin{aligned}[t]
    \mathrm{minimize}\quad
    & \operatorname{tr}({N})
    \\
    \mathrm{over~variables}\quad
    &
    {N}\geq 0,\ {O}\ {,}
    \\
    \mathrm{subject~to}\quad
    &
      \begin{gathered}[t]
        {O} + {O}^\dagger = D\ {,}\\
        \begin{bmatrix} \rho & {O} \\
          {O}^\dagger & {N} \end{bmatrix} \geq 0\ .
      \end{gathered}
  \end{aligned}
  \label{z:IzVNtZCA8V9E}
\end{align}
Using the Schur complement argument again (\cref{z:pZMoi26flsRb}), we find that
${n}\geq {O}^\dagger \rho^{-1} {O}$, and for the same reason as above, there is an
optimal choice of variables with ${Y} = {O}^\dagger \rho^{-1} {O}$.  Hence
\begin{align}
  \text{\eqref{z:IzVNtZCA8V9E}}
  &=
  \begin{aligned}[t]
    \mathrm{minimize}:\quad
    & \operatorname{tr}({O}^\dagger \rho^{-1} {O})
    \\
    \mathrm{over~variables}:\quad
    &
    {O}~\mathrm{arb.}
    \\
    \mathrm{subject~to}:\quad
    & {O} + {O}^\dagger = D \\
    & P_\rho^\perp {O} = 0\ .
  \end{aligned}
  \label{z:XLC-MgBnHNnF}
\end{align}
We may introduce the variable $L = \rho^{-1/2} {O}$, which yields
\begin{align}
  \text{\eqref{z:XLC-MgBnHNnF}}
  &=
  \begin{aligned}[t]
    \mathrm{minimize}:\quad
    & \operatorname{tr}(L^\dagger L)
    \\
    \mathrm{over~variables}:\quad
    &
    L~\mathrm{arb.}
    \\
    \mathrm{subject~to}:\quad
    & \rho^{1/2} L + L^\dagger \rho^{1/2} = D\ .
  \end{aligned}
  \label{z:-VreJyyLKXzL}
\end{align}
We recognize the optimization
in~\eqref{z:V4lWxgFVfrPu}.
At this point we have shown that both optimizations in the claim,
\cref{z:YdgtbMB30Vi6,z:V4lWxgFVfrPu}, are equal thanks to
semidefinite programming duality.  It remains to show that the common optimal
value is $\Ftwo{\rho}{D}$ as given
by~\eqref{z:HWLNvE.Hsjmk}.

To find optimal variables, we examine the complementary slackness
conditions~\cite{R105} corresponding to the primal-dual
problem pair~\eqref{z:YR2AWiJacgL7}
and~\eqref{z:QBNy4lzhmGKl}.  Namely, taking the product of
an inequality constraint with the corresponding dual variable turns the
inequality into an equality for optimal primal and dual choices of variables.
For the primal constraint this gives us the equalities
\begin{align}
  \begin{aligned}
    Q{M} - S{O}^\dagger &= 0\ ;
    &\qquad\qquad
    Q{O} - S{N} &= 0\ ;
    \\
    -S{M} + {O}^\dagger &= 0\ ;
    &\qquad\qquad
    -S{O} + {N} &= 0\ .
  \end{aligned}
\end{align}
From $-S{M} + {O}^\dagger = 0$ along with the optimal ${M}=\rho$ we deduce that
$\rho S = {O}$ and thus $\rho^{1/2}S = \rho^{-1/2}{O} = L$.  Plugging this into the
constraint in~\eqref{z:-VreJyyLKXzL} we find
\begin{align}
  \rho S + S \rho = D\ .
  \label{z:DFfRoGaKlwjQ}
\end{align}
The solutions of this anticommutator equation {have been} studied in
\cref{z:V.Yy8m55j7Q5}, leading us to the primal candidate
\begin{align}
  S = \frac12 \mathcal{R}_\rho^{-1}\bigl( D \bigr)\ .
\end{align}
Plugging this choice into~\eqref{z:YR2AWiJacgL7}, along
with the choice $Q=S^2$, we obtain
\begin{align}
  \text{\eqref{z:YR2AWiJacgL7}}
  &\geq \frac12 \operatorname{tr}\bigl( D \mathcal{R}_\rho^{-1}(D)  \bigr)
    - \frac14 \operatorname{tr}\bigl( \rho \bigl[ \mathcal{R}_\rho^{-1}(D) \bigr]^2 \bigr)
  = \frac14 \Ftwo{\rho}{D}\ ,
    \label{z:2.feaXnqFkjS}
\end{align}
where we have used the fact that
$\operatorname{tr}\bigl( D \mathcal{R}_\rho^{-1}(D)  \bigr)
= \llangle {D}\mkern 1.5mu\relax \vert \mkern 1.5mu\relax { \mathcal{R}_\rho^{-1} }\mkern 1.5mu\relax \vert \mkern 1.5mu\relax {D}\rrangle 
= \llangle {D}\mkern 1.5mu\relax \vert \mkern 1.5mu\relax { \mathcal{R}_\rho^{-1} \mathcal{R}_\rho \mathcal{R}_\rho^{-1} }\mkern 1.5mu\relax \vert \mkern 1.5mu\relax {D}\rrangle 
= \llangle {\mathcal{R}_\rho^{-1}(D)}\mkern 1.5mu\relax \vert \mkern 1.5mu\relax { \mathcal{R}_\rho }\mkern 1.5mu\relax \vert \mkern 1.5mu\relax {\mathcal{R}_\rho^{-1}(D)}\rrangle 
= \operatorname{tr}\bigl(\rho\,\bigl[ \mathcal{R}_\rho^{-1}(D) \bigr]^2 \bigr)$.

By construction, $L = \rho^{1/2} S$ satisfies the constraint
in~\eqref{z:-VreJyyLKXzL}, noting that we have used
the assumption that $P_\rho^\perp D P_\rho^\perp=0$ as in the proposition
statement.  The corresponding value attained in the dual problem is
\begin{align}
  \text{\eqref{z:-VreJyyLKXzL}}
  &\leq \operatorname{tr}(L^\dagger L)
    = \frac14\operatorname{tr}\bigl(\rho\bigl[ \mathcal{R}_\rho^{-1}(D) \bigr]^2\bigr)
    = \frac14 \Ftwo{\rho}{D}\ ,
  \label{z:NyVP614mSrMk}
\end{align}
noting that $\operatorname{tr}(L^\dagger L) = \operatorname{tr}(\rho S^2)$.
Combining \cref{z:2.feaXnqFkjS,z:NyVP614mSrMk} with the above
statement %
that
$\text{\eqref{z:YR2AWiJacgL7}} =
\text{\eqref{z:-VreJyyLKXzL}}$ proves the first part of the claim.

The alternative form~\eqref{z:HD7JuvyAuF39}
is nothing else than~\eqref{z:IzVNtZCA8V9E}.
Now we show the alternative
form~\eqref{z:CJMl45Ex5Bss}.  Consider the
optimization~\eqref{z:IzVNtZCA8V9E}.  Decompose
${O}={O}_R+i{O}_I$ into Hermitian and anti-Hermitian parts with
${O}_R = ({O}+{O}^\dagger)/2 ={O}_R^\dagger$, ${O}_I = -i({O}-{O}^\dagger)/2 ={O}_I^\dagger$.
The constraint on ${O}$ indicates that the Hermitian part ${O}_R$ of ${O}$ must
satisfy $2{O}_R = D$.  The second constraint then becomes
\begin{align}
  \begin{bmatrix}\rho & D/2 + i{O}_I\\ D/2-i{O}_I & {N}\end{bmatrix}\geq 0\ .
\end{align}
Conjugating by $\begin{bsmallmatrix}\mathds{1}& 0\\ 0 & 2\mathds{1}\end{bsmallmatrix}$,
we see that this condition is equivalent to
\begin{align}
  \begin{bmatrix}\rho & D + 2i {O}_I\\ D-2i{O}_I & 4{N}\end{bmatrix}\geq 0\ .
\end{align}
Now we set $K=2{O}_I$ and $J=4{N}$, showing that the
optimization~\eqref{z:CJMl45Ex5Bss} is
equivalent to~\eqref{z:IzVNtZCA8V9E} (up to a factor
of $4$), and therefore equal to $\Ftwo{\rho}{D}$.

For completeness, we exhibit optimal choices for $K,J$.  Choose $K$ to be the
anti-Hermitian part of $\rho\mathcal{R}_\rho^{-1}(D)$, 
i.e.,
$K = (\rho\mathcal{R}_\rho^{-1}(D) - \mathcal{R}_\rho^{-1}(D)\rho)/(2i)$.  Then
\begin{align}
  D + iK = \frac12\bigl( \rho\mathcal{R}_\rho^{-1}(D) + \mathcal{R}_\rho^{-1}(D)\rho \bigr)
  + \frac12\bigl(\rho\mathcal{R}_\rho^{-1}(D) - \mathcal{R}_\rho^{-1}(D)\rho\bigr)
  = \rho\mathcal{R}_\rho^{-1}(D)\ ,
\end{align}
and its Hermitian conjugate is $D - iK = \mathcal{R}_\rho^{-1}(D)\rho$.  Now
choose
$J = (D-iK)\rho^{-1}(D+iK) =
\mathcal{R}_\rho^{-1}(D)\rho\mathcal{R}_\rho^{-1}(D)$; the constraint
in~\eqref{z:CJMl45Ex5Bss} is satisfied thanks
to \cref{z:pZMoi26flsRb}.  The value reached by this choice of candidates is
then the optimal value $\operatorname{tr}(J)=\Ftwo{\rho}{D}$.
\end{proof}

The expressions in \cref{z:xifd0Y80UQtS} lead to simple
proofs of elementary properties of the Fisher information.

\begin{proposition}[Simple bounds for the Fisher information]
  \label{z:4lI8Rm6bxlme}
  Let $\rho$ be a subnormalized quantum state and $D$ be a Hermitian operator
  that satisfies $P_\rho^\perp D P_\rho^\perp = 0$.  Then we have
  \begin{align}
    \opnorm{D}^2 \leq \Ftwo{\rho}{D} \leq \operatorname{tr}(\rho^{-1} D'^2)\ ,
    \label{z:aSKeUo91pZRw}
  \end{align}
  where $D' = 2D - P_\rho D P_\rho$.
\end{proposition}
\begin{proof}[**z:4lI8Rm6bxlme]
  First we show the lower bound.  Let $\lvert {\phi}\rangle $ be a (normalized) eigenvector
  associated with the largest eigenvalue of $D$ (in magnitude), such that
  $\langle {\phi}\mkern 1.5mu\relax \vert \mkern 1.5mu\relax {D}\mkern 1.5mu\relax \vert \mkern 1.5mu\relax {\phi}\rangle  = \opnorm{D}$.  For some $s\in\mathbb{R}$ to be determined
  later, we choose the optimization candidate $S=s\lvert {\phi}\rangle \mkern -1.8mu\relax \langle{\phi}\rvert $
  in~\eqref{z:YdgtbMB30Vi6}.  Then the
  corresponding objective value is
  \begin{align}
    \Ftwo{\rho}{D}
    &\geq
    4\operatorname{tr}(DS) - 4\operatorname{tr}(\rho S^2)
      = 4s\opnorm{D} - 4s^2\langle {\phi}\mkern 1.5mu\relax \vert \mkern 1.5mu\relax {\rho}\mkern 1.5mu\relax \vert \mkern 1.5mu\relax {\phi}\rangle \ .
  \end{align}
  The latter expression is maximal when
  $0 = (d/ds)(~{\cdots}~) = 4\opnorm{D} - 8s\langle {\phi}\mkern 1.5mu\relax \vert \mkern 1.5mu\relax {\rho}\mkern 1.5mu\relax \vert \mkern 1.5mu\relax {\phi}\rangle $, i.e., when
  $s=\opnorm{D}/(2\langle {\phi}\mkern 1.5mu\relax \vert \mkern 1.5mu\relax {\rho}\mkern 1.5mu\relax \vert \mkern 1.5mu\relax {\phi}\rangle )$.  We obtain the bound
  \begin{align}
    \Ftwo{\rho}{D}
    &\geq
      2\frac{\opnorm{D}^2}{\langle {\phi}\mkern 1.5mu\relax \vert \mkern 1.5mu\relax {\rho}\mkern 1.5mu\relax \vert \mkern 1.5mu\relax {\phi}\rangle } - \frac{\opnorm{D}^2}{\langle {\phi}\mkern 1.5mu\relax \vert \mkern 1.5mu\relax {\rho}\mkern 1.5mu\relax \vert \mkern 1.5mu\relax {\phi}\rangle }
      = \frac{\opnorm{D}^2}{\langle {\phi}\mkern 1.5mu\relax \vert \mkern 1.5mu\relax {\rho}\mkern 1.5mu\relax \vert \mkern 1.5mu\relax {\phi}\rangle } \geq \opnorm{D}^2\ ,
  \end{align}
  recalling furthermore that $\langle {\phi}\mkern 1.5mu\relax \vert \mkern 1.5mu\relax {\rho}\mkern 1.5mu\relax \vert \mkern 1.5mu\relax {\phi}\rangle \leq 1$.

  For the upper bound, consider the optimization
  problem~\eqref{z:V4lWxgFVfrPu} and choose the
  candidate $L=\rho^{-1/2} D'/2$.  This is a feasible candidate because
  \begin{align}
    \rho^{1/2}L + L^\dagger \rho^{1/2}
    &= P_\rho D' + (\text{h.c.})
    \quad=\quad 2 P_\rho D - P_\rho D P_\rho + (\text{h.c.})
    \nonumber\\
    &= 2 P_\rho D(P_\rho + P_\rho^\perp) - P_\rho D P_\rho + (\text{h.c.})
      \nonumber\\
    &= P_\rho D P_\rho + 2 P_\rho D P_\rho^\perp + (\text{h.c.})
      \nonumber\\
    &= 2(P_\rho D P_\rho + P_\rho D P_\rho^\perp + P_\rho^\perp D P_\rho)
      \quad=\quad 2D\ ,
  \end{align}
  where $\text{h.c.}$ stands for the Hermitian conjugate of the entire preceding
  expression, and where we furthermore recall that
  $P_\rho^\perp D P_\rho^\perp=0$.  The objective value attained by this choice
  of candidate is
  $\Ftwo{\rho}{D} \leq 4 \operatorname{tr}\bigl(L^\dagger L\bigr) = \operatorname{tr}\bigl(\rho^{-1}D'^2\bigr)$.
\end{proof}

\begin{proposition}[Right logarithmic derivative (RLD)
  bound~\cite{R88}]
  \label{z:Dcj2SGgHn1GH}
  Let $\rho$ be a subnormalized quantum state and let $D$ be a Hermitian
  operator satisfying $P_\rho^\perp D P_\rho^\perp = 0$.  Let $G$ be any
  operator (possibly non-Hermitian) that satisfies
  $\bigl(\rho G + G^\dagger \rho\bigr)/2 = D$.  Then
  \begin{align}
    \Ftwo{\rho}{D} \leq \operatorname{tr}(\rho GG^\dagger)\ .
  \end{align}
\end{proposition}
\begin{proof}[**z:Dcj2SGgHn1GH]
  Use $L = \rho^{1/2}G/2$
  in~\eqref{z:V4lWxgFVfrPu}.
\end{proof}

\begin{proposition}[Fisher information under parameter rescaling]
  \label{z:xptYZVgGge6R}
  Let $\rho$ be a subnormalized quantum state and $D$ be a Hermitian operator
  that satisfies $P_\rho^\perp D P_\rho^\perp = 0$.  Then for any
  $\alpha\leq1$, $\beta\in\mathbb{R}$,
  \begin{align}
    \Ftwo{\alpha\rho}{\beta D} = \frac{\beta^2}{\alpha} \Ftwo{\rho}{D}\ .
  \end{align}
\end{proposition}
\begin{proof}[**z:xptYZVgGge6R]
  Let $S,L$ be optimal variables
  in~\eqref{z:YdgtbMB30Vi6}
  and~\eqref{z:V4lWxgFVfrPu} for
  $\Ftwo{\alpha\rho}{\beta D}$.  Let $S'=(\alpha/\beta) S$ and
  $L'=(\sqrt{\alpha}/\beta)L$.  Then
  \begin{subequations}
    \begin{align}
      \frac14 \Ftwo{\alpha\rho}{\beta D}
      &= \operatorname{tr}(\beta DS)-\operatorname{tr}(\alpha\rho S^2)
        = \frac{\beta^2}{\alpha}\,\bigl[ \operatorname{tr}(DS') - \operatorname{tr}(\rho S'^2) \bigr]
        \leq \frac{\beta^2}{\alpha}\,\frac14 \Ftwo{\rho}{D} \ ;
      \\
      \frac14 \Ftwo{\alpha\rho}{\beta D}
      &= \operatorname{tr}(LL^\dagger)
        = \frac{\beta^2}{\alpha} \operatorname{tr}(L'L'^\dagger)
        \geq \frac{\beta^2}{\alpha}\,\frac14 \Ftwo{\rho}{D}\ ,
    \end{align}
  \end{subequations}
  noting that $L'$ is a valid choice of optimization candidate
  in~\eqref{z:V4lWxgFVfrPu} for $\Ftwo{\rho}{D}$
  because
  $\rho^{1/2} L' + L'^\dagger \rho^{1/2} = (1/\beta) \bigl((\alpha\rho)^{1/2} L
  + L^\dagger (\alpha\rho)^{1/2}\bigr) = D$.
\end{proof}

\begin{proposition}[Fisher information bound for trace-decreasing maps]
  \label{z:6-noPFS3Rdzn}
  Let $\lvert {\psi}\rangle $ be a (normalized) state vector and let $\lvert {\xi}\rangle $ be any vector
  such that $\langle {\psi}\mkern 1.5mu\relax \vert\mkern 1.5mu\relax {\xi}\rangle =0$.  Let $\mathcal{N}$ be any completely positive,
  trace-nonincreasing map and let $0\leq\alpha\leq1$ such that
  $\mathcal{N}^\dagger(\mathds{1})\leq\alpha\mathds{1}$.  Then
  \begin{align}
    \Ftwo{\mathcal{N}(\lvert {\psi}\rangle \mkern -1.8mu\relax \langle{\psi}\rvert )}{\mathcal{N}(\lvert {\xi}\rangle \mkern -1.8mu\relax \langle{\psi}\rvert +\lvert {\psi}\rangle \mkern -1.8mu\relax \langle{\xi}\rvert )}
    \leq 4\alpha\langle {\xi}\mkern 1.5mu\relax \vert\mkern 1.5mu\relax {\xi}\rangle \ .
  \end{align}
\end{proposition}
\begin{proof}[**z:6-noPFS3Rdzn]
  Let ${O} = \mathcal{N}(\lvert {\psi}\rangle \mkern -1.8mu\relax \langle{\xi}\rvert )$ and 
  ${N} = \mathcal{N}(\lvert {\xi}\rangle \mkern -1.8mu\relax \langle{\xi}\rvert )$.
  These choices are feasible
  in~\eqref{z:HD7JuvyAuF39} because
  applying the completely positive map $\mathcal{N}\otimes{\mathrm{id}}_{2}$ onto
  the positive semidefinite matrix
  \begin{align}
  \begin{bmatrix} \lvert {\psi }\rangle \mkern -1.8mu\relax \langle{\psi }\rvert & \lvert {\psi}\rangle \mkern -1.8mu\relax \langle{\xi }\rvert \\ \lvert {\xi}\rangle \mkern -1.8mu\relax \langle{\psi }\rvert & \lvert {\xi }\rangle \mkern -1.8mu\relax \langle{\xi }\rvert \end{bmatrix}
  = \begin{bmatrix} \lvert {\psi }\rangle \\ \lvert {\xi }\rangle \end{bmatrix}
  \begin{bmatrix} \langle {\psi }\rvert & \langle {\xi }\rvert \end{bmatrix} \geq 0
  \end{align}
  gives again a
  positive semidefinite matrix.  This choice of variables yields the objective
  value
  $\operatorname{tr}\bigl(\mathcal{N}(\lvert {\xi}\rangle \mkern -1.8mu\relax \langle{\xi}\rvert )\bigr) =
  \operatorname{tr}\bigl(\mathcal{N}^\dagger(\mathds{1})\,\lvert {\xi}\rangle \mkern -1.8mu\relax \langle{\xi}\rvert \bigr) \leq\alpha\langle {\xi}\mkern 1.5mu\relax \vert\mkern 1.5mu\relax {\xi}\rangle $,
  proving the claim.
\end{proof}

\begin{proposition}[Joint convexity of the Fisher information]
  \label{z:qvWV7mc6QtTI}
  Let $\{ \rho_k \}$ be a set of subnormalized states and $\{ D_k \}$ be a set of
  Hermitian operators such that $P_{\rho_k}^\perp D_k P_{\rho_k}^\perp = 0$.
  Let $\{\alpha_k\}$ be a real positive coefficients such that
  $\sum_k \alpha_k \operatorname{tr}(\rho_k) \leq 1$.  Then
  \begin{align}
    \Ftwo*{\sum_k \alpha_k \rho_k}{\sum_k \alpha_k D_k}
    \leq \sum_k \alpha_k \Ftwo{\rho_k}{D_k}\ .
  \end{align}
\end{proposition}
\begin{proof}[**z:qvWV7mc6QtTI]
  For each $k$, let $K_k,J_k$ be optimal choices
  in~\eqref{z:CJMl45Ex5Bss} for
  $\Ftwo{\rho_k}{D_k}$.  Set $K = \sum_k \alpha_k K_k$ 
  and $J = \sum_k \alpha_k J_k$.  Then
  \begin{align}
    \begin{bmatrix} \rho & D+iK \\ D-iK & J \end{bmatrix}
     = \sum_k \alpha_k \begin{bmatrix} \rho_k & D_k + iK_k \\ D_k - iK_k & J_k \end{bmatrix}
     \geq 0\ ,
  \end{align}
  and so $K,J$ are feasible candidates in the
  problem~\eqref{z:CJMl45Ex5Bss} for
  $\Ftwo{\rho}{D}$.  The objective value achieved for this choice of variables
  gives the bound
  $\Ftwo{\rho}{D} \leq \operatorname{tr}(J) = \sum_k \alpha_k \operatorname{tr}(J_k) = \sum_k \alpha_k
  \Ftwo{\rho_k}{D_k}$.
\end{proof}

\begin{proposition}[Additivity of independent probes]
  \label{z:BCjxitFfXsSP}
  Let $\rho_A$, $\rho'_B$ be two subnormalized quantum states on two systems
  $A,B$, and let $D_A$, $D'_B$ be two traceless Hermitian operators such that
  $P_{\rho_A}^\perp D_A P_{\rho_A}^\perp = 0$ and
  $P_{\rho_B'}^\perp D_B' P_{\rho_B'}^\perp = 0$.  Then
  \begin{align}
    \Ftwo{\rho_A\otimes\rho_B'}{D_A\otimes\rho_B' + \rho_A\otimes D_B'}
    =
    \Ftwo{\rho_A}{D_A} + \Ftwo{\rho_B'}{D_B'}\ .
  \end{align}
\end{proposition}
Observe that the second argument on the left-hand side corresponds to the
derivative of the state of a composite system that remains in a tensor product,
$(d/dt)(\rho_A\otimes\rho_B') = (d\rho_A/dt)\otimes\rho_B' +
\rho_A\otimes(d\rho_B'/dt)$.
\begin{proof}[**z:BCjxitFfXsSP]
  Here we may directly guess a solution $R$ to
  $\mathcal{R}_{\rho_A\otimes\rho_B'}(R) = D_A\otimes\rho_B' + \rho_A\otimes
  D_B'$.  Compute first
  \begin{align}
    \mathcal{R}_{\rho_A\otimes\rho_B'}\bigl(\mathds{1}_A\otimes {M}_B\bigr)
    = \frac12\,\bigl\{\rho_A\otimes\rho_B', \mathds{1}_A\otimes {M}_B\bigr\}
    = \frac12\,\rho_A\otimes\bigl\{\rho_B', {M}_B\bigr\}
    \ ,
  \end{align}
  so we see that, setting
  \begin{align}
    R = \mathds{1}_A\otimes\mathcal{R}_{\rho_B'}^{-1}(D_B') +
    \mathcal{R}_{\rho_A}^{-1}(D_A)\otimes\mathds{1}_B\ ,
  \end{align}
  we have
  $\mathcal{R}_{\rho_A\otimes\rho_B'}(R) = \rho_A\otimes D_B' +
  D_A\otimes\rho_B'$.  Then
  \begin{align}
    \hspace*{3em}&\hspace*{-3em}
    \Ftwo{\rho_A\otimes\rho_B'}{D_A\otimes\rho_B' + \rho_A\otimes D_B'}
    =
    \operatorname{tr}\bigl( (\rho_A\otimes\rho_B')\, R^2\bigr)
      \nonumber\\
    &=
    \begin{alignedat}[t]{1}
      \operatorname{tr}\Bigl( (\rho_A\otimes\rho_B')
    \Bigl(&\mathds{1}_A\otimes\bigl[\mathcal{R}_{\rho_B'}^{-1}(D_B')\bigr]^2 +
    \bigl[\mathcal{R}_{\rho_A}^{-1}(D_A)\bigr]^2\otimes\mathds{1}_B
    \\
    &+ 2\,\mathcal{R}_{\rho_A}^{-1}(D_A)\otimes\mathcal{R}_{\rho_B'}^{-1}(D_B')
    \Bigr)\Bigr)
  \end{alignedat}
  \nonumber\\
    &= \Ftwo{\rho_A}{D_A} + \Ftwo{\rho_B'}{D_B'}\ ,
\end{align}
where in the last line we have used
$\operatorname{tr}\bigl(\rho_A \mathcal{R}_{\rho_A}^{-1}(D_A)\bigr) = \operatorname{tr}\bigl(D_A -
P_{\rho_A}^\perp D_A P_{\rho_A}^\perp\bigr) = \operatorname{tr}(D_A) = 0$.
\end{proof}

\begin{proposition}[Fisher information for pure states]
  \label{z:Sj25VG-tiqC5}
  Let $\lvert {\psi}\rangle $ be a subnormalized state vector and let $D$ be a Hermitian
  operator satisfying $\operatorname{tr}(D)=0$ and $P_\psi^\perp D P_\psi^\perp = 0$.  Then
  $\langle {D}\rangle _\psi = 0$ and
  \begin{align}
    \Ftwo{\psi}{D} = \frac{1}{(\operatorname{tr}\psi)^2}\,\Bigl[4\operatorname{tr}\bigl(\psi D^2\bigr)\Bigr]\ .
  \end{align}
  Furthermore, if $\operatorname{tr}(\psi)=1$ and $D = -i[H,\psi]$ for some Hermitian operator
  $H$, then
  \begin{align}
    \Ftwo{\psi}{D} = 4\sigma_H^2 = 4\bigl(\langle {H^2}\rangle _\psi - \langle {H}\rangle _\psi^2\bigr)\ .
  \end{align}
\end{proposition}
\begin{proof}[**z:Sj25VG-tiqC5]
  First of all thanks to \cref{z:xptYZVgGge6R} we
  assume without loss of generality that $\operatorname{tr}(\psi)=1$.  Then, to see that
  $\langle {D}\rangle  = 0$ we write
  \begin{align}
    0 = \operatorname{tr}(D) = \operatorname{tr}\bigl[ (\psi + P_\psi^\perp) D \bigr]
    = \langle {D}\rangle  + \operatorname{tr}\bigl[ P_\psi^\perp D P_\psi^\perp \bigr] = \langle {D}\rangle \ .
  \end{align}
  Using~\eqref{z:HWLNvE.Hsjmk} and
  \cref{z:qtGZ5Y7B.DE1}, we then find
  \begin{align}
    \Ftwo{\psi}{D} = \operatorname{tr}\bigl(\psi\,( 2D )^2 \bigr) = 4\operatorname{tr}(\psi D^2)\ .
  \end{align}
  If furthermore $D = -i[H,\psi]$ for some Hermitian $H$, then we
  use~\eqref{z:XrleY6eW0oFd} to see that
  $\operatorname{tr}(\psi D^2) = \langle {H^2}\rangle _\psi - \langle {H}\rangle _\psi^2$.
\end{proof}

\begin{proposition}[Data-processing inequality for the Fisher
  information~\cite{R18}]
  \label{z:x2KvR6Lo6ilw}
  Let $\rho$ be a subnormalized quantum state and $D$ be a Hermitian operator
  that satisfies $P_\rho^\perp D P_\rho^\perp = 0$.  Let $\mathcal{E}$ be any
  completely positive, trace-nonincreasing map.
  Then
  \begin{align}
    \Ftwo{\rho}{D} \geq \Ftwo{\mathcal{E}(\rho)}{\mathcal{E}(D)}\ .
    \label{z:EdcL.q0qiNK-}
  \end{align}
\end{proposition}
\begin{proof}[**z:x2KvR6Lo6ilw]
  First we show that
  $P_{\mathcal{E}(\rho)}^\perp \mathcal{E}(D) P_{\mathcal{E}(\rho)}^\perp = 0$,
  ensuring that the right-hand side
  in~\eqref{z:EdcL.q0qiNK-} is well defined.  Decompose
  $D = P_\rho D P_\rho + P_\rho^\perp D P_\rho + P_\rho D P_\rho^\perp = D_0 +
  D_0^\dagger$, defining $D_0 = (P_\rho D P_\rho)/2 + P_\rho D P_\rho^\perp$
  such that $P_\rho^\perp D_0 = 0$.  For $c>0$ large enough, we have
  $\begin{bsmallmatrix} \rho & D_0 \\ D_0^\dagger & c\mathds{1}\end{bsmallmatrix}
  \geq 0$ thanks to \cref{z:pZMoi26flsRb}.  Applying the completely positive map
  ${\mathrm{id}}_{2}\otimes\mathcal{E}$ we obtain
  $\begin{bsmallmatrix} \mathcal{E}(\rho) & \mathcal{E}(D_0) \\
    \mathcal{E}(D_0^\dagger) & c\mathcal{E}(\mathds{1}) \end{bsmallmatrix} \geq 0 $,
  and therefore thanks to \cref{z:pZMoi26flsRb},
  $P_{\mathcal{E}(\rho)}^\perp \mathcal{E}(D_0) = 0$.  Then
  $P_{\mathcal{E}(\rho)}^\perp \mathcal{E}(D) P_{\mathcal{E}(\rho)}^\perp = 0$
  recalling $D=D_0+D_0^\dagger$.

  Let $S$ be optimal in~\eqref{z:YdgtbMB30Vi6}
  for $\Ftwo{\mathcal{E}(\rho)}{\mathcal{E}(D)}$, thus satisfying
  $\Ftwo{ \mathcal{E}(\rho) }{ \mathcal{E}(D) } =
  4\bigl[\operatorname{tr}\bigl(S\,\mathcal{E}(D)\bigr) - \operatorname{tr}\bigl(\mathcal{E}(\rho) \, S^2\bigr)\bigr]$.
  Choosing the candidate $\mathcal{E}^\dagger(S)$
  in~\eqref{z:YdgtbMB30Vi6} for $\Ftwo{\rho}{D}$ we
  obtain
  \begin{align}
    \Ftwo{\rho}{D}
    &\geq 4\bigl[
    \operatorname{tr}\bigl(D \, \mathcal{E}^\dagger(S)\bigr) - \operatorname{tr}\bigl(\rho \, [\mathcal{E}^\dagger(S)]^2\bigr) \bigr]
      \nonumber\\
    &\geq 4\bigl[
      \operatorname{tr}\bigl(\mathcal{E}(D)\, S\bigr) - \operatorname{tr}\bigl(\rho \, \mathcal{E}^\dagger(S^2)\bigr) \bigr]
      \nonumber\\
    &= 4\bigl[
      \operatorname{tr}\bigl(\mathcal{E}(D)\, S\bigr) - \operatorname{tr}\bigl(\mathcal{E}(\rho) \, S^2\bigr) \bigr]
      \nonumber\\
    &= \Ftwo{\mathcal{E}(\rho)}{\mathcal{E}(D)}\ ,
  \end{align}
  where we have used \cref{z:2McXzF-3arvz} in the second inequality.
\end{proof}

In the case of commuting state and differential, the symmetric logarithmic
derivative reduces to a matrix inverse as described by the following
proposition.
\begin{proposition}[Fisher information for commuting state and derivative]
  \label{z:Pm5buQF.JQ.f}
  Let $\rho$ be any subnormalized quantum state and let $D$ be a Hermitian
  operator that satisfies $P_\rho^\perp D P_\rho^\perp = 0$.  Suppose that
  $\rho$ and $D$ commute.  Then
  \begin{align}
    \Ftwo{\rho}{D} = \operatorname{tr}\bigl( \rho^{-1} D^2\bigr)\ .
    \label{z:KQz7f4Ou6xIF}
  \end{align}
\end{proposition}
\begin{proof}[**z:Pm5buQF.JQ.f]
  This can be shown from the properties of the symmetric logarithmic derivative,
  but we give a simple alternative proof using our convex optimizations for fun.
  Choose $S = \rho^{-1} D / 2$
  in~\eqref{z:YdgtbMB30Vi6}, which we note is a
  Hermitian operator because $\rho$ and $D$ commute.  This gives
  $\Ftwo{\rho}{D} \geq \text{\eqref{z:KQz7f4Ou6xIF}}$.
  Similarly, the choice $L = \rho^{-1/2} D / 2$
  in~\eqref{z:V4lWxgFVfrPu} provides the opposite
  bound.
\end{proof}

The following proposition interprets the Fisher information for subnormalized
states according to the definition~\eqref{z:.gRi0KAw40V9} as the Fisher
information of a normalized state that was projected onto a smaller subspace.
This interpretation works as long as the subnormalized state does not change
trace along its evolution, meaning that the derivative $D$ has zero trace.

\begin{proposition}[Fisher information for subnormalized and normalized states]
  \label{z:ShHtwoFJPzYY}
  Let $\rho$ be any subnormalized quantum state and let $D$ be any Hermitian
  operator that satisfies both $\operatorname{tr}(D)=0$ and $P_\rho^\perp D P_\rho^\perp = 0$.
  Define $\rho',D'$, with an additional new Hilbert space dimension, as
  \begin{align}
    \rho' &= \begin{pmatrix}
      ~{\fboxsep=1.5em\fbox{\(\rho\)}}
 & 0\\ 0 & 1-\operatorname{tr}(\rho) \end{pmatrix}\ ;
    & %
    D' &= \begin{pmatrix}
      ~{\fboxsep=1.5em\fbox{\(D\)}}
 & 0\\ 0 & 0~ \end{pmatrix}\ .
  \end{align}
  Then
  \begin{align}
    \Ftwo{\rho}{D} = \Ftwo{\rho'}{D'}\ .
  \end{align}
\end{proposition}
\begin{proof}[**z:ShHtwoFJPzYY]
  Let $P$ denote the projector onto the subspace of the Hilbert space on which
  the upper left block of $\rho',D'$ acts.  %
  Let $R = \mathcal{R}_\rho^{-1}(D)$, and define
  \begin{align}
    R' = \begin{pmatrix}
      ~{\fboxsep=1.5em\fbox{\(R\)}}
      & 0\\ 0 & 0~ \end{pmatrix}\ .
  \end{align}
  Multiplying together block-diagonal matrices preserves the block-diagonal
  structure, hence
  \begin{align}
    \frac12\{\rho', R'\}
    &= \begin{pmatrix}
      ~{\fboxsep=1.5em\fbox{\(\displaystyle \frac12\{\rho,R\}\)}}
      & 0\\ 0 & 0~ \end{pmatrix}
    = D'\ .
  \end{align}
  Then with the definition~\eqref{z:.gRi0KAw40V9},
  \begin{align}
    \Ftwo{\rho'}{D'} = \operatorname{tr}\bigl(\rho' R'^2\bigr) = \operatorname{tr}\bigl(\rho R^2\bigr) = \Ftwo{\rho}{D}\ .
    \tag*\qedhere
  \end{align}
\end{proof}

We can furthermore prove a relation between the Fisher information of two
different directions in state space that might be associated with two different
parametrized evolutions.

\begin{proposition}[Relation between the Fisher information of two directions]
  \label{z:wXcXJs-gCbik}
  Let $\rho$ be a subnormalized quantum state and let $D,D'$ be two Hermitian
  operators that satisfy
  $P_\rho^\perp D P_\rho^\perp = P_\rho^\perp D' P_\rho^\perp = 0$.  Then
  \begin{align}
    \Ftwo{\rho}{D}
    \leq
    \Ftwo{\rho}{D'}
    + \Bigl[\Ftwo{\rho}{D+D'} \,
    \Ftwo{\rho}{D-D'}\Bigr]^{1/2}\ .
    \label{z:22LPlD1sh7iv}
  \end{align}
  Consequently,
  \begin{align}
    \bigl \lvert {\Ftwo{\rho}{D} - \Ftwo{\rho}{D'}}\bigr \rvert 
    \leq \Bigl[\Ftwo{\rho}{D+D'} \, \Ftwo{\rho}{D-D'}\Bigr]^{1/2}\ .
    \label{z:Z82-JAjXWpU5}
  \end{align}
  Furthermore, equality holds
  in~\eqref{z:Z82-JAjXWpU5} if and
  only if $D, D'$ are linearly dependent.
\end{proposition}
\begin{proof}[**z:wXcXJs-gCbik]
  Define the shorthand $\Delta_\pm = D \pm D'$.  We
  compute
  \begin{align}
    \Ftwo{\rho}{D}
    &= \operatorname{tr}\Bigl( \rho \,  \Bigl[\mathcal{R}_\rho^{-1}\bigl(D'\bigr) +
      \mathcal{R}_\rho^{-1}\bigl(\Delta_-\bigr)\Bigr]^2 \Bigr)
      \nonumber\\
    &= \Ftwo{\rho}{D'} +
      \operatorname{tr}\Bigl( \rho \,  \Bigl[
      \Bigl(\mathcal{R}_\rho^{-1}\bigl(\Delta_-\bigr)\Bigr)^2 +
      \Bigl\{ \mathcal{R}_\rho^{-1}\bigl(D'\bigr) \, , \,
      \mathcal{R}_\rho^{-1}\bigl(\Delta_-\bigr) \Bigr\} \Bigr]\Bigr)
      \nonumber\\
    &= \Ftwo{\rho}{D'} +
      \operatorname{tr}\Bigl( \rho \,
      \Bigl\{ \frac12\mathcal{R}_\rho^{-1}\bigl(\Delta_-\bigr) +
      \mathcal{R}_\rho^{-1}\bigl(D'\bigr) \, , \,
      \mathcal{R}_\rho^{-1}\bigl(\Delta_-\bigr) \Bigr\} \Bigr)\ ,
      \label{z:L2FRyKxEq3Nv}
  \end{align}
  where in the last equality we have used ${M}^2 = \{\frac12 {M},{M}\}$ for any operator ${M}$
  along with the linearity of the anticommutator in the first argument.
  Furthermore, we see from the definition of  $\Delta_-$ that
  \begin{align}
    D' + \frac12 \Delta_-
    = \frac12\bigl(D + D'\bigr) = \frac12\Delta_+\ .
  \end{align}
  Then
  \begin{align}
    \text{\eqref{z:L2FRyKxEq3Nv}}
    &= \Ftwo{\rho}{D'} +
      \frac12 \operatorname{tr}\Bigl( \rho \,
      \Bigl\{ \mathcal{R}_\rho^{-1}\bigl(\Delta_+\bigr) \, , \,
      \mathcal{R}_\rho^{-1}\bigl(\Delta_-\bigr) \Bigr\} \Bigr)
      \nonumber\\
    &= \Ftwo{\rho}{D'} +
      \operatorname{Re}\operatorname{tr}\Bigl( \rho \;
      \mathcal{R}_\rho^{-1}\bigl(\Delta_+\bigr) \,
      \mathcal{R}_\rho^{-1}\bigl(\Delta_-\bigr)  \Bigr)
      \nonumber\\
    &=: \Ftwo{\rho}{D'} + C_\rho\bigl(\Delta_+,\Delta_-\bigr)\ ,
      \label{z:zQH4pqxbUtSJ}
  \end{align}
  where $C_\rho\bigl(\Delta_+,\Delta_-\bigr)$ is defined as the second term in the
  above expression.  From the Cauchy-Schwarz inequality,
  \begin{align}
    \bigl \lvert {C_{\rho}\bigl(\Delta_+,\Delta_-\bigr)}\bigr \rvert ^2
    &\leq
      \operatorname{tr}\Bigl(\rho\Bigl[ \mathcal{R}_\rho^{-1}\bigl(\Delta_+\bigr) \Bigr]^2 \Bigr)
      \operatorname{tr}\Bigl(\rho\Bigl[ \mathcal{R}_\rho^{-1}\bigl(\Delta_-\bigr) \Bigr]^2 \Bigr)
      \ .
  \end{align}
  Hence
  \begin{align}
    \text{\eqref{z:zQH4pqxbUtSJ}}
    &\leq \Ftwo{\rho}{D} + \Bigl[
      \Ftwo{\rho}{\Delta_+} \, \Ftwo{\rho}{\Delta_-} \Bigr]^{1/2}\ .
  \end{align}
  \Cref{z:Z82-JAjXWpU5} follows by
  repeating the argument while inverting the roles of $D$ and $D'$.

  Equality in~\eqref{z:Z82-JAjXWpU5}
  is equivalent to the Cauchy-Schwarz inequality being tight.  In turn is
  equivalent to the operators $\rho^{1/2}\mathcal{R}_\rho^{-1}\bigl(\Delta_+\bigr)$
  and $\rho^{1/2}\mathcal{R}_\rho^{-1}\bigl(\Delta_-\bigr)$ being linearly dependent,
  i.e., there exist $\alpha_1, \alpha_2\in\mathbb{R}$,
  $(\alpha_1,\alpha_2)\neq(0,0)$, such that
  \begin{align}
    \alpha_1 \rho^{1/2}\mathcal{R}_\rho^{-1}\bigl(\Delta_+\bigr)
    + \alpha_2 \rho^{1/2}\mathcal{R}_\rho^{-1}\bigl(\Delta_-\bigr)
    = 0\ .
    \label{z:HkMNBKMKD7MB}
  \end{align}
  Since the operator $\mathcal{R}_\rho^{-1}\bigl(\Delta_\pm\bigr)$ vanishes on the
  operator subspace spanned by $P_\rho^\perp (\cdot) P_\rho^\perp$, we have
  that~\eqref{z:HkMNBKMKD7MB} is equivalent to
  \begin{align}
    \alpha_1 \mathcal{R}_\rho^{-1}\bigl(\Delta_+\bigr)
    + \alpha_2 \mathcal{R}_\rho^{-1}\bigl(\Delta_-\bigr)
    = 0\ ,
  \end{align}
  and therefore to
  \begin{align}
    \mathcal{R}_\rho^{-1}\Bigl[
        \alpha_1 \Delta_+ + \alpha_2 \Delta_-
    \Bigr] = 0\ .
    \label{z:6xbhKzahB-Vy}
  \end{align}
  Because the kernel of the superoperator $\mathcal{R}_\rho^{-1}$ is spanned by
  $P_\rho^{\perp}\,(\cdot)\,P_\rho^{\perp}$, onto which $\Delta_\pm$ have no
  support by assumption, then \eqref{z:6xbhKzahB-Vy} is further
  equivalent to
  \begin{align}
    \alpha_1 \Delta_+ + \alpha_2 \Delta_- = 0\ .
  \end{align}
  Therefore, equality in
  \eqref{z:Z82-JAjXWpU5} is achieved
  if and only if $\Delta_\pm$ are linearly dependent, which is equivalent to the
  linear dependence of $D$ with $D'$.
\end{proof}

Using a similar idea, we can also prove a continuity bound on the Fisher
information with respect to its second argument.

\begin{proposition}[A continuity bound of the Fisher information in its second
  argument]
  \label{z:Y81L4FyklNd5}
  Let $\rho$ be any subnormalized quantum state and let $D,\Delta$ be any
  Hermitian operators such that
  $P_\rho^\perp D P_\rho^\perp = 0 = P_\rho^\perp \Delta P_\rho^\perp$.  Then
  \begin{align}
    \Bigl \lvert { \Ftwo{\rho}{D+\Delta} - \Ftwo{\rho}{D} - \Ftwo{\rho}{\Delta} }\Bigr \rvert 
    \leq
    2 \mathopen{}\left[ \Ftwo{\rho}{D} \Ftwo{\rho}{\Delta} \right]\mathclose{}^{1/2}\ .
  \end{align}
  As a consequence,
  \begin{align}
    \Bigl \lvert { \Ftwo{\rho}{D+\Delta} - \Ftwo{\rho}{D} }\Bigr \rvert 
    \leq
    \Ftwo{\rho}{\Delta} + 2 \mathopen{}\left[ \Ftwo{\rho}{D} \Ftwo{\rho}{\Delta} \right]\mathclose{}^{1/2}\ .
  \end{align}
\end{proposition}
\begin{proof}[**z:Y81L4FyklNd5]
  Using the formula $\Ftwo{\rho}{D'} = \operatorname{tr}(D'\,\mathcal{R}_\rho^{-1}(D'))$ for
  the Fisher information, we write
  \begin{align}
    \Ftwo{\rho}{D+\Delta}
    &= \operatorname{tr}\mathopen{}\left( (D+\Delta)\, \mathcal{R}_\rho^{-1}(D+\Delta) \right)\mathclose{}
    \nonumber\\
    &= \Ftwo{\rho}{D} + \Ftwo{\rho}{\Delta} + 2\operatorname{tr}\bigl( D\, \mathcal{R}_\rho^{-1}(\Delta)\bigr)\ ,
  \end{align}
  recalling that $\mathcal{R}_\rho^{-1}$ is superoperator self-adjoint.  The
  claim follows by bounding the last term in the above expression using the
  Cauchy-Schwarz inequality, to get
  \begin{align}
    \mathopen{}\left \lvert { \operatorname{tr}\bigl( D\, \mathcal{R}_\rho^{-1}(\Delta)\bigr) }\right \rvert \mathclose{}
    &\leq
      \sqrt{ \operatorname{tr}\bigl(D\,\mathcal{R}_\rho^{-1}(D)\bigr) \operatorname{tr}\bigl(\Delta\,\mathcal{R}_\rho^{-1}(\Delta)\bigr) }
    \nonumber\\
    &= \sqrt{ \Ftwo{\rho}{D} \Ftwo{\rho}{\Delta} }\ .
      \tag*\qedhere
  \end{align}
\end{proof}

We can consider more precisely how $\FIqty{Bob}{t}$ behaves when seen
as a function of the noise channel $\mathcal{N}$, for channels $\mathcal{N}$
that are close to the identity channel ${\mathrm{id}}$.
More specifically, we prove a continuity bound for the quantum Fisher
information $\Ftwo{\mathcal{N}(\psi)}{\mathcal{N}(\partial_t \psi)}$ at the
point $\mathcal{N}={\mathrm{id}}$, when that quantity is seen as a function of
$\mathcal{N}$.

\begin{proposition}
  \label{z:-97yhn1ACkRa}
  Let $\lvert {\psi}\rangle $ be a pure state and let $D$ be a Hermitian operator such that
  $\langle {D}\rangle _\psi = 0$ and $P_\psi^\perp D P_\psi^\perp = 0$.  Let $\epsilon>0$
  and let $\mathcal{N}$ be a channel with
  $\lVert { \mathcal{N} - {\mathrm{id}} }\rVert _\diamond \leq \epsilon$.  Then
  \begin{align}
    F\bigl(\psi, D\bigr)
    &\geq
    F\bigl(\mathcal{N}(\psi), \mathcal{N}(D)\bigr)
      \geq
    F\bigl(\psi, D\bigr)
    - 8\epsilon \lVert {D}\rVert _1\lVert {D}\rVert _\infty\ .
  \end{align}
\end{proposition}
Observe that the stated conditions on $D$ are satisfied if $D = -i[H,\psi]$ for
some Hermitian operator $H$.

\begin{proof}[**z:-97yhn1ACkRa]
  Let $\epsilon>0$ and let $\mathcal{N} = {{\mathrm{id}}} + \Delta$ where
  $\Delta$ is a Hermiticity-preserving superoperator with
  $\lVert {\Delta}\rVert _\diamond \leq \epsilon$.  The first claimed inequality
  immediately follows from the data-processing inequality.  We now prove the
  second inequality.  Using \cref{z:qtGZ5Y7B.DE1}, let
  $S = \frac12 \mathcal{R}_{\psi}^{-1}(D) = D$.  Since this $S$ is known to be
  optimal in \cref{z:YdgtbMB30Vi6} for
  $\Ftwo{\psi}{D}$, we can compute
  \begin{align}
    \Ftwo{\psi}{D}
    & = 4 \Bigl\{ \operatorname{tr}\bigl[(D) \, S\bigr] - \operatorname{tr}\bigl[\psi\, S^2\bigr] \Bigr\}
      \nonumber\\
    &= 4 \Bigl\{ \operatorname{tr}\bigl[\mathcal{N}(D) \, S\bigr] - \operatorname{tr}\bigl[\mathcal{N}(\psi)\, S^2\bigr] \Bigr\}
      - 4 \Bigl\{ \operatorname{tr}\bigl[\Delta(D) \, S\bigr] - \operatorname{tr}\bigl[\Delta(\psi)\, S^2\bigr] \Bigr\}
      \nonumber\\
    &\leq F\bigl( \mathcal{N}(\psi) , \mathcal{N}(D) \bigr)
      + 4 \bigl \lVert {\Delta(D)}\bigr \rVert _1\lVert {S}\rVert _\infty
      + 4\bigl \lVert {\Delta(\psi)}\bigr \rVert _1 \lVert { S}\rVert _\infty^2
      \nonumber\\
    &\leq F\bigl( \mathcal{N}(\psi) , \mathcal{N}(D) \bigr)
      + 8\epsilon \lVert {D }\rVert _1 \lVert {D}\rVert _\infty\ ,
  \end{align}
  using $\lVert {D}\rVert _\infty \leq \lVert {D}\rVert _1$, and
  thus proving the claim.
\end{proof}

\section{Optimal local-sensing and the Cram\'er-Rao bound}
\label{z:9mLm-0LXJIol}
\label{z:eQl8.oS.4.na}

Here we review which operators achieve the optimal variance in estimating an
unknown parameter~\cite{R22,R23,R0,R11,R88}.  An unknown parameter $t$ of an evolution
$\rho_t$ of a (normalized) quantum state is estimated locally around $t_0$ using
an observable $T$, whose measurement outcomes are the estimates of the
parameter.  We ask for the observable to have the correct average and first
order deviation, $\langle {T}\rangle _{\rho_{t_0+dt}} = t_0 + dt + O(dt^2)$; except in edge
cases, this condition can be enforced by a suitable scaling factor and a
suitable shift by the identity.  The conditions then become
$\langle {T}\rangle _{\rho_{t_0}} = t_0$ and $\operatorname{tr}\bigl\{ (\partial_t \rho_t|_{t_0})\, T\bigr\} = 1$.
We seek to minimize the operator $T$'s variance
$\langle {T^2}\rangle _{\rho_{t_0}} - \langle {T}\rangle _{\rho_{t_0}}^2$.  We call such an operator
with minimal variance an \emph{optimal local-sensing operator}, and the square
root of the minimal variance is the \emph{optimal estimation error}
$\Delta t_{\mathrm{unc}}(t_0)$ locally at $t_0$.  That is, the optimal
estimation error locally at $t_0$, along with an optimal local-sensing operator
at $t_0$, are given by the following optimization problem:
\begin{align}
  \Delta t_{\mathrm{unc}}^2(t_0)
  = \begin{array}[t]{rl}
      \min\limits_{T=T^\dagger}
      &\quad \operatorname{tr}\bigl\{\rho_{t_0} (T - t_0\mathds{1})^2\bigr\}, \\
      \mathrm{s.t.}
      &\quad \operatorname{tr}\bigl\{ \rho_{t_0} T\bigr\} = t_0\ ,\ \operatorname{tr}\bigl\{ (\partial_t\rho_t|_{t_0})\, T\bigr\} = 1\ .
    \end{array}
  \label{z:URKHO5wxhm5z}
\end{align}
In the event that $\partial_t \rho_t |_{t_0} = 0$, there is no operator $T$ that
satisfies the given conditions.  We conventionally set
$\Delta t_{\mathrm{unc}} = \infty$, since the state is locally stationary and no
observable is able to detect a first-order deviation in the parameter $t$.

A more general scheme would enable an agent to use a generalized measurement
given by a POVM instead of an observable $T$.  However, as shown in e.g.\@
Ref.~\cite{R0}, the optimal POVM can in fact be chosen
to be a projective measurement.  Therefore one cannot sense a parameter more
accurately using a POVM instead of an observable.

The following proposition fully characterizes the locally optimal sensing
observables (cf., e.g.,~\cite{R0}).  In the following,
we write as a shorthand $\rho$ and $\partial_t\rho$ instead of $\rho_{t_0}$ and
$\partial_t \rho_t |_{t_0}$.

\begin{proposition}[Locally optimal sensing]
  \label{z:ZoysZzb.FHNr}
  Assume $\partial_t\rho\neq0$.  Then any operator $T$ that is optimal
  in~\eqref{z:URKHO5wxhm5z} is of the
  form
  \begin{align}
    T = t\mathds{1}+
    (\Delta t_{\mathrm{unc}}^2) \, 
    \mathcal{R}_\rho^{-1}\bigl(\partial_t\rho\bigr)
    + P_\rho^\perp {M} P_\rho^\perp\ ,
  \end{align}
  for some Hermitian operator ${M}$.
  
  If $P_\rho^\perp (\partial_t\rho) P_\rho^\perp = 0$, then
  $\Delta t_{\mathrm{unc}}^2 = [ \Ftwo{\rho}{\partial_t\rho} ]^{-1}$ with the
  Fisher information defined in~\eqref{z:.gRi0KAw40V9}, and ${M}$ can be
  arbitrary.

  If $P_\rho^\perp (\partial_t\rho) P_\rho^\perp \neq 0$, then
  $\Delta t_{\mathrm{unc}}^2 = 0$ and ${M}$ satisfies
  $\operatorname{tr}\bigl({M}\,P_\rho^\perp\partial_t\rho P_\rho^\perp\bigr) = 1$.
\end{proposition}

Let us further note that if $\partial_t\rho = 0$, we have
$F(\rho; \partial_t\rho) = 0$.  Therefore, provided that
$P_\rho^\perp \partial_t\rho P_\rho^\perp = 0$, we can in full generality write
\begin{align}
  \Delta t_{\mathrm{unc}}^2 = \frac1{ \Ftwo{\rho}{\partial_t\rho} }\ ,
\end{align}
along with the convention that $\Delta t_{\mathrm{unc}}=\infty$ if
$\Ftwo{\rho}{\partial_t\rho} = 0$.
In our setting, the optimal sensing scheme always achieves the value of the
Cram\'er-Rao bound.

\begin{proof}[*z:ZoysZzb.FHNr]
  Without loss of generality, we assume $t_0 = 0$ throughout this proof; this is
  achieved by shifting the parameter to center it at zero, implying the
  corresponding shift $T\to T' = T - t_0\mathds{1}$.
  We thus consider the optimization problem
  \begin{align}
    \Delta t_{\mathrm{unc}}^2
  = \begin{array}[t]{rl}
      \min\limits_{T=T^\dagger}
      &\quad \operatorname{tr}\bigl\{\rho_{t_0} T^2\bigr\}, \\
      \mathrm{s.t.}
      &\quad \operatorname{tr}\bigl\{ \rho_{t_0} T\bigr\} = 0\ ,\ \operatorname{tr}\bigl\{ (\partial_t\rho_t|_{t_0})\, T\bigr\} = 1\ .
    \end{array}
  \end{align}

  First of all we observe that the first condition, $\operatorname{tr}(\rho T) = 0$, can be
  ignored without changing the optimal value of the problem.  Indeed, for any
  $T$ that satisfies $\operatorname{tr}\bigl( (\partial_t\rho) T\bigr) = 1$ but with
  $\operatorname{tr}(\rho T)\neq 0$, we can define $T' = T - \operatorname{tr}(\rho T)\,\mathds{1}$, with
  $\operatorname{tr}(\rho T') = 0$ and
  $\operatorname{tr}\bigl( (\partial_t\rho)\,T \bigr) = \operatorname{tr}\bigl( (\partial_t\rho)\,T' \bigr)$ since
  $\operatorname{tr}(\partial_t\rho) = \partial_t \operatorname{tr}(\rho) = 0$; then
  $\operatorname{tr}(\rho T'^2) = \operatorname{tr}(\rho T^2) - [\operatorname{tr}(\rho T)]^2 \leq \operatorname{tr}(\rho T^2)$,
  meaning that $T'$ not only satisfies $\operatorname{tr}(\rho T') = 0$ in addition to the
  other condition, but it achieves a better objective function value.  
  
  We can recast this optimization as semidefinite problem, following
  Refs.~\cite{R88,R106}, by using
  Schur complements (\cref{z:pZMoi26flsRb}):
  \begin{align}
    \Delta t_{\mathrm{unc}}^2
    = \begin{array}[t]{rl}
        \min\limits_{Q\geq0,\ T=T^\dagger} &\quad \operatorname{tr}(\rho Q) \\
        \mathrm{s.t.:} &\quad \operatorname{tr}\bigl( (\partial_t\rho)\, T\bigr) = 1;\\
                                     &\quad \begin{bmatrix}Q & -T\\ -T & \mathds{1}\end{bmatrix} \geq 0\ .
      \end{array}
  \end{align}

  The associated dual problem takes the following form, noting that strong duality
  holds thanks to Slater's
  conditions~\cite{R105,R75}.
  \begin{align}
    \Delta t_{\mathrm{unc}}^2
    &= \begin{array}[t]{rl}
         \max\limits_{A,C\geq0,\ B\,\mathrm{arb.},\ \mu\in\mathbb{R}}
         &\quad  \mu - \operatorname{tr}(C) \\
         \mathrm{s.t.:} 
         &\quad A \leq \rho\\
         &\quad B + B^\dagger = \mu\,\partial_t\rho\\
         &\quad \begin{bmatrix}A & B\\ B^\dagger & C \end{bmatrix} \geq 0
       \end{array}
                                                   \nonumber\\
    &= \begin{array}[t]{rl}
         \max\limits_{B\,\mathrm{arb.},\ \mu\in\mathbb{R}}
         &\quad  \mu - \operatorname{tr}(B^\dagger\rho^{-1}B) \\
         \mathrm{s.t.:} &\quad B + B^\dagger = \mu\,\partial_t\rho\\
         &\quad P_\rho B = B
       \end{array}
           \label{z:WymaBHSPBQn6}\\
    &= \begin{array}[t]{rl}
         \max\limits_{L\,\mathrm{arb.},\ \mu\in\mathbb{R}}
         &\quad  \mu - \operatorname{tr}(L^\dagger L) \\
         \mathrm{s.t.:} &\quad \rho^{1/2}L + L^\dagger\rho^{1/2} = \mu\,\partial_t\rho\ ,\\
       \end{array}
    \label{z:aoMRLgpSr1sD}
  \end{align}
  using again Schur complements and where we introduced the variable $L$ via
  $B = \rho^{1/2}L$, and where $P_\rho = \mathds{1}- P_\rho^\perp$ is the projector
  onto the support of $\rho$.

A powerful characterization of the whole family of optimal solutions to a
semidefinite problem with strong duality are the complementary slackness
relations.  An inequality constraint multiplied by the corresponding dual
variable becomes an equality for any choice of primal and dual optimal
solutions~\cite{R75,R105}.  Here, this means
that
\begin{align}
\left[\begin{smallmatrix}Q & -T\vphantom{{}^\dagger} \\ -T & \mathds{1}\end{smallmatrix}\right]
\left[\begin{smallmatrix}A & B \\ B^\dagger & C\end{smallmatrix}\right]=0.
\end{align}
This
gives us the following relations that must be satisfied for any choice of
optimal variables:
\begin{align}
    Q\rho &= TB^\dagger\ ; &    QB &=TC\ ; &
    B^\dagger &= T\rho\ ; &   C &= TB\ .
\end{align}
The third equality ($B^\dagger=T\rho$) along with the dual constraint
in~\eqref{z:WymaBHSPBQn6} implies that
$\rho T + T\rho = \mu\,\partial_t\rho$.
\Cref{z:uWDssQGxkLXS} asserts that the solutions are
necessarily of the form
$T = (\mu/2)\mathcal{R}_\rho^{-1}\bigl(\partial_t\rho\bigr) + P_\rho^\perp {M}
P_\rho^\perp$ for some Hermitian ${M}$.

Now first suppose that $P_\rho^\perp \,(\partial_t\rho)\, P_\rho^\perp = 0$.  The
primal value achieved for a $T$ of this form, and for any $\mu$ and ${M}$, is
\begin{align}
  \text{primal achieved}
  = \operatorname{tr}(\rho T^2) = \frac{\mu^2}{4} F(t)\ ,
\end{align}
with $F(t)$ as in~\eqref{z:.eb-akuquBNd}.  From complementary slackness
we have $B^\dagger=T\rho$ and hence
$\operatorname{tr}(B^\dagger\rho^{-1}B) = \operatorname{tr}(\rho T^2) = \mu^2 F(t)/4$.  The dual problem
therefore reaches the value 
\begin{align}
  \text{dual achieved} = \mu- \mu^2 F(t)/4\ .
\end{align}
Optimality implies that the primal and dual values are equal,
$\mu^2 F(t)/4 = \mu - \mu^2 F(t)/4$ and therefore $\mu = 2/F(t)$ (note $\mu=0$
is ruled out because the primal constraint $\operatorname{tr}\bigl((\partial_t\rho)\,T\bigr)=1$
would be impossible to satisfy).  Therefore the optimal solution to the problem
is
\begin{align}
  \Delta t_{\mathrm{unc}}^2 = \frac1{F(t)}\ .
\end{align}

Now suppose that $P_\rho^\perp (\partial_t\rho) P_\rho^\perp \neq 0$.
Then there cannot be any solution for $L$ in the constraint
in~\eqref{z:aoMRLgpSr1sD} unless $\mu=0$ (the left-hand side vanishes
entirely if we hit it with $P_\rho^\perp (\cdot) P_\rho^\perp$ but not the right-hand
 side if $\mu\neq0$).  Then $T=P_\rho^\perp X P_\rho^\perp$, which implies
$\operatorname{tr}(\rho T^2)=0$, and furthermore ${M}$ must satisfy
$\operatorname{tr}\bigl((\partial_t\rho) P_\rho^\perp {M} P_\rho^\perp\bigr)=1$ from the primal
constraint. The dual candidate $L=0$ yields objective value of zero in the dual
problem, and therefore the optimal value of the optimization problem is zero,
$\Delta t_{\mathrm{unc}}^2 = 0$.
\end{proof}

\section{Proof of the sensitivity uncertainty relation}
\label{z:.B4MZrnrq.9S}

The goal of this section is to prove the statements made in
\cref{z:6pGtSPe4CZWB}.  The setting is the one introduced
in \cref{z:CKirppebju2f}.  %
We provide two independent proofs of the uncertainty relation.  The first proof
is more intuitive and straightforward.  The second proof is slightly more
general and provides greater insight into some technicalities that underpin the
uncertainty relation. The second proof directly relates the semidefinite
characterizations of the quantities $\FIqty{Bob}{t}$ and
$\FIqty{Eve}{\eta}$, making it easier to analyze edge cases, to gain
insight on what choices of semidefinite variables are optimal, and to consider
the more general situation where $\mathcal{N}$ is a trace-nonincreasing map.

\subsection{Proof via the second-order expansion of the fidelity}

The strategy of our first proof of our uncertainty relation is to provide a
direct proof of the statement presented as
\cref{z:dq0nj3WKijUb}; we have already seen in the main
text that the statement in \cref{z:T-KNuYhGRM7I} is
equivalent to \cref{z:dq0nj3WKijUb}.

First observe that without loss of generality, we can assume that the
Hamiltonian is time independent.  This is because the Fisher information depends
only on the state and its local time derivative at $t$, which is given
by \cref{z:bs.I9ZrrQvK2} and depends only on the value of the
Hamiltonian at the fixed value $t$ of interest.

Our proof proceeds in a similar fashion to that of the \emph{channel-extension
  bound} developed in Refs.~\cite{R77,R29,R30}.  %
While our uncertainty relation could also be derived from the results in those
references, we provide a self-contained proof for completeness and consistency
of notation.

A remarkable property of the Fisher information is that it is directly related
to the Bures distance and the fidelity of quantum
states~\cite{R0,R107,R11,R24}
{according to}
\begin{align}
  \FIqty{Bob}{t}
  &= -4 \left.\frac{d^2}{dt'^2}\right|_{t'=t}\,
    F(\rho_{B}(t), \rho_{B}(t'))\ ,
    \label{z:eFGTLHPb0smy}
\end{align}
where
$F(\rho,\rho') = \onenorm{\rho^{1/2}\rho'^{1/2}} =
\operatorname{tr}\bigl[(\rho^{1/2}\rho'\rho^{1/2})^{1/2}\bigr]$ is the root fidelity between
two quantum states~\cite{R62}, where $\onenorm{A}$ denotes 
trace norm, i.e., the
sum of the singular values of $A$.  Note that at $t'=t$, the fidelity reaches
its maximum value $1$.  We assume that $\rho(t)$ is does not change rank at
$t'=t$, avoiding edge cases where the
expression~\eqref{z:eFGTLHPb0smy} is
incomplete~\cite{R31,R32,R33}.

By Uhlmann's theorem, and writing
$\lvert {\rho(t)}\rangle _{BE} = V_{A\to BE}\lvert {\psi(t)}\rangle _A$ in terms of the Stinespring
dilation $V_{A\to BE}$ of $\mathcal{N}$ given
in~\eqref{z:1i4Ojzjjs9zf}, we have that
\begin{align}
  F(\rho_{B}(t), \rho_{B}(t'))
  = \max_{W_E\text{ unitary}} \,
  \bigl \lvert {\langle {\rho(t')}\rvert _{BE} \, W_E\, \lvert {\rho(t)}\rangle _{BE} }\bigr \rvert \ ,
\end{align}
where $W_E$ is a unitary operation on $E$.  We therefore have the following
equivalent expressions:
\begin{subequations}
  \begin{align}
    F(\rho_{B}(t), \rho_{B}(t'))
    &= \max_{W_E}\,
      \operatorname{Re}\, \langle {\psi(t')}\rvert _{A} \, V^\dagger \, W_E\, V_{A\to BE} \, \lvert {\psi(t)}\rangle _{A}
      \\
    &= \max_{W_E}\,
      \operatorname{Re}\, \langle {\psi(t')}\rvert _A \, {\widehat{\mathcal{N}}}^\dagger(W_E) \,
      \lvert {\psi(t)}\rangle _A
      \label{z:Vq56q8ktxGPs}
    \\
    &= \max_{W_E}\,
      \operatorname{Re}\, \langle {\psi(t)}\rvert _A \, {e}^{iH(t'-t)} {\widehat{\mathcal{N}}}^\dagger(W_E) \,
      \lvert {\psi(t)}\rangle _A
      \label{z:BVHbTNZ44pG0}
    \\
    &= \max_{W_E}\,
      \operatorname{Re}\operatorname{tr}\bigl( \widehat{\mathcal{N}}\bigl(\psi_A(t)\, {e}^{iH(t'-t)}\bigr)\, W_E \bigr)
      \\
    &= \onenorm[\big]{\widehat{\mathcal{N}}\bigl(\psi {e}^{iH(t'-t)}\bigr)}\ ,
  \end{align}
\end{subequations}
where the complementary channel $\widehat{\mathcal{N}}$ is given by~\eqref{z:L-vRq0ojGGDS}.
In the above expressions, the maximization can be taken over operators $W_E$
that are unitary, or equivalently, it can be relaxed to all operators $W_E$
satisfying $\opnorm{W_E}\leq 1$.

The optimal unitary $W_E$ is given by the polar decomposition of the operator
$\widehat{\mathcal{N}}\bigl(\psi {e}^{iH(t'-t)}\bigr)$.  For $t' = t + dt$ with a
small $dt$, we have that the optimal $W_E$ is close to the identity, which is
the optimal for $t'=t$.  Let us expand
$W_E = \mathds{1}-i dt S - (1/2) dt^2 S_2 + O(dt^3)$ for general matrices $S$ and
$S_2$ to be determined.  The unitary constraint $W_E^\dagger W_E = \mathds{1}_E$ for
all $dt$ implies that $S=S^\dagger$ and that
$S_2 + S_2^\dagger = 2S^\dagger S = 2S^2$.  Starting
from~\eqref{z:BVHbTNZ44pG0} and expanding up to
order $dt^2$ we find
\begin{align}
  \hspace*{1em}&\hspace*{-1em}
  F(\rho_B(t), \rho_B(t'))
                 \nonumber\\
  & = \max_{S=S^\dagger,\;S_2} \operatorname{Re}\operatorname{tr}\mathopen{}\left\{
    \psi \mathopen{}\left(\mathds{1}+ idt H - \frac{H^2}{2} dt^2\right)\mathclose{}\,
    {\widehat{\mathcal{N}}}^\dagger\mathopen{}\left(\mathds{1}- i dt S - \frac{S_2}{2} dt^2\right)\mathclose{}
    \right\}\mathclose{} + O(dt^3)
    \nonumber\\
  \begin{split}
    &=
    1 + \max_{S=S^\dagger,\;S_2} \Bigl\{
      dt\operatorname{Re}\operatorname{tr}\mathopen{}\left[i\psi H - i\psi \widehat{\mathcal{N}}^\dagger(S)\right]\mathclose{}
      \\
      &\hspace*{7em}
        + dt^2 \operatorname{Re}\operatorname{tr}\Bigl[ - \frac12 \psi H^2 -
        \frac12 \psi \widehat{\mathcal{N}}^\dagger(S_2) +
        \psi H \widehat{\mathcal{N}}^\dagger(S) \Bigr]
        + O(dt^3)
        \Bigr\}
    \end{split}
        \nonumber\\
  &= 1 + \frac{dt^2}{2} \max_{S=S^\dagger,\;S_2} \operatorname{Re}\operatorname{tr}\Bigl\{
    -  \psi H^2 -
     \psi \widehat{\mathcal{N}}^\dagger(S_2) +
    2\psi H \widehat{\mathcal{N}}^\dagger(S) \Bigr\}
    + O(dt^3)\ ,
    \label{z:u.dlhHNOgsWs}
\end{align}
recalling that $\widehat{\mathcal{N}}^\dagger(\mathds{1}) = \mathds{1}$, and where the
first-order term vanishes because a product of two Hermitian operators has a
real trace; with the factor $i$ the term is killed by taking the real part.
Continuing with only the second-order term we find
\begin{align}
  \hspace*{2em}
  &\hspace*{-2em}
    \left.\frac{d^2}{dt'^2}\right|_{t'=t} F(\rho_B(t), \rho_B(t'))
    \nonumber\\
  &= \max_{S=S^\dagger,\;S_2} \mathopen{}\left\{
    - \operatorname{tr}\bigl(\psi H^2\bigr)
    - \frac12 \operatorname{tr}\bigl(\widehat{\mathcal{N}}(\psi)\, \bigl(S_2 + S_2^\dagger\bigr)\bigr)
    + \operatorname{tr}\bigl[\{\psi, H\} \, \widehat{\mathcal{N}}^\dagger(S)\bigr]
    \right\}\mathclose{}
    \nonumber\\
  &= \max_{S=S^\dagger} \mathopen{}\left\{
    -  \operatorname{tr}\bigl(\psi H^2\bigr)
    -  \operatorname{tr}\bigl(\widehat{\mathcal{N}}(\psi)\, S^2\bigr)
    + \operatorname{tr}\bigl[\widehat{\mathcal{N}}\bigl(\{\psi, H\}\bigr) \, S\bigr]
    \right\}\mathclose{}\ ,
\label{z:XH4FrCCdyG-h}
\end{align}
where we have used the identity $2\operatorname{Re}\operatorname{tr}(A{O}) = \operatorname{tr}(A({O}+{O}^\dagger))$ for Hermitian
$A$, the identity $2\operatorname{Re}\operatorname{tr}(ABC) = \operatorname{tr}(\{A,B\}\,C)$ for Hermitian $A,B,C$, as well
as the condition $S_2+S_2^\dagger = 2S^2$ that came from enforcing the unitarity
of $W_E$.

It is instructive to briefly comment on the situation of a time-dependent
Hamiltonian.  The derivation of the above expression,
especially~\eqref{z:BVHbTNZ44pG0} and the
expansion of the time-evolution operator leading up
to~\eqref{z:u.dlhHNOgsWs}, looks like it necessitated the
assumption of time independence of the Hamiltonian and that a time-dependent
Hamiltonian might have led to a different result.  In fact, we obtain the same
result with a time-dependent Hamiltonian, which can be seen as follows.  Write
\begin{align}
  H(t) = H + t H' + O(t^2) 
\end{align}
and expand the time-evolution operator via the
time-ordered exponential as
$U^\dagger(t'-t) = \mathcal{T}{e}^{i\int_t^{t'} dt'' H(t'')} = 1 + i\int_t^{t'}
dt'' H(t'') - \int_t^{t'} dt'' H(t'')\int_{t}^{t''} dt''' H(t''') + O(t'^3) = 1
+ idt H + (dt^2/2)\bigl(iH'-H^2\bigr) + O(dt^3)${,} then we see that the only
difference in the expressions leading up to~\eqref{z:u.dlhHNOgsWs}
is an additional term $\operatorname{Re}\operatorname{tr}\bigl\{ \psi iH' dt^2 \bigr\}$ which is equal to zero.

Now, we proceed to prove the uncertainty relation.  With the definition
$\FIloss{Bob}{t} = \FIqty{Alice}{t} - \FIqty{Bob}{t}$, we have
\begin{align}
  \FIloss{Bob}{t}
  &= 4\bigl( \operatorname{tr}(\psi H^2) - \bigl(\operatorname{tr}(\psi H)\bigr)^2 \bigr) + 4
    \left.\frac{d^2}{dt'^2}\right|_{t'=t} F(\rho_B(t), \rho_B(t'))
    \nonumber\\
  &=\max_{S=S^\dagger} \Bigl\{
  4\operatorname{tr}\bigl[ \widehat{\mathcal{N}}\bigl(\{\psi, H\}\bigr)\, S \bigr]
  - 4\operatorname{tr}\bigl[ \widehat{\mathcal{N}}(\psi)\, S^2 \bigr]
  - 4\bigl(\operatorname{tr}(\psi H)\bigr)^2
    \Bigr\}\ ,
    \label{z:EZd9QXhrnKt-}
\end{align}
recalling that
$\FIqty{Alice}{t} = 4\sigma_H^2 = 4\bigl(\langle {H^2}\rangle  - \langle {H}\rangle ^2\bigr)$ and using
the expression~\eqref{z:XH4FrCCdyG-h}.  Observe that $\FIloss{Bob}{t}$ is
necessarily invariant under a constant shift of the Hamiltonian
$H\mapsto H+c\mathds{1}$, because such a shift does not influence the evolution
$\psi(t)$ and therefore both $\FIqty{Alice}{t}$ and $\FIqty{Bob}{t}$
are invariant under such shifts.  [This invariance can also be checked
explicitly by carrying out the corresponding transformations $H\mapsto H+c\mathds{1}$
and $S\to S+c\mathds{1}$ in~\eqref{z:EZd9QXhrnKt-}.]
Applying the shift $H \mapsto H-\langle {H}\rangle _\psi$ yields
\begin{align}
  \text{\eqref{z:EZd9QXhrnKt-}}
  &=
    \max_{S=S^\dagger} \Bigl\{
    4\operatorname{tr}\bigl[ \widehat{\mathcal{N}}\bigl(\{\psi, \bar{H}\}\bigr)\, S \bigr]
    - 4\operatorname{tr}\bigl[ \widehat{\mathcal{N}}(\psi)\, S^2 \bigr]
    \Bigr\}\ ,
    \label{z:uvDL0FEIgWX0}
\end{align}
using the shorthand $\bar{H} := H-\langle {H}\rangle _\psi$.  At this point we recognize the
expression of the Fisher information from
\cref{z:xifd0Y80UQtS}, with
$\rho = \widehat{\mathcal{N}}(\psi)$ and
$D = \widehat{\mathcal{N}}\bigl(\{\psi, \bar H\}\bigr)$.  Let us briefly check that the
requirement $P_\rho^\perp D P_\rho^\perp=0$ in
\cref{z:xifd0Y80UQtS} and in the definition of the Fisher
information~\eqref{z:.gRi0KAw40V9} is satisfied.  Thanks to
\cref{z:pZMoi26flsRb}, we have
$\left[\begin{smallmatrix} \psi & \psi \bar{H} \\ \bar{H} \psi & \bar{H} \psi
    \bar{H} \end{smallmatrix}\right] \geq 0$, and furthermore, by
\cref{z:TV-YjOuDas4P}, $\left[\begin{smallmatrix}
    \widehat{\mathcal{N}}(\psi) & \widehat{\mathcal{N}}(\psi \bar{H}) \\
    \widehat{\mathcal{N}}(\bar{H} \psi) & \widehat{\mathcal{N}}(\bar{H} \psi
    \bar{H}) \end{smallmatrix}\right] \geq 0$; by \cref{z:pZMoi26flsRb} again,
this implies that
$P_\rho^{\perp} \widehat{\mathcal{N}}\bigl(\psi \bar{H} \bigr) = 0$.  Therefore
$P_\rho^\perp \widehat{\mathcal{N}}\bigl( \{ \psi, \bar{H} \}\bigr) P_\rho^\perp = 0$.
It follows that
\begin{align}
  \FIloss{Bob}{t} &=
  \text{\eqref{z:uvDL0FEIgWX0}}
  = \Ftwo[\Big]{ \widehat{\mathcal{N}}(\psi) }{
    \widehat{\mathcal{N}}\bigl( \{ \psi, \bar{H} \}\bigr) }\ ,
         \label{z:wHfLRDAY1JLh}
\end{align}
as claimed.

\subsection{Direct proof using the semidefinite characterization
  of the Fisher information}
\label{z:VM5-jRd.0oMe}
\label{z:xqliGh9KZXzi}
\label{z:3LQPjij2PMi6}

For this section, we fix $\lvert {\psi}\rangle , \lvert {\xi}\rangle $ be such that $\langle {\psi}\mkern 1.5mu\relax \vert\mkern 1.5mu\relax {\psi}\rangle =1$
and $\langle {\psi}\mkern 1.5mu\relax \vert\mkern 1.5mu\relax {\xi}\rangle =0$, and let $\mathcal{N}$ be a completely positive, trace
nonincreasing map.  Let $V_{A\to BE}$ be a Stinespring dilation of
$\mathcal{N}$, i.e., $\mathcal{N}(\cdot) = \operatorname{tr}_E\bigl(V\,(\cdot)\,V^\dagger\bigr)$,
and let $\widehat{\mathcal{N}}(\cdot) = \operatorname{tr}_B\bigl( V\,(\cdot)\,V^\dagger\bigr)$.
Let
\begin{align}
  D_A^Y &= -i \bigl(\lvert {\xi}\rangle \mkern -1.8mu\relax \langle{\psi }\rvert - \lvert {\psi}\rangle \mkern -1.8mu\relax \langle{\xi}\rvert \bigr)\ ;
  &
    D_A^Z &= \lvert {\xi}\rangle \mkern -1.8mu\relax \langle{\psi }\rvert + \lvert {\psi}\rangle \mkern -1.8mu\relax \langle{\xi}\rvert \ .
\end{align}

Suppose that $\lvert {\Phi_{B:E}}\rangle $ is a maximally entangled ket between two
suitable subspaces of $B$ and $E$ that are sufficiently large to ensure that
there exist $M,\Lambda$ matrices on $B$ satisfying
\begin{align}
  V\lvert {\phi }\rangle &= (\Lambda\otimes\mathds{1})\lvert {\Phi_{B:E}}\rangle \ ;
  &
  V\lvert {\xi }\rangle &= (M\otimes\mathds{1})\lvert {\Phi_{B:E}}\rangle  \ .
\end{align}
(Alternatively, one can embed both $B$ and $E$ into larger systems $B',E'$ with
$B'\simeq E'$, on which one can consider the canonical maximally entangled ket
$\lvert {\Phi'_{B':E'}}\rangle  = \sum \lvert {k}\rangle _{B'}\lvert {k}\rangle _{E'}$ with respect to the canonical bases of $B',E'$.  We then define $\lvert {\Phi_{B:E}}\rangle $ by projecting
down $\lvert {\Phi'_{B':E'}}\rangle $ onto $B\otimes E$.)
Throughout the following, we only ever consider operators that are in the
support of the reduced operators of $\Phi_{B:E}$ on $B$ and $E$.

We define the operation
$t_{B\to E}(\cdot) := \operatorname{tr}_B\bigl\{ \Phi_{B:E} \,[(\cdot)\otimes\mathds{1}_E] \bigr\}$ which
is the partial transpose operation with respect to the bases used to define
$\lvert {\Phi_{B:E}}\rangle $.  Equivalently, a defining property of this operation is that
for any operator $X_B$, we have $(X_B\otimes\mathds{1}_E)\lvert {\Phi_{B:E}}\rangle  = %
(\mathds{1}_B\otimes t_{B\to{}E}(X_B))\lvert {\Phi_{B:E}}\rangle $.  Furthermore, for any ${M}$,
we have $t_{B\to{}E}({M}^\dagger) = \bigl[t_{B\to{}E}({M})\bigr]^\dagger$ and for any
$X,Y$ we have $t_{B\to E}({M}{N}) = t_{B\to E}({N}) \, t_{B\to E}({M})$.  Similarly, we
define the inverse operation
$t_{E\to B}(\cdot) = \operatorname{tr}\bigl\{\Phi_{B:E}\,[\mathds{1}_B\otimes(\cdot)]\bigr\}$ which has the
same properties.

Observe that $\Lambda\Lambda^\dagger = \mathcal{N}(\lvert {\psi}\rangle \mkern -1.8mu\relax \langle{\psi}\rvert ) = \rho_B$ and
$MM^\dagger = \mathcal{N}(\lvert {\xi}\rangle \mkern -1.8mu\relax \langle{\xi}\rvert )$.  Furthermore, we define $W$ via the polar
decomposition of $\Lambda = \rho^{1/2} W$, with
\begin{align}
  \Lambda &= \rho^{1/2} W\ ;
  &
    \Lambda^\dagger &= W^\dagger \rho^{1/2}\ ;
  &
    \Lambda^{-1} &= W^\dagger\rho^{-1/2}\ ;
  &
    \Lambda^{-\dagger} &= \rho^{-1/2} W\ .
\end{align}
The operators $\Lambda^{-1}$ and $\Lambda^{-\dagger}$ are the Moore-Penrose
pseudoinverses of $\Lambda$ and $\Lambda^\dagger$, respectively, as can be seen
by computing $\Lambda\Lambda^{-1} = P_\rho$ and
$\Lambda^{-1}\Lambda = W^\dagger P_\rho W$ as well as
$\Lambda^\dagger\Lambda^{-\dagger} = W^\dagger P_\rho W$ and
$\Lambda^{-\dagger}\Lambda^\dagger = P_\rho$.
Furthermore, we have
\begin{subequations}
  \begin{align}
    \mathcal{N}(D_A^Y)
    &= -i (M\Lambda^\dagger - \Lambda M^\dagger)\ ,
    \\
    \widehat{\mathcal{N}}(D_A^Z)
    &= \operatorname{tr}_B\bigl\{M_B \Phi_{B:E}\Lambda_B^\dagger +  \Lambda_B \Phi_{B:E} M_B^\dagger\bigr\}
      \nonumber\\
    &= t_{B\to E}\bigl[\Lambda^\dagger M + M^\dagger\Lambda\bigr]\ .
  \end{align}
\end{subequations}
We may also relate these objects to the state on Eve's system, via the partial
transpose operation $t_{B\to E}$.  Observe that
$V\lvert {\psi }\rangle = (\Lambda\otimes\mathds{1})\lvert {\Phi_{B:E}}\rangle  = (\mathds{1}\otimes
t_{B\to{}E}(\Lambda))\lvert {\Phi_{B:E}}\rangle $, and therefore
$\rho_E = \operatorname{tr}_B(V\psi V^\dagger) =
\bigl[t_{B\to{}E}(\Lambda)\bigr]\bigl[t_{B\to{}E}(\Lambda)\bigr]^\dagger %
= t_{B\to E}(\Lambda^\dagger\Lambda) %
= t_{B\to E}(W^\dagger \rho_B W)$.  Then
$P_{\rho_E} = t_{B\to{}E}(W^\dagger P_{\rho_B} W)$ and
$P_{\rho_E}^\perp = t_{B\to{}E}(W^\dagger P_{\rho_B}^\perp W)$.
We begin with a characterization of when our uncertainty relation holds with
equality.
\begin{proposition}[Conditions for uncertainty relation equality]
  \label{z:L9cC7f4RKO0S}
  The following statements are equivalent:
  \begin{enumerate}[label=(\roman*)]
  \item\label{z:tPvG.efzGNQh}
    $(P_{\rho_B}^\perp \otimes P_{\rho_E}^\perp) V \lvert {\xi }\rangle = 0$\ .
  \item\label{z:TpPY7RBUzBIR}
    We have $P_{\rho_B}^\perp M W^\dagger P_{\rho_B}^\perp = 0$\ .
  \item\label{z:yf8Wki00SdWM} We have
    $P_{\rho_B}^\perp\mathcal{N}(\lvert {\xi}\rangle \mkern -1.8mu\relax \langle{\xi}\rvert )P_{\rho_B}^\perp = P_{\rho_B}^\perp
    \mathcal{N}(\lvert {\xi}\rangle \mkern -1.8mu\relax \langle{\psi}\rvert )\rho_B^{-1} \mathcal{N}(\lvert {\psi}\rangle \mkern -1.8mu\relax \langle{\xi}\rvert )
    P_{\rho_B}^\perp$.
  \item\label{z:8Z44eS9wHEqZ} We have
    $P_{\rho_E}^\perp\widehat{\mathcal{N}}(\lvert {\xi}\rangle \mkern -1.8mu\relax \langle{\xi}\rvert )P_{\rho_E}^\perp =
    P_{\rho_E}^\perp \widehat{\mathcal{N}}(\lvert {\xi}\rangle \mkern -1.8mu\relax \langle{\psi}\rvert )\rho_E^{-1}
    \widehat{\mathcal{N}}(\lvert {\psi}\rangle \mkern -1.8mu\relax \langle{\xi}\rvert ) P_{\rho_E}^\perp$.
  \item\label{z:qaKYvNwsiBmk} Let $\{ E_k \}$ are Kraus
    operators for $\mathcal{N}$.  For any linear combination
    $E = \sum_k c_k E_k$ with $c_k\in\mathbb{C}$ and such that $E\lvert {\psi}\rangle =0$, we
    have $P_{\rho_B}^\perp E \lvert {\xi }\rangle = 0$.
  \end{enumerate}
  Furthermore, consider the setting of
  \cref{z:T-KNuYhGRM7I} and suppose that $\lvert {\xi}\rangle $ is
  defined as $\lvert {\xi }\rangle = (H-\langle {H}\rangle )\lvert {\psi}\rangle $.  Then
  \labelcref{z:tPvG.efzGNQh,z:TpPY7RBUzBIR,z:yf8Wki00SdWM,z:8Z44eS9wHEqZ,z:qaKYvNwsiBmk} are furthermore equivalent to:
  \begin{enumerate}[label=(\roman*),resume]
  \item \label{z:KTO32PgtQRRR} For any eigenvalue $p_k(t)$ of
    $\mathcal{N}(\psi(t))$ such that $p_k(t_0)=0$, we have
    $\partial_t^2 p_k\,(t_0) = 0$\ .
  \end{enumerate}
\end{proposition}

Observe that all the conditions above do not depend on the choice of Stinespring
dilation and/or on the choice of the Kraus operator representation, as all such
choices differ by a partial isometry on the $E$ system.  In other words, if the
conditions above hold for particular choices of $V$, $\widehat{\mathcal{N}}$ and
$\{ E_k \}$, they hold for all other choices as well.

\begin{proof}[*z:L9cC7f4RKO0S]
  We have the following implications.

  \noindent\textit{\hbox{$\text{\ref{z:tPvG.efzGNQh}} \Leftrightarrow
      \text{\ref{z:TpPY7RBUzBIR}}$}}:
  Consider
  \begin{align}
    (P_{\rho_B}^\perp\otimes P_{\rho_E}^\perp) V \lvert {\xi
    }\rangle &= \bigl( (P_{\rho_B}^\perp M)\otimes P_{\rho_E}^\perp\bigr) \lvert {\Phi}\rangle _{BE}
    = \bigl( (P_{\rho_B}^\perp M\, t_{E\to B}(P_{\rho_E}^\perp))\otimes\mathds{1}\bigr)
    \lvert {\Phi}\rangle _{BE}\ .
    \label{z:BI1R6TmoKDBK}
  \end{align}
  Since $V\lvert {\psi }\rangle = (\Lambda\otimes\mathds{1})\lvert {\Phi }\rangle %
  = (\mathds{1}\otimes t_{B\to E}(\Lambda))\lvert {\Phi}\rangle $, we have
  $\rho_E = t_{B\to E}(\Lambda)\, t_{B\to E}(\Lambda)^\dagger %
  = t_{B\to E}(\Lambda^\dagger \Lambda) = t_{B\to E}(W^\dagger\rho_B W)$.  Then
  $P_{\rho_E} = t_{B\to E}(W^\dagger P_{\rho_B} W)$ and
  $P_{\rho_E}^\perp = t_{B\to E}(W^\dagger P_{\rho_B}^\perp W) = %
  t_{B\to E}(\tilde{P}_{\rho_B}^\perp)$, and
  \begin{align}
    \text{\eqref{z:BI1R6TmoKDBK}}
    &= \bigl( (P_{\rho_B}^\perp M\, \tilde{P}_{\rho_B}^\perp)\otimes\mathds{1}\bigr)
      \lvert {\Phi}\rangle _{BE}\ .
  \end{align}
  Therefore we have that $(P_{\rho_B}^\perp\otimes P_{\rho_E}^\perp)V\lvert {\xi}\rangle =0$
  is equivalent to $0 = P_{\rho_B}^\perp M \tilde{P}_{\rho_B}^\perp$.

  \noindent\textit{\hbox{$\text{\ref{z:TpPY7RBUzBIR}} \Rightarrow
      \text{\ref{z:yf8Wki00SdWM}}$}}: Let
  $K = P_\rho^\perp \mathcal{N}(\lvert {\xi}\rangle \mkern -1.8mu\relax \langle{\psi}\rvert )\rho_B^{-1/2} = P_\rho^\perp
  \operatorname{tr}_E\bigl( M \Phi_{BE} \Lambda^\dagger \bigr) \rho_B^{-1/2} = P_\rho^\perp M
  W^\dagger P_\rho$.  Now assume that~\ref{z:TpPY7RBUzBIR} holds;
  then $K = P_\rho^\perp M W^\dagger$ and we have
  $KK^\dagger = P_\rho^\perp M M^\dagger P_\rho^\perp = P_\rho^\perp
  \mathcal{N}(\lvert {\xi}\rangle \mkern -1.8mu\relax \langle{\xi}\rvert ) P_\rho^\perp$,
  showing~\ref{z:yf8Wki00SdWM}.

  \noindent\textit{\hbox{$\text{\ref{z:yf8Wki00SdWM}} \Rightarrow
      \text{\ref{z:TpPY7RBUzBIR}}$}}:
  Conversely, assuming~\ref{z:yf8Wki00SdWM} and if
  $K=P_{\rho_B}^\perp \mathcal{N}(\lvert {\xi}\rangle \mkern -1.8mu\relax \langle{\psi}\rvert )\rho_B^{-1/2} %
  = P_\rho^\perp M W^\dagger P_\rho$, we have by assumption that
  $KK^\dagger = P_\rho^\perp M M^\dagger P_\rho^\perp %
  = (P_\rho^\perp M W^\dagger P_\rho) (P_\rho M W^\dagger P_\rho)^\dagger %
  + P_\rho^\perp M W^\dagger P_\rho^\perp W M^\dagger P_\rho^\perp$.  This means
  that
  $0 = (P_\rho^\perp M W^\dagger P_\rho^\perp) (P_\rho^\perp W M^\dagger
  P_\rho^\perp)$.  The latter equation can only hold if
  $P_\rho^\perp M W^\dagger P_\rho^\perp = 0$,
  showing~\ref{z:TpPY7RBUzBIR}.

  \noindent\textit{\hbox{$\text{\ref{z:tPvG.efzGNQh}} \Leftrightarrow
      \text{\ref{z:8Z44eS9wHEqZ}}$}}:
  Condition~\ref{z:tPvG.efzGNQh} is symmetric if we replace
  $B \leftrightarrow E$ (and correspondingly
  $\mathcal{N} \leftrightarrow \widehat{\mathcal{N}}$), meaning that the
  condition holds if and only if the condition with $B$ and $E$ swapped also
  holds.  Therefore, we can swap $B \leftrightarrow E$ in the other conditions
  and those will also hold if and only if~\ref{z:tPvG.efzGNQh}
  holds.  Condition~\ref{z:8Z44eS9wHEqZ} is obtained by
  performing this transformation on~\ref{z:yf8Wki00SdWM}.

  \noindent\textit{\hbox{$\text{\ref{z:tPvG.efzGNQh}} \Rightarrow
      \text{\ref{z:qaKYvNwsiBmk}}$}}: 
  We choose the representation $V = \sum E_k \otimes \lvert {k}\rangle _E$ and assume
  $\text{\ref{z:tPvG.efzGNQh}}$, i.e., that we have
  \begin{align}
    (P_{\rho_B}^\perp \otimes P_{\rho_E}^\perp) V\lvert {\xi }\rangle = 0 .
    \label{z:TRgyTHWihOgu}
  \end{align}
  Let $\{ c_k \}$ with $c_k\in\mathbb{C}$ such that $\sum_k c_k E_k \lvert {\psi }\rangle = 0$
  and let $E = \sum_k c_k E_k$.  Define $\lvert {\mathrm{e}}\rangle _E = \sum c_k^* \lvert {k}\rangle _E$.
  We have 
  \begin{align}
    \langle {\mathrm{e}}\mkern 1.5mu\relax \vert \mkern 1.5mu\relax {\widehat{\mathcal{N}}(\psi)}\mkern 1.5mu\relax \vert \mkern 1.5mu\relax {\mathrm{e}}\rangle 
    = \sum_{k,k'}
    \langle {\mathrm{e}}\mkern 1.5mu\relax \vert\mkern 1.5mu\relax {k}\rangle  \operatorname{tr}\bigl( E_k \psi E_{k'}^\dagger \bigr) \langle {k'}\mkern 1.5mu\relax \vert\mkern 1.5mu\relax {\mathrm{e}}\rangle 
    = \operatorname{tr}\Bigl[ \Bigl(\sum c_k E_k\Bigr) \, \psi\, \Bigl(\sum c_k E_k\Bigr){}^\dagger \Bigr]
    = \bigl \lVert { E\lvert {\psi }\rangle }\bigr \rVert ^2 = 0\ ,
    \label{z:8EM7fsR2.Q48}
  \end{align}
  which implies that $\lvert {\mathrm{e}}\rangle _E \in \ker \widehat{\mathcal{N}}(\psi)$,
  i.e., $P_{\rho_E}^\perp \lvert {\mathrm{e}}\rangle _E = \lvert {\mathrm{e}}\rangle _E$.  Applying
  $\bigl(\mathds{1}\otimes\langle {\mathrm{e}}\rvert \bigr)$ onto~\eqref{z:TRgyTHWihOgu}
  we find
  \begin{align}
    0 &= 
        \bigl(P_{\rho_B}^\perp \otimes \langle {\mathrm{e}}\rvert  P_{\rho_E}^\perp\bigr) V\lvert {\xi
        }\rangle = \sum_k \bigl(P_{\rho_B}^\perp E_k\lvert {\xi}\rangle \bigr)\, \langle {\mathrm{e}}\mkern 1.5mu\relax \vert\mkern 1.5mu\relax {k}\rangle 
        = P_{\rho_B}^\perp E \lvert {\xi}\rangle \ ,
  \end{align}
  showing that $\text{\ref{z:qaKYvNwsiBmk}}$ holds.

  \noindent\textit{\hbox{$\text{\ref{z:tPvG.efzGNQh}} \Leftarrow
      \text{\ref{z:qaKYvNwsiBmk}}$}}: 
  We now suppose that condition $\text{\ref{z:qaKYvNwsiBmk}}$
  holds.  Let $\lvert {\chi_j}\rangle _E$ be a set of orthonormal states that span the support of
  $P_{\rho_E}^\perp$, i.e., $P_{\rho_E}^\perp = \sum_j \lvert {\chi_j}\rangle \mkern -1.8mu\relax \langle{\chi_j}\rvert _E$.
  Fix any such $\lvert {\chi_j}\rangle $ and define
  $E^{(j)} = \sum \langle {\chi_j}\mkern 1.5mu\relax \vert\mkern 1.5mu\relax {k}\rangle \,E_k$.  We
  repeat~\eqref{z:8EM7fsR2.Q48} by replacing
  $\lvert {\mathrm{e}}\rangle  \to \lvert {\chi_j}\rangle $, $c_k \to \langle {\chi_j}\mkern 1.5mu\relax \vert\mkern 1.5mu\relax {k}\rangle $ to find
  \begin{align}
    0 = \langle {\chi_j}\mkern 1.5mu\relax \vert \mkern 1.5mu\relax { \widehat{\mathcal{N}}(\psi) }\mkern 1.5mu\relax \vert \mkern 1.5mu\relax {\chi_j}\rangle 
    =\ \ldots\ = \bigl \lVert { E^{(j)} \lvert {\psi }\rangle }\bigr \rVert ^2\ ,
  \end{align}
  which implies that $E^{(j)}\lvert {\psi }\rangle = 0$.  We use the assumption that
  $\text{\ref{z:qaKYvNwsiBmk}}$ holds to deduce that
  $P_{\rho_B}^\perp E^{(j)}\lvert {\xi }\rangle = 0$; we note the latter expression holds for
  all $j$ by repeating this argument for each $j$ individually.  Then
  \begin{align}
    \bigl(P_{\rho_B}^\perp \otimes P_{\rho_E}^\perp\bigr) \, V \lvert {\xi
    }\rangle &= \Bigl(P_{\rho_B}^\perp \otimes \sum \lvert {\chi_j}\rangle \mkern -1.8mu\relax \langle{\chi_j}\rvert \Bigr) \, V \lvert {\xi
    }\rangle = \sum_{k,j} (P_{\rho_B}^\perp E_k \lvert {\xi}\rangle )\otimes(\lvert {\chi_j}\rangle  \langle {\chi_j}\mkern 1.5mu\relax \vert\mkern 1.5mu\relax {k}\rangle )
    \nonumber\\
    &= \sum_j (P_{\rho_B}^\perp E^{(j)} \lvert {\xi}\rangle )\otimes\lvert {\chi_j}\rangle  = 0\ ,
  \end{align}
  showing that $\text{\ref{z:tPvG.efzGNQh}}$ holds.

  \noindent\textit{\hbox{$\text{\ref{z:yf8Wki00SdWM}} \Leftrightarrow
      \text{\ref{z:KTO32PgtQRRR}}$}}: Now consider the setting
  of \cref{z:T-KNuYhGRM7I} and suppose that $\lvert {\xi}\rangle $ is
  defined as $\lvert {\xi }\rangle = (H-\langle {H}\rangle )\lvert {\psi}\rangle $.  We invoke
  Ref.~\cite[Eq.~(B15)]{R32}, which in the present
  context reads
  \begin{align}
    \operatorname{tr}\bigl(P_\rho^\perp\,\partial_t^2\rho\bigr)
    = \sum_{k:\,p_k=0} \partial_t^2 p_k +
    2 \sum_{\substack{k,\ell:\\p_k > 0\\p_\ell=0}}
    \frac{\lvert {\langle {\lambda_k}\mkern 1.5mu\relax \vert \mkern 1.5mu\relax {\partial_t\rho}\mkern 1.5mu\relax \vert \mkern 1.5mu\relax {\lambda_{\ell}}\rangle }\rvert ^2}{p_k}\ ,
    \label{z:7I2ojQMuazIH}
  \end{align}
  where $\{\lvert {\lambda_k}\rangle  \}$ is a complete eigenbasis of $\rho$ with
  eigenvalues $p_k$.  Using \cref{z:d-w60nHacIQJ} one can check
  that the second term on the right-hand side satisfies
  \begin{align}
    2 \sum_{\raisebox{-2ex}[0pt][0pt]{\(\substack{k,\ell:\\p_k > 0\\p_\ell=0}\)}}
    \frac{\lvert {\langle {\lambda_k}\mkern 1.5mu\relax \vert \mkern 1.5mu\relax {\partial_t\rho}\mkern 1.5mu\relax \vert \mkern 1.5mu\relax {\lambda_\ell}\rangle }\rvert ^2}{p_k}
    &= 2\operatorname{tr}\bigl\{ \rho^{-1}\,(\partial_t\rho)\,P_\rho^\perp\,(\partial_t\rho) \bigr\}
      \nonumber\\[1ex]
    &= 2\operatorname{tr}\bigl\{ \rho^{-1}\,\mathcal{N}(\lvert {\psi}\rangle \mkern -1.8mu\relax \langle{\xi}\rvert )\,P_\rho^\perp\,
      \mathcal{N}(\lvert {\xi}\rangle \mkern -1.8mu\relax \langle{\psi}\rvert ) \bigr\}\ ,
      \label{z:R6j6cdZDtuKD}
  \end{align}
  using the fact that
  $\partial_t\rho = \mathcal{N}(-i[H,\psi]) =
  \mathcal{N}(-i\lvert {\xi}\rangle \mkern -1.8mu\relax \langle{\psi}\rvert +i\lvert {\psi}\rangle \mkern -1.8mu\relax \langle{\xi}\rvert )$ and that
  $\mathcal{N}(X\psi)P_\rho^\perp=0$ for any $X$.
  On the other hand we can see that
  $\partial_t^2\rho = \partial_t \mathcal{N}(-i[H,\psi]) =
  \mathcal{N}\bigl(-i[\partial_t H,\psi] -[H, [H, \psi]]\bigr)$, and recalling that
  $\mathcal{N}(X\psi)P_\rho^\perp=0$ for any $X$ we obtain
  \begin{align}
    \operatorname{tr}\bigl(P_\rho^\perp\,\partial_t^2\rho\bigr)
    = \operatorname{tr}\bigl(P_\rho^\perp\,\mathcal{N}(2H\psi H)\bigr)
    = 2\operatorname{tr}\bigl(P_\rho^\perp\,\mathcal{N}(\bar{H}\psi\bar{H})\bigr)
    = 2\operatorname{tr}\bigl(P_\rho^\perp\,\mathcal{N}(\lvert {\xi}\rangle \mkern -1.8mu\relax \langle{\xi}\rvert )\bigr)\ ,
    \label{z:X.oYAZe-jUxy}
  \end{align}
  writing $\bar{H} = H-\langle {H}\rangle _\psi$ and where $\lvert {\xi}\rangle =\bar{H}\lvert {\psi}\rangle $.

  Now suppose that \ref{z:yf8Wki00SdWM} holds.  Then
  \begin{align}
    \text{\eqref{z:X.oYAZe-jUxy}}
    &= 2\operatorname{tr}\bigl(P_\rho^\perp \mathcal{N}(\lvert {\xi}\rangle \mkern -1.8mu\relax \langle{\psi}\rvert )\,
      \rho^{-1}\mathcal{N}(\lvert {\psi}\rangle \mkern -1.8mu\relax \langle{\xi}\rvert )\bigr) = \text{\eqref{z:R6j6cdZDtuKD}}\ ,
  \end{align}
  and therefore the first term on the right-hand side
  of~\eqref{z:7I2ojQMuazIH} must vanish, and since
  $\partial_t^2 p_k \geq 0$ for all $k$ for which $p_k=0$ as $p_k$ reaches a
  minimum at that point, we must necessarily have that $\partial_t^2 p_k = 0$
  for all those $k$.

  Conversely, if the first term on the right-hand side
  of~\eqref{z:7I2ojQMuazIH} vanishes, then we have
  \begin{align}
    \operatorname{tr}\bigl(P_\rho^\perp\,
    \mathcal{N}(\lvert {\xi}\rangle \mkern -1.8mu\relax \langle{\psi}\rvert )\,\rho^{-1}\,\mathcal{N}(\lvert {\psi}\rangle \mkern -1.8mu\relax \langle{\xi}\rvert ) \bigr) =
    \operatorname{tr}\bigl(P_\rho^\perp \mathcal{N}(\lvert {\xi}\rangle \mkern -1.8mu\relax \langle{\xi}\rvert )\bigr)\ .
    \label{z:BFXzMcq.gU8f}
  \end{align}
  By applying the completely positive map ${\mathrm{id}}_{2}\otimes\mathcal{N}$
  onto the matrix
  $\begin{psmallmatrix}\lvert {\psi }\rangle \mkern -1.8mu\relax \langle{\psi }\rvert & \lvert {\psi}\rangle \mkern -1.8mu\relax \langle{\xi }\rvert \\ \lvert {\xi}\rangle \mkern -1.8mu\relax \langle{\psi }\rvert &
    \lvert {\xi }\rangle \mkern -1.8mu\relax \langle{\xi }\rvert \end{psmallmatrix}$ and further conjugating by
  $\begin{psmallmatrix} \mathds{1}&\\& P_\rho^\perp \end{psmallmatrix}$ we find
  that
  \begin{align}
    \begin{bmatrix}
      \rho & \mathcal{N}(\lvert {\psi}\rangle \mkern -1.8mu\relax \langle{\xi}\rvert ) P_\rho^\perp \\
      P_\rho^\perp \mathcal{N}(\lvert {\xi}\rangle \mkern -1.8mu\relax \langle{\psi}\rvert )  & P_\rho^\perp \mathcal{N}(\lvert {\xi}\rangle \mkern -1.8mu\relax \langle{\xi}\rvert ) P_\rho^\perp
    \end{bmatrix}
    \geq 0\ .
  \end{align}
  From the Schur complement (\cref{z:pZMoi26flsRb}) we find that
  \begin{align}
    P_\rho^\perp \Bigl[
    \mathcal{N}(\lvert {\xi}\rangle \mkern -1.8mu\relax \langle{\xi}\rvert ) 
    -
    \mathcal{N}(\lvert {\xi}\rangle \mkern -1.8mu\relax \langle{\psi}\rvert ) \; \rho^{-1} \;
    \mathcal{N}(\lvert {\psi}\rangle \mkern -1.8mu\relax \langle{\xi}\rvert ) \Bigr] P_\rho^\perp
    \geq 0\ .
  \end{align}
  But a positive semidefinite operator has trace zero if and only if it is
  identically equal to zero, so with~\eqref{z:BFXzMcq.gU8f} we find
  that
  $P_\rho^\perp\bigl[ \mathcal{N}(\lvert {\xi}\rangle \mkern -1.8mu\relax \langle{\xi}\rvert ) -
  \mathcal{N}(\lvert {\xi}\rangle \mkern -1.8mu\relax \langle{\psi}\rvert )\,\rho^{-1}\,\mathcal{N}(\lvert {\psi}\rangle \mkern -1.8mu\relax \langle{\xi}\rvert ) \bigr]
  P_\rho^\perp = 0$, showing that~\ref{z:yf8Wki00SdWM} holds.
\end{proof}

Our main technical theorem is the following.

\begin{theorem}[Time-energy uncertainty relation in the virtual metrological qubit picture]
  \label{z:4VI7KY5eDBCv}
  Let $A$, $B$ and $E$ be finite-dimensional quantum systems.  Let
  $\mathcal{N}_{A\to B}$ be a completely positive, trace-nonincreasing map.  Let
  $V_{A\to BE}$ be such that
  $\mathcal{N}_{A\to B}(\cdot) = \operatorname{tr}_E\bigl( V(\cdot)V^\dagger\bigr)$ and
  $V^\dagger V \leq \mathds{1}$, i.e., $V$ is a Stinespring dilation of
  $\mathcal{N}$.  Let
  $\widehat{\mathcal{N}}_{A\to E}(\cdot) = \operatorname{tr}_B\bigl( V(\cdot)V^\dagger\bigr)$.  Let
  $\lvert {\psi}\rangle $ be any subnormalized state on $A$, and let $\lvert {\xi}\rangle $ be any
  vector on $A$ such that $\langle {\psi}\mkern 1.5mu\relax \vert\mkern 1.5mu\relax {\xi}\rangle =0$.  Define
  $D_A^Y=-i\bigl(\lvert {\xi}\rangle \mkern -1.8mu\relax \langle{\psi}\rvert  - \lvert {\psi}\rangle \mkern -1.8mu\relax \langle{\xi}\rvert \bigr)$ and
  $D_A^Z = \lvert {\xi}\rangle \mkern -1.8mu\relax \langle{\psi}\rvert +\lvert {\psi}\rangle \mkern -1.8mu\relax \langle{\xi}\rvert $.  Then
  \begin{align}
    \Ftwo{\mathcal{N}(\psi)}{\mathcal{N}(D_A^Y)}
    + \Ftwo{\widehat{\mathcal{N}}(\psi)}{\widehat{\mathcal{N}}(D_A^Z)}
    \leq 4\langle {\xi}\mkern 1.5mu\relax \vert \mkern 1.5mu\relax {\mathcal{N}^\dagger(\mathds{1})}\mkern 1.5mu\relax \vert \mkern 1.5mu\relax {\xi}\rangle \ .
    \label{z:.6Ao-K6k6kYf}
  \end{align}
  Furthermore, if $(P_{\rho_B}^\perp \otimes P_{\rho_E}^\perp) V\lvert {\xi }\rangle = 0$,
  then equality holds.
\end{theorem}

First, we remark that both Fisher information expressions
in~\eqref{z:.6Ao-K6k6kYf} are well-defined
in that we always have
$P_{\mathcal{N}(\psi)}^\perp \mathcal{N}(D_A^Y) P_{\mathcal{N}(\psi)}^\perp = 0$ and
$P_{\widehat{\mathcal{N}}(\psi)}^\perp \widehat{\mathcal{N}}(D_A^Z)
P_{\widehat{\mathcal{N}}(\psi)}^\perp = 0$ as required in the
definition~\eqref{z:.gRi0KAw40V9}.  These conditions can be verified by
first noting that the following matrix is positive semidefinite,
\begin{align}
    \begin{bmatrix} \lvert {\psi }\rangle \mkern -1.8mu\relax \langle{\psi }\rvert & \lvert {\psi}\rangle \mkern -1.8mu\relax \langle{\xi }\rvert \\ \lvert {\xi}\rangle \mkern -1.8mu\relax \langle{\psi }\rvert & \lvert {\xi
    }\rangle \mkern -1.8mu\relax \langle{\xi
    }\rvert \end{bmatrix} 
    = \begin{bmatrix} \lvert {\psi }\rangle \\ \lvert {\xi }\rangle \end{bmatrix}
    \begin{bmatrix} \langle {\psi }\rvert & \langle {\xi }\rvert \end{bmatrix}
    \geq 0\ ,
\end{align}
and applying either completely positive map ${\mathrm{id}}_{2}\otimes\mathcal{N}$
or ${\mathrm{id}}_{2}\otimes\widehat{\mathcal{N}}$ to obtain
\begin{align}
    \begin{bmatrix} \mathcal{N}(\lvert {\psi}\rangle \mkern -1.8mu\relax \langle{\psi}\rvert ) & \mathcal{N}(\lvert {\psi}\rangle \mkern -1.8mu\relax \langle{\xi}\rvert ) \\
      \mathcal{N}(\lvert {\xi}\rangle \mkern -1.8mu\relax \langle{\psi}\rvert ) & \mathcal{N}(\lvert {\xi}\rangle \mkern -1.8mu\relax \langle{\xi}\rvert )
    \end{bmatrix}
    &\geq 0\ ;
    &
    \begin{bmatrix} \widehat{\mathcal{N}}(\lvert {\psi}\rangle \mkern -1.8mu\relax \langle{\psi}\rvert ) & \widehat{\mathcal{N}}(\lvert {\psi}\rangle \mkern -1.8mu\relax \langle{\xi}\rvert ) \\
      \widehat{\mathcal{N}}(\lvert {\xi}\rangle \mkern -1.8mu\relax \langle{\psi}\rvert ) & \widehat{\mathcal{N}}(\lvert {\xi}\rangle \mkern -1.8mu\relax \langle{\xi}\rvert )
    \end{bmatrix}
    &\geq 0\ .
\end{align}
Then, \cref{z:pZMoi26flsRb} ensures that
$P_{\mathcal{N}(\psi)}^\perp \mathcal{N}(\lvert {\psi}\rangle \mkern -1.8mu\relax \langle{\xi}\rvert ) = 0$ and therefore
$P_{\mathcal{N}(\psi)}^\perp \mathcal{N}(D_A^Y) P_{\mathcal{N}(\psi)}^\perp = 0$;
likewise
$P_{\widehat{\mathcal{N}}(\psi)}^\perp \widehat{\mathcal{N}}(D_A^Z)
P_{\widehat{\mathcal{N}}(\psi)}^\perp = 0$.

\begin{proof}[*z:4VI7KY5eDBCv]
  Let $\Lambda,M$ be operators acting on $B$ such that
  $V\lvert {\psi }\rangle = (\Lambda\otimes\mathds{1})\lvert {\Phi}\rangle $ and
  $V\lvert {\xi }\rangle = (M\otimes\mathds{1})\lvert {\Phi}\rangle $.  We can write
  \begin{subequations}
    \begin{align}
      D_B
      &=
        \mathcal{N}(D_A^Y)
        = \operatorname{tr}_E\bigl( -i\bigl(V\lvert {\xi }\rangle \mkern -1.8mu\relax \langle{\psi }\rvert V^\dagger - V\lvert {\psi }\rangle \mkern -1.8mu\relax \langle{\xi }\rvert V^\dagger\bigr)\bigr)
        = -i\bigl( M \Lambda^\dagger - \Lambda M^\dagger\bigr)\ ,
        \label{z:POoOh29X6gj.}
      \\
      \widehat{D}_E
      &= \widehat{\mathcal{N}}(D_A^Z) =
        \operatorname{tr}_B(V\lvert {\xi }\rangle \mkern -1.8mu\relax \langle{\psi }\rvert V^\dagger + V\lvert {\psi }\rangle \mkern -1.8mu\relax \langle{\xi }\rvert V^\dagger)
        = \operatorname{tr}_B(M\Phi \Lambda^\dagger + \Lambda\Phi M^\dagger)\ ;
        \label{z:o27hGghQi7Bi}
    \end{align}
  \end{subequations}
  where in~\eqref{z:o27hGghQi7Bi} the operators $M,\Lambda$
  act only on $B$ with a tensor product with the identity on $E$ implied but
  $\Phi = \Phi_{BE} = \lvert {\Phi}\rangle \mkern -1.8mu\relax \langle{\Phi}\rvert _{BE}$.
  Now consider
  \begin{align}
    \frac14\Bigl\{
    4\langle {\xi}\mkern 1.5mu\relax \vert \mkern 1.5mu\relax {\mathcal{N}^\dagger(\mathds{1})}\mkern 1.5mu\relax \vert \mkern 1.5mu\relax {\xi}\rangle  - \Ftwo{\rho_E}{\widehat{D}_E}
    \Bigr\}
    = \operatorname{tr}(MM^\dagger) - %
    \max_{S_E=S_E^\dagger} \Bigl\{ \operatorname{tr}(\widehat{D}_E S_E) - \operatorname{tr}(\rho_E S_E^2) \Bigr\}\ ,
    \label{z:srWjLtb0R-We}
  \end{align}
  using \cref{z:YdgtbMB30Vi6} and noting that
  $\langle {\xi}\mkern 1.5mu\relax \vert \mkern 1.5mu\relax {\mathcal{N}^\dagger(\mathds{1})}\mkern 1.5mu\relax \vert \mkern 1.5mu\relax {\xi}\rangle  = \operatorname{tr}(\mathcal{N}(\lvert {\xi}\rangle \mkern -1.8mu\relax \langle{\xi}\rvert )) =
  \operatorname{tr}(MM^\dagger)$.  Then, using~\eqref{z:o27hGghQi7Bi},
  and writing $t(\cdot) = t_{E\to B}(\cdot)$ as a shorthand,
  \begin{align}
    \text{\eqref{z:srWjLtb0R-We}}
    &= \min_{S_E=S_E^\dagger} \Bigl\{
      \operatorname{tr}(MM^\dagger) -
      \operatorname{tr}\bigl((M\Phi\Lambda^\dagger + \Lambda\Phi M^\dagger) S_E\bigr)
      + \operatorname{tr}\bigl(S_E \Lambda \Phi \Lambda^\dagger S_E\bigr) \Bigr\}
      \nonumber\\
      &= \min_{S_E=S_E^\dagger} \Bigl\{
        \operatorname{tr}(MM^\dagger) -
        \operatorname{tr}\bigl( M\, t(S_E) \, \Lambda^\dagger 
        + \Lambda\, t(S_E) \, M^\dagger \bigr)
        + \operatorname{tr}\bigl(\Lambda (t(S_E))^2 \Lambda^\dagger\bigr)
      \Bigr\}
      \nonumber\\
    &= \min_{S'=S'^\dagger} \Bigl\{
      \operatorname{tr}(MM^\dagger) -
      \operatorname{tr}\bigl(M S' \Lambda^\dagger + \Lambda S' M^\dagger\bigr)
      + \operatorname{tr}\bigl(\Lambda S'^2 \Lambda^\dagger\bigr)
      \Bigr\}\ ,
      \nonumber\\
    &= \min_{S'=S'^\dagger} \operatorname{tr}\Bigl( (M - \Lambda S') (M - \Lambda S')^\dagger \Bigr)\ ,
      \label{z:9Y8aa8ek2C5Y}
  \end{align}
  where the optimization now ranges over all Hermitian operators $S'$ acting on
  $B$.
  On the other hand, using \cref{z:V4lWxgFVfrPu},
  \begin{align}
    \frac14\Ftwo{\rho}{D}
    &= \min \Bigl\{
      \operatorname{tr}(L^\dagger L)\ :\quad
      \rho^{1/2} L + L^\dagger \rho^{1/2} = D
      \Bigr\}\ ,
      \label{z:5igUwN9xh2MZ}
  \end{align}
  where $\rho,D$ refer to operators on $B$.  To prove the
  inequality~\eqref{z:.6Ao-K6k6kYf}, which
  is the first part of our main theorem claim, our strategy is to show that for
  any candidate $S'$ in~\eqref{z:9Y8aa8ek2C5Y}, there is
  a valid candidate $L$ in~\eqref{z:5igUwN9xh2MZ} that achieves the
  same value.  This statement then implies that
  $\text{\eqref{z:5igUwN9xh2MZ}} \leq
  \text{\eqref{z:9Y8aa8ek2C5Y}}$ as desired.

  Recall that $\Lambda\Lambda^\dagger = \rho$ (where $\rho\equiv\rho_B$ for
  short in this proof), and therefore the polar decomposition of $\Lambda$ can
  be written as $\Lambda = \rho^{1/2}W$ for some unitary matrix $W$.  Let $S'$
  be any Hermitian operator that is candidate in the
  optimization~\eqref{z:9Y8aa8ek2C5Y}, and let
  $L=iW(M^\dagger - S'\Lambda^\dagger)$.  Then one can verify that
  \begin{align}
    \rho^{1/2}L + L^\dagger \rho^{1/2}
    = i\Lambda( M^\dagger - S'\Lambda^\dagger) - i (M - \Lambda S') \Lambda^\dagger
    = -i\bigl(M\Lambda^\dagger - \Lambda M^\dagger\bigr) = D\ ,
  \end{align}
  and thus $L$ is a feasible candidate in~\eqref{z:5igUwN9xh2MZ}.
  Furthermore it holds that
  $\operatorname{tr}(L^\dagger L) = \operatorname{tr}\bigl((M - \Lambda S') (M - \Lambda S')^\dagger\bigr)$, thus
  proving the
  inequality~\eqref{z:.6Ao-K6k6kYf}.

  We now show that, assuming
  $(P_{\rho_B}^\perp \otimes P_{\rho_E}^\perp) V \lvert {\xi }\rangle = 0$, the inequality
  becomes an equality.  The proof strategy is to go in reverse direction above,
  starting with an optimal candidate $L$ in~\eqref{z:5igUwN9xh2MZ},
  and constructing a candidate $S'$
  in~\eqref{z:9Y8aa8ek2C5Y} that achieves the same
  value.
  From \cref{z:L9cC7f4RKO0S} we see that
  $(P_{\rho_B}^\perp \otimes P_{\rho_E}^\perp) V \lvert {\xi }\rangle = 0$ is equivalent to
  \begin{align}
    P_\rho^\perp M W^\dagger P_\rho^\perp = 0\ .
    \label{z:oryJp.qQf2mk}
  \end{align}

  Let $L$ be an optimal candidate in~\eqref{z:5igUwN9xh2MZ}, i.e.,
  such that $\rho^{1/2}L + L^\dagger \rho^{1/2} = D$ and
  $\Ftwo{\rho}{D} = 4\operatorname{tr}(L^\dagger L)$.  Without loss of generality, we may
  assume that $P_\rho L = L$, since otherwise $P_\rho L$ would yield a better
  optimization candidate in~\eqref{z:5igUwN9xh2MZ}.
  Denoting by $P_X^{\mathrm{supp}}$ and $P_X^{\mathrm{rng}}$ the projectors onto
  the support and the range of an operator $X$, and defining
  $\tilde{P}_\rho = W^\dagger P_\rho W$, we have
  \begin{align}
    \begin{aligned}
      P_\Lambda^{\mathrm{rng}} &= P_{\Lambda^\dagger}^{\mathrm{supp}}
      = P_\rho\ ,\qquad
      &
      P_{\Lambda}^{\mathrm{supp}} &= P_{\Lambda^\dagger}^{\mathrm{rng}}
      = W^\dagger P_\rho W = \tilde{P}_\rho\ ,
      \\
      P_\rho^\perp &= \mathds{1}- P_\rho\ ,
      &
      \tilde{P}_\rho^\perp &= \mathds{1}- \tilde{P}_\rho\ .
    \end{aligned}
  \end{align}
  Let us compute the object $L P_\rho^\perp$:
  \begin{align}
    L P_\rho^\perp
    &= P_\rho L P_\rho^\perp
    = \rho^{-1/2} \bigl(\rho^{1/2} L + L^\dagger \rho^{1/2}\bigr) P_\rho^\perp
    \nonumber\\
    &= \rho^{-1/2} \bigl(-i(M\Lambda^\dagger - \Lambda M^\dagger)\bigr) P_\rho^\perp
      \nonumber\\
    &= i\rho^{-1/2} \bigl(\Lambda M^\dagger\bigr) P_\rho^\perp
    = i P_\rho W M^\dagger P_\rho^\perp
      = i W \tilde{P}_\rho M^\dagger P_\rho^\perp
      = i W M^\dagger P_\rho^\perp\ ,
      \label{z:dNzzdYPJQIgB}
  \end{align}
  where we have employed~\eqref{z:oryJp.qQf2mk} in the last
  equality.
  
  Now let us get started with constructing $S'$.  Our goal is to find a Hermitian
  matrix $S'$ such that
  \begin{align}
    L \stackrel{!}{=} iW(M^\dagger - S'\Lambda^\dagger)\ .
    \label{z:8OxxomcUuQaZ}
  \end{align}
  Indeed, this would ensure a valid candidate
  in~\eqref{z:9Y8aa8ek2C5Y} reaching the same value as
  $\operatorname{tr}(L^\dagger L)$.  The
  equality~\eqref{z:8OxxomcUuQaZ} is equivalent to
  both simultaneous conditions
    \begin{align}
      L P_\rho &\stackrel{!}{=} iW(M^\dagger - S'\Lambda^\dagger) P_\rho\ ;
      &%
      L P_\rho^\perp &\stackrel{!}{=} iW(M^\dagger - S'\Lambda^\dagger) P_\rho^\perp\ .
      \label{z:f-AyG.quPkz1}
    \end{align}
  The latter follows immediately from~\eqref{z:dNzzdYPJQIgB},
  noting that $\Lambda^\dagger P_\rho^\perp=0$.  It suffices, therefore, to find
  a Hermitian matrix $S'$ such that
  the first equality in~\eqref{z:f-AyG.quPkz1} is satisfied.

  Let $\Lambda^{-1} = W^\dagger \rho^{-1/2}$ noting that
  $\Lambda^{-1}\Lambda = \tilde P_\rho$ and $\Lambda \Lambda^{-1} = P_\rho$.  Define
  \begin{align}
    S' = \Lambda^{-1} \Bigl[
    \Lambda M^\dagger + i\Lambda W^\dagger L
    \Bigr] (\Lambda^{-1})^\dagger + \tilde{P}_\rho^\perp M^\dagger (\Lambda^{-1})^\dagger
    + \Lambda^{-1} M \tilde{P}_\rho^\perp\ .
    \label{z:tdSmtm7KkiX.}
  \end{align}
  First we show that $S'$ is Hermitian by proving that the term in brackets in
  the first term above is, in fact, Hermitian.  Using
  $\Lambda W^\dagger=\rho^{1/2}$ we can compute
  \begin{multline}
    \Bigl[\Lambda M^\dagger + i \Lambda W^\dagger L\Bigr]
    - \Bigl[\Lambda M^\dagger + i \Lambda W^\dagger L\Bigr]^\dagger
    = \bigl(\Lambda M^\dagger - M\Lambda^\dagger\bigr)
      + i\bigl(\rho^{1/2} L + L^\dagger \rho^{1/2}\bigr)
      \\
       = -iD + iD = 0\ ,
  \end{multline}
  using properties of $L$ noted above and
  using~\eqref{z:POoOh29X6gj.}.  Therefore $S'$ is
  Hermitian.  Then
  \begin{align}
    iW(M^\dagger - S'\Lambda^\dagger) P_\rho
    &=
      iWM^\dagger P_\rho
      -i W \tilde{P}_\rho M^\dagger P_\rho
      + P_\rho L P_\rho
      -iW\tilde{P}_\rho^\perp M^\dagger P_\rho
      = L P_\rho\ ,
  \end{align}
  noting that $\tilde{P}_\rho^\perp \Lambda^\dagger = 0$,
  $(\Lambda^{-1})^\dagger\Lambda^\dagger = (\Lambda\Lambda^{-1})^\dagger =
  P_\rho$, and recalling that $P_\rho L = L$.  With this choice of $S'$, the
  first equality in~\eqref{z:f-AyG.quPkz1} is thus also satisfied,
  thereby completing the proof.
\end{proof}

\subsection{Additional equivalent conditions for zero sensitivity loss}
\label{z:r0U85JMqlcN8}

The following theorem provides additional conditions under which zero
sensitivity loss is achieved (see \cref{z:JZtfKZv25too}),
leading to an explicit form of Bob's optimal sensing observable whenever these
conditions are satisfied.
\begin{theorem}
  \label{z:0dFBhZ.ueJZC}
  We use the notation of \cref{z:xqliGh9KZXzi}.
  Suppose that the conditions for our uncertainty relation equality
  (\cref{z:L9cC7f4RKO0S}) hold.
  Then the following statements are equivalent:
  \begin{enumerate}[label=(\roman*)]
  \item\label{z:SOGX2iW5cR10} %
    We have
    $\Ftwo{ \mathcal{N}(\psi) }{ \mathcal{N}(D_A^Y) } =
    4\langle {\xi}\mkern 1.5mu\relax \vert \mkern 1.5mu\relax {\mathcal{N}^\dagger(\mathds{1})}\mkern 1.5mu\relax \vert \mkern 1.5mu\relax {\xi}\rangle $.
  \item\label{z:N9RFx.mOO7dA} We have
    $\operatorname{tr}\bigl(E_{k'}^\dagger E_k D_A^Z\bigr) = 0$ for all $k,k'$, where $\{E_k\}$ is any set of
    Kraus operators for $\mathcal{N}$.
  \item\label{z:OMFhQmpPsAej}
    We have $\widehat{\mathcal{N}}(D_A^Z) = 0$.
  \item\label{z:BmfQ3yH1Z041}
    We have $\Lambda^\dagger M + M^\dagger\Lambda = 0$.
  \item\label{z:mYcXo.8M1Sw-}
    The operator $i\rho^{1/2} M W^\dagger$ is Hermitian.
  \item\label{z:mRLHU-vpcau5}
    The operator
    $i \rho\,\mathcal{N}(\lvert {\xi}\rangle \mkern -1.8mu\relax \langle{\psi}\rvert )$ is Hermitian and
    $\mathcal{N}(\lvert {\xi}\rangle \mkern -1.8mu\relax \langle{\xi}\rvert ) =
    \mathcal{N}(\lvert {\xi}\rangle \mkern -1.8mu\relax \langle{\psi}\rvert )\rho^{-1}\mathcal{N}(\lvert {\psi}\rangle \mkern -1.8mu\relax \langle{\xi}\rvert )$.
  \item\label{z:OMT8MgjCIU1t} The operator
    $i \rho\,\mathcal{N}(\lvert {\xi}\rangle \mkern -1.8mu\relax \langle{\psi}\rvert )$ is Hermitian and
    $\langle {\xi}\mkern 1.5mu\relax \vert \mkern 1.5mu\relax {\mathcal{N}^\dagger(\mathds{1})}\mkern 1.5mu\relax \vert \mkern 1.5mu\relax {\xi}\rangle  =
    \operatorname{tr}\bigl[\mathcal{N}(\lvert {\xi}\rangle \mkern -1.8mu\relax \langle{\psi}\rvert )\rho^{-1}\mathcal{N}(\lvert {\psi}\rangle \mkern -1.8mu\relax \langle{\xi}\rvert )\bigr]$.
  \end{enumerate}

  Furthermore, if these conditions are satisfied then
  \begin{align}
    \mathcal{R}_{\rho_B}^{-1}\bigl( \mathcal{N}(D_A^Y) \bigr)
    = -2i\mathcal{N}(\lvert {\xi}\rangle \mkern -1.8mu\relax \langle{\psi}\rvert )\rho^{-1}
    + 2i \rho^{-1} \mathcal{N}(\lvert {\psi}\rangle \mkern -1.8mu\relax \langle{\xi}\rvert )P_\rho^\perp\ .
    \label{z:Zky0DlQ1CJeC}
  \end{align}
\end{theorem}

\begin{proof}[**z:0dFBhZ.ueJZC]
  The proof of (i)$\Leftrightarrow$(ii)$\Leftrightarrow$(iii) is presented in the
  main text (\cref{z:JZtfKZv25too}).

  \textit{\hbox{\ref{z:OMFhQmpPsAej}~\(\Rightarrow\)~%
      \ref{z:BmfQ3yH1Z041}:}} %
  Write $0 = \widehat{\mathcal{N}}(D_A^Z) %
  = \operatorname{tr}_{B}\bigl\{\Phi_{B:E}\bigl[\Lambda^\dagger M + M^\dagger\Lambda\bigr]\bigr\}$.
  Observe that $\operatorname{tr}_{B}\bigl\{\Phi_{B:E}\,(\cdot)\bigr\}$ is the partial transpose map
  with respect to the bases used to define $\Phi_{B:E}$; therefore
  $\Lambda^\dagger M + M^\dagger\Lambda = 0$.

  \textit{\hbox{\ref{z:BmfQ3yH1Z041}~\(\Leftrightarrow\)~%
      \ref{z:mYcXo.8M1Sw-}:}} %
  We compute
  \begin{align}
    i\rho^{1/2} M W^\dagger - ( i\rho^{1/2} M W^\dagger )^\dagger
    = i W \bigl( \Lambda^\dagger M + M^\dagger \Lambda \bigr) W^\dagger\ ,
    \label{z:tlmteKhL-z.T}
  \end{align}
  which vanishes thanks to the assumption that
  \ref{z:BmfQ3yH1Z041} holds.  Conversely, because
  $W$ is unitary we may only have $\text{\eqref{z:tlmteKhL-z.T}} = 0$ if
  $\Lambda^\dagger M + M^\dagger \Lambda = 0$.

  \textit{\hbox{\ref{z:BmfQ3yH1Z041}~\(\Rightarrow\)~%
      \ref{z:mRLHU-vpcau5}:}} %
  Recall that $\rho=\Lambda\Lambda^\dagger$ and
  $\mathcal{N}(\lvert {\xi}\rangle \mkern -1.8mu\relax \langle{\psi}\rvert ) = \operatorname{tr}_E\bigl(V\lvert {\xi }\rangle \mkern -1.8mu\relax \langle{\psi }\rvert V^\dagger\bigr) = M
  \Lambda^\dagger$.  Then
  $i\rho \mathcal{N}(\lvert {\xi}\rangle \mkern -1.8mu\relax \langle{\psi}\rvert ) = i\Lambda \Lambda^\dagger M
  \Lambda^\dagger$.  To check that $i\rho\mathcal{N}(\lvert {\xi}\rangle \mkern -1.8mu\relax \langle{\psi}\rvert )$ is
  Hermitian we compute
  \begin{align}
    i\rho \mathcal{N}(\lvert {\xi}\rangle \mkern -1.8mu\relax \langle{\psi}\rvert ) - (i\rho \mathcal{N}(\lvert {\xi}\rangle \mkern -1.8mu\relax \langle{\psi}\rvert ))^\dagger
    = i\Lambda\bigl( \Lambda^\dagger M + M^\dagger \Lambda \bigr) \Lambda^\dagger
    = 0\ ,
    \label{z:7nCcgLfY1YhO}
  \end{align}
  using our assumption that \ref{z:BmfQ3yH1Z041}
  holds.  Furthermore, we have
  \begin{align}
    0 = \rho^{-1/2} W \bigl( \Lambda^\dagger M + M^\dagger \Lambda \bigr)
    W^\dagger P_\rho^\perp
    = P_\rho M W^\dagger P_\rho^\perp\ ;
  \end{align}
  recalling point \ref{z:TpPY7RBUzBIR} of
  \cref{z:L9cC7f4RKO0S}, we find that
  \begin{align}
    M W^\dagger P_\rho^\perp = 0\ .
  \end{align}
  Then
  \begin{align}
    \mathcal{N}(\lvert {\xi}\rangle \mkern -1.8mu\relax \langle{\xi}\rvert ) = MM^\dagger
    = M W^\dagger (P_\rho + P_\rho^\perp) W M^\dagger
    = \mathcal{N}(\lvert {\xi}\rangle \mkern -1.8mu\relax \langle{\psi}\rvert ) \rho^{-1} \mathcal{N}(\lvert {\psi}\rangle \mkern -1.8mu\relax \langle{\xi}\rvert )\ .
  \end{align}

  \textit{\hbox{\ref{z:mRLHU-vpcau5}%
      ~\(\Rightarrow\)~%
      \ref{z:OMT8MgjCIU1t}:}} %
  This implication follows immediately from
  $\langle {\xi}\mkern 1.5mu\relax \vert \mkern 1.5mu\relax {\mathcal{N}^\dagger(\mathds{1})}\mkern 1.5mu\relax \vert \mkern 1.5mu\relax {\xi}\rangle  = \operatorname{tr}(\mathcal{N}(\lvert {\xi}\rangle \mkern -1.8mu\relax \langle{\xi}\rvert ))$.

  \textit{\hbox{\ref{z:OMT8MgjCIU1t}%
      ~\(\Rightarrow\)~%
      \ref{z:SOGX2iW5cR10}:}} %
  Our proof strategy for this implication is to show that the expression of the
  symmetric logarithmic derivative in~\eqref{z:Zky0DlQ1CJeC} is
  correct, and that the corresponding Fisher information at Bob's end has no
  sensitivity loss.  Let
  \begin{align}
    R = -2i\mathcal{N}(\lvert {\xi}\rangle \mkern -1.8mu\relax \langle{\psi}\rvert )\rho^{-1} +
    2i\rho^{-1}\mathcal{N}(\lvert {\psi}\rangle \mkern -1.8mu\relax \langle{\xi}\rvert ) P_\rho^\perp\ .
  \end{align}
  We can see that $R$ is Hermitian by writing
  \begin{align}
    R &= -2i \, (P_\rho + P_\rho^\perp)\, \mathcal{N}(\lvert {\xi}\rangle \mkern -1.8mu\relax \langle{\psi}\rvert )\rho^{-1} +
    2i\rho^{-1}\mathcal{N}(\lvert {\psi}\rangle \mkern -1.8mu\relax \langle{\xi}\rvert ) P_\rho^\perp
    \nonumber\\
    &= -2i \rho^{-1}\bigl[ \rho \mathcal{N}(\lvert {\xi}\rangle \mkern -1.8mu\relax \langle{\psi}\rvert ) \bigr] \rho^{-1}
      + \bigl( -2i P_\rho^\perp \mathcal{N}(\lvert {\xi}\rangle \mkern -1.8mu\relax \langle{\psi}\rvert )\rho^{-1} + \textrm{h.c.} \bigr)\ .
  \end{align}
  The first term is Hermitian by assumption and the second term is manifestly
  Hermitian.
  We note for convenience that
  $R P_\rho = -2i \mathcal{N}(\lvert {\xi}\rangle \mkern -1.8mu\relax \langle{\psi}\rvert )\rho^{-1}$ and
  $P_\rho R = 2i \rho^{-1} \mathcal{N}(\lvert {\psi}\rangle \mkern -1.8mu\relax \langle{\xi}\rvert )$.
  We can compute
  \begin{align}
    \frac12\bigl(\rho R + R\rho\bigr)
    = i\mathcal{N}(\lvert {\psi}\rangle \mkern -1.8mu\relax \langle{\xi}\rvert ) - i\mathcal{N}(\lvert {\xi}\rangle \mkern -1.8mu\relax \langle{\psi}\rvert )
    = \mathcal{N}(D_A^Y)\ .
  \end{align}
  Combining with the fact that $P_\rho^\perp R P_\rho^\perp = 0$ we have that
  $\mathcal{R}_\rho^{-1}(\mathcal{N}(D_Y)) = R$ (see also
  \cref{z:uWDssQGxkLXS}), thus
  proving~\eqref{z:Zky0DlQ1CJeC}.
  The Fisher information at the output of the mapping $\mathcal{N}$ is therefore
  \begin{align}
    \Ftwo{\mathcal{N}(\psi)}{\mathcal{N}(D_A^Y)}
    &= \operatorname{tr}(\rho R^2)
      = \operatorname{tr}\bigl[ \rho \, \bigl(2i\rho^{-1}\mathcal{N}(\lvert {\psi}\rangle \mkern -1.8mu\relax \langle{\xi}\rvert )\bigr)
      \bigl(-2i\mathcal{N}(\lvert {\xi}\rangle \mkern -1.8mu\relax \langle{\psi}\rvert )\rho^{-1}\bigr) \bigr]
      \nonumber\\
      &= 4\operatorname{tr}\bigl(\mathcal{N}(\lvert {\xi}\rangle \mkern -1.8mu\relax \langle{\psi}\rvert )\rho^{-1}\mathcal{N}(\lvert {\psi}\rangle \mkern -1.8mu\relax \langle{\xi}\rvert )\bigr)
        = 4\langle {\xi}\mkern 1.5mu\relax \vert \mkern 1.5mu\relax {\mathcal{N}^\dagger(\mathds{1})}\mkern 1.5mu\relax \vert \mkern 1.5mu\relax {\xi}\rangle \ .
  \end{align}
  We conclude that
  \ref{z:SOGX2iW5cR10} holds.
\end{proof}

\subsection{Proof of the generalized bipartite Fisher information uncertainty
  relation for any two parameters}
\label{z:UxjsHdSwwZRl}

In this Appendix, we prove the generalized uncertainty
relation~\eqref{z:4438TdZJTrfX} that
applies to any two parameters generated by unitary evolutions.

\begin{proposition}[Uncertainty relation for any two parameters with associated generators]
  \label{z:55zQoE3Yx2Kr}
  Let $\lvert {\psi}\rangle $ be a state vector on Alice's system, and let $A,B$ be two
  Hermitian operators.  The latter generate two respective parametrized
  evolutions
  \begin{align}
    \partial_a\psi &= -i[A,\psi]\ ;
    &
    \partial_b \psi &= -i[B,\psi]\ .
    \label{z:JnfYHAxNRnH0}
  \end{align}
  Consider the setting depicted in \cref{z:FnHui0ahNLxW}, where
  $\mathcal{N}$ can be any completely positive, trace-nonincreasing map.
  Then
  \begin{align}
    \frac{\FIqty{Bob}{a}}{\FIqty{Alice}{a}} +
    \frac{\FIqty{Eve}{b}}{\FIqty{Alice}{b}}
    \leq
    1 + 2\sqrt{ 1 - \frac{ \bigl \langle { i[A,B] }\bigr \rangle ^2 }{4 \, \sigma_A^2\sigma_B^2} }\ .
    \label{z:L01BMwGDb6HM}
  \end{align}

  Furthermore, assume that $\mathcal{N}[\psi(a)]$ does not change rank locally
  as a function of $a$ and that there exists $\beta\in\mathbb{R}, \beta\neq 0$
  such that
  \begin{align}
    \widehat{\mathcal{N}}\biggl( -i \biggl[ \frac{B}{\sigma_B} \,,\; \psi \biggr] \biggr)
    = \beta\, \widehat{\mathcal{N}}\biggl( \biggl\{ \frac{A-\langle {A}\rangle }{\sigma_A} \,,\; \psi \biggr\} \biggr)\ .
    \label{z:zYQep-v2EUWm}
  \end{align}
  Then
  \begin{align}
    \frac{\FIqty{Bob}{a}}{\FIqty{Alice}{a}} +
    \frac1{\beta^2}\frac{\FIqty{Eve}{b}}{\FIqty{Alice}{b}}
    = 1\ .
    \label{z:j1xGpDvr5yV5}
  \end{align}
\end{proposition}

\begin{corollary}[Uncertainty relation for any two parameters]
  \label{z:WNVeFyPbIaLG}
  Let $\lvert {\psi(a,b)}\rangle $ be any state vector depending on parameters $a,b$.
  Then
  \begin{align}
    \frac{ \FIqty{Bob}{a} }{ \FIqty{Alice}{a} }
    +
    \frac{ \FIqty{Eve}{b} }{ \FIqty{Alice}{b} }
    \leq
    1 + 2\sqrt{ 1 - 
      \frac{ \bigl \langle { i [ \partial_a\psi, \partial_b\psi ] }\bigr \rangle ^2 }{ 4\, \bigl \langle {(\partial_a\psi)^2}\bigr \rangle   \bigl \langle {(\partial_b\psi)^2}\bigr \rangle  }
    }\ .
  \end{align}
\end{corollary}
We first prove the following lemma.

\begin{lemma}
  \label{z:sHEqk0uiteCv}
  Let $\lvert {\psi}\rangle $ be any state vector and let $\mathcal{M}$ be any completely
  positive, trace-nonincreasing map.  Consider two Hermitian operators $C,B$
  generating respective evolutions
  \begin{align}
    \partial_c\psi &= -i[C,\psi]\ ,
    &
      \partial_b\psi &= -i[B,\psi]\ .
  \end{align}
  We write $\rho=\mathcal{M}(\psi)$, $\partial_c \rho = \mathcal{M}(-i[C,\psi])$
  and $\partial_b \rho = \mathcal{M}(-i[B,\psi])$.  Then for any $x,y>0$,
  \begin{align}
    \frac{y}{\sigma_B^2} \Ftwo{\rho}{\partial_b\rho}
    \leq
    \frac{x}{\sigma_C^2} \Ftwo{\rho}{\partial_c\rho}
    + 4(x+y)\sqrt{
    1 - \frac{xy}{(x+y)^2}%
    \frac{4\bigl[\operatorname{Re}\langle {\bar C\bar B}\rangle \bigr]^2 }{ \sigma_C^2\sigma_B^2 }
    }\ ,
    \label{z:TqCnGtoCA6EM}
  \end{align}
  where $\bar{C} = C - \langle {C}\rangle \mathds{1}$ and $\bar{B} = B - \langle {B}\rangle \mathds{1}$.  In
  addition, suppose that $C$ can be written as $C = i \alpha[A, \psi]$ for some
  Hermitian operator $A$ and some $\alpha\in\mathbb{R}$.  Then the above
  inequality takes the form
  \begin{align}
    \frac{y}{\sigma_B^2} \Ftwo{\rho}{\partial_b\rho}
    \leq
    \frac{x}{\alpha^2\sigma_A^2} \Ftwo{\rho}{\partial_c\rho}
    + 4(x+y)\sqrt{
    1 - \frac{xy}{(x+y)^2}%
    \frac{ \bigl \langle { i[A,B] }\bigr \rangle ^2 }{ \sigma_A^2\sigma_B^2 }
    }\ .
    \label{z:gBgdUl4f-5Hj}
  \end{align}

  Furthermore, let $x,y>0$.  If there exists $s\in \{+1,-1\}$ such that
  \begin{align}
    \mathcal{M}\biggl(-i\biggl[\frac{\sqrt{y} B}{\sigma_B}
        + s\frac{\sqrt{x} C}{\sigma_C}, \psi\biggr]\biggr) = 0\ ,
    \label{z:xZU8ZXiZTQIk}
  \end{align}
  then
  \begin{align}
    \frac{y}{\sigma_B^2} \Ftwo{\rho}{\partial_b\rho}
    =
    \frac{x}{\sigma_C^2} \Ftwo{\rho}{\partial_c\rho}\ .
    \label{z:NJ6HBldDH3Zy}
  \end{align}
\end{lemma}
\begin{proof}[*z:sHEqk0uiteCv]
  For any $x,y>0$, define the shorthands
  \begin{align}
    \tilde{C} &= \frac{\sqrt x}{\sigma_C}\,\bigl(C-\langle {C}\rangle \bigr)\ ,
    &
      \tilde{B} &= \frac{\sqrt y}{\sigma_B}\,\bigl(B-\langle {B}\rangle \bigr)\  .
  \end{align}
  Observe that $\sigma_{\tilde{C}}^2 = x$ and $\sigma_{\tilde{B}}^2 = y$.
  Furthermore, we define for convenience
  $D_{(\cdot)} = \mathcal{M}(-i[(\cdot),\psi])$, observing that
  $D_C = \partial_c\rho$ and $D_B = \partial_b\rho$.  Then using
  \cref{z:xptYZVgGge6R} we see that
  \begin{align}
  \Ftwo{\rho}{D_{\tilde{C}}}
    &= \frac{x}{\sigma_C^2} \Ftwo{\rho}{D_{C}}\ ;
    &
  \Ftwo{\rho}{D_{\tilde{B}}}
    &= \frac{y}{\sigma_B^2} \Ftwo{\rho}{D_{B}}\ .
      \label{z:Zw.HWelFzZZR}
  \end{align}
  Invoking \cref{z:wXcXJs-gCbik},
  \begin{align}
    \Ftwo{\rho}{D_{\tilde{B}}}
    \leq
    \Ftwo{\rho}{D_{\tilde{C}}}
    + \Bigl[ \Ftwo{\rho}{\Delta_+} \Ftwo{\rho}{\Delta_-} \Bigr]^{1/2}\ ,
    \label{z:TTYLIpK9z0SK}
  \end{align}
  where $\Delta_\pm = D_{\tilde{C}}\pm D_{\tilde{B}} = D_{\tilde{C}\pm\tilde{B}}$.
  We proceed to compute the second term on the right-hand side of this
  inequality.  The data-processing inequality
  (\cref{z:x2KvR6Lo6ilw}), along with
  \cref{z:Sj25VG-tiqC5}, gives us
  \begin{align}
    \Ftwo{\rho}{\Delta_\pm}
    \leq \Ftwo{\psi}{ -i[\tilde{C}\pm \tilde{B}, \psi] }
    = 4\Var_\psi\bigl(\tilde{C} \pm \tilde{B}\bigr)\ ,
    \label{z:V.44V5cbsSxA}
  \end{align}
  where we write $\Var_\rho(X) = \langle {X^2}\rangle _\rho - (\langle {X}\rangle _\rho)^2$.
  We find
  \begin{align}
    4\Var_\psi\bigl(\tilde{C}\pm\tilde{B}\bigr)
    = 4\,\langle { (\tilde{C}\pm \tilde{B})^2}\rangle 
    = 4\,\langle { {\tilde{C}}^2 + {\tilde{B}}^2 \pm \{\tilde{C},\tilde{B}\} }\rangle 
    = 4\bigl(x+y\bigr) \pm 8\operatorname{Re}\,\langle {\tilde{C}\tilde{B}}\rangle \ .
  \end{align}
  Then
  \begin{align}
    4^2\Var_\psi\bigl(\tilde{C} + \tilde{B}\bigr) \Var_\psi\bigl(\tilde{C} - \tilde{B}\bigr)
    = 4^2\bigl(x+y\bigr)^2 - 8^2\frac{xy}{\sigma_C^2\sigma_B^2}\bigl[\operatorname{Re}\,\langle {\bar{C}\bar{B}}\rangle \bigr]^2\ ,
  \end{align}
  where $\bar{C} = C-\langle {C}\rangle $ and $\bar{B} = B-\langle {B}\rangle $.  Combining the above,
  \begin{align}
    \Bigl[\Ftwo{\rho}{\Delta_+} \Ftwo{\rho}{\Delta_-}\Bigr]^{1/2}
    \leq
    4(x+y)\,\sqrt{ 1
    - \frac{xy}{(x+y)^2}\,
      \frac{4\bigl[\operatorname{Re}\,\langle {\bar C\bar B}\rangle \bigr]^2}{\sigma_C^2\sigma_B^2}
    } {.}
  \end{align}
  Plugging this expression back into~\eqref{z:TTYLIpK9z0SK}, along
  with~\eqref{z:Zw.HWelFzZZR},
  proves~\eqref{z:TqCnGtoCA6EM}.
  Now suppose that $C = i\alpha [A,\psi]$ for some Hermitian operator $A$ and
  for a real number $\alpha$.  Then $\langle {C}\rangle =0$ so $\bar{C}=C$ and
  \begin{align}
    \operatorname{Re}\,\langle {\bar{C}\bar{B}}\rangle 
    &= \alpha \operatorname{Re}\,\bigl \langle { i[A,\psi]\,\bar{B} }\bigr \rangle 
      = \alpha \operatorname{Re}\,\bigl( \langle {\psi}\mkern 1.5mu\relax \vert \mkern 1.5mu\relax {iH\lvert {\psi }\rangle \mkern -1.8mu\relax \langle{\psi }\rvert \bar{B}}\mkern 1.5mu\relax \vert \mkern 1.5mu\relax {\psi}\rangle  - \langle {\psi}\mkern 1.5mu\relax \vert \mkern 1.5mu\relax {i A \bar{B}}\mkern 1.5mu\relax \vert \mkern 1.5mu\relax {\psi}\rangle  \bigr)
    \nonumber\\
    &
      = \alpha \operatorname{Re}\bigl( - i \langle {A \bar{B}}\rangle  \bigr)
      = \frac\alpha2\bigl( - i\langle {A \bar{B}}\rangle  + i\langle {\bar{B} A}\rangle  \bigr)
      = \frac\alpha2\bigl \langle { i[A, \bar{B}] }\bigr \rangle 
      = \frac\alpha2\bigl \langle { i[A, B] }\bigr \rangle \ .
  \end{align}
  \Cref{z:gBgdUl4f-5Hj} follows from
  this and using the fact that
  $\sigma_C^2 = \langle {\bar{C}^2}\rangle  = \alpha^2\langle {-(A\psi -\psi A)^2}\rangle  =
  \alpha^2\bigl(\langle {A^2}\rangle -\langle {A}\rangle ^2\bigr) = \alpha^2\sigma_A^2$.

  Now assume that
  \cref{z:xZU8ZXiZTQIk}
  is satisfied.  Recalling that
  $ \Delta_\pm = D_{\tilde C} \pm D_{\tilde B} = \mathcal{M}\bigl(-i[ \tilde{C}
  \pm \tilde{B}, \psi ]\bigr) $, we find that
  condition~\eqref{z:xZU8ZXiZTQIk}
  immediately implies that either $\Delta_+=0$ or $\Delta_-=0$ and therefore
  either $\Ftwo{\rho}{\Delta_+}=0$ or $\Ftwo{\rho}{\Delta_-}=0$.  In this case,
  \cref{z:wXcXJs-gCbik} immediately implies that
  $\Ftwo{\rho}{D_{\tilde C}} = \Ftwo{\rho}{D_{\tilde B}}$.  We conclude
  that~\eqref{z:NJ6HBldDH3Zy}
  holds, recalling~\eqref{z:Zw.HWelFzZZR}.
\end{proof}

\begin{proof}[*z:55zQoE3Yx2Kr]
  Consider the evolution $\psi(a,c)$, where the parameter $a$ is generated by
  the first given Hermitian operator $A$ and where the parameter $c$ is
  generated by the complementary generator $C$ (as per
  \cref{z:KI2eOYZK2.Ln} in the main text) given by
  \begin{align}
    \partial_c\psi
    &= i[C, \psi]\ ,
    &
      C &= \frac1{2\Var_\psi(A)}\,\bigl(-i[A,\psi]\bigr)\ .
  \end{align}
  Recall $\FIqty{Alice}{c} = 4\sigma_C^2 = \sigma_A^{-2}$
  from~\eqref{z:lN.pbrjbNave} with $H\to A$ and $T\to C$.  Our time-energy
  uncertainty relation, in its form of
  \cref{z:4VI7KY5eDBCv}, asserts that
  \begin{align}
    \frac1{4\sigma_A^2}\,\FIqty{Bob}{a}
    + \sigma_A^2\, \FIqty{Eve}{c}
    \leq 1\ .
    \label{z:gcLqGPLmRz0.}
  \end{align}
  Now we invoke \cref{z:sHEqk0uiteCv}, with
  $\mathcal{M} = \widehat{\mathcal{N}}$, $c$, $b$, $C = i\alpha [A,\psi]$,
  $\alpha=-(2\sigma_A^2)^{-1}$, $B$, and $x=y=1/4$.
  From~\eqref{z:gBgdUl4f-5Hj} we
  find
  \begin{align}
    \frac{1}{4\sigma_B^2}\,\FIqty{Eve}{b}
    &\leq \sigma_A^2 \FIqty{Eve}{c}
    + 2\sqrt{1 - \frac{\bigl \langle { i[A,B] }\bigr \rangle ^2}{4\sigma_A^2\sigma_B^2} }\ .
      \label{z:8kXHjm6xYWCe}
  \end{align}
  We find, applying~\eqref{z:gcLqGPLmRz0.}
  and~\eqref{z:8kXHjm6xYWCe} in succession,
  \begin{align}
    \frac1{4\sigma_A^2} \FIqty{Bob}{a} + \frac1{4\sigma_B^2} \FIqty{Eve}{b}
    &\leq
    1 - \sigma_A^2 \FIqty{Eve}{c}
      + \frac1{4\sigma_B^2} \FIqty{Eve}{b} 
      \nonumber\\
     &
       \leq
       1+ 2\sqrt{1 - \frac{ \bigl \langle { i [A,B] }\bigr \rangle ^2 }{4\sigma_A^2\sigma_B^2} }\ .
  \end{align}
  This shows the desired uncertainty relation.

  Now assume that $\mathcal{N}[\psi]$ does not change rank locally as a
  function of $a$ and that
  \cref{z:zYQep-v2EUWm} holds.  Let
  $\mathcal{M} = \widehat{\mathcal{N}}$, %
  $\rho_E = \widehat{\mathcal{N}}[\psi]$,
  $C = i\alpha [A,\psi]$, and
  $\alpha=-(2\sigma_A^2)^{-1}$. %
  Then as computed above $\sigma_C = \lvert {\alpha}\rvert \sigma_A=1/(2\sigma_A)$.  Let us
  compute now
  \begin{align}
    -i\biggl[\sqrt{x} \frac{C}{\sigma_C}, \psi\biggr]
    = -i\biggl[ \sqrt{x}\frac{-i[A,\psi]}{\sigma_A} , \psi \biggr]
    = -\sqrt{x}\biggl\{ \frac{A - \langle {A}\rangle }{\sigma_A}, \psi \biggr\} \ ,
  \end{align}
  recalling that $\bigl[[A, \psi], \psi\bigr] = \{A-\langle {A}\rangle ,\psi\}$.  Let $y=1$,
  $x=\lvert {\beta}\rvert ^2$ and $s=\operatorname{sign}(\beta)$ such that
  $s\sqrt{x}/\sqrt{y} = \beta$.  We then have
  \begin{align}
    \mathcal{M}\biggl(-i \biggl[ \sqrt{y}\frac{B}{\sigma_B} + s \sqrt{x}\frac{C}{\sigma_C}\,,\; \psi \biggr]\biggr)
    &= \widehat{\mathcal{N}}\biggl(-i \biggl[ \sqrt{y}\frac{B}{\sigma_B} , \psi \biggr]\biggr)
      + s \widehat{\mathcal{N}}\biggl(-i \biggl[ \sqrt{x}\frac{C}{\sigma_C} , \psi \biggr]\biggr)
      \nonumber\\
    &= \frac{\sqrt y}2\,\Biggl\{
      \widehat{\mathcal{N}}\biggl(-i \biggl[ \frac{B}{\sigma_B} , \psi \biggr]\biggr)
      - \beta
      \widehat{\mathcal{N}}\biggl( \biggl\{ \frac{A - \langle {A}\rangle }{\sigma_A}, \psi \biggr\} \biggr) \Biggr\}
      \nonumber\\
    &= 0\ .
  \end{align}
  The latter expression then vanishes thanks to our assumption that
  \cref{z:zYQep-v2EUWm} holds.  Thanks
  to \cref{z:sHEqk0uiteCv} we find
  \begin{align}
    \frac{\FIqty{Eve}{b}}{\FIqty{Alice}{b}} =
    \frac{1}{4\sigma_B^2}\,\Ftwo{\rho_E}{\partial_b\rho_E}
    = \frac{x}{4\sigma_C^2}\,\Ftwo{\rho_E}{\partial_c\rho_E}
    = \beta^2\frac{\FIqty{Eve}{c}}{\FIqty{Alice}{c}}\ .
    \label{z:jnjDUJ4.BheE}
  \end{align}
  Thanks to our assumption that $\mathcal{N}[\psi]$ does not change rank locally
  as a function of $a$, we know that our main uncertainty relation
  (\cref{z:4VI7KY5eDBCv}) holds with
  equality:
  \begin{align}
    \frac1{4\sigma_A^2}\,\FIqty{Bob}{a} + \sigma_A^2\, \FIqty{Eve}{c}
    = 1\ .
  \end{align}
  We therefore find, recalling $\sigma_A^2 = 1/(4\sigma_C^2)$,
  \begin{align}
    \frac{\FIqty{Bob}{a}}{\FIqty{Alice}{a}} + \frac1{\beta^2}\frac{\FIqty{Eve}{b}}{\FIqty{Alice}{b}}
    &= \biggl[1 - \sigma_A^2 \FIqty{Eve}{c} \biggr] + \frac1{\beta^2}\frac{\FIqty{Eve}{b}}{4\sigma_B^2}
      = 1 - \frac1{\beta^2}\frac{\FIqty{Eve}{c}}{4\sigma_C^2}
         + \frac1{\beta^2}\frac{\FIqty{Eve}{b}}{4\sigma_B^2}
      = 1\ ,
  \end{align}
  using~\eqref{z:jnjDUJ4.BheE}, thus proving the claim.
\end{proof}

\begin{proof}[*z:WNVeFyPbIaLG]
  The main idea of this corollary is to note that the Fisher information depends
  only on the state and its first derivative with respect to the parameter,
  and that any derivative $\partial_a\psi$ can be written in the form
  $\partial_a\psi = -i[A,\psi]$ for some Hermitian generator $A$.  Therefore we
  seek Hermitian operators $A,B$ such that $\partial_a\psi = -i[A,\psi]$ and
  $\partial_b\psi = -i[B,\psi]$, such that we can apply
  \cref{z:55zQoE3Yx2Kr}.
  We let $A = i\bigl[ \partial_a\psi, \psi \bigr]$ and
  $B = i\bigl[ \partial_b\psi, \psi \bigr]$, and we compute
  \begin{align}
    -i[A, \psi] = -i\bigl[ i [\partial_a\psi, \psi], \psi\bigr]
    = \bigl\{ \partial_a\psi , \psi \bigr\}
    = \partial_a(\psi^2) = \partial_a\psi\ ,
    \label{z:5cmpYLxoPvTl}
  \end{align}
  using~\eqref{z:v.5MaZHkqw71} and the fact that
  $\langle {\partial_a\psi}\rangle  = \operatorname{tr}[\partial_a\psi] = \partial_a \operatorname{tr}(\psi) = 0$.
  Similarly,
  \begin{align}
    -i[B, \psi] = \partial_b\psi\ .
  \end{align}
  We can therefore apply \cref{z:55zQoE3Yx2Kr}.  It
  remains to compute the quantities appearing in the right-hand side
  of~\eqref{z:L01BMwGDb6HM}.  We have
  \begin{align}
    \langle { i[A, B] }\rangle 
    &= i\operatorname{tr}\Bigl\{ \psi\, \bigl[ i[ \partial_a\psi, \psi], i[\partial_b\psi, \psi] \bigr] \Bigr\}
      \nonumber\\
    &= i\operatorname{tr}\Bigl\{\bigl[ \psi \,,\,  i[ \partial_a\psi, \psi] \bigr] \, \bigl(i[\partial_b\psi, \psi]\bigr) \Bigr\}
      \nonumber\\
    &= \operatorname{tr}\Bigl\{-i\bigl[ i[ \partial_a\psi, \psi]  \,,\, \psi  \bigr] \, \bigl(i[\partial_b\psi, \psi]\bigr) \Bigr\}
      \nonumber\\
    &= \operatorname{tr}\Bigl\{ (\partial_a\psi) \, \bigl(i[\partial_b\psi, \psi]\bigr) \Bigr\}
      \nonumber\\
    &= \bigl \langle {i [ \partial_a\psi, \partial_b\psi ] }\bigr \rangle \ ,
  \end{align}
  using the cyclicity of the trace and invoking~\eqref{z:5cmpYLxoPvTl}
  for the fourth equality.
  Furthermore
  \begin{align}
    \sigma_A^2
    &= \langle {A^2}\rangle  - \langle {A}\rangle ^2
      = \bigl \langle { (i[\partial_a\psi, \psi])^2 }\bigr \rangle 
      - \bigl \langle {i[\partial_a\psi, \psi]}\bigr \rangle ^2
      \nonumber\\&
    = -\bigl \langle { \bigl((\partial_a\psi)\,\psi - \psi\,(\partial_a\psi)\bigr)
      \bigl((\partial_a\psi)\, \psi - \psi\,(\partial_a\psi)\bigr) }\bigr \rangle 
    \nonumber\\
    &= \bigl \langle { \psi\,(\partial_a\psi)\,(\partial_a\psi)\,\psi }\bigr \rangle 
      = \bigl \langle { (\partial_a\psi)^2 }\bigr \rangle \ ,
  \end{align}
  where we have made use of $\psi \,(\partial_a\psi)\, \psi = 0$.  Similarly
  $\sigma_B^2 = \bigl \langle { ( \partial_b\psi )^2 }\bigr \rangle $, which ends the proof.
\end{proof}

\section{Generalizations to infinite-dimensional Hilbert spaces}
\label{z:3vQhWOZ7civK}

While the main text has put an emphasis on discussing notions of quantum
metrology making use of finite-dimensional quantum systems, in this section,
we %
generalize the above findings to the setting of infinite-dimensional Hilbert
spaces.  A specific attention is given to %
unbounded operators, as many physical systems of practical use fall under this
category.

\subsection{Uncertainty relation for any two parameters}

\newcommand*{\opA}{\mathtt{A}}
\newcommand*{\opB}{\mathtt{B}}

We start with a generalisation of \cref{z:incm1nGt-mIE} to infinite dimensions (c.f. \cref{z:55zQoE3Yx2Kr}).

\begin{theorem}[Uncertainty relation for infinite-dimensional systems]
  \label{z:Tja8FC4elFJW}  
  Let $\opA,\opB$ be %
  two self-adjoint operators (possibly unbounded) on a separable Hilbert space
  $\mathscr{H}_A$ with domains $\mathcal{D}(\opA)$ and $\mathcal{D}(\opB)$, respectively. Let
  $\lvert {\psi}\rangle \in\mathcal{D}(\opA)\cap\mathcal{D}(\opB)$ and
  $\lvert {\psi(a)}\rangle  \in\mathcal{D}(\opA)$, $\lvert {\psi(b)}\rangle  \in\mathcal{D}(\opB)$ for some
  $b,a\in\mathbb{R}$ where $\lvert {\psi(a)}\rangle :=e^{-ia \opA}\lvert {\psi}\rangle $,
  $\lvert {\psi(b)}\rangle :=e^{-ib \opB }\lvert {\psi}\rangle $.  Let $V_{A\to BE}$ be any isometry
  $\mathscr{H}_A\to\mathscr{H}_B\otimes\mathscr{H}_E$, where the Hilbert spaces $\mathscr{H}_B, \mathscr{H}_E$
  associated with Bob and Eve are also separable and possibly of infinite
  dimensions.  Consider the two pure state evolutions given
  by~\eqref{z:JnfYHAxNRnH0}.
  Then~\eqref{z:L01BMwGDb6HM} holds,
  with the following quantities defined by
  \begin{align}
	\bigl \langle {i [\opA,\opB ] }\bigr \rangle &:= i\bigl \langle { \opA\psi, \opB \psi }\bigr \rangle - i\bigl \langle { \opB \psi, \opA\psi }\bigr \rangle ,\\
	\bigl \langle { \opA^2 }\bigr \rangle  & := \bigl \langle { \opA\psi, \opA\psi }\bigr \rangle 
	\end{align}
and
 \begin{align}
F_{{M}}(y)&:= \liminf_{l\to\infty} \operatorname{tr}\left[\rho_{{M}}^{(l)}(y)R^2\right]\in\mathbb{R},
  \end{align}
  where ${M}\in\{B,E\},\, y\in\{a,b\}$ and $\rho_X^{(l)}$ is an $l$-dimensional subnormalised density operator and $R=R(l)$ is defined in \cref{z:d-w60nHacIQJ} on an 
  $l$-dimensional Hilbert space for $\rho_{{M}}^{(l)}$. Specifically, 
  \begin{align}
  \rho_{{M}}^{(l)}(y)&:= \operatorname{tr}_{\backslash {M}} \left[P_{BE}^{(l)} V_{A\to BE} \rho_A(y) V^\dag_{A\to BE} P_{BE}^{(l)} \right],
  \end{align}
  where $\backslash E:=B$, $\backslash B:=E$, and $P_{BE}^{(l)}$ is the orthogonal projection onto the first $l$ basis elements of a basis for $\mathscr{H}_B\otimes\mathscr{H}_E$. 
  Furthermore, the derivative of $\rho_{{M}}^{(l)}(y)$ is defined via
  \begin{align}
  \frac{d}{dy}	\rho_{{M}}^{(l)}(y)&:= \operatorname{tr}_{\backslash {M}} \left[P_{BE}^{(l)} V_{A\to BE} \frac{d}{dy} \rho_A(y) V^\dag_{A\to BE} P_{BE}^{(l)} \right],
  \end{align}
where 
  \begin{align}
  \frac{d}{da} \rho_A(a)&:=  i \lvert {\psi(a)}\rangle  \langle {\psi(a)}\rvert  \opA- i \opA \lvert {\psi(a)}\rangle  \langle {\psi(a)}\rvert ,\\ \vspace{-5cm}\nonumber\\
  \frac{d}{db} \rho_A(b)&:=  i \lvert {\psi(b)}\rangle  \langle {\psi(b)}\rvert  \opB - i \opB  \lvert {\psi(b)}\rangle  \langle {\psi(b)}\rvert .
  \end{align} 
\end{theorem}

\begin{proof}[*z:Tja8FC4elFJW]
The proof will proceed in two steps. First we will approximate $\opB $ and $\opA$ by bounded operators (if they are already bounded, then this first step is not necessary, although the approximation will nevertheless be well defined). Second, we will approximate these bounded operators by finite dimensional operators. Then we will apply \cref{z:L01BMwGDb6HM} before taking a sequence of limits in which the approximations vanish.
We start with a few elementary definitions and results which will be necessary for our proof.

Let $A$, $(A_n)_n$, be bounded operators on a Hilbert space $\mathscr{H}$. We define all
bounded operators we consider to have domain equal to the entire Hilbert space.
We say that $A_n$ converges (as $n\to\infty$) to $A$ in the strong limit if
$A_n\Psi \to A\Psi$ as $n\to \infty$ for any $\Psi\in \mathscr{H}$. We denote this as
$A_n \stackrel{s}{\to} A$. Some properties are the following.
\begin{itemize}
	\item [i)] Let $A$, $(A_n)_n$, $B$, $(B_n)_n$, $C$, $(C_n)_n$, be bounded operators on a Hilbert space $\mathscr{H}$. $A_n \stackrel{s}{\to} A$, $B_n \stackrel{s}{\to} B$ and  $C_n \stackrel{s}{\to} C$  imply $A_nB_n \stackrel{s}{\to} AB $ and $A_nB_nC_n \stackrel{s}{\to} ABC$.

	\emph{Proof.} $(A_nB_n - AB)\Psi= A_n(B_n-B)\Psi + (A_n-A)B\Psi$. By the uniform {boundedness} principle, $A_n \stackrel{s}{\to} A$ implies $\| A_n\| \leq c$ for some $c\in \mathbb{R}$ for all $n$. Therefore, 
	\begin{align}
	\| (A_nB_n - AB)\Psi \| \leq c \|(B_n-B)\Psi\| + \|(A_n-A)B\Psi\|,
	\end{align}
	where the r.h.s.\@ tends to zero as $n\to \infty$. This {proves} the first claim. For the second, simply define $\bar A_n:= A_n B_n$. Hence $\bar A_n \stackrel{s}{\to} AB$ and thus $\bar A_nC_n \stackrel{s}{\to} (AB) C$, hence proving the second claim.
	\item [ii)] $A_n \stackrel{s}{\to} A$ implies $e^{-i A_n t}\stackrel{s}{\to}e^{-i A t}$ for $t\in\mathbb{R}$.
	
	\emph{Proof.} $e^{-i A_n t}- e^{-i A t}= e^{-i A_n s} e^{-i A (t-s)}\Big{|}_{s=0}^{s=t}= -i \int_0^t ds e^{-i A_n s} (A_n-A) e^{-i A(t-s)} $. But we have $ (A_n-A) e^{-i A(t-s)}\stackrel{s}{\to} \bar 0$ pointwise in $s$, where $\bar 0$ is the bounded operator mapping all vectors in $\mathscr{H}$ to the zero vector in $\mathscr{H}$. Thus via i), $ e^{-i A s} (A_n-A) e^{-i A(t-s)}\stackrel{s}{\to} \bar 0$  pointwise in $s$ and  the result follows by dominated convergence. 
	\item [iii)] Let $A$ be self-adjoint and possibly unbounded. Let $f$, $(f_n)_n$ $:\mathbb{R}\to\mathbb{C}$ be uniformly bounded functions with $f_n \to f$ as $n\to\infty$ pointwise. Then $f_n(A) \stackrel{s}{\to} f(A)$.\\
	\emph{Proof.} See
	Ref.~\cite{R108}. %
\end{itemize}
 We can now prove the theorem. Let $(P_{{N}}^{(n)})_n$ be the orthogonal projections onto the span of the first $n$ basis elements of a separable Hilbert space $\mathscr{H}_{{N}}$. Consider two bounded operators $\tilde \opA$ and 
 $\tilde{\opB }$ 
 on $\mathscr{H}_A$ and define $\tilde \opA_n$, 
 $\tilde{\opB }_n$ by
 \begin{align}
 \tilde \opA_n := P_{A}^{(n)} \tilde \opA P_{A}^{(n)}, \quad  \tilde{\opB }_n := P_{A}^{(n)} 
 \tilde{\opB } P_{A}^{(n)}.
 \end{align} 
Furthermore, consider the sequence of states $(\rho^{(n,l)}_B)_{n,l}$ on $\mathscr{H}_B$, and  $(\rho^{(n,l)}_E)_{n,l}$ on $\mathscr{H}_E$, 
where 
\begin{align}
\begin{split}
	\rho^{(n,l)}_B(a)&:=\operatorname{tr}_E\left[P_{BE}^{(l)} V_{A\to BE} P_{A}^{(n)} \left(    \lvert {\psi_n(a)}\rangle  \langle {\psi_n(a)}\rvert  \right) P_{A}^{(n)} V_{A\to BE}^\dag P_{BE}^{(l)} \right], \\
    \rho^{(n,l)}_E(b)&:=\operatorname{tr}_B\left[P_{BE}^{(l)} V_{A\to BE} P_{A}^{(n)} \left(   \lvert {\psi_n(b)}\rangle  \langle {\psi_n(b)}\rvert   \right) P_{A}^{(n)} V_{A\to BE}^\dag P_{BE}^{(l)}  \right].
\end{split}\label{z:xqLpOGHC26c.}
\end{align}
where $\lvert {\psi_n(a)}\rangle :=e^{-i a\tilde \opA_n} \lvert {\psi}\rangle $, $\lvert {\psi_n(b)}\rangle :=e^{-i b\tilde{\opB }_n} \lvert {\psi}\rangle $ 
and the sequences of derivatives,  $(\frac{d}{da} \rho^{(n,l)}_B(t))_{n,l}$ on 
$\mathscr{H}_B$, and  $(\frac{d}{da} \rho^{(n,l)}_E(a))_{n,l}$ on $\mathscr{H}_E$ are
\begin{align}
\begin{split}
\frac{d}{da}\rho^{(n,l)}_B(a)&=\operatorname{tr}_E\left[P_{BE}^{(l)} V_{A\to BE} P_{A}^{(n)} \left(   i \lvert {\psi_n(a)}\rangle  \langle {\psi_n(a)}\rvert  \tilde \opA_n- i \tilde \opA_n \lvert {\psi_n(a)}\rangle  \langle {\psi_n(a)}\rvert  \right) P_{A}^{(n)} V_{A\to BE}^\dag P_{BE}^{(l)}  \right], \\
\frac{d}{db}\rho^{(n,l)}_E(b)&=\operatorname{tr}_B\left[P_{BE}^{(l)} V_{A\to BE} P_{A}^{(n)} \left(   i \lvert {\psi_n(b)}\rangle  \langle {\psi_n(b)}\rvert  
\tilde{\opB }_n- i \tilde{\opB }_n \lvert {\psi_n(b)}\rangle  \langle {\psi_n(b)}\rvert  \right) P_{A}^{(n)} V_{A\to BE}^\dag P_{BE}^{(l)}\right].
\end{split}\label{z:Vz6mzHb5RLbX}
\end{align}
We can use \cref{z:xqLpOGHC26c.,z:Vz6mzHb5RLbX} to construct the Fisher information for these states.
Since $V_{A\to BE}V_{A\to BE}^\dag= \mathds{1}_{BE}$, where $\mathds{1}_{BE}$ is the identity operator on $\mathscr{H}_{BE}$, it follows that \begin{align}
P_{A}^{(n)}- \left(P_{BE}^{(l)} V_{A\to BE} P_{A}^{(n)} \right)^\dag \left(P_{BE}^{(l)} V_{A\to BE} P_{A}^{(n)}\right) \geq 0
\end{align}
for all $l,n$. Hence, by Kraus' theorem, \cref{z:xqLpOGHC26c.} are completely
positive and trace-nonincreasing maps evaluated on inputs
$\lvert {\psi_n(a)}\rangle  \langle {\psi_n(a)}\rvert $.
Since \cref{z:55zQoE3Yx2Kr} holds for any completely positive, trace-nonincreasing map, we can apply it to our setup. This yields
\begin{align}
\frac{F\left(\rho^{(n,l)}_B(a)\right)}{F\left(\rho_A^{(n)}(a)\right)} +
\frac{F\left(\rho^{(n,l)}_E(b)\right)}{F\left(\rho_A^{(n)}(b)\right)}
\leq
1 + 2\sqrt{ 1 - \frac{ \bigl \langle { i[\tilde \opA_n,\tilde{\opB }_n] }\bigr \rangle ^2 }{4 \, {\tilde\sigma_{\opA,n}}^2{\tilde\sigma_{\opB ,n}}^2} }\ ,
\label{z:Qu6W-Cu1JXfQ}
\end{align}
recalling that the uncertainty relation also
applies to subnormalized positive operators, and
where
\begin{align}
F\left(\rho_A^{(n)}(a)\right)&= 4\, \tilde\sigma_{\opA,n}^2,\\%
F\left(\rho_A^{(n)}(b)\right)&= 4\, \tilde\sigma_{\opB ,n}^2,\\%
\tilde\sigma_{\opA,n}&:= \left( \bigl \langle { \tilde \opA_n \psi,\tilde \opA_n \psi }\bigr \rangle -  \bigl \langle { \psi,\tilde \opA_n \psi }\bigr \rangle ^2 \right)^{1/2},\\
\tilde\sigma_{\opB ,n}&:= \left( \bigl \langle { \tilde{\opB }_n \psi,\tilde{\opB }_n \psi }\bigr \rangle -  \bigl \langle { \psi,\tilde{\opB }_n \psi }\bigr \rangle ^2 \right)^{1/2}.
\end{align}
We can now take the limit $n\to \infty$ on both sides of \cref{z:Qu6W-Cu1JXfQ}. Due to property i), it follows  
\begin{align}
\frac{\lim_{n\to\infty} F\left(\rho^{(n,l)}_B(a)\right)}{F\left(\rho_A^{(\infty)}(a)\right)} +
\frac{\lim_{n\to\infty} F\left(\rho^{(n,l)}_E(b)\right)}{F\left(\rho_A^{(\infty)}(b)\right)}
\leq
1 + 2\sqrt{ 1 - \frac{ \bigl \langle { i[\tilde \opA,\tilde{\opB }] }\bigr \rangle ^2 }{4 \, {\tilde\sigma_{\opA}}^2{\tilde\sigma_{\opB }}^2} }\ ,
\label{z:nmMW9OwgBbPX}
\end{align}
where 
\begin{align}
F\left(\rho_A^{(\infty)}(a)\right)&:= 4\, \tilde\sigma_{\opA}^2,\\%
F\left(\rho_A^{(\infty)}(b)\right)&:= 4\, \tilde\sigma_{\opB }^2,\\%
\tilde\sigma_{\opA}&:= \left( \bigl \langle { \tilde \opA \psi,\tilde \opA \psi }\bigr \rangle -  \bigl \langle { \psi,\tilde \opA \psi }\bigr \rangle ^2 \right)^{1/2},\\
\tilde\sigma_{\opB }&:= \left( \bigl \langle { \tilde{\opB } \psi,\tilde{\opB } \psi }\bigr \rangle -  \bigl \langle { \psi,\tilde{\opB } \psi }\bigr \rangle ^2 \right)^{1/2}.
\end{align}
Observe that the quantities $\lim_{n\to\infty} F(\rho^{(n,l)}_B(a))$,
$\lim_{n\to\infty} F(\rho^{(n,l)}_E(b))$ cannot diverge, since it would
contradict the inequality (since the Fisher information is
nonnegative). %
This observation follows alternatively from applying the data-processing
  inequality \eqref{z:.Hy7vObZJO-E} to bound Bob's Fisher information in
  terms of Alice's, followed by talking the $n\to \infty$ limit. Similarly for
  Eve's Fisher information.
By direct calculation, we observe that the Fisher information $F$ of a state $\rho$ on a  $d$-dimensional Hilbert space, according to \cref{z:.eb-akuquBNd,z:d-w60nHacIQJ}, is given by
\begin{align}
F = \sum_{\substack{k,k'=1 \\ \text{s.t. } p_k+p_{k'}>0}}^d
\frac{p_k}{(p_k + p_{k'})^2}\,\left| \bigl \langle {k}\mkern 1.5mu\relax \big \vert \mkern 1.5mu\relax {\frac{d\rho}{da}}\mkern 1.5mu\relax \big \vert \mkern 1.5mu\relax {k'}\bigr \rangle  \right|^2 \ , \label{z:oWJ0dqqkNfon}
\end{align}
where $\rho = \sum_{k=1}^d p_k \lvert {k}\rangle \langle {k}\rvert $.
Hence 
\begin{align}
\lim_{n\to\infty} F\left(\rho^{(n,l)}_B(a)\right) = \lim_{n\to\infty} \sum_{\substack{k,k'=1 \\ \text{s.t. } p_k^{(n,l)}+p_{k'}^{(n,l)}>0}}^{d_B(l)}
\frac{p_k^{(n,l)}}{(p_k^{(n,l)} + p_{k'}^{(n,l)})^2}\,\left| \bigl \langle {k,n,l}\mkern 1.5mu\relax \big \vert \mkern 1.5mu\relax {\frac{d \rho^{(n,l)}_B(a) }{da}}\mkern 1.5mu\relax \big \vert \mkern 1.5mu\relax {k',n,l}\bigr \rangle  \right|^2 \ ,\label{z:ZKqtUEagDlCs}
\end{align}
where $d_B(l)$ is the dimension of Bob's reduced 
system, which is $l$ independent, and $\rho^{(n,l)}_B(a)= \sum_{k=1}^{d_B(l)} p_k^{(n,l)} \lvert {k,n,l}\rangle \langle {k,n,l}\rvert $. Observe that all terms in the summation must be finite in the limit, since they are all nonnegative and we are guaranteed that the r.h.s.\@ of \cref{z:ZKqtUEagDlCs} does not diverge.
Observe that for terms in the summation for which  $\lim_{n\to\infty}  p_k^{(n,l)}+p_{k'}^{(n,l)} > 0 $, the  summation can be interchanged with the limit. However, while for terms such that $p_k^{(n,l)}+p_{k'}^{(n,l)} > 0 $ for all $n$, but $\lim_{n\to\infty}  p_k^{(n,l)}+p_{k'}^{(n,l)} = 0 $, the summation and integration cannot be interchanged, the interchange of the limit and summation will result in the
lower bound
\begin{align}
\lim_{n\to\infty} F\left(\rho^{(n,l)}_B(a)\right) \geq   \sum_{\substack{k,k'=1 \\ \text{s.t. } p_k^{(\infty,l)}+p_{k'}^{(\infty,l)}>0}}^{d_B(l)}
\frac{p_k^{(\infty,l)}}{(p_k^{(\infty,l)} + p_{k'}^{(\infty,l)})^2}\,\left| \bigl \langle {k,\infty,l}\mkern 1.5mu\relax \big \vert \mkern 1.5mu\relax {\frac{d \rho^{(\infty,l)}_B(a) }{da}}\mkern 1.5mu\relax \big \vert \mkern 1.5mu\relax {k',\infty,l}\bigr \rangle  \right|^2 \ ,\label{z:Qwo1LVG1og6j}
\end{align}
where $\rho^{(\infty,l)}_B(a)= \sum_{k=1}^{d_B(l)} p_k^{(\infty,l)} \lvert {k,\infty,l}\rangle \langle {k,\infty,l}\rvert $, with
\begin{align}
\rho^{(\infty,l)}_B(a)&:= \lim_{n\to\infty}\operatorname{tr}_E\left[P_{BE}^{(l)} V_{A\to BE} P_{A}^{(n)} \left(    \lvert {\psi_n(a)}\rangle  \langle {\psi_n(a)}\rvert  \right) P_{A}^{(n)} V_{A\to BE}^\dag P_{BE}^{(l)} \right]\\
&= \operatorname{tr}_E\left[P_{BE}^{(l)} V_{A\to BE}  \left(    \lvert {\tilde \psi(a)}\rangle  \langle {\tilde\psi(a)}\rvert  \right)  V_{A\to BE}^\dag P_{BE}^{(l)} \right],
\end{align}
where $\lvert {\tilde \psi(a)}\rangle :=e^{-i \tilde \opA a }\lvert {\psi}\rangle $ and using properties i) and ii). Similarly, use properties i) and ii) again to obtain
\begin{align}
	\frac{d}{da}\rho^{(\infty,l)}_B(a)&:=\lim_{n\to \infty} \operatorname{tr}_E\left[P_{BE}^{(l)} V_{A\to BE} P_{A}^{(n)} \left(   i \lvert {\psi_n(a)}\rangle  \langle {\psi_n(a)}\rvert  \tilde \opA_n- i \tilde \opA_n \lvert {\psi_n(a)}\rangle  \langle {\psi_n(a)}\rvert  \right) P_{A}^{(n)} V_{A\to BE}^\dag P_{BE}^{(l)}  \right]\\
	&= \operatorname{tr}_E\left[P_{BE}^{(l)} V_{A\to BE} \left(   i \lvert {\tilde \psi(a)}\rangle  \langle {\tilde\psi(a)}\rvert  \tilde \opA- i \tilde \opA \lvert {\tilde\psi(a)}\rangle  \langle {\tilde\psi(a)}\rvert  \right) V_{A\to BE}^\dag P_{BE}^{(l)}  \right].
\end{align}
Likewise, we obtain the same expression for $\lim_{n\to\infty} F(\rho^{(n,l)}_E(b))$ that we have achieved for 
\begin{align}
\lim_{n\to\infty} F\left(\rho^{(n,l)}_B(a)\right), 
\end{align}
but interchanging $a\mapsto b$, $\tilde \opA\mapsto \tilde \opB $ and $\operatorname{tr}_E \mapsto \operatorname{tr}_B$.

Now that we have an expression for the bound which holds for bounded operators
$\tilde \opA$ and $\tilde \opB $, our next step is to move to unbounded operators. For
this task, we define sequences of bounded operators $( \opA_m)_m$ and $( \opB _m)_m$ as
\begin{align}
	\opA_m:= \frac{\opA}{1+\opA^2/m}, \quad 	\opB _m:= \frac{\opB }{1+\opB ^2/m}\ .
\end{align}

We now evaluate \cref{z:nmMW9OwgBbPX} choosing $\tilde \opA$ equal to $\tilde \opA_m$ and $\tilde \opA$ equal to $\tilde \opA_m$, followed by taking the limit $m\to\infty$ on both sides of the equation. Since by iii), it follows that $1/(1+\opA^2/m) \stackrel{s}{\to} \mathds{1}_{A}$ and $1/(1+\opB ^2/m) \stackrel{s}{\to} \mathds{1}_{A}$, where $\mathds{1}_{A}$ is the identity operator on $\mathscr{H}_A$, we have that $\opA_m \psi \to \opA\psi$ and $\opB _m\psi \to \opB \psi$  for all $\psi\in\mathcal{D}(\opA) \cap \mathcal{D}(\opB )$, and we find
\begin{align}
\frac{\lim_{m\to\infty}\lim_{n\to\infty} F\left(\rho^{(n,l)}_B(a)\right)}{F\left(\rho_A^{(\infty,\infty)}(a)\right)} +
\frac{\lim_{m\to\infty}\lim_{n\to\infty} F\left(\rho^{(n,l)}_E(b)\right)}{F\left(\rho_A^{(\infty,\infty)}(b)\right)}
\leq
1 + 2\sqrt{ 1 - \frac{ \bigl \langle { i[ \opA, \opB ] }\bigr \rangle ^2 }{4 \, {\sigma_{\opA}}^2{\sigma_{\opB }}^2} }\ ,
\label{z:J5vP1snzrlMV}
\end{align}
where 
\begin{align}
F\left(\rho_A^{(\infty,\infty)}(a)\right)&:= 4\,\sigma_{\opA}^2,\\%
F\left(\rho_A^{(\infty,\infty)}(b)\right)&:= 4\,\sigma_{{O}}^2,\\%
\sigma_{\opA}&:= \left( \bigl \langle { \opA \psi, \opA \psi }\bigr \rangle -  \bigl \langle { \psi, \opA \psi }\bigr \rangle ^2 \right)^{1/2},\\
\sigma_{\opB }&:= \left( \bigl \langle {  \opB  \psi, \opB  \psi }\bigr \rangle -  \bigl \langle { \psi, \opB  \psi }\bigr \rangle ^2 \right)^{1/2}.
\end{align}
The r.h.s.\@ of this inequality is now of the form in the corollary statement. We continue with the l.h.s. First observe that
\begin{align}
&\lim_{m\to\infty}\lim_{n\to\infty} F\left(\rho^{(n,l)}_B(a)\right) \geq  \\
& \sum_{\substack{k,k'=1 \\ \text{s.t. } p_k^{(\infty,\infty,l)}+p_{k'}^{(\infty,\infty,l)}>0}}^{d_B(l)}
\frac{p_k^{(\infty,\infty,l)}}{(p_k^{(\infty,\infty,l)} + p_{k'}^{(\infty,\infty,l)})^2}\,\left| \bigl \langle {k,\infty,\infty,l}\mkern 1.5mu\relax \big \vert \mkern 1.5mu\relax {\frac{d \rho^{(\infty,\infty,l)}_B(a) }{da}}\mkern 1.5mu\relax \big \vert \mkern 1.5mu\relax {k',\infty,\infty,l}\bigr \rangle  \right|^2 \ ,\label{z:V4tfZhJI2Thz}
\end{align}
where $\rho^{(\infty,\infty,l)}_B(a)= \sum_{k=1}^{d_B(l)} p_k^{(\infty,\infty,l)} \lvert {k,\infty,\infty,l}\rangle \langle {k,\infty,\infty,l}\rvert $, with
\begin{align}
\rho^{(\infty,\infty,l)}_B(a)&:= \lim_{m\to\infty}\operatorname{tr}_E\left[P_{BE}^{(l)} V_{A\to BE} \left(    \lvert {\tilde\psi_m(a)}\rangle  \langle {\tilde\psi_m(a)}\rvert  \right)  V_{A\to BE}^\dag P_{BE}^{(l)} \right],\\
\frac{d}{da}\rho^{(\infty,\infty,l)}_B(a)&:=\lim_{m\to\infty} \operatorname{tr}_E\left[P_{BE}^{(l)} V_{A\to BE} \left(   i \lvert {\tilde\psi_m(a)}\rangle  \langle {\tilde\psi_m(a)}\rvert  \opA_m- i \opA_m \lvert {\tilde\psi_m(a)}\rangle  \langle {\tilde\psi_m(a)}\rvert  \right) V_{A\to BE}^\dag P_{BE}^{(l)}  \right],
\end{align}
and $\lvert {\tilde\psi_m(a)}\rangle = e^{-i a \opA_m} \lvert {\psi}\rangle $. To see that \cref{z:V4tfZhJI2Thz} holds, observe that the same reasoning to why the limit and summation could be interchanged going from \cref{z:oWJ0dqqkNfon} to \cref{z:ZKqtUEagDlCs}, holds for the limit $m\to \infty$ also.  Now define $f_m(x)= e^{-i a x/(1+x^2/m)}$ and $f(x)=e^{-i a x}$. Assumptions in iii) hold, thus $e^{-i a \opA_m} \stackrel{s}{\to}  e^{-i a \opA}$, hence using ii) $1/(1+\opA^2/m) e^{-i a \opA_m} \stackrel{s}{\to} e^{-i a \opA}$. Furthermore, since, by definition $e^{-ia \opA}\lvert {\psi}\rangle \in\mathcal{D}(\opA)$, we have $H\lvert {\psi(a)}\rangle \in\mathscr{H}_A$. Taking all these things into account, we conclude that  
\begin{align}
\rho^{(\infty,\infty,l)}_B(a)&= \operatorname{tr}_E\left[P_{BE}^{(l)} V_{A\to BE} \left(    \lvert {\psi(a)}\rangle  \langle {\psi(a)}\rvert  \right)  V_{A\to BE}^\dag P_{BE}^{(l)} \right],\label{z:iMYddbx33T1N}\\
\frac{d}{da}\rho^{(\infty,\infty,l)}_B(a)&= \operatorname{tr}_E\left[P_{BE}^{(l)} V_{A\to BE} \left(   i \lvert {\psi(a)}\rangle  \langle {\psi(a)}\rvert  \opA- i \opA \lvert {\psi(a)}\rangle  \langle {\psi(a)}\rvert  \right) V_{A\to BE}^\dag P_{BE}^{(l)}  \right].\label{z:ClLvM3NN4a.J}
\end{align}
Likewise, we obtain the same expression for $\lim_{m\to\infty}\lim_{n\to\infty} F(\rho^{(n,l)}_E(b))$ that we {have} achieved for $\lim_{m\to\infty}\lim_{n\to\infty} F(\rho^{(n,l)}_B(a))$, but interchanging $a\mapsto b$, $ \opA\mapsto \opB $ and $\operatorname{tr}_E \mapsto \operatorname{tr}_B$.
Lastly, by comparing the r.h.s.\ of \cref{z:V4tfZhJI2Thz} with the 
r.h.s.\ of \cref{z:oWJ0dqqkNfon}, one {sees} that  $\lim_{m\to\infty}\lim_{n\to\infty} F(\rho^{(n,l)}_B(a))$ is given by evaluating the Fisher information for $\rho^{(\infty,\infty,l)}_B(a)$ (defined by \cref{z:iMYddbx33T1N}) with derivative $\frac{d}{da}\rho^{(\infty,\infty,l)}_B(a)$ (defined by \cref{z:iMYddbx33T1N}) according to \cref{z:.eb-akuquBNd,z:d-w60nHacIQJ}. The same observation holds for Eve's Fisher information. Hence to conclude the proof, we take $\liminf_{l\to\infty}$ on both sides of the equation.
\end{proof}

\subsection{Time-energy uncertainty equality in infinite dimensions}
\label{z:Klv0VTxb9iim}

In fact, building on the previous result, we get the following statement in the case where the commutator in the previous theorem vanishes. This can be viewed as a generalization of 
\cref{z:T-KNuYhGRM7I}
to the unbounded operator case.

\begin{theorem}[Time-energy uncertainty relation for infinite-dimensional systems]
  \label{z:DoGpU2OXKOtL}  
  Let $\lvert {\psi}\rangle $ be a state vector in a separable Hilbert space $\mathscr{H}_A$ of
  possibly infinite dimensions, 
  let $H,X$ be self-adjoint 
  operators (possibly unbounded) with domains $\mathcal{D}(H)$
  and $\mathcal{D}(X)$, respectively,
  so that
$  \lvert {\psi}\rangle \in\mathcal{D}(H)\cap \mathcal{D}(X)$.
  Define $\sigma_H := [ \langle H \psi, H \psi\rangle - \langle\psi, H \psi\rangle^2]^{1/2}$, 
  which is 
  finite due to
  $\lvert {\psi}\rangle \in\mathcal{D}(H)$,
  and $P_\rho^\perp := \mathds{1}- P_\rho$, 
  where $P_\rho$ denotes the
projector onto the support 
of $\rho$. Define analogously as before
  \begin{align}\label{z:4OyWcr5srJFP}
  T := t_0 - \frac{i [ H, \psi ]}{2\sigma_H^2}
  + P_\psi^\perp X P_\psi^\perp\ ,
\end{align}
where $X$ 
captures the 
freedom left when defining the optimal local time-sensing observable, and
consider for real $t,\eta$
and $t_0,\eta_0$
the two-parameter family 
$\lvert {\psi(t,\eta)}\rangle $ with
$\lvert {\psi(t_0,\eta_0)}\rangle =\lvert {\psi}\rangle $,
again $\psi=\lvert {\psi}\rangle 
\langle {\psi}\rvert $
and
\begin{align}
  \lvert {\psi(t,\eta)}\rangle  = \exp\bigl\{-i[ (t-t_0) H - (\eta-\eta_0) T]\bigr\}\,\lvert {\psi}\rangle \ .
\end{align}
  Let, as in 
  \cref{z:Tja8FC4elFJW},
  $V_{A\to BE}$ be any
  isometry $\mathscr{H}_A\to\mathscr{H}_B\otimes\mathscr{H}_E$, where the Hilbert spaces $\mathscr{H}_B, \mathscr{H}_E$
  associated with Bob and Eve are also separable and possibly of infinite
  dimensions, and define
  $F_A,F_B, F_E$ analogously as in
  \cref{z:Tja8FC4elFJW}.
Then the uncertainty principle
  \begin{align}
  \frac{F_B(t)}{F_A(t)} 
  +
  \frac{F_E(\eta)}{F_A(\eta)}
  \leq 1\ 
\end{align}
  holds.
\end{theorem}

Indeed, even in the infinite-dimensional setting for unbounded operators,
the uncertainty principle can be attained with equality, 
so that 
\begin{align}
  \frac{F_B(t)}{F_A(t)} 
  +
  \frac{F_E(\eta)}{F_A(\eta)}
  = 1 
\end{align}
still holds true.  
\begin{proof}[*z:DoGpU2OXKOtL]
The proof follows the same line of thought as 
that of \cref{z:Tja8FC4elFJW}, with some differences. To start with, consider the bounded operators $\tilde H$ 
and $\tilde X$ on $\mathscr{H}_A$ 
and define 
for a positive integer $n$ 
the truncated operators 
$\tilde H_n$ as
 \begin{align}
 \tilde H_n := P_{A}^{(n)} \tilde H P_{A}^{(n)}.
 \end{align}
 and
 \begin{align}
 \tilde T_n := P_{A}^{(n)} \tilde T P_{A}^{(n)},
 \end{align} 
 with $\tilde T$
 being defined as in
 Eq.\ (\ref{z:4OyWcr5srJFP}) with
 $T$ being replaced by
 $\tilde T$ and $X$
by $\tilde X$.
As above, one can define the
time-evolved states as
\begin{align}
  \lvert {\psi_n(t,\eta)}\rangle  = \exp\bigl\{-i[ (t-t_0) \tilde H_n - (\eta-\eta_0) \tilde T_n]\bigr\}\,\lvert {\psi}\rangle \ ,
\end{align}
with 
$\lvert {\psi_n}\rangle :=
\lvert {\psi_n(t_0,\eta_0)}\rangle $.
In the same way as before, for positive integers $l$ (and $n$), 
we can consider the sequence of 
positive operators $(\rho^{(n,l)}_B)_{n,l}$ on $\mathscr{H}_B$ defined as
\begin{align}
	\rho^{(n,l)}_B(t,\eta)&:=\operatorname{tr}_E\left[P_{BE}^{(l)} V_{A\to BE} P_{A}^{(n)} \left(    \lvert {\psi_n(t,\eta)}\rangle  
	\langle {\psi_n(t,\eta)}\rvert  \right) P_{A}^{(n)} V_{A\to BE}^\dag P_{BE}^{(l)} \right], 
\end{align}	
and
\begin{align}
   \rho^{(n,l)}_E(t,\eta)&:=\operatorname{tr}_B\left[P_{BE}^{(l)} V_{A\to BE} P_{A}^{(n)} \left(   \lvert {\psi_n(t,\eta)}\rangle  \langle {\psi_n(t,\eta)}\rvert   \right) P_{A}^{(n)} V_{A\to BE}^\dag P_{BE}^{(l)}  \right].
 \end{align}	  
Using these quantities, and proceeding as in the proof of \cref{z:Tja8FC4elFJW}, 
since this is a valid finite-dimensional setting in which the above proof in terms of a semidefinite program holds true,
one has 
\begin{align}
  \frac{\FIqty{Bob}{t}}{4
  \langle \tilde H \psi_n, \tilde H_n \psi_n\rangle - \langle\psi_n, \tilde H_n \psi_n\rangle^2} +
  \langle \tilde H_n \psi_n, \tilde H_n \psi_n\rangle - \langle\psi_n, \tilde H_n \psi_n\rangle^2
  \,\FIqty{Eve}{\eta}
  = 1\ ,
\end{align}
with equality, since 
$  \lvert {\psi}\rangle \in\mathcal{D}(H)\cap \mathcal{D}(X)$ and hence the state vector is in the domains of $H$ and $X$. Here,
\begin{align}
  \FIqty{Bob}{t} := \Ftwo{\rho_B^{(n,l)}(t_0)}{\partial_t \rho_B^{(n,l)}(t_0)}\ .
\end{align}
with 
\begin{align}
\rho^{(n,l)}_B(.):= 
\rho^{(n,l)}_B(.,\eta_0),
\end{align}
and
$\FIqty{Eve}{\eta}$
defined analogously based on
$\rho^{(n,l)}_E(\eta)$
with
$\rho^{(n,l)}_E(.):
=\rho^{(n,l)}_E(t_0,.)$.
The limit to the infinite-dimensional
setting involving the suitable
limit of $n\rightarrow\infty$
and $l\rightarrow\infty$
can be performed as in 
\cref{z:Tja8FC4elFJW}, while maintaining equality for each $n$
and $l$. %
\end{proof}

\section{Calculations for the case of continuous Lindbladian noise}
\label{z:nnVBeRY6bPwI}

\subsection{Sensing an unknown parameter in the Hamiltonian}

Consider a probe initialized in the state {vector} $\lvert {\psi_{\mathrm{init}}}\rangle $ and
subject to the Lindblad dynamics
\begin{align}
  \dot\rho = \mathcal{L}_{\mathrm{tot}}^{(\omega)}(\rho)\ ,
\end{align}
with
\begin{gather}
  \begin{aligned}
    \mathcal{L}_{\mathrm{tot}}^{(\omega)}
    &= \mathcal{L}_{\mathrm{sig}}^{(\omega)} + \mathcal{L}_{\mathrm{rest}}\ ;
    &\qquad
    \mathcal{L}_{\mathrm{sig}}^{(\omega)}(\rho)
    &= -i[\omega G, \rho]\ ;
  \end{aligned}
      \nonumber\\
  \mathcal{L}_{\mathrm{rest}}(\rho)\
    = -i[H_{\mathrm{rest}},\rho]
      + \sum_j\mathopen{}\left[ L_j\rho L_j^\dagger - \frac12\bigl\{L_j^\dagger L_j, \rho\bigr\}\right]\mathclose{}\ .
    \label{z:Vps6oPLlgW0g}
\end{gather}
Here, $\omega$ is the unknown parameter to be estimated.  The overall evolution
up to some total time $T$ is given by
\begin{align}
  \mathcal{E}_T^{(\omega)}
  = {e}^{T[\mathcal{L}_{\mathrm{sig}}^{(\omega)} + \mathcal{L}_{\mathrm{rest}}]}\ .
\end{align}
As we did earlier, we can decompose the overall evolution into the unitary
evolution driven by the signal (which depends on the unknown parameter
$\omega$), followed by an effective instantaneous {noisy} channel
$\mathcal{N}_{T,\omega}$:
\begin{align}
    \mathcal{E}_T^{(\omega)}
  &= \mathcal{N}_{T,\omega}\, {e}^{T\mathcal{L}_{\mathrm{sig}}^{(\omega)}}\ ,
\end{align}
where $\mathcal{N}_{T,\omega}$ is given by
\begin{align}
    \mathcal{N}_{T,\omega}
  &= \mathcal{E}_T^{(\omega)}\, {e}^{-T\mathcal{L}_{\mathrm{sig}}^{(\omega)}}\ .
\end{align}
We are interested in the sensitivity of the probe to the parameter $\omega$,
locally around $\omega_0$, after letting the probe evolve for some fixed time
$T$.  The sensitivity is given in terms of the Fisher information
\begin{align}
  \FIqtyp{T}{\omega}(\omega_0)
  &= \Ftwo{\rho_{T,\omega_0}}{(\partial_\omega \rho_{T,\omega})\,(\omega_0)}\ .
\end{align}
Defining the (fictitious) family of states
\begin{align}
  \begin{split}
  \psi_{T,\omega}
  &= {e}^{-iT\omega G}\,\psi_{\mathrm{init}}\,{e}^{iT\omega G}\ ;
  \\
  \partial_\omega \psi_{T,\omega}
  &= -iT[G, \psi_{T,\omega} ]\ ,
\end{split}
\end{align}
we may write
\begin{align}
  \FIqtyp{T}{\omega}(\omega_0)
  &= \Ftwo{\rho_{T}}{%
    \mathcal{N}_{T}\bigl( \partial_\omega \psi_{T,\omega} \bigr)
    + (\partial_\omega \mathcal{N}_{T,\omega})\bigl(\psi_{T}\bigr) }\ ,
    \label{z:iVeeWRE7OYTk}
\end{align}
where we omit the subscript $(\cdot)_{\omega_0}$ on all objects which are
ultimately evaluated at $\omega=\omega_0$.

Again as earlier we assume that we can neglect the second term in the derivative
in~\eqref{z:iVeeWRE7OYTk}, and carry on with
the approximation
\begin{align}
  \FIqtyp{T}{\omega}
  \approx
  \Ftwo{\mathcal{N}_{T,\omega_0}\bigl(\psi_{T,\omega_0}\bigr)}{
  \mathcal{N}_{T,\omega_0}\bigl(\partial_\omega \psi_{T,\omega}\bigr)
  }
  =:
  \FIqtyp[unit.]{T}{\omega}\ .
  \label{z:bu65yR9LrywX}
\end{align}

As above we are now in the setting of our main uncertainty relation; we can
identify the above quantity with $\FIqty{Bob}{t}$ in
\cref{z:T-KNuYhGRM7I}, where now the relevant evolution
generator is $TG$.  \Cref{z:T-KNuYhGRM7I} then implies
that
\begin{align}
  \FIqtyp[unit.]{T}{\omega}
  &= 4T^2\sigma_{G}^2 - \FIlossp[unit.]{T}{\omega}\ ;
  \nonumber\\
  \FIlossp[unit.]{T}{\omega}
  &= T^2 \Ftwo{\widehat{\mathcal{N}}_{T,\omega_0}\bigl(\psi_{T,\omega_0}\bigr)}{
  \widehat{\mathcal{N}}_{T,\omega_0}\bigl(\bigl\{ \bar{G}, \psi_{T,\omega_0} \bigr\}\bigr)
  }\ ,
\end{align}
where $\widehat{\mathcal{N}}_{T,\omega_0}$ is a channel that is complementary to
$\mathcal{N}_{T,\omega_0}$, and where $\bar{G} = G - \langle {G}\rangle $ with
$\langle {G}\rangle  = \operatorname{tr}[ G\,\psi_{T,\omega_0} ]$.
As earlier, the complementary channel can be written
$\widehat{\mathcal{N}}_{T,\omega_0} =
\widehat{\mathcal{E}}_{T,\omega_0}\,{e}^{-T\mathcal{L}_{\mathrm{sig}}^{(\omega)}}$.

The absolute error $\delta$ in the
approximation~\eqref{z:bu65yR9LrywX} can be bounded as earlier using
\cref{z:Y81L4FyklNd5} in
\cref{z:KH4B5FtzQ1kE} as
\begin{align}
  \lvert {\delta}\rvert 
  \leq \Ftwo{\rho}{(\partial_\omega\mathcal{N}_{T,\omega})(\psi_{T,\omega_0})}
  + \bigl[
  \Ftwo{\rho}{(\partial_\omega\mathcal{N}_{T,\omega})(\psi_{T,\omega_0})}
  \, \FIqtyp[unit.]{T}{\omega}
  \bigr]^{1/2}\ .
  \label{z:kTvdT8KoO1n6}
\end{align}
Similar arguments to those presented earlier apply when computing
$\partial_\omega \mathcal{N}_{T,\omega}$ in order to bound $\delta$; we have
\begin{align}
  \bigl(\partial_\omega \mathcal{N}_{T,\omega}\bigr) (\psi_{T,\omega})
  = \partial_\omega \rho - \mathcal{E}_{T}^{(\omega)}(-iT[G, \psi_{0}])\ .
\end{align}
Any numerical or analytical upper bound on
$\Ftwo{\rho}{(\partial_\omega\mathcal{N}_{T,\omega})(\psi_{T,\omega_0})}$ then
directly gives an upper bound to $\lvert {\delta}\rvert $
in~\eqref{z:kTvdT8KoO1n6}.

\subsection{Example: continuous dephasing noise along the $Z$ axis}
A qubit is initialized in the state vector
\begin{align}
  \lvert {\psi_{\mathrm{init}}}\rangle  = \lvert {+}\rangle 
  = \frac1{\sqrt2}\bigl[ \lvert {\uparrow }\rangle + \lvert {\downarrow }\rangle \bigr]\ ,
\end{align}
and evolves according to the Hamiltonian $H = \omega
{Z}/2$.
Suppose that the qubit is subject to continuous dephasing along the $Z$ axis.
This noise is represented by the Lindbladian jump operators
\begin{align}
  L_0 &= \sqrt\gamma\lvert {\uparrow}\rangle \mkern -1.8mu\relax \langle{\uparrow}\rvert \ ,
  &
    L_1 &= \sqrt\gamma\lvert {\downarrow}\rangle \mkern -1.8mu\relax \langle{\downarrow}\rvert \ .
\end{align}
In vectorized operator notation (same conventions as in the appendices of our
work, i.e., row-major convention), we have
\begin{gather*}
  \mathcal{L}_1
  = \sum \mathopen{}\left[ L_j\otimes L_j^T -
  \frac12\bigl[L_j^\dagger L_j\otimes\mathds{1}+ \mathds{1}\otimes(L_j^\dagger L_j)^T\bigr] \right]\mathclose{}
  = \begin{bmatrix} 0 &&& \\ & -\gamma && \\ && -\gamma & \\ &&& 0\end{bmatrix}\ ;
 \\
 \mathcal{L}_0
 = (\ldots) = \begin{bmatrix} 0 &&& \\ & -i\omega && \\ && i\omega & \\ &&& 0\end{bmatrix}\ ,
\qquad
  \mathcal{E}_t
  = {e}^{t(\mathcal{L}_0+\mathcal{L}_1)}
  = \begin{bmatrix}
    1 &&& \\ & {e}^{-\gamma t-it\omega} && \\ && {e}^{-\gamma t+it\omega} & \\ &&& 1
  \end{bmatrix}
  \ .
\end{gather*}
The full evolution map, represented as an operator in terms of matrix elements
$\rho_{ij} = \langle {i}\mkern 1.5mu\relax \vert \mkern 1.5mu\relax {\rho}\mkern 1.5mu\relax \vert \mkern 1.5mu\relax {j}\rangle $, is
\begin{align}
  \mathcal{E}_t(\rho)
  &= \begin{bmatrix}
    \rho_{00}  &  \rho_{01}{e}^{-it\omega -\gamma t}\\
    \rho_{10}{e}^{it\omega-\gamma t} & \rho_{11}
  \end{bmatrix}\ .
\end{align}
The next steps for this example are:
(a)~a direct computation of Bob's sensitivity;
(b)~a calculation of Eve's sensitivity to energy via our effective picture; and
(c)~an assessment of the error made in the
  approximation~\eqref{z:svEq8bs.6Q93}.

\paragraph{Direct computation of the sensitivity of the noisy probe.}
At a time $t$, the state is
\begin{align}
  \rho(t)
  &= \frac12 \begin{bmatrix}
    1 & {e}^{-it\omega-\gamma t} \\ {e}^{it\omega-\gamma t} & 1
  \end{bmatrix}
  \nonumber\\
  &= U_t \begin{bmatrix}
    1 & {e}^{-\gamma t} \\
    {e}^{-\gamma t} & 1
  \end{bmatrix} U_t^\dagger
  = U_t \frac{\mathds{1}+ {e}^{-\gamma t}
  {X}}{2} U_t^\dagger
    \nonumber\\
  &= \frac{1+{e}^{-\gamma t}}{2} U_t\lvert {+}\rangle \mkern -1.8mu\relax \langle{+}\rvert U_t^\dagger +
    \frac{1-{e}^{-\gamma t}}{2} U_t\lvert {-}\rangle \mkern -1.8mu\relax \langle{-}\rvert U_t^\dagger\ ,
\label{z:IfGEjMeWAQGx}
\end{align}
where we use the shorthand $U_t = {e}^{-iHt}$.  The last expression
in~\eqref{z:IfGEjMeWAQGx} provides a diagonal form for $\rho$
which will serve in the calculation of the Fisher information.  The derivative
of the state is
\begin{align}
  \mathcal{L}_{\mathrm{tot}}[\rho(t)]
  &= \dot\rho(t) = \frac12 \begin{bmatrix}
    0 & (-i\omega-\gamma){e}^{-it\omega-\gamma t} \\
    (i\omega-\gamma){e}^{it\omega-\gamma t} & 0
  \end{bmatrix}
    \nonumber\\
  &= \frac{{e}^{-\gamma t}}{2}
    \mathopen{}\left[ -\gamma\, U_t 
    {X}
    U_t^\dagger + \omega\, U_t 
    {Y}
    U_t^\dagger\right]\mathclose{}\ ,
\end{align}
noting that
$U_t %
{X}
U_t^\dagger = %
\begin{bsmallmatrix}0 & {e}^{-it\omega}\\{e}^{it\omega} & 0\end{bsmallmatrix}$ and
$U_t %
{Y}
U_t^\dagger = \begin{bsmallmatrix}0 & -i{e}^{-it\omega}\\
  i{e}^{it\omega} & 0\end{bsmallmatrix}$.  We can interpret this derivative in terms
of two different dynamics: One $\propto \omega\, U_t %
{Y} 
U_t^\dagger$, which drives
the rotation around the Bloch sphere, and one
$\propto -\gamma U_t
{X}
U_t^\dagger$, which drives decoherence
(\cref{z:kTNaoRKyf8rO}).
\begin{figure}
  \centering
  \includegraphics{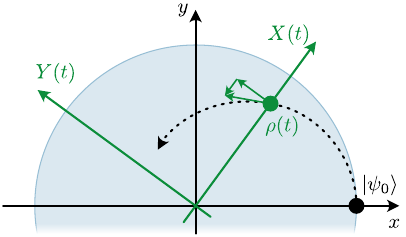}
  \caption{\label{z:kTNaoRKyf8rO} Top view of the Bloch sphere
    for a single qubit prepared in the $+X$ eigenstate, evolving under the
    Hamiltonian 
    {$H=\omega Z/2$ }
    and subject to continuous dephasing along the
    $Z$ axis.  The derivative of the state can be decomposed into a
    ``longitudinal part'' along $\sigma_Y(t)$ associated with the Hamiltonian
    dynamics, and a ``radial part'' along
    {$-X(t)$}
    associated with the
    noise terms.  The assumption that enables the mapping from the Lindblad
    setting to our bipartite uncertainty relation setting is that the noise
    component (``radial'' component) contributes negligibly to the overall time
    sensitivity of the clock.}
\end{figure}
The matrix elements of the derivative in
the eigenbasis $\{ U_t\lvert {\pm }\rangle \}$ of $\rho$ are
\begin{align}
  \langle {+}\rvert U_t^\dagger\, \dot\rho\, U_t\lvert {+}\rangle 
  &= -\gamma\frac{{e}^{-\gamma t}}{2}\ ,
  &
  \langle {+}\rvert U_t^\dagger\, \dot\rho\, U_t\lvert {-}\rangle 
  &= i\omega \frac{{e}^{-\gamma t}}{2}\ ,
    \nonumber\\
  \langle {-}\rvert U_t^\dagger\, \dot\rho\, U_t\lvert {+}\rangle 
  &= -i\omega \frac{{e}^{-\gamma t}}{2}\ ,
  &
  \langle {-}\rvert U_t^\dagger\, \dot\rho\, U_t\lvert {-}\rangle 
  &= \gamma\frac{{e}^{-\gamma t}}{2}\ .
\end{align}
where
$\langle {+}\mkern 1.5mu\relax \vert \mkern 1.5mu\relax {
{Y}}\mkern 1.5mu\relax \vert \mkern 1.5mu\relax {-}\rangle  = \langle {+}\mkern 1.5mu\relax \vert \mkern 1.5mu\relax {
{YZ}}\mkern 1.5mu\relax \vert \mkern 1.5mu\relax {+}\rangle  =
i\langle {+}\mkern 1.5mu\relax \vert \mkern 1.5mu\relax {
{X}
}\mkern 1.5mu\relax \vert \mkern 1.5mu\relax {+}\rangle =i$.  Now we compute the Fisher information using
\cref{z:r7GWSaRrLfNh} as
\begin{align}
  \FIqty{clock}{t}
  &=
  \Ftwo{\rho(t_0)}{\dot\rho(t_0)}
    \nonumber\\
  &=
  \frac2{1+{e}^{-\gamma t_0}}\, \mathopen{}\left \lvert { \gamma \frac{{e}^{-\gamma t_0}}{2} }\right \rvert \mathclose{}^2
  +
  2%
  \mathopen{}\left \lvert { i\omega \frac{{e}^{-\gamma t_0}}{2} }\right \rvert \mathclose{}^2
  +
  2%
  \mathopen{}\left \lvert { i\omega \frac{{e}^{-\gamma t_0}}{2} }\right \rvert \mathclose{}^2
  +
  \frac2{1-{e}^{-\gamma t_0}}\, \mathopen{}\left \lvert { \gamma \frac{{e}^{-\gamma t_0}}{2} }\right \rvert \mathclose{}^2
  \nonumber\\
  &= \omega^2{e}^{-2\gamma t_0}
    + \gamma^2\,\frac{2{e}^{-2\gamma t_0}}{1-{e}^{-2\gamma t_0}}\ .
    \label{z:fY39BjBZh2sa}
\end{align}

\paragraph{Eve's Fisher information with respect to energy.}
Now we turn to using the methods of our paper to characterize the
sensitivity of the noisy probe.  As described in
\cref{z:.YRCo8Z9DEX4}, we turn to computing
\begin{align}
  \FIqty{clock,U}{t}
  = \Ftwo{\rho(t_0)}{\mathcal{N}_{t_0}(\partial_t \psi (t_0))}\ ,
\end{align}
for the instantaneous effective {noisy} channel $\mathcal{N}_t$ and fictitious
unitary evolution $\psi(t)$ defined in
\cref{z:.YRCo8Z9DEX4}.  We will then later discuss
how good of an approximation $\FIqty{clock,U}{t}$ is to the original
desired quantity $\FIqty{clock}{t}$.

We decompose the full evolution $\mathcal{E}_t$ as
in~\eqref{z:XMdxyV4u6n1N}. Since
$[H,L_j] = 0$, we have
\begin{align}
  \mathcal{N}_t
  &= {e}^{t\mathcal{L}_1}
  = \begin{bmatrix}
    1 &&& \\ & {e}^{-\gamma t} && \\ && {e}^{-\gamma t} & \\ &&& 1
  \end{bmatrix}
  \quad\to\quad&
  \mathcal{N}_t(\rho)
  &=
  \begin{bmatrix}
    \rho_{00} & \rho_{01}\,{e}^{-\gamma t}\\
    \rho_{10}\,{e}^{-\gamma t} & \rho_{11}
  \end{bmatrix}\ .
  \label{z:pCFecsLeCcJH}
\end{align}
This channel can be described by the two Kraus operators
\begin{align}
  E_0^{(t)}
  &= \sqrt{\frac{1+{e}^{-t\gamma}}{2}}\,\mathds{1}\ ;
  &
    E_1^{(t)}
  &= \sqrt{\frac{1-{e}^{-t\gamma}}{2}}\,
  {Z}.
\end{align}
The (fictitious) pure unitary evolution of the initial state vector 
$\lvert {\psi_{\mathrm{init}}}\rangle =\lvert {+}\rangle $ is
\begin{align}
  \psi(t)
  = U_t\,\psi_{\mathrm{init}}\,U_t^\dagger
  = \frac12 \begin{bmatrix}
    1 & {e}^{-it\omega} \\
    {e}^{it\omega} & 1
  \end{bmatrix}\ .
\end{align}

We compute Eve's Fisher information with respect to energy, which characterizes
the sensitivity loss of the noisy probe.
For any $t$, a complementary channel
to~\eqref{z:pCFecsLeCcJH} is given by
\begin{align}
  \widehat{\mathcal{N}}_t(\rho)
  &= \begin{bmatrix}
    \frac{1+{e}^{-\gamma t}}{2}\,\operatorname{tr}(\rho)
    & \frac{\sqrt{1-{e}^{-2\gamma t}}}{2}\,\operatorname{tr}(
    {Z}
    \rho)\\
    \frac{\sqrt{1-{e}^{-2\gamma t}}}{2}\,\operatorname{tr}(%
     {Z}\rho)
    & \frac{1-{e}^{-\gamma t}}{2}\,\operatorname{tr}(\rho)
  \end{bmatrix}\ .
\end{align}
We would like to compute
\begin{align}
  \Ftwo{\widehat{\mathcal{N}}_t(\psi)}{\widehat{\mathcal{N}}_t(\{H-\langle {H}\rangle , \psi\})}\ .
\end{align}
Noting that $\langle {H}\rangle _{\psi(t)} = 0$ for all $t$ and that
$\psi_{\mathrm{init}} = (\mathds{1}+
{X}
)/2$, we can compute
\begin{align}
  \{ H - \langle {H}\rangle , \psi \}
  &= \Bigl\{ \frac{\omega}{2}
  {Z}, U_t\,\psi_{\mathrm{init}}\,U_t^\dagger \Bigr\}
  = \frac{\omega}{2}\,U_t\,\Bigl\{ %
  {Z}
  , \frac{1+
  {X} }{2} \Bigr\} \,U_t^\dagger
    = \frac{\omega}{2}\,
    {Z}\ .
\end{align}
We then see that
\begin{align}
  \widehat{\mathcal{N}}_t(\psi)
  &= \begin{bmatrix}
    \frac{1+{e}^{-\gamma t}}{2} & 0 \\ 0 & \frac{1-{e}^{-\gamma t}}{2}
  \end{bmatrix}\ ;
  &
    \widehat{\mathcal{N}}_t\Bigl(\frac \omega 2 \, 
    {Z} \Bigr)
  &= \frac \omega 2 \begin{bmatrix}
    0 & \sqrt{1-{e}^{-2\gamma t}} \\ \sqrt{1-{e}^{-2\gamma t}} & 0
  \end{bmatrix}\ .
\end{align}
Then using \cref{z:r7GWSaRrLfNh} we find
\begin{align}
\hspace*{2em}&\hspace*{-2em}
  \Ftwo{\widehat{\mathcal{N}}(\psi)}{\widehat{\mathcal{N}}(\{H-\langle {H}\rangle , \psi\})}
\nonumber\\
  &= 0 +
    2\,%
    \mathopen{}\left[\frac{\omega^2}{4}\bigl(1-{e}^{-2\gamma t_0}\bigr)\right]\mathclose{} + (\textup{same term}) + 0
    \nonumber\\
  &= \omega^2\,\bigl( 1 - {e}^{-2\gamma t_0} \bigr)\ .
\end{align}
In the present picture of the effective {noisy} channel being applied instantly
after unitary evolution of duration $t_0$, we see that Eve obtains no
information about the energy direction for $t\approx 0$.  However, for large $t$
Eve obtains near-perfect information which hinders Bob's sensitivity.  Since the
noiseless Fisher information is $\omega^2$, we have via our uncertainty relation that
\begin{align}
  \FIqty{Bob}{t} = \omega^2 {e}^{-2\gamma t_0}\ .
\end{align}

Our method therefore correctly gives us the first term
in~\eqref{z:fY39BjBZh2sa}.
We can also check by direct calculation that the first term
in~\eqref{z:fY39BjBZh2sa} is indeed
the Fisher information of the noisy clock state if we neglect the term in the
derivative that is associated with the time derivative of the effective noise
channel itself.  First observe that
\begin{align}
  \partial_t\psi &= \frac12 \begin{bmatrix}
   0 & -i\omega{e}^{-it\omega} \\ i\omega{e}^{it\omega} & 0
 \end{bmatrix}\ ,
  \\
  \mathcal{N}(\partial_t\psi)
  &= \frac12 \begin{bmatrix}
   0 & -i\omega{e}^{-it\omega-\gamma t} \\ i\omega{e}^{it\omega-\gamma t} & 0
 \end{bmatrix}
 = \frac{\omega{e}^{-\gamma t}}{2}\,U_t\,
 {Y}
 \,U_t^\dagger
 \ .
\end{align}
We see that the object $\mathcal{N}(\partial_t \psi)$ is exactly the part of the
derivative $\dot\rho$ with respect to the full dynamics that is associated with
the Hamiltonian evolution of $\rho$, i.e., it is the ``longitudinal'' component
of the derivative depicted in \cref{z:kTNaoRKyf8rO}.

We use again \cref{z:r7GWSaRrLfNh} of our manuscript,
recalling the diagonal form for $\rho$ given
in~\eqref{z:IfGEjMeWAQGx}:
\begin{align}
  \FIqty{clock,U}{t}
  =  \Ftwo{\rho}{\mathcal{N}(\partial_t \psi)}
  = 0 + 2 \mathopen{}\left \lvert { \frac{\omega{e}^{-\gamma t_0}}{2} }\right \rvert \mathclose{} + 2 \mathopen{}\left \lvert { \frac{\omega{e}^{-\gamma t_0}}{2} }\right \rvert \mathclose{} + 0
  = \omega^2{e}^{-2\gamma t_0}\ .
\end{align}
The difference between $\FIqty{clock,U}{t}$ and $\FIqty{clock}{t}$ is
\begin{align}
  \delta = \FIqty{clock}{t} - \FIqty{clock,U}{t}
  = \gamma^2 \frac{2{e}^{-2\gamma t_0}}{1-{e}^{-2\gamma t_0}}\ .
\end{align}
The relative error of the approximation is
\begin{align}
  \frac{\delta}{\FIqty{clock,U}{t}}
  = \frac{\gamma^2}{\omega^2}\,\frac1{1-{e}^{-2\gamma t_0}}\ .
\end{align}
(We computed the relative error with respect to $\FIqty{clock,U}{t}$
because it is simpler.)  We can see that $\delta$ is small relative to
$\FIqty{clock,U}{t}$ if the ratio $\gamma/\omega$ of the loss rate
to the qubit's energy gap is small.

Numerical plots for $\omega=1, \gamma=0.1$ are presented in
\cref{z:axU8zJ6hAklL}.
\begin{figure}
  \centering
  \includegraphics[width=13cm]{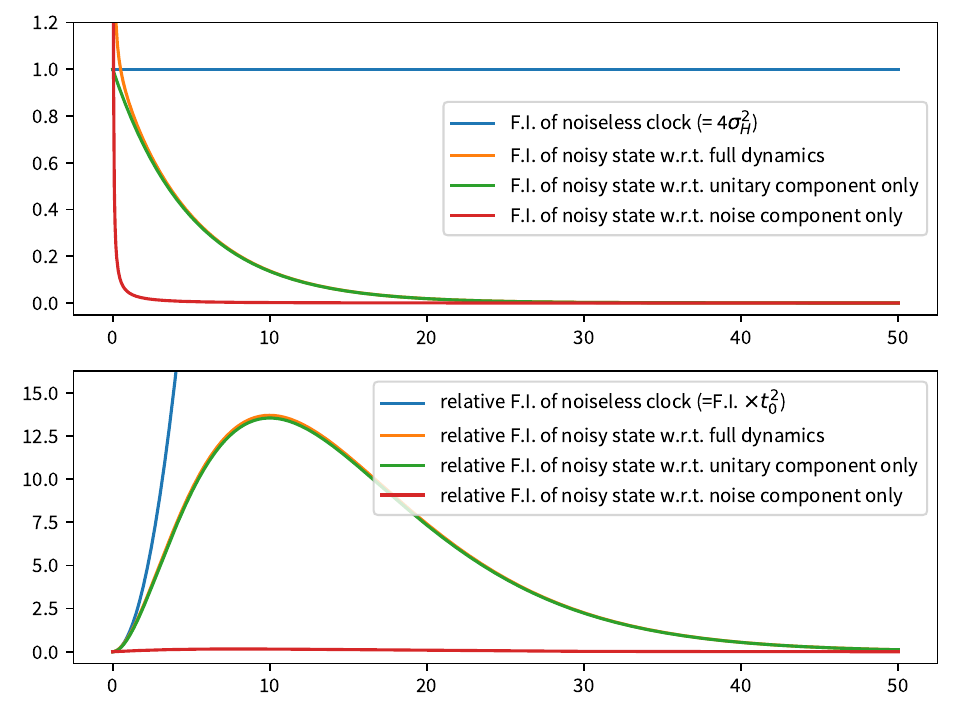}
  \caption{Fisher information (F.I.) of a single qubit prepared in a $+X$
    eigenstate evolving according to the Hamiltonian %
    {$H=\omega Z/2$}
    and
    subject to continuous dephasing along the $Z$ axis.  The horizontal axis
    represents the time $t_0$ at which we consider the clock sensitivity, and
    the vertical axis is the value of the different versions of the Fisher
    information (top plot) and relative Fisher information (bottom plot). The
    relative Fisher information is the Fisher information times $t_0^2$, which
    is relevant if we are interested in the relative sensitivity to time.  In
    these plots we have set $\omega=1$ and $\gamma=0.1$ (see main text).  We
    verify from these plots that the time dependency $\partial_t\mathcal{N}_t$
    of the effective noisy channel contributes negligibly to the overall Fisher
    information; this example in the setting of continuous noise can therefore
    be reduced to a setting as in \cref{z:FnHui0ahNLxW}.}
  \label{z:axU8zJ6hAklL}
\end{figure}

\paragraph{Error bound for the mapping from the Lindblad master equation to our
  setting.}
As a sanity check we compute the error
bound~\eqref{z:raWZngbty0t4}.  We have
\begin{align}
  \partial_t \mathcal{N}
  = \begin{bmatrix}
    0 &&&
    \\ & -\gamma{e}^{-\gamma t} &&\\
    &&-\gamma{e}^{-\gamma t} &\\
    &&& 0
  \end{bmatrix}
  \ ,
\end{align}
and thus
\begin{align}
  (\partial_t \mathcal{N})( \psi(t) )
  &= \frac12\begin{bmatrix}
    0 & -\gamma{e}^{-\gamma t}{e}^{-i\omega t} \\
    -\gamma{e}^{-\gamma t}{e}^{i\omega t} & 0
  \end{bmatrix}
  =
  -\gamma{e}^{-\gamma t} U_t
  {X}
  U_t^\dagger\ .
\end{align}
The matrix elements in the state's eigenbasis are
\begin{align}
  \langle {+}\rvert  U_t^\dagger (\partial_t \mathcal{N}) U_t \lvert {+}\rangle 
  &= -\gamma{e}^{-\gamma t}\ ,
  &
  \langle {+}\rvert  U_t^\dagger (\partial_t \mathcal{N}) U_t \lvert {-}\rangle 
  &= 0\ ,
  \nonumber\\
  \langle {-}\rvert  U_t^\dagger (\partial_t \mathcal{N}) U_t \lvert {+}\rangle 
  &= 0\ , 
  &
  \langle {-}\rvert  U_t^\dagger (\partial_t \mathcal{N}) U_t \lvert {-}\rangle 
  &= \gamma{e}^{-\gamma t}\ .
\end{align}
Then we can compute
\begin{align}
  \Ftwo{\rho}{\partial_t \mathcal{N}(\psi)}
  &= \frac{2}{1+{e}^{-\gamma t_0}} \mathopen{}\left \lvert {\gamma{e}^{-\gamma t_0}}\right \rvert \mathclose{}^2
  + \frac{2}{1-{e}^{-\gamma t_0}} \mathopen{}\left \lvert {\gamma{e}^{-\gamma t_0}}\right \rvert \mathclose{}^2
    = 4\gamma^2 \frac{{e}^{-2\gamma t_0}}{1 - {e}^{-2\gamma t_0}}\ .
\end{align}
Our bound~\eqref{z:raWZngbty0t4} on the error $\delta$ becomes
\begin{align}
  \delta
  &\leq
  4\gamma^2 \frac{{e}^{-2\gamma t_0}}{1 - {e}^{-2\gamma t_0}}
  + 2\gamma \omega\frac{{e}^{-2\gamma t_0}}{\sqrt{1 - {e}^{-2\gamma t_0}}} \ .
\end{align}
The bound is consistent with our computed value of $\delta$.  However, in this
case our bound is loose: The second term in our bound would suggest that the
relative error with respect to $F_{\mathrm{clock,U},\,t_0}$ behaves only as
$\gamma/\omega$ (if $\gamma\ll \omega$), whereas we know from our explicit
calculation of $\delta$ that the behavior of this relative error is
$\gamma^2/\omega^2$.

\subsection{Example: continuous dephasing noise along the
  transversal $X$ axis}

Consider the qubit state vector
\begin{align}
  \lvert {\psi }\rangle = \lvert {+}\rangle  = \frac1{\sqrt2}\bigl[ \lvert {\uparrow }\rangle + \lvert {\downarrow }\rangle \bigr]\ .
\end{align}
Suppose that the evolution of the qubit is given by the
Lindbladian~\eqref{z:BsZSNRM3uYvl} with
\begin{align}
    H &= \frac\omega2 
    {Z},
  &
    L_0 &= \sqrt\gamma\lvert {+}\rangle \mkern -1.8mu\relax \langle{+}\rvert \ ,
  &
    L_1 &= \sqrt\gamma\lvert {-}\rangle \mkern -1.8mu\relax \langle{-}\rvert \ .
\end{align}
One checks that the action of $\mathcal{L}_{\mathrm{tot}}$ on the Pauli
operators and the identity are
\begin{align}
  \mathcal{L}_{\mathrm{tot}}(\mathds{1})
  &= 0\ ,
  &
  \mathcal{L}_{\mathrm{tot}}(
  {X})
  &= \omega
  {Y}\ ,
  &
  \mathcal{L}_{\mathrm{tot}}(
  {Y})
  &= -\omega
  {X}- \gamma
  {Y} \ ,
  &
  \mathcal{L}_{\mathrm{tot}}(
  {Z})
  &= -\gamma 
  {Z}\ .
\end{align}
Therefore, $\mathcal{L}_{\mathrm{tot}}$ can be represented in the orthonormal
basis
$\{\lvert {\mathds{1}}\rrangle /\sqrt{2},\allowbreak \lvert {
{X}}\rrangle /\sqrt{2},
\allowbreak
\lvert {
{Y} }\rrangle /\sqrt{2},\allowbreak \lvert {
{Z} }\rrangle /\sqrt{2} \}$ of Pauli
operators (denoted with subscript $\textit{P}$) as
\begin{align}
  \bigl[\mathcal{L}_{\mathrm{tot}}\bigr]_{\textit{P}}
  = \begin{pmatrix}
    0 &0 & 0 & 0\\
    0 &0 & -\omega & 0\\
    0 &\omega & -\gamma & 0\\
    0 &0 & 0 & -\gamma\\
  \end{pmatrix}_{\textit{P}}
\end{align}
One can verify that this matrix is diagonalized as
\begin{gather}
  \bigl[\mathcal{L}_{\mathrm{tot}}\bigr]_{\textit{P}} = S\,
  \begin{pmatrix} 0 &&& \\ &\lambda_+&&\\&&\lambda_-&\\&&&-\gamma\end{pmatrix}
  \, S^{-1}\ ,
  \\
  \begin{aligned}
  \lambda_\pm
  &= -\frac\gamma{2}\pm i\alpha\ ,
  &
  \alpha
  &= \frac12\sqrt{4\omega^2 - \gamma^2}\ ,
  &
  \lambda_+\lambda_- &=\omega^2\ ,
  &
  \lambda_++\lambda_-&=-\gamma\ ,
  \end{aligned}
  \nonumber\\
  \begin{aligned}
    S &=
    \begin{pmatrix}
      1 & 0 & 0 & 0 \\
      0 & -\frac{\lambda_-}{\omega} & -\frac{\lambda_+}{\omega} & 0 \\
      0 & 1 & 1 & 0 \\
      0 & 0 & 0 & 1
    \end{pmatrix}\ ,\qquad
    &
    S^{-1} &= \begin{pmatrix}
      1 & 0 & 0 & 0 \\
      0 & \frac{-i\omega}{2\alpha} &
      \frac{-i\lambda_+}{2\alpha} & 0 \\
      0 & \frac{i\omega}{2\alpha} &
      \frac{i\lambda_-}{2\alpha} & 0 \\
      0 & 0 & 0 & 1
    \end{pmatrix}\ .
  \end{aligned}
\end{gather}
We can solve the dynamics analytically using this diagonal representation to
compute the matrix exponential as
\begin{gather}
  \mathcal{E}_t
  = \bigl[ {e}^{t\mathcal{L}_{\mathrm{tot}}} \bigr]_{\textit{P}}
  = S \begin{pmatrix} 1&&&\\
    &{e}^{t\lambda_+}&&\\
    &&{e}^{t\lambda_-}&\\
    &&&{e}^{-\gamma t}
  \end{pmatrix} S^{-1}
  = \begin{pmatrix}
    1 & 0 & 0 & 0\\
    0 & e_{xx} & e_{xy} & 0 \\
    0 & e_{yx} & e_{yy} & 0 \\
    0 & 0 & 0 & {e}^{-\gamma t}
  \end{pmatrix}_{\textit{P}}\ ,
  \nonumber\\
  \begin{aligned}
    e_{xx} &=
    {e}^{-\frac{\gamma t}{2}}\mathopen{}\left[ \cos(\alpha t) + \frac{\gamma}{2\alpha} \sin(\alpha t)\right]\mathclose{}\ ,
    &
    e_{xy} &=
    -e_{yx}\ ,
    \\
    e_{yx} &=
    \frac{\omega}{\alpha} {e}^{-\frac{\gamma t}{2}}\sin(\alpha t)\ ,
    &
    e_{yy} &= e_{xx}\ .
  \end{aligned}
\end{gather}
This gives us a useful expression of the linear operator $\mathcal{E}_t$ acting
on the operator basis of Pauli operators.  If we let the initial state
$\psi_{\mathrm{init}} = \lvert {+}\rangle \mkern -1.8mu\relax \langle{+}\rvert $ evolve for a time $t$, we obtain
\begin{align}
  \rho(t)
  &= \mathcal{E}_t(\psi_{\mathrm{init}}) = 
  \mathcal{E}_t\mathopen{}\left(\frac{1+
  {X} }{2}\right)\mathclose{}
    \nonumber\\
  &= \frac{\mathds{1}}{2} +
    \frac{{e}^{-\frac{\gamma t}{2}}}{2}\mathopen{}\left[
    \mathopen{}\left(\cos(\alpha t) + \frac{\gamma}{2\alpha} \sin(\alpha t)\right)\mathclose{}\,
    {X}
    +\frac{\omega}{\alpha} \sin(\alpha t)\,
    {Y}
    \right]\mathclose{}\ .
    \label{z:HwyCzho4gsTn}
\end{align}
See \cref{z:FhWQk7X.reUf} for a plot of the trajectory of
the state $\rho(t)$ in the $X$-$Y$ plane of the Bloch sphere.
\begin{figure}
  \centering
  \includegraphics{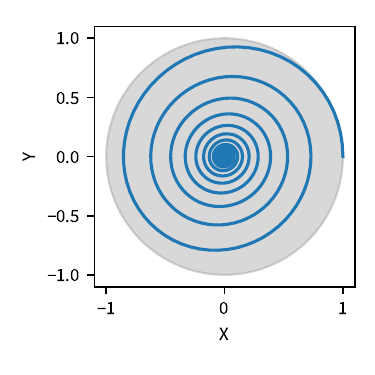}
  \caption{Trajectory on the equatorial slice of the Bloch sphere of the state
    of a qubit initialized in the state {vector} $\lvert {+}\rangle $, evolving under the Hamiltonian
    {$H=(\omega/2)Z$} 
    and subject to continuous dephasing along the $X$
    axis. Here $\omega=1$ and $\gamma=0.1$.}
  \label{z:FhWQk7X.reUf}
\end{figure}

We can compute the derivative $\partial_t\rho$ by directly differentiating the
expression~\eqref{z:HwyCzho4gsTn} or by simply applying the Lindbladian
since we have determined its action in the Pauli basis:
\begin{align}
  \partial_t\rho
  &= \frac{{e}^{-\frac{\gamma t}{2}}}{2}\Bigl[
    \Bigl(\cos(\alpha t) + \frac\gamma{2\alpha}\sin(\alpha t)\Bigr) \omega 
    {Y}
    + \frac\omega\alpha\sin(\alpha t)\bigl(-\omega
    {X}
    - \gamma
    {Y}\bigr)
    \Bigr]
    \nonumber\\
  &= \frac{\omega}{2}{e}^{-\frac{\gamma t}{2}}\Bigl[
    -\frac{\omega}{\alpha}\sin(\alpha t)\,
    {X}+
    \Bigl(\cos(\alpha t) - \frac{\gamma}{2\alpha}\sin(\alpha t)\Bigr) %
    {Y}
    \Bigr].
    \label{z:NcoMf18XMCmh}
\end{align}
The approximation we make to apply our uncertainty relation is to replace this
expression for $\partial_t\rho$ by
\begin{align}
  \mathcal{E}_{t}\mathopen{}\left( -i[H,\psi_{\mathrm{init}}] \right)\mathclose{}
  &= \mathcal{E}_{t}\mathopen{}\left( -i\mathopen{}\left[\frac{\omega}{2}
  {Z}
  ,\frac{1+
  {X}
  }{2}\right]\mathclose{} \right)\mathclose{}
  = \frac{\omega}{2}\, \mathcal{E}_t\mathopen{}\left(
  {Y}\right)\mathclose{}
  \nonumber\\
  &= \frac{\omega}{2} {e}^{-\frac{\gamma t}{2}} \mathopen{}\left[
    -\frac{\omega}\alpha\sin(\alpha t)\,
    {X}
    + \mathopen{}\left(\cos(\alpha t) + \frac{\gamma}{2\alpha}\sin(\alpha t)\right)\mathclose{}\,
    {Y}
    \right]\mathclose{}\ .
    \label{z:ziHoT1fE9Dl7}
\end{align}
We see that the two expressions
\labelcref{z:NcoMf18XMCmh,z:ziHoT1fE9Dl7} differ by a term $(2\alpha)^{-1}\gamma\omega{e}^{-\gamma t/2}\sin(\alpha t) 
{Y}
$, which is small as long as $\gamma\ll\omega$.

The Fisher information $\FIqty{clock}{t}$ given
by~\eqref{z:RquV6F2IJPiu} and $\FIqty{clock,U}{t}$ given
by~\eqref{z:svEq8bs.6Q93} are plotted in
\cref{z:bK7Cmm5UAXz3} as a function of $t_0$.
\begin{figure}
  \centering
  \includegraphics{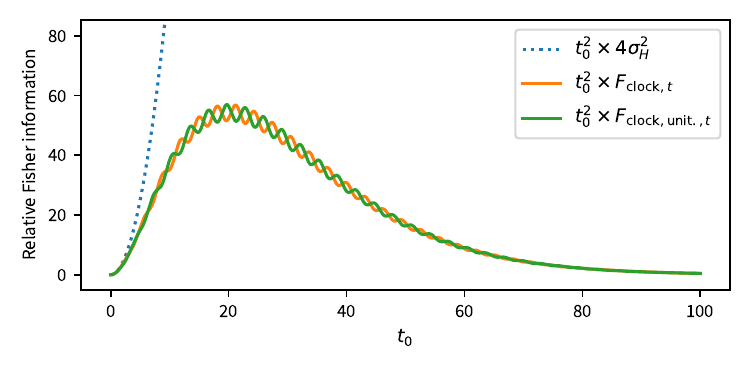}
  \caption{Relative Fisher information with respect to time of a single qubit
    prepared in $\lvert {+}\rangle $ and evolving according to the Hamiltonian
     {$(\omega/2) Z$} 
    and exposed to continuous dephasing along the X axis at
    a rate $\gamma$.  The blue curve shows the sensitivity as a function of time
    $t_0$ of the probe to the signal if we turn off the noise.  In orange, the
    exact Fisher information $\Ftwo{\rho}{\partial_\omega\rho}$ is computed
    directly.  In green, an approximation to the desired Fisher information
    ignores the contribution of the time dependency $\partial_t\mathcal{N}_t$ of
    the effective {noisy} channel.  This approximation is the quantity that
    appears in our trade-off relation in the alternative setting where Alice
    sends a noiseless quantum clock over a noisy channel to Bob.  Because the
    unitary and noise parts of the Lindbladian do not commute as superoperators,
    invoking our the trade-off relation requires the channel $\mathcal{N}_t$ to
    be determined via~\eqref{z:p1QyEW2YDGue}.  Here
    $\omega = 1$, $\gamma=0.1$.}
  \label{z:bK7Cmm5UAXz3}
\end{figure}
Our approximation matches the exact Fisher information well, except for an
out-of-phase oscillation of relatively small amplitude.  This error to the
contribution of the phase damping is expected to be attributable to the
difference in sign of the smaller terms in
\cref{z:NcoMf18XMCmh,z:ziHoT1fE9Dl7}.

\section{Perturbing the noisy channel to restore equality
  in the uncertainty relation for metrological codes}
\label{z:DQLLjQ-XUmxW}

In this Appendix, we study how to perturb a noisy channel $\mathcal{N}$ in order
to restore uncertainty relation equality for a metrological code.
We prove \cref{z:F6rfqyfKNr41} of the
main text, which shows that equality in the uncertainty relation can be restored
by an infinitesimal perturbation of the Stinespring isometry, all while
preserving the zero sensitivity-loss conditions~\eqref{z:zFMncm5bkXDt} (it might be necessary to
enlarge Bob's system with an auxiliary qubit).
The proposition is slightly reformulated to emphasize the fact that we can apply
the same construction also without regards to the zero sensitivity-loss
condition.
\begin{proposition}
  \label{z:QmMajA-tmAuf}
  Let $V_{A\to BE}$ be an isometry, let $\lvert {\psi}\rangle _A,\lvert {\xi}\rangle _A$ with
  $\langle {\psi}\mkern 1.5mu\relax \vert\mkern 1.5mu\relax {\xi}\rangle _A=0$ and let
  $\mathcal{N}(\cdot) = \operatorname{tr}_E\bigl(V\,(\cdot)\,V^\dagger\bigr)$,
  $\widehat{\mathcal{N}}(\cdot) = \operatorname{tr}_B\bigl(V\,(\cdot)\,V^\dagger\bigr)$.  For any
  $\epsilon>0$, there exists an isometry $V'_{A\to BE}$ with
  $\lVert {V' - V}\rVert  \leq \epsilon$ and such that
  $\bigl(P_{\rho_B'}^\perp\otimes P_{\rho_E'}^\perp\bigr) V'\lvert {\xi }\rangle = 0$, where
  $\rho_B' = \operatorname{tr}_E\bigl\{ V'\psi V'^\dagger \bigr\}$ and
  $\rho_E' = \operatorname{tr}_B\bigl\{ V'\psi V'^\dagger \bigr\}$.

  Furthermore, assume that
  $\widehat{\mathcal{N}}(\lvert {\xi}\rangle \mkern -1.8mu\relax \langle{\psi}\rvert +\lvert {\psi}\rangle \mkern -1.8mu\relax \langle{\xi}\rvert ) = 0$ and assume that
  there exists a unitary operator $G_B$ acting on the system $B$ with the
  properties that $P_{\rho_B} G_B P_{\rho_B} = 0$,
  $P_{\zeta_B} G_B P_{\zeta_B} = 0$, $P_{\rho_B} G_B P_{\zeta_B} = 0$, and
  $P_{\zeta_B} G_B P_{\rho_B} = 0$, where $\zeta_B = \mathcal{N}(\lvert {\xi}\rangle \mkern -1.8mu\relax \langle{\xi}\rvert )$ and
  $\zeta_E = \widehat{\mathcal{N}}(\lvert {\xi}\rangle \mkern -1.8mu\relax \langle{\xi}\rvert )$.  Then the perturbed isometry $V'$
  can be chosen to also satisfy
  $\widehat{\mathcal{N}}'\bigl(\lvert {\xi}\rangle \mkern -1.8mu\relax \langle{\psi}\rvert +\lvert {\psi}\rangle \mkern -1.8mu\relax \langle{\xi}\rvert \bigr) = 0$, where
  $\widehat{\mathcal{N}}'(\cdot) = \operatorname{tr}_B\bigl\{ V' \, (\cdot)\, V'^\dagger \bigr\}$.
\end{proposition}

\begin{proof}[**z:QmMajA-tmAuf]
  Write $\mathcal{N}(\cdot) = \operatorname{tr}_E\bigl(V\,(\cdot)\,V^\dagger\bigr)$ and let
  \begin{align}
    \rho_B &= \mathcal{N}(\psi)\ ;
    &
      \rho_E &= \widehat{\mathcal{N}}(\psi)\ .
  \end{align}

  The strategy to perturb $V$ is to include an infinitesimal rotation that
  rotates the state $V\lvert {\psi}\rangle $ into the direction of another suitably chosen
  state $\lvert {\chi}\rangle _{BE}$.
  We first compute some properties of a general such rotation, and then we will
  prove the stated claims.

  Let $\epsilon>0$.  
  Let $\alpha>0$ such that $4\sin^2(\alpha/2)\leq\epsilon$. 
  Let $\lvert {\chi}\rangle _{BE}$ be a state with the property that the reduced state on $B$
  lies in a subspace that is orthogonal to the reduced state $\rho_B$ of
  $V\lvert {\psi}\rangle $, i.e., $P_{\rho_B} \lvert {\chi }\rangle = 0$, or equivalently, $\lvert {\chi}\rangle _{BE}$
  lies in the support of $P_{\rho_B}^\perp \otimes \mathds{1}$.  The state
  $\lvert {\chi}\rangle _{BE}$ will be fixed later.
  Let $\bigl\{ \lvert {\mu^{(j)}}\rangle  \bigr\}_j$ be a basis of $BE$ with
  $\lvert {\mu^{(1)}}\rangle _{BE} = V\lvert {\psi}\rangle $ and $\lvert {\mu^{(2)}}\rangle _{BE} = \lvert {\chi}\rangle _{BE}$.  Let
  \begin{align}
    W_{BE\to BE}
    &= 
      \Bigl( \cos(\alpha)\lvert {\mu^{(1)}}\rangle  + \sin(\alpha)\lvert {\mu^{(2)}}\rangle  \Bigr) \langle {\mu^{(1)}}\rvert 
      +  \Bigl( \cos(\alpha)\lvert {\mu^{(2)}}\rangle  -\sin(\alpha)\lvert {\mu^{(1)}}\rangle  \Bigr) \langle {\mu^{(2)}}\rvert 
      \nonumber\\
    &\qquad + \sum_{j=3,\ldots} \lvert {\mu^{(j)}}\rangle \mkern -1.8mu\relax \langle{\mu^{(j)}}\rvert \ ,
  \end{align}
  and note that $W_{BE\to BE}$ is a unitary close to the identity, effecting the
  rotation $W_0 = \left[\begin{smallmatrix}
          \cos(\alpha) &  -\sin(\alpha) \\
          \sin(\alpha) & \cos(\alpha)
        \end{smallmatrix}\right]$
  between $\lvert {\mu^{(1)}}\rangle _{BE}$ and $\lvert {\mu^{(2)}}\rangle _{BE}$.
  The eigenvalues $\theta_1,\theta_2$ of $W_0$ are determined from
  $\theta_1+\theta_2 = \operatorname{tr}(W_0) = 2\cos(\alpha)$ and
  $\theta_1\theta_2 = \det(W_0) = 1$ as $\theta_1 = \theta_2^* = {e}^{i\alpha}$.
  As the operator norm is the maximal singular value, we find
  $\lVert {W_0 - \mathds{1}}\rVert _\infty = \max\bigl\{ \lvert {{e}^{i\alpha} - 1}\rvert ,
    \lvert {{e}^{-i\alpha} - 1}\rvert  \bigr\} = (1 - \cos\alpha)^2 + (\sin\alpha)^2 = 2 -
  2\cos(\alpha) = 4\sin^2(\alpha/2) \leq \epsilon$, and
  $\bigl \lVert { W - \mathds{1}}\bigr \rVert _\infty \leq \epsilon$.
  Now let $V' = W_{BE} V$, with
  \begin{align}
    \lVert {V' - V}\rVert _\infty \leq \lVert {W - \mathds{1}}\rVert _\infty \lVert {V}\rVert _\infty \leq \epsilon\ .
  \end{align}
  We find
  \begin{align}
    \rho_E'
    &= \operatorname{tr}_B\bigl(V'\psi V'^\dagger\bigr)
    \nonumber\\
    &= 
      \operatorname{tr}_B\Bigl[\cos^2(\alpha)\, \lvert {\mu_1}\rangle \mkern -1.8mu\relax \langle{\mu_1}\rvert 
      + \cos(\alpha)\sin(\alpha)\, \bigl(\lvert {\mu_1}\rangle \mkern -1.8mu\relax \langle{\mu_2}\rvert  + \lvert {\mu_2}\rangle \mkern -1.8mu\relax \langle{\mu_1}\rvert \bigr)
      + \sin^2(\alpha)\, \lvert {\mu_2}\rangle \mkern -1.8mu\relax \langle{\mu_2}\rvert  \Bigr]
      \nonumber\\
    &=
      \cos^2(\alpha)\,\rho_E
      + \cos(\alpha)\sin(\alpha)
          \operatorname{tr}_B\Bigl[V\lvert {\psi}\rangle \mkern -1.8mu\relax \langle{\chi}\rvert  + \lvert {\chi}\rangle \mkern -1.8mu\relax \langle{\psi}\rvert V^\dagger \Bigr]
      + \sin^2(\alpha) \, \chi_E\ .
      \nonumber\\
    &=
      \cos^2(\alpha)\,\rho_E
      + \sin^2(\alpha) \, \chi_E\ .
      \label{z:ABSNrv9JDtxw}
  \end{align}
  The last equality holds thanks to our assumption that
  $P_{\rho_B} \lvert {\chi }\rangle = 0$.

  We now prove the first part of the proposition.  We can assume without loss of
  generality that $\operatorname{rank}( P_{\rho_E}^\perp ) \leq \operatorname{rank}( P_{\rho_B}^\perp ) $,
  by exchanging the roles of the $B$ and $E$ systems if necessary.
  Let $\{ \lvert {\chi_k}\rangle _B \}_{k=1}^{K}$, $\{ \lvert {\chi_k'}\rangle _E \}_{k=1}^{K}$ be two
  orthonormal families of states lying in the support of $P_{\rho_B}^\perp$ and
  $P_{\rho_E}^\perp$ respectively, with
  $K = \min\bigl\{ \operatorname{rank}(P_{\rho_B}^\perp) , \operatorname{rank}(P_{\rho_E}^\perp) \bigr\} =
  \operatorname{rank}(P_{\rho_E}^\perp)$.
  Define
  $\lvert {\chi}\rangle _{BE} = (1/\sqrt{K})\,\sum_{k} \lvert {\chi_k}\rangle _B\otimes\lvert {\chi_k}\rangle _E$.
  By construction, we have that
  $\chi_E = (1/K) \sum_{k=1}^{K} \lvert {k}\rangle \mkern -1.8mu\relax \langle{k}\rvert _E = P_{\rho_E}^\perp /
  \operatorname{tr}(P_{\rho_E}^\perp)$.
  It follows that the state~\eqref{z:ABSNrv9JDtxw} has full rank, and
  therefore our conditions for our uncertainty relation equality are fulfilled.

  Now we prove the second part of the proposition, and we assume that
  $\widehat{\mathcal{N}}(D_A^Z) = 0$, with
  $D_A^Z = \lvert {\xi}\rangle \mkern -1.8mu\relax \langle{\psi}\rvert +\lvert {\psi}\rangle \mkern -1.8mu\relax \langle{\xi}\rvert $.  The proof strategy is similar to
  above, to introduce a small ``rotation'' to fix the support of the state
  $\rho_E$ all while preserving the zero sensitivity-loss conditions~\eqref{z:zFMncm5bkXDt}.

  Without loss of generality, we may assume that $\bigl \lVert {\lvert {\xi}\rangle }\bigr \rVert  = 1$.  We
  define for later convenience
  \begin{align}
    Z_L &= \lvert {\xi}\rangle \mkern -1.8mu\relax \langle{\psi }\rvert + \lvert {\psi}\rangle \mkern -1.8mu\relax \langle{\xi }\rvert \ ;
    &
    \Pi_L &= \lvert {\psi}\rangle \mkern -1.8mu\relax \langle{\psi }\rvert + \lvert {\xi}\rangle \mkern -1.8mu\relax \langle{\xi }\rvert \ ;
    &
    \widetilde{Z}_L &= Z_L + (\mathds{1}- \Pi_L)\ ,
  \end{align}
  noting that $\widetilde{Z}_L$ is the unitary operator that flips the
  normalized states $\lvert {\psi}\rangle $ and $\lvert {\xi}\rangle $ and acts as the identity on the
  subspace that is orthogonal to $\lvert {\psi}\rangle ,\lvert {\xi}\rangle $.

  As stated in the claim, we assume that there exists a unitary operator $G_B$
  with the properties that 
  $P_{\rho_B} G_B P_{\rho_B} = 0$, $P_{\zeta_B} G_B P_{\zeta_B} = 0$,
  $P_{\rho_B} G_B P_{\zeta_B} = 0$, and $P_{\zeta_B} G_B P_{\rho_B} = 0$.

  Let $0<\epsilon\leq 1$.  Let $\alpha = \epsilon/2$ with $0<\alpha\leq 1/2$ and
  let
  \begin{align}
    V' &= \bigl( \cos(\alpha) V + \sin(\alpha)\,G_B\,V\,\widetilde{Z}_L \bigr)\ .
  \end{align}
 Then
  \begin{align}
    \lVert {V' - V}\rVert 
    &= \bigl \lVert { (\cos(\alpha) - 1)\, V + \sin(\alpha)\, G_B V \widetilde{Z} }\bigr \rVert 
      \nonumber\\
    &\leq (1 - \cos(\alpha))\bigl \lVert {V}\bigr \rVert  + \sin(\alpha)\,\bigl \lVert {G_B V \widetilde{Z}}\bigr \rVert 
      \nonumber\\
    &
      \leq  2\sin^2(\alpha/2) + \sin(\alpha)
      \leq 2\lvert {\alpha}\rvert \ ,
  \end{align}
  using $\sin(\alpha)\leq\lvert {\alpha}\rvert $ and with $\lvert {\alpha}\rvert \leq 1/2$.

  We first show that the perturbed isometry $V'$ also satisfies the zero
  sensitivity-loss conditions.  Let
  $\widehat{\mathcal{N}}'(\cdot) = \operatorname{tr}_B\Bigl\{ V'\,(\cdot)\,V'^\dagger \Bigr\}$ and we
  compute
  \begin{align}
    \widehat{\mathcal{N}}'(Z_L)
    &= \operatorname{tr}_B\Bigl\{ V'\, Z_L\, V'^\dagger \Bigr\}
      \nonumber\\
    \begin{split}
      &= \operatorname{tr}_B\Bigl\{
        \cos^2(\alpha)\, V\, Z_L\, V^\dagger
        \\
        &\qquad\qquad
        + \cos(\alpha)\sin(\alpha)\, \Bigl[
          V\,Z_L\, \widetilde{Z}_L \, V^\dagger G_B^\dagger
          + G_B V \widetilde{Z}_L \, Z_L \, V^\dagger
        \Bigr]
        \\
        &\qquad\qquad
        + \sin^2(\alpha)\, G_B V\, \widetilde{Z}_L Z_L \widetilde{Z}_L \, 
        V^\dagger G_B^\dagger
    \Bigr\}
    \end{split}
      \nonumber\\
    &= 0\ ,
  \end{align}
  using $\operatorname{tr}_B\bigl\{ V Z_L V^\dagger \bigr\} = \widehat{\mathcal{N}}(Z_L) = 0$ and
  $\widetilde{Z}_L Z_L \widetilde{Z}_L = Z_L$, as well as the fact that
  \begin{align}
    \operatorname{tr}_B\bigl[ G_B V \widetilde{Z}_L Z_L  V^\dagger\bigr]
    &= \operatorname{tr}_B\bigl[ G_B V [\lvert {\psi }\rangle \mkern -1.8mu\relax \langle{\psi }\rvert + \lvert {\xi}\rangle \mkern -1.8mu\relax \langle{\xi}\rvert ] V^\dagger \bigr]
      \nonumber\\
    &= \operatorname{tr}_B \bigl[ G_B P_{\rho_B} V\psi V^\dagger P_{\rho_B}
    + G_B P_{\zeta_B} V \xi V^\dagger P_{\zeta_B} \bigr]
      = 0\ ,
  \end{align}
  using the fact that
  $P_{\rho_B} G_B P_{\rho_B} = 0 = P_{\zeta_B} G_B P_{\zeta_B}$.

  We then have
  \begin{align}
    \rho_E' &= \operatorname{tr}_B\Bigl\{ V'\,\psi\,V'^\dagger \Bigr\}
              \nonumber\\
    &
      =\operatorname{tr}_B\Bigl\{ \cos^2(\alpha)\,V\psi V^\dagger
      + \cos(\alpha)\sin(\alpha)\Bigl[ V\lvert {\psi }\rangle \mkern -1.8mu\relax \langle{\xi }\rvert V^\dagger G_B^\dagger
      + G_B V \lvert {\xi }\rangle \mkern -1.8mu\relax \langle{\psi }\rvert V^\dagger \Bigr]
      + \sin^2(\alpha)\,G_B V\xi V G_B^\dagger
      \Bigr\}
      \nonumber\\
    &
      = \cos^2(\alpha)\,\rho_E + \sin^2(\alpha)\,\zeta_E\ ,
  \end{align}
  where the two middle terms in the long expression vanish because
  $\operatorname{tr}_B\bigl\{ G_B V \lvert {\xi }\rangle \mkern -1.8mu\relax \langle{\psi }\rvert V^\dagger \bigr\} = \operatorname{tr}_B\bigl\{ P_{\rho_B} G_B
    P_{\zeta_B} V \lvert {\xi }\rangle \mkern -1.8mu\relax \langle{\psi }\rvert V^\dagger \bigr\} = 0$.

  Similarly,
  \begin{align}
    \zeta_E'
    &= \operatorname{tr}_B\Bigl\{ V' \, \xi\, V'^\dagger \Bigr\}
      \nonumber\\
    &= \operatorname{tr}_B\Bigl\{
      \cos^2(\alpha)\,V\xi V^\dagger
      + \cos(\alpha)\sin(\alpha) \Bigl[
        V\lvert {\xi }\rangle \mkern -1.8mu\relax \langle{\psi }\rvert V^\dagger G_B^\dagger
        + G_B V\lvert {\psi }\rangle \mkern -1.8mu\relax \langle{\xi }\rvert V^\dagger
      \Bigr]
      + \sin^2(\alpha)\, G_B V \psi V^\dagger G_B^\dagger
      \Bigr\}
      \nonumber\\
    &= \cos^2(\alpha)\, \zeta_E + \sin^2(\alpha)\,\rho_E\ .
  \end{align}
  Any state $\lvert {c}\rangle _E$ that lies in the kernel of $\rho_E'$ must satisfy
  \begin{align}
    0 = \langle {c}\mkern 1.5mu\relax \vert \mkern 1.5mu\relax { \rho_E' }\mkern 1.5mu\relax \vert \mkern 1.5mu\relax {c}\rangle 
    = \cos^2(\alpha)\,\langle {c}\mkern 1.5mu\relax \vert \mkern 1.5mu\relax {\rho_E}\mkern 1.5mu\relax \vert \mkern 1.5mu\relax {c}\rangle 
    + \sin^2(\alpha)\,\langle {c}\mkern 1.5mu\relax \vert \mkern 1.5mu\relax {\zeta_E}\mkern 1.5mu\relax \vert \mkern 1.5mu\relax {c}\rangle \ ,
  \end{align}
  which in turn implies $0 = \langle {c}\mkern 1.5mu\relax \vert \mkern 1.5mu\relax {\rho_E}\mkern 1.5mu\relax \vert \mkern 1.5mu\relax {c}\rangle  = \langle {c}\mkern 1.5mu\relax \vert \mkern 1.5mu\relax {\zeta_E}\mkern 1.5mu\relax \vert \mkern 1.5mu\relax {c}\rangle $.
  We then find
  \begin{align}
    \bigl \lVert { (\mathds{1}_B\otimes\langle {c}\rvert _E)\, V'\lvert {\xi }\rangle }\bigr \rVert ^2
    &= \langle {c}\rvert _E \, \operatorname{tr}_B\bigl(V'\xi V'^\dagger\bigr) \, \lvert {c}\rangle _E
      \nonumber\\
    &= 
      \langle {c}\mkern 1.5mu\relax \vert \mkern 1.5mu\relax {\zeta_E'}\mkern 1.5mu\relax \vert \mkern 1.5mu\relax {c}\rangle 
      \nonumber\\
    &= \langle {c}\rvert _E \, \Bigl[
      \cos^2(\alpha)\,\zeta_E + \sin^2(\alpha)\,\rho_E
      \Bigr]\, \lvert {c}\rangle _E
      \nonumber\\
    &= 0\ .
  \end{align}
  Therefore $\bigl(\mathds{1}_B\otimes P_{\rho_E'}^\perp\bigr)\,V'\lvert {\xi }\rangle = 0$, implying
  that $\bigl(P_{\rho_B'}^\perp \otimes P_{\rho_E'}^\perp\bigr)\,V'\lvert {\xi }\rangle = 0$ and
  our uncertainty relation equality conditions are satisfied.
\end{proof}

\section{Behavior of metrological codes for weak i.i.d.\@ noise; metrological
  codes, uncertainty relation equality, and discontinuities of the quantum
  Fisher information}
\label{z:z.6FF4fEKkkd}

In this Appendix, we consider a metrological code $(\lvert {\psi}\rangle ,\lvert {\xi}\rangle )$ on $n$
qubits, with a metrological distance $d_m > 1$.  For any noise channel that acts
on fewer than $d_m$ qubits, we have seen in
\cref{z:JZtfKZv25too} that $\FIloss{Bob}{t} =0$.  Instead of
noise acting on few qubits, we now consider examples of i.i.d.\@ noise channels
$[\mathcal{N}_1^{(p)}]^{\otimes n}$, where each channel $\mathcal{N}_1^{(p)}$
acts on a single qubit and depends on a noise parameter $p$ such that
$\mathcal{N}_1^{(p=0)} = {\mathrm{id}}$.
We ask: For constant $n$, to what order in $p$ is the loss in quantum Fisher
information $\FIloss{Bob}{t}$ suppressed?

Let us first consider a similar question in the conventional setting of quantum
error correction, where a logical state is encoded into a physical state, is
exposed to a noise channel, and is subsequently decoded to attempt to recover
the initial state.  If a state $\lvert {\psi}\rangle $, encoded with a distance-$d$ quantum
error-correcting code, is exposed to a weak i.i.d.\@ noise channel in which a
single-site error happens with probability $p$, then after a subsequent decoding
operation, the fidelity of the state with respect to the original state differs
with the ideal value one by at most $O( p^{d/2} )$.  I.e., the fidelity loss is
suppressed by the quantum error correction procedure to an order in the noise
parameter that is proportional to the distance of the code.  This suppressed
fidelity loss is explained by a fundamental principle in quantum information:
Two states (respectively two channels) that are $\epsilon$-close in trace
distance (respectively diamond distance) may not be distinguished by any
physical operation, except with probability of the order at most $O(\epsilon)$.
In the case of weak i.i.d.\@ noise, any error operator whose weight is larger
than $(d-1)/2$ occurs only with probability at most $O(p^{d/2})$.  Consequently,
no experiment should be able to distinguish the weak i.i.d.\@ noise from a noise
operator with only weight-$[(d-1)/2]$ operators with probability better than
$O(p^{d/2})$, for which the quantum error-correction scheme enables perfect
recovery.

By analogy, it is natural to expect that the quantum Fisher information loss
$\FIloss{Bob}{t}$ should scale as $\sim p^{c d_m}$, where $d_m$ is the
metrological distance of the metrological code, and where $c$ is some constant.
However, this is not the case, as we will see in the remainder of this appendix.
While $\FIloss{Bob}{t}$ exhibits the expected behavior for certain examples of
metrological codes, we can find counterexamples in which the quantum
Fisher information loss scales as $\FIloss{Bob}{t} \sim p$ despite the state
forming a metrological code of an arbitrarily large, but fixed, metrological
distance $d_m$.
This counterexample shows that when measuring the accuracy of Bob's estimate to
the time parameter in terms of the quantum Fisher information, the code distance
is not necessarily related to the loss in sensitivity of the state.  This might
be worrying, since the metrological distance of the metrological code would not
be related to the degree of protection offered by such codes in suppressing the
sensitivity loss.  We argue, however, that the quantum Fisher information might
not be the relevant sensitivity measure to study in such regimes.  More
specifically, we know that there are regimes in which we should question the
operational relevance of the quantum Fisher information, because infinitesimal
perturbations in the state or the noise channel result in observable
consequences in the purported sensitivity as reported by the quantum Fisher
information.  We attribute this behavior to the fact that it ignores the error
associated with the estimation of the expectation value of the optimal sensing
observable from a finite number of measurement repetitions.
Based on our examples, we hypothesize that the settings where
$\FIloss{Bob}{t} \not \leq O(p^{d_m/2})$ fall into this regime.
While confirming this
hypothesis would invalidate known counterexamples in which a high metrological
distance can still lead to a high accuracy loss, a full proof of the protection
offered by metrological codes in the general setting remains elusive.  Such a
result would further require (a)~establishing a measure of sensitivity that is
robust to perturbations of the physical setting by accounting for limits on the
number of available measurement repetitions and (b)~showing that its loss is
suppressed as a function of the metrological distance of the metrological code.

In the following, we first compute the quantum Fisher information loss of some
states that form metrological codes after exposure to weak i.i.d.\@ noise.  In
order to explore the cause of the behavior of some examples that appear
problematic, we study more closely some properties of the quantum Fisher
information: We argue that there are regimes in which the quantum Fisher
information, being discontinuous, cannot be a representative measure of
sensitivity, and we attribute this problematic behavior to the failure to
account for the number of finite available measurement repetitions.  Finally, we
consider a restricted setting with additional assumptions on the state and the
noise channel, in which we prove the expected bound on the quantum Fisher
information loss $\FIloss{Bob}{t} \leq O(p^{d_m/2})$.

\subsection{Examples of metrological codes exposed to weak i.i.d.\@ noise}

We now consider three single-site noise channels: the amplitude-damping channel, the
dephasing channel in the $Z$ basis, and the bit-flip channel.  In the basis
$\{ \lvert {\uparrow}\rangle , \lvert {\downarrow}\rangle \}$, the single-qubit amplitude-damping channel
has Kraus operators
\begin{align}
  E_{\mathrm{a.d.}, \, 0}^{(p)}
  &=
    \begin{pmatrix}
    \sqrt{1-p} & 0\\
    0 & 1
  \end{pmatrix}\ ;
  &
  E_{\mathrm{a.d.}, \, 1}^{(p)}
  &=
    \begin{pmatrix}
    0 & 0\\
    \sqrt{p} & 0
  \end{pmatrix}\ .
\end{align}
The second noise channel we consider is the dephasing channel in the $Z$ basis,
described by the Kraus operators
\begin{align}
  E_{\textrm{dephas.}, \, 0}^{(p)}
  &=
    \sqrt{1-\frac{p}{2}}
    \begin{pmatrix}
    1 & 0\\
    0 & 1
  \end{pmatrix}\ ;
  &
  E_{\textrm{dephas.}, \, 1}^{(p)}
  &=
    \sqrt{\frac{p}{2}}
    \begin{pmatrix}
    1 & 0\\
    0 & -1
  \end{pmatrix}\ .
\end{align}
Finally, the bit-flip channel is described by the Kraus operators
\begin{align}
  E_{\textrm{bit-flip}, \, 0}^{(p)}
  &=
    \sqrt{1-\frac{p}{2}}
    \begin{pmatrix}
    1 & 0\\
    0 & 1
  \end{pmatrix}\ ;
  &
  E_{\textrm{bit-flip}, \, 1}^{(p)}
  &=
    \sqrt{\frac{p}{2}}
    \begin{pmatrix}
    0 & 1\\
    1 & 0
  \end{pmatrix}\ .
\end{align}

\subsubsection{Four-qubit code state based on the $[[4,2,2]]$ code}

Consider the state vector introduced
in~\cref{z:AVGHDyn1nGzm,z:oz0py5ItCeCD},
\begin{align}
  \lvert {\psi_{\textrm{code}}}\rangle 
  = \frac12\Bigl[
  \lvert {\uparrow\uparrow\uparrow\uparrow}\rangle  +
  \lvert {\downarrow\downarrow\downarrow\downarrow}\rangle  +
  \lvert {\uparrow\downarrow\uparrow\downarrow}\rangle  +
  \lvert {\downarrow\uparrow\downarrow\uparrow}\rangle 
  \Bigr]\ .
  \label{z:boEeCtdT1jUf}
\end{align}
Consider the Hamiltonian consisting of $ZZ$ terms on the edges connecting the
four qubits when they are arranged in a square, as in
\cref{z:q6y-xgqvsyUu}\textbf{a}; with a suitable normalization we
obtain
\begin{align}
  \lvert {\xi_{\textrm{code}}}\rangle 
  = \frac12\Bigl[
  \lvert {\uparrow\uparrow\uparrow\uparrow}\rangle  +
  \lvert {\downarrow\downarrow\downarrow\downarrow}\rangle 
  -\lvert {\uparrow\downarrow\uparrow\downarrow}\rangle 
  -\lvert {\downarrow\uparrow\downarrow\uparrow}\rangle 
  \Bigr]\ .
\end{align}
We have seen that $(\lvert {\psi_{\textrm{code}}}\rangle , \lvert {\xi_{\textrm{code}}}\rangle )$ forms
a metrological code of metrological distance $2$.

Let us consider how the quantum Fisher information of this state drops when
exposed to i.i.d.\@ amplitude-damping noise and to i.i.d.\@ dephasing noise.
The quantum Fisher information loss $\FIloss{Bob}{t}$ is plotted in a log-log
plot as a function of $p$ in
\cref{z:4-cGO3oXubrx}.
\begin{figure}
  \centering
  \includegraphics{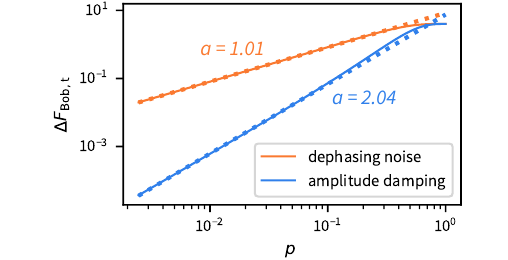}
  \caption{Quantum Fisher information loss $\FIloss{Bob}{t}$ after exposure
    of $\lvert {\psi_{\textrm{code}}}\rangle $ [cf.\@
    \cref{z:boEeCtdT1jUf}] to i.i.d.\@
    amplitude-damping or dephasing noise in the $Z$ basis, as a function of the
    noise parameter $p$.  Based on our intuition of standard error-correcting
    codes, we might have expected that $\FIloss{Bob}{t}$ depends only on an
    order in $p$ (for $p\to 0$) that is directly related to $d_m$ (or $d_m/2$).
    In the case of either noise model, we fit the data points where $p<0.1$ to
    $\ln(y) = a\ln(p) + b$ (which corresponds to a power law $y\propto p^a$) to
    obtain the order in $p$ to which $\FIloss{Bob}{t}$ is affected.  We see that
    for amplitude-damping noise, the loss in quantum Fisher information is
    suppressed to depend only on $p$ to second order; for dephasing noise, the
    loss is affected to first order in $p$.  The quantum Fisher information loss
    due to an i.i.d.\@ bit-flip noise channel (not shown) behaves very similarly to the
    dephasing noise.}
  \label{z:4-cGO3oXubrx}
\end{figure}
We fit the computed values for points with $p<0.1$ to the model
$\ln(y) = a \ln(p) + b$ in order to determine the quantum Fisher information
loss order (as $y \propto p^a$).
We observe that while the quantum Fisher information loss is indeed affected
only to second order in $p$ for amplitude-damping noise, it is directly
affected to first order for dephasing noise.
The behavior of this small-scale example is not necessarily surprising, although
it rules out an optimistic conjecture that states of the
form~\eqref{z:boEeCtdT1jUf} could have
their loss in quantum Fisher information be protected to second order in $p$
against any i.i.d.\@ noise channel, as could have been suggested from
\cref{z:eNrtlC61cj3z}.

\subsubsection{Repetition code in the $+/-$ basis}

Now we investigate a larger example that shows that the metrological distance is
not always indicative of the order of quantum Fisher information loss in the
noise parameter.
On $n$ qubits, let
\begin{align}
  \lvert {\psi }\rangle &= \lvert {+}\rangle ^{\otimes n}\ ;
  &
    \lvert {\xi }\rangle &= \lvert {-}\rangle ^{\otimes n}\ .
\end{align}
Here, the Hamiltonian corresponding to these states is the nonlocal operator
$H = Z^{\otimes n}$.  (Note that this example differs starkly from a standard
ensemble of $n$ spins where the Hamiltonian is as a sum of $Z$ terms on each
site.  In that case, $\lvert {\xi}\rangle $ would be a superposition of strings that consist
of all $\lvert {+}\rangle $ state vectors and a single $\lvert {-}\rangle $ state vector.)
The $(\lvert {\psi}\rangle , \lvert {\xi}\rangle )$ given above form a metrological code of distance
$d_m = n$.  Indeed, any operator $O$ with $\wgt(O)<n$ cannot make $\lvert {\psi}\rangle $
nonorthogonal to $\lvert {\xi}\rangle $, and the
conditions~\eqref{z:y2lXSpYS4Z-t} are
satisfied.

We show that if we expose this state to i.i.d.\@ dephasing noise along the $Z$
axis, the quantum Fisher information loss is indeed suppressed to the order
$O(p^{n/2})$, as we would expect.  On the other hand, if we expose the state to
i.i.d.\@ bit-flip noise, which can be seen as dephasing noise along the $X$
axis, then the quantum Fisher information loss is not suppressed as expected and
we find $\FIloss{Bob}{t} \sim p$.

Let us first consider i.i.d.\@ dephasing noise along the $Z$ axis.  We show that
the quantum Fisher information loss is indeed suppressed to order $O(p^{n/2})$
for this noise channel.  We now prove this statement.  We may choose for the
noise channel
$\mathcal{N}_{\textrm{dephas}}^{(p)}(\cdot) = (1-p/2)(\cdot) + (p/2) Z(\cdot)Z$
the Stinespring isometry
\begin{align}
  V_{A\to BE}
  &= \sqrt{1-\frac{p}{2}} \, \mathds{1}\otimes\lvert {0}\rangle _E + \sqrt{\frac{p}{2}} \, Z\otimes\lvert {1}\rangle _E
    \nonumber\\
  &= \lvert {p_+}\rangle _E\langle {\uparrow}\rvert _A\otimes\lvert {\uparrow}\rangle _B +
    \lvert {p_-}\rangle _E\langle {\downarrow}\rvert _A\otimes\lvert {\downarrow}\rangle _B\ ,
\end{align}
with respect to some basis $\lvert {0}\rangle ,\lvert {1}\rangle $ on $E$, and with
\begin{align}
  \lvert {p_\pm}\rangle  = \sqrt{1 - \frac{p}{2}} \, \lvert {0}\rangle  \pm \sqrt{\frac{p}{2}}\, \lvert {1}\rangle \ .
\end{align}
This choice leads to the complementary channel
\begin{align}
  \widehat{\mathcal{N}}_{\textrm{dephas}}^{(p)}(\cdot)
  = \langle {\uparrow}\mkern 1.5mu\relax \vert \mkern 1.5mu\relax {\cdot}\mkern 1.5mu\relax \vert \mkern 1.5mu\relax {\uparrow}\rangle _A\,\lvert {p_+}\rangle \mkern -1.8mu\relax \langle{p_+}\rvert _E
  + \langle {\downarrow}\mkern 1.5mu\relax \vert \mkern 1.5mu\relax {\cdot}\mkern 1.5mu\relax \vert \mkern 1.5mu\relax {\downarrow}\rangle _A\,\lvert {p_-}\rangle \mkern -1.8mu\relax \langle{p_-}\rvert _E \ .
\end{align}
We find
\begin{align}
  \rho_{E_1}
  = \widehat{\mathcal{N}}_{\textrm{dephas}}^{(p)}(\lvert {+}\rangle \mkern -1.8mu\relax \langle{+}\rvert )
  &= \frac12\Bigl[ \lvert {p_+}\rangle \mkern -1.8mu\relax \langle{p_+}\rvert _E + \lvert {p_-}\rangle \mkern -1.8mu\relax \langle{p_-}\rvert _E \Bigr]
  = \begin{pmatrix}
    1-p & 0 \\ 0 & p
  \end{pmatrix}\ ;
  \nonumber\\
  \widehat{\mathcal{N}}_{\textrm{dephas}}^{(p)}(\lvert {+}\rangle \mkern -1.8mu\relax \langle{-}\rvert )
  &= \frac12\Bigl[ \lvert {p_+}\rangle \mkern -1.8mu\relax \langle{p_+}\rvert _E - \lvert {p_-}\rangle \mkern -1.8mu\relax \langle{p_-}\rvert _E \Bigr]
  = \sqrt{\frac p2 \Bigl(1-\frac p2\Bigr)} \begin{pmatrix}
    0 & 1 \\ 1 & 0
  \end{pmatrix}\ .
\end{align}
We can then compute
\begin{align}
  \FIloss{Bob}{t}
  & = \Ftwo[\Big]{
      \bigl[ \widehat{\mathcal{N}}_{\textrm{dephas}}^{(p)}(\lvert {+}\rangle \mkern -1.8mu\relax \langle{+}\rvert ) \bigr]^{\otimes n}
    }{
      \bigl[\widehat{\mathcal{N}}_{\textrm{dephas}}^{(p)}\bigr]^{\otimes n}\bigl(
          [\lvert {+}\rangle \mkern -1.8mu\relax \langle{-}\rvert ]^{\otimes n} + [\lvert {-}\rangle \mkern -1.8mu\relax \langle{+}\rvert ]^{\otimes n}
      \bigr)
    }
    \nonumber\\
  &=
    \Ftwo[\bigg]{
    \begin{pmatrix}
      1-p & 0 \\ 0 & p
    \end{pmatrix}
                     ^{\otimes n}
  }{
    \Bigl[ \sqrt{\frac p2 \Bigl(1-\frac p2\Bigr)} \begin{pmatrix}
    0 & 1 \\ 1 & 0
  \end{pmatrix} \Bigr]^{\otimes n} + \textrm{h.c.}
  }
    \nonumber\\
  &=
    4 \Bigl[\frac{p}{2}\Bigl(1 - \frac{p}{2}\Bigr)\Bigr]^n
    \Ftwo[\bigg]{
      \begin{pmatrix}
        1-p & 0 \\ 0 & p
      \end{pmatrix}
                     ^{\otimes n}
      }{
        X^{\otimes n}
      }
      \nonumber\\
  &=
    4\,p^n \, 2^{-n} \Bigl(1 - \frac{p}{2}\Bigr)^n \,
    \sum_{\boldsymbol{x}, \boldsymbol{x}'}
    \frac2{\lambda_{\boldsymbol x} + \lambda_{\boldsymbol x'}}
    \Bigl \lvert {
      \langle {\boldsymbol x}\mkern 1.5mu\relax \vert \mkern 1.5mu\relax { X^{\otimes n} }\mkern 1.5mu\relax \vert \mkern 1.5mu\relax {\boldsymbol x'}\rangle 
    }\Bigr \rvert ^2\ ,
    \label{z:AGHy7t7l-7rX}
\end{align}
where $\boldsymbol x, \boldsymbol x'$ are bit strings and where
$\lambda_{\boldsymbol x} = (1-p)^{\lvert {\boldsymbol x}\rvert } p^{n - \lvert {\boldsymbol
    x}\rvert }$ is the eigenvalue of $\rho_{E_1}^{\otimes n}$ associated with the
eigenvector $\lvert {\boldsymbol x}\rangle $.  Observe that
$\langle {\boldsymbol x}\mkern 1.5mu\relax \vert \mkern 1.5mu\relax { X^{\otimes n} }\mkern 1.5mu\relax \vert \mkern 1.5mu\relax {\boldsymbol x'}\rangle  =
\delta_{\widetilde{\boldsymbol x}, \boldsymbol x'}$, where
$\widetilde{\boldsymbol x}$ is the bit string obtained by flipping all the bits
of $\boldsymbol x$.  Then
\begin{align}
  \lambda_{\boldsymbol x} + \lambda_{\widetilde{\boldsymbol x}}
  &=
  (1-p)^{\lvert {\boldsymbol x}\rvert } p^{n - \lvert {\boldsymbol x}\rvert }
  + (1-p)^{n-\lvert {\boldsymbol x}\rvert } p^{\lvert {\boldsymbol x}\rvert }
  =
    \Omega( p^{\min(\lvert {\boldsymbol x}\rvert , n-\lvert {\boldsymbol x}\rvert )} )
  =
  \Omega( p^{n/2} )\ ,
\end{align}
noting that $\min(\lvert {\boldsymbol x}\rvert , n-\lvert {\boldsymbol x}\rvert ) \leq n/2$.
Therefore,
\begin{align}
  \text{\eqref{z:AGHy7t7l-7rX}}
  &=
    p^n\,
    \sum_{\boldsymbol x} O\bigl(p^{-n/2}\bigr)
    = O(p^{n/2})\ .
\end{align}
I.e., the quantum Fisher information on Bob's end after exposure of the state to
i.i.d.\@ dephasing noise along the $Z$ axis is well protected, in that the loss
is suppressed to the order $O(p^{n/2})$.
Observe that $\rho_E$ is full rank, and therefore our uncertainty relation holds
with equality in this setting.

Consider now the i.i.d.\@ bit-flip noise channel
$[\mathcal{N}_{\textrm{bit-flip}}^{(p)}]^{\otimes n}$ determined by the
single-site Kraus operators $E_{\textrm{bit-flip},\,0}^{(p)}$ and
$E_{\textrm{bit-flip},\,1}^{(p)}$.  We find
\begin{align}
  \mathcal{N}_{\textrm{bit-flip}}^{(p)}(\lvert {+}\rangle \mkern -1.8mu\relax \langle{+}\rvert )
  &= \lvert {+}\rangle \mkern -1.8mu\relax \langle{+}\rvert \ ;
  &
  \mathcal{N}_{\textrm{bit-flip}}^{(p)}(\lvert {+}\rangle \mkern -1.8mu\relax \langle{-}\rvert )
  &= (1-p)\lvert {+}\rangle \mkern -1.8mu\relax \langle{-}\rvert \ .
\end{align}
We would like to compute
\begin{align}
  \FIqty{Bob}{t}
  &=
    \Ftwo[\Big]{
      [\mathcal{N}_{\textrm{bit-flip}}^{(p)}]^{\otimes n}\bigl(\lvert {+}\rangle \mkern -1.8mu\relax \langle{+}\rvert ^{\otimes n}\bigr)
    \ 
    }{\ 
      [\mathcal{N}_{\textrm{bit-flip}}^{(p)}]^{\otimes n}\bigl(
        -i[\lvert {-}\rangle \mkern -1.8mu\relax \langle{+}\rvert ]^{\otimes n} + i[\lvert {+}\rangle \mkern -1.8mu\relax \langle{-}\rvert ]^{\otimes n}
      \bigr)
    }
    \nonumber\\
  &=
    \Ftwo[\Big]{
     \lvert {+}\rangle \mkern -1.8mu\relax \langle{+}\rvert ^{\otimes n}
    \ 
    }{\ 
    -i [(1-p)\lvert {-}\rangle \mkern -1.8mu\relax \langle{+}\rvert ]^{\otimes n} + \textrm{h.c.}
    }
    \nonumber\\
  &=
    4\; \langle {+^n}\mkern 1.5mu\relax \vert \mkern 1.5mu\relax {
    \Bigl[ -i [(1-p)\lvert {-}\rangle \mkern -1.8mu\relax \langle{+}\rvert ]^{\otimes n} + \textrm{h.c.} \Bigr]^2
    }\mkern 1.5mu\relax \vert \mkern 1.5mu\relax {+^n}\rangle \ ,
    \label{z:JbUkup2h0D9R}
\end{align}
where the last equality follows from \cref{z:Sj25VG-tiqC5}.
With
\begin{align}
  \langle {+^n}\rvert \Bigl[ -i[(1-p)\lvert {-}\rangle \mkern -1.8mu\relax \langle{+}\rvert ]^{\otimes n} + \textrm{h.c.} \Bigr]
  = i (1-p)^n \langle {-^n}\rvert \ ,
\end{align}
we find
\begin{align}
  \text{\eqref{z:JbUkup2h0D9R}}
  &= 4 (1-p)^{2n} = 4 - 8np + O(p^2)\ .
\end{align}
Therefore, for bit-flip i.i.d.\@ noise, we have
\begin{align}
  \FIloss{Bob}{t} [\textrm{bit-flip}] = 8np + O(p^2)\ ,
  \label{z:Cloty.JHWOCt}
\end{align}
meaning that the quantum Fisher information 
loss is linear in $p$ despite the
high metrological distance $d_m$.

Note that, in the case of i.i.d.\@ bit-flip noise, our uncertainty relation
equality conditions are not satisfied, since the rank of $\rho_B$ changes
locally as a function of time.  I.e., we should not expect our uncertainty
relation to hold with equality.  This fact does not impact our calculation of
the quantum Fisher information loss~\eqref{z:Cloty.JHWOCt}, since we
determined this value by direct computation on Bob's side.  However, based on
this example, we are tempted to hypothesize that settings in which a high
metrological distance does not inhibit a high accuracy loss under weak i.i.d.\@
noise coincide with the settings in which our uncertainty relation does not hold
with equality.  In the remainder of this Appendix, we provide additional
indications in favor of this hypothesis.

\subsection{Discontinuities of the quantum Fisher and uncertainty
  relation equality conditions}

We briefly return to study the behavior of the quantum Fisher information in a
simple example in which our uncertainty relation equality conditions are not
satisfied.  In such cases, the state on Bob's side changes rank, and it is known
that the quantum Fisher information can be
discontinuous~\cite{R32,R33,R109}.

The definition of the quantum Fisher information that we use
[\cref{z:.eb-akuquBNd}], which can differ from the expression stemming
from the second-order expansion of the Bures
metric~\cite{R32,R33,R109}, directly expresses the
accuracy to which one can sense an unknown parameter via an observable that
reveals the true value of the parameter locally in expectation value (see
\cref{z:ZoysZzb.FHNr} in
\cref{z:9mLm-0LXJIol}).

It is a fundamental principle in quantum information that a quantity that is
measurable in a physical setting should be robust to infinitesimal perturbations
of the quantum state.  Yet, how is possible that the quantum Fisher information
is discontinuous, if it directly corresponds to the physically operational
sensitivity to which one can estimate an unknown parameter locally?
We attribute this discontinuity to the assumption, in
\cref{z:ZoysZzb.FHNr} in
\cref{z:9mLm-0LXJIol}, that the sensing observable
reveals the true parameter value \emph{in expectation value}.  An expectation
value needs to be estimated using multiple rounds of measurements, and depending
on the outcome distribution of the observable, an arbitrary large number of
measurements might be required to accurately estimate its expectation value.
In the following example, we study how the optimal sensing observable diverges
close to discontinuity points of the quantum Fisher information; namely, the
discontinuity can be associated with diverging eigenvalues of the observable
associated with eigenstates that are outside the support of the state at the
discontinuity point.

Overall, this example indicates that the operational relevance of the quantum
Fisher information might break down in certain regimes where it is not possible
to accurately estimate the expectation value of the optimal sensing observable.

The following example is based on Refs.~\cite{R32,R33,R109}.
Consider the example of
\cref{z:TaWot2w4x4dl}: A qubit state
evolving along the equator of the Bloch sphere is collapsed by the noise channel
along the $X$ axis of the Bloch sphere.  Bob's quantum Fisher information is
constant and equal to $\omega^2$ almost all the time, except when the state is
exactly a $\pm X$ eigenstate, in which case Bob's quantum Fisher information is
equal to zero.  The state on Bob's end is given by
\cref{z:qdCKlozHjX8s} as
\begin{align}
  \rho_B &= p_+\lvert {+}\rangle \mkern -1.8mu\relax \langle{+}\rvert  + p_-\lvert {-}\rangle \mkern -1.8mu\relax \langle{-}\rvert \ ;
  &
    p_+ &= \cos^2\Bigl(\frac{\omega t_0}{2}\Bigr)\ ;
  &
    p_- &= \sin^2\Bigl(\frac{\omega t_0}{2}\Bigr)\ .
\end{align}
When Bob's quantum Fisher information $\FIqty{Bob}{t}$ is nonzero, there is
always an observable $O$ whose expectation value reveals the true parameter
value locally, i.e.\@ $\langle {O}\rangle _{\rho(t_0+dt)} = t_0 + dt + O(dt^2)$, and whose
variance is $\langle {O^2}\rangle  - \langle {O}\rangle ^2 = 1/\omega^2$ (cf.\@
\cref{z:9mLm-0LXJIol}).  The optimal sensing observable
is given by the suitably normalized symmetric logarithmic derivative
(\cref{z:ZoysZzb.FHNr}) and can be computed, when $\omega t_0$ is
not a multiple of $\pi$, as follows:
\begin{align}
  O - t_0\mathds{1}&= \frac1{\omega^2} \mathcal{R}_{\rho_{t_0}}^{-1}\bigl( \partial_t \rho \bigr)
  = \frac1{\omega^2}
    \sum_{k,k'=\pm} \frac2{p_k + p_{k'}} \langle {k}\mkern 1.5mu\relax \vert \mkern 1.5mu\relax { (\partial_t \rho) }\mkern 1.5mu\relax \vert \mkern 1.5mu\relax {k'}\rangle  \lvert {k}\rangle \mkern -1.8mu\relax \langle{k'}\rvert 
    \nonumber\\
  &= \frac2{2\omega^2 \cos^2(\frac{ \omega t_0}{2} ) }
    \langle {+}\mkern 1.5mu\relax \vert \mkern 1.5mu\relax { (\partial_t \rho) }\mkern 1.5mu\relax \vert \mkern 1.5mu\relax {+}\rangle  \lvert {+}\rangle \mkern -1.8mu\relax \langle{+}\rvert 
    + \frac2{2\omega^2 \sin^2(\frac{\omega t_0}{2})}
    \langle {-}\mkern 1.5mu\relax \vert \mkern 1.5mu\relax { (\partial_t\rho) }\mkern 1.5mu\relax \vert \mkern 1.5mu\relax {-}\rangle  \lvert {-}\rangle \mkern -1.8mu\relax \langle{-}\rvert 
    \nonumber\\
  &=
    -\frac1\omega \tan\Bigl(\frac{\omega t_0}{2}\Bigr) \lvert {+}\rangle \mkern -1.8mu\relax \langle{+}\rvert 
    + \frac1\omega \Bigl[ \tan\Bigl(\frac{\omega t_0}{2}\Bigr) \Bigr]^{-1} \lvert {-}\rangle \mkern -1.8mu\relax \langle{-}\rvert \ ,
\end{align}
using the relation $(\partial_t\rho) = -(\omega/2)\sin(\omega t_0)\, X$ [cf.\@
\cref{z:5xfiUugKQWOt}], which implies
$\langle {\pm}\mkern 1.5mu\relax \vert \mkern 1.5mu\relax {(\partial_t\rho)}\mkern 1.5mu\relax \vert \mkern 1.5mu\relax {\pm}\rangle  = \mp(\omega/2)\sin(\omega t_0) = \mp\omega
\sin(\omega t_0/2)\cos(\omega t_0/2)$.

As a sanity check, we can verify that $O$ satisfies
\begin{align}
  \langle {O}\rangle _{\rho(t_0+dt)} = t_0 + dt + O(dt^2)\ ,
\end{align}
as well as
\begin{align}
  \sigma_O^2 = \langle {O^2}\rangle _{\rho_{t_0}} - \langle {O}\rangle _{\rho_{t_0}}^2
  = \frac1{\omega^2}\ .
\end{align}

As the state gets closer to a discontinuity (for instance at $t_0 = 0$), this
optimal sensing observable has one eigenvalue that diverges (for $t_0=0$, this
eigenvalue is associated with the eigenvector $\lvert {-}\rangle $). At the discontinuous
point, the derivative is zero locally, so no observable will ever be able to
correctly reveal the true value of the parameter to first order locally. The
state does not change to first order in $t$ at all!
We can attribute the discontinuity to the fact that an optimal sensing
observable for one state might turn out to no longer be an acceptable sensing
observable for a neighboring point.  In other words, while the variance of an
observable is continuous both as a function of the state and of the observable,
the \emph{optimal} variance in the local-sensing scenario is discontinuous
because the conditions of the
optimization~\eqref{z:URKHO5wxhm5z} are
discontinuous.

At the discontinuity $t_0 = 0$, the derivative $\partial_t\rho$ vanishes
locally, and it is impossible to find an observable $O$ such that
$\langle {O}\rangle _{\rho(t_0 + dt)} = t_0 + dt + O(dt^2)$.  By convention we set the
corresponding quantum Fisher information to be zero; first, it is convenient
because we do not have to modify the definition of the quantum Fisher
information, and second, it expresses the fact that we cannot have any
sensitivity locally to first order in the parameter by measuring the expectation
value of an observable.  If the quantum Fisher information is defined starting
from the Bures distance, a mismatch will be observed; this mismatch could be
interpreted as a failure of the Cramér-Rao bound.

Operationally, even for $t_0$ not at one of the discontinuities, the use of the
expectation value as the way of reading out the parameter in the estimation
process might be problematic.  Estimating the expectation value of $O$ to good
accuracy, for $t_0 \approx 0$, requires that we observe sufficiently many times
the $\lvert {-}\rangle $ outcome, even though the latter only appears with the vanishing
probability $\sin^2(\omega(t_0+dt)/2)$.  If we do not repeat the measurement on
enough copies, we would only empirically observe $\lvert {+}\rangle $ events and we would
erroneously estimate the expectation value of $O$ to be equal to
$-\bigl[ \tan(\omega t_0/2) \bigr]/\omega$, and that its variance is zero. Not only
this result would be wrong as it does not depend on the actual value $dt$ that we
wanted to measure, but the variance is certainly incorrect since the optimal
variance when an infinite number of measurements is available is $1/\omega^2$.
There might be opportunities for defining and investigating refined measures of
sensitivity that can account for the finite amount of measurement outcomes that
can be collected in the estimation process.

The above example illustrates that the quantum Fisher information can be
problematic to interpret in certain regimes close to points where the rank of
the state can change.  This type of regime can occur for metrological codes, if
the noise happens to fix the state vector $\lvert {\psi}\rangle $ while not fixing other states that
are infinitesimally close to $\lvert {\psi}\rangle $, resulting in a rank change for Bob and
Eve's states.  We observe that in the context of metrological codes exposed to
weak i.i.d.\@ noise, the quantum Fisher information is not actually
discontinuous as a function of the noise parameter; rather, it is the order in
$p$ of the Fisher information loss that can behave unexpectedly.  That the
quantum Fisher information loss must be suppressed at least to the order $O(p)$
follows from our continuity bound
\cref{z:-97yhn1ACkRa},
noting that the weak i.i.d.\@ noise channel is $O(p)$-close to the identity
channel.

\subsection{Suppression of quantum Fisher information loss in a restricted
  setting}

Here we show that, when considering a metrological code exposed to weak i.i.d.\@
noise in a restricted setting with additional assumptions, the quantum Fisher
information loss $\FIloss{Bob}{t}$ is suppressed to the expected order
$O(p^{d_m/2})$, where $d_m$ is the metrological distance of the metrological
code.

\begin{proposition}
  \label{z:gjPZkX-ZdRa-}
  Let $\lvert {\psi}\rangle ,\lvert {\xi}\rangle $ define a metrological code of metrological distance
  $d_m$.  Let $\mathcal{N}_1$ be a single-site noise operator with a Kraus
  representation $\{ E_1^{(k)} \}_{k=1}^K$ that is such
  that %
  $\lVert { E_1^{(k')} }\rVert _\infty = O\bigl(\sqrt{p}\bigr)$ for $k'\neq 1$.  Furthermore,
  if $\boldsymbol x$ denotes a string of Kraus operator labels with
  $x_i \in \{ 1, \ldots, K\}$, and if
  $E_{\boldsymbol x} = \bigl(\bigotimes_{i=1}^n E_1^{(x_i)}\bigr)$, we assume that
  the states $\bigl\{ E_{\boldsymbol x} \lvert {\psi }\rangle \bigr\}_{\boldsymbol x}$ are all
  nonzero and orthogonal, and that
  $\lVert { E_{\boldsymbol x}\lvert {\psi }\rangle }\rVert  \geq \Omega(p^{\lvert {x}\rvert /2})$.
  Then $\Delta F_{\mathrm{Bob},t} = O\bigl(p^{d_m/2}\bigr)$.
\end{proposition}

This result follows fairly straightforwardly from
\cref{z:1qgah0zhTYQY} in \cref{z:sZUQYA.IesEQ}.

\begin{proof}[**z:gjPZkX-ZdRa-]
  Using the notation in \cref{z:1qgah0zhTYQY}, with
  $\epsilon=p$, we have that $\Delta F_{\mathrm{Bob},t} = O(p^{m})$ with
  \begin{align}
    m = \min_{\boldsymbol{x},\boldsymbol{x}'}
    \Bigl\{ 2 q_{\boldsymbol{x}, \boldsymbol{x}'} - 
    \min\bigl(r_{\boldsymbol{x}}, r_{\boldsymbol{x}'}\bigr) \Bigr\}\ ,
  \end{align}
  where $r_{\boldsymbol x}$ and $q_{\boldsymbol x, \boldsymbol x'}$ are defined
  via
  \begin{align}
    \langle {\psi}\mkern 1.5mu\relax \vert \mkern 1.5mu\relax {E_{\boldsymbol x}^\dagger E_{\boldsymbol x}}\mkern 1.5mu\relax \vert \mkern 1.5mu\relax {\psi}\rangle 
    &= \Omega\bigl( p^{r_{\boldsymbol x}}\bigr)\ ;
    &
      \operatorname{tr}\bigl\{ E_{\boldsymbol x'}^\dagger E_{\boldsymbol x}
      \bigl( \lvert {\xi}\rangle \mkern -1.8mu\relax \langle{\psi }\rvert + \lvert {\psi}\rangle \mkern -1.8mu\relax \langle{\xi }\rvert \bigr) \bigr\}
      = O\bigl( p^{q_{\boldsymbol x, \boldsymbol x'}} \bigr)\ ,
  \end{align}
  setting by convention $q_{\boldsymbol x, \boldsymbol x'}=\infty$ whenever we
  have
  $\operatorname{tr}\bigl\{ E_{\boldsymbol x'}^\dagger E_{\boldsymbol x} \bigl( \lvert {\xi}\rangle \mkern -1.8mu\relax \langle{\psi
    }\rvert + \lvert {\psi}\rangle \mkern -1.8mu\relax \langle{\xi }\rvert \bigr) \bigr\} = 0$.  From our assumption that
  $E_{\boldsymbol x}\lvert {\psi}\rangle \neq 0$, we see that $r_{\boldsymbol x}$ is always
  finite.

  We now consider different cases for $\boldsymbol x, \boldsymbol x'$.  Suppose
  first that $\lvert {\boldsymbol x}\rvert  + \lvert {\boldsymbol x'}\rvert  < d_m$.  Then, since
  $\lvert {\psi}\rangle ,\lvert {\xi}\rangle $ form a metrological code of metrological distance $d_m$, we
  have $q_{\boldsymbol x, \boldsymbol x'} = \infty$.  Now suppose instead that
  $\lvert {\boldsymbol x}\rvert  + \lvert {\boldsymbol x'}\rvert  \geq d_m$, implying that either
  $\lvert {\boldsymbol{x}}\rvert  \geq d_m/2$ or $\lvert {\boldsymbol{x}'}\rvert  \geq d_m/2$.
  Then, since $\lVert {E_{\boldsymbol x} \lvert {\psi}\rangle }\rVert  %
  = \Omega\bigl( p^{\lvert {\boldsymbol x}\rvert /2} \bigr)$, we find
  \begin{align}
    \langle {\psi}\mkern 1.5mu\relax \vert \mkern 1.5mu\relax { E_{\boldsymbol x}^\dagger E_{\boldsymbol x} }\mkern 1.5mu\relax \vert \mkern 1.5mu\relax {\psi}\rangle 
    &= \lVert { E_{\boldsymbol x} \lvert {\psi}\rangle }\rVert ^2 = 
      \Omega\bigl(p^{\lvert {\boldsymbol x}\rvert }\bigr)\ ,
  \end{align}
  so we can pick $r_{\boldsymbol x} = \lvert {\boldsymbol x}\rvert $.  Since
  $E_{1}^{(x_i)} = O(\sqrt{p})$ for each $x_i\neq 0$, we have
  \begin{align}
    \operatorname{tr}\bigl\{ E_{\boldsymbol x'}^\dagger E_{\boldsymbol x}
      \bigl( \lvert {\xi}\rangle \mkern -1.8mu\relax \langle{\psi }\rvert + \lvert {\psi}\rangle \mkern -1.8mu\relax \langle{\xi }\rvert \bigr) \bigr\}
    &= \langle {\xi}\mkern 1.5mu\relax \vert \mkern 1.5mu\relax { E_{\boldsymbol x'}^\dagger E_{\boldsymbol x} }\mkern 1.5mu\relax \vert \mkern 1.5mu\relax {\psi}\rangle 
    +\langle {\psi}\mkern 1.5mu\relax \vert \mkern 1.5mu\relax { E_{\boldsymbol x'}^\dagger E_{\boldsymbol x} }\mkern 1.5mu\relax \vert \mkern 1.5mu\relax {\xi}\rangle 
      \nonumber\\
    &= (\sqrt{p})^{\lvert {\boldsymbol x'}\rvert  + \lvert {\boldsymbol x}\rvert }\,O(1)
      = O\bigl( p^{ (\lvert {\boldsymbol x'}\rvert  + \lvert {\boldsymbol x}\rvert )/2 } \bigr)\ ,
  \end{align}
  so we can pick $q_{\boldsymbol x, \boldsymbol x'} = %
  (\lvert {\boldsymbol x'}\rvert  + \lvert {\boldsymbol x}\rvert )/2 $.  Then
  \begin{align}
    2 q_{\boldsymbol x, \boldsymbol x'} - \min\{ r_{\boldsymbol x}, r_{\boldsymbol x'} \}
    = \lvert {\boldsymbol x}\rvert  + \lvert {\boldsymbol x'}\rvert 
    - \min\{ \lvert {\boldsymbol x}\rvert , \lvert {\boldsymbol x'}\rvert  \}
    = \max\{ \lvert {\boldsymbol x}\rvert , \lvert {\boldsymbol x'}\rvert  \} \geq d_m/2\ .
  \end{align}
  
  In all cases, we have
  $2 q_{\boldsymbol x, \boldsymbol x'} - \min\{ r_{\boldsymbol x},
    r_{\boldsymbol x'} \} \geq d_m/2$ and thus
  \begin{align}
    \FIloss{Bob}{t} \leq O\bigl(p^{d_m/2}\bigr)\ ,
  \end{align}
  as claimed.
\end{proof}

There are two strong assumptions made in the above proposition.  First, we
assume that the Kraus operator representation satisfies
$\operatorname{tr}\{ E_{\boldsymbol x'}^\dagger E_{\boldsymbol x} \,\psi \} \propto
\delta_{\boldsymbol x, \boldsymbol x'}$, or equivalently, that $\rho_E$ is
diagonal; such a representation always exists but might be difficult to find.
Second, the state on Eve must not be rank-deficient, or equivalently, there is
no Kraus operator $E_{\boldsymbol x}$ that has zero probability of occurring
when the channel is applied onto the state $\psi$.  It is not immediately clear
to us how to generalize the above proposition to weaken either of these
assumptions.

\renewcommand\selectlanguage[1]{}
\bibsep=2pt plus 1pt\relax
\catcode`\&=12\relax%
\def\ {\unskip\space\ignorespaces}%

\end{document}